\documentclass{article}

\usepackage[pdftex]{graphicx, color}
\usepackage{ascmac}
\usepackage{amsthm,amsmath,amsfonts,amssymb}
\usepackage{algorithm,algorithmic}
\usepackage{bm}
\usepackage[title]{appendix}
\usepackage[colorlinks=true, linkcolor=blue, filecolor=blue, urlcolor=blue, citecolor=red]{hyperref}
\usepackage{authblk}
\usepackage{indentfirst}
\usepackage{geometry}
\usepackage{bbm}

\newtheorem{theorem}{Theorem}[section]
\newtheorem{lemma}{Lemma}[section]
\newcommand{\norm}[1]{\left\lVert#1\right\rVert}
\allowdisplaybreaks[2]

\makeatletter
\long\def\@makecaption#1#2{%
  \normalsize
  \vskip\abovecaptionskip
  \sbox\@tempboxa{#1: #2}%
  \ifdim \wd\@tempboxa >\hsize
    #1: #2\par
  \else
    \global \@minipagefalse
    \hb@xt@\hsize{\hfil\box\@tempboxa\hfil}%
  \fi
  \vskip\belowcaptionskip}
\makeatother

\title{A Goodness-of-fit Test on the Number of Biclusters in a Relational Data Matrix}
\author[1]{Chihiro Watanabe\thanks{watanabe-chihiro763@g.ecc.u-tokyo.ac.jp}} 
\author[1,2]{Taiji Suzuki\thanks{taiji@mist.i.u-tokyo.ac.jp}}
\affil[1]{{\normalsize Graduate School of Information Science and Technology, The University of Tokyo, Tokyo, Japan}}
\affil[2]{{\normalsize Center for Advanced Intelligence Project (AIP), RIKEN, Tokyo, Japan}}
\date{}

\begin{document}
\maketitle

\begin{abstract}
Biclustering is a method for detecting homogeneous submatrices in a given observed matrix, and it is an effective tool for relational data analysis. Although there are many studies that estimate the underlying bicluster structure of a matrix, few have enabled us to determine the appropriate number of biclusters in an observed matrix. Recently, a statistical test on the number of biclusters has been proposed for a regular-grid bicluster structure, where we assume that the latent bicluster structure can be represented by row-column clustering. However, when the latent bicluster structure does not satisfy such regular-grid assumption, the previous test requires a larger number of biclusters than necessary (i.e., a finer bicluster structure than necessary) for the null hypothesis to be accepted, which is not desirable in terms of interpreting the accepted bicluster structure. In this study, we propose a new statistical test on the number of biclusters that does not require the regular-grid assumption and derive the asymptotic behavior of the proposed test statistic in both null and alternative cases. We illustrate the effectiveness of the proposed method by applying it to both synthetic and practical relational data matrices. 

\end{abstract}

\section{Introduction}

Relational data is a kind of matrix data, the entries of which reflect some kind of relationship between two (generally different) objects. For example, the rows and columns, respectively, of an observed matrix $A \in \mathbb{R}^{n \times p}$ represent customers and products, and each entry $A_{ij}$ is a number of times for which the $i$th customer purchased the $j$th product. It has been shown that we can successfully model various kinds of relational data matrices, including customer-item transaction/rating data \cite{Shan2008,Symeonidis2007}, document-word co-occurrence data \cite{Dhillon2001,Franca2012}, and gene expression data \cite{Madeira2004,Oghabian2014,Prelic2006,Tanay2002}, by assuming the existence of a latent \textit{bicluster} or ``homogeneous'' submatrix (e.g., the entries in each bicluster are identically distributed). In the example above, concerning the customer-item relational data matrix, this assumption corresponds to that there are some groups of customers $I \subseteq \{ 1, \dots, n \}$ and some groups of items $J \subseteq \{ 1, \dots, p \}$, and that the customers in $I$ tend to purchase the items in $J$ at a similar frequency. 

Regarding the bicluster structure of a relational data matrix, the following two problems have been extensively studied in the literature: \textit{submatrix detection} and \textit{localization}. Submatrix detection serves to detect the existence of such biclusters in a given observed matrix $A$ (i.e., whether or not matrix $A$ contains at least one bicluster) \cite{Butucea2013,Hartigan1972,Ma2015,Shabalin2009}. In this paper, as in a number of previous studies \cite{Cai2017,Chen2016}, we distinguish such a task from submatrix localization (which is also known as \textit{biclustering}), the purpose of which is to recover the exact position of such biclusters. 
So far, many biclustering methods have been proposed for a fixed number of biclusters $K$ \cite{Cai2017,Chen2016,Hajek2018,Hochreiter2010,Shabalin2009}. In most practical cases, however, there would not be any prior knowledge about $K$ in a given data matrix. Therefore, it is an important task to develop some method to appropriately determine $K$ from the observed data $A$. In the next two paragraphs, we outline some related studies that propose methods for choosing $K$. 

\paragraph{A statistical test on the number of biclusters $K$.} Although many studies have tested whether an observed matrix $A$ contains any large average submatrix \cite{Brennan2019,Butucea2013,Cai2020,Liu2018,Ma2015}, few statistical test methods have been proposed for ascertaining the number of biclusters $K$ in a given matrix $A$. 
Recently, statistical tests on $K$ have been proposed in \cite{Bickel2016,Hu2020,Lei2016,Watanabe2021} with the constraint that the underlying bicluster structure should be represented by a regular grid (as shown in Figure \ref{fig:submatrix} (b-2))\footnote{Particularly, in \cite{Bickel2016,Hu2020,Lei2016}, the observed matrix $A$ (and thus its bicluster structure) is assumed to be square symmetric.}. 
However, if the latent bicluster structure does not satisfy the regular grid constraint (as shown in Figure \ref{fig:submatrix} (b-1)), such a test needs a larger hypothetical number of biclusters $K_0$ than necessary (i.e., a finer bicluster structure than necessary) to accept the null hypothesis $K = K_0$, which is not desirable from the perspective of interpreting the accepted bicluster structure. To cope with such a problem, a more flexible model is required, one which can represent the existence of local biclusters \cite{Shabalin2009}. 
For a singular-value-decomposition-based biclustering, a stopping criterion has been proposed for detecting multiple biclusters (which determines $K$) based on stability selection in \cite{Sill2011}. This method has made it possible to detect a bicluster structure with Type I error control and without the regular grid constraint. However, unlike the method we propose in this paper, its Type I error has been guaranteed only in terms of an upper bound, not the null distribution of a test statistic. Therefore, in this previous study \cite{Sill2011}, no means is provided to perform a statistical test on the number of biclusters. Moreover, this method has no theoretical guarantee for the alternative cases (i.e., statistical power). 
In this study, we address these problems by developing a new statistical test on $K$, which does not require the regular grid constraint and whose test statistic $T$ shows a good property in an alternative case (i.e., with high probability, $T$ asymptotically increases with the matrix size and thus the Type II error converges in probability to zero\footnote{It must be noted that Theorem \ref{th:unrealizable} shows the asymptotic behavior of test statistic $T$ under the assumption that an observed matrix has some latent bicluster structure. To derive its behavior in the case that an observed matrix cannot be represented by such a model is beyond the scope of this paper.}), as shown later in Theorem \ref{th:unrealizable}. 

\paragraph{An information criterion on the number of biclusters $K$.} Aside from the statistical-test-based methods, some studies have proposed that $K$ can be determined based on the minimum description length \cite{Sakai2013,Tepper2016,Yamanishi2019} and modified DIC for the biclustering problem \cite{Chekouo2015a,Chekouo2015b}. 
Particularly, under the regular grid constraint of the bicluster structure, an information criterion called integrated completed likelihood (ICL) has been proposed for determining $K$ \cite{Corneli2015,Lomet2012,Wyse2017}, which approximates the maximum marginal likelihood of a given $K$. These methods aim to select the optimal number of biclusters $K$ from a given set of candidates, in terms of some criterion (e.g., marginal likelihood). This purpose is different from that of a statistical test, which aims to judge whether we accept a hypothetical number of biclusters $K_0$ with a specific significance level given by a user. 

\paragraph{Other approaches for determining the number of biclusters $K$.} Aside from the above information-criterion-based and statistical-test-based methods, some studies have proposed the construction of a generative model of the bicluster structure including the number of biclusters $K$ and the subsequent selection of an optimal model in terms of some measure (e.g., choosing a MAP estimator) \cite{Moran2019,Raff2020}. 
There have also been some heuristic criteria for determining the number of biclusters in an observed matrix, which have been proposed as stopping rules for top-down division-based biclustering algorithms \cite{Duffy1991,Hartigan1972,Tibshirani1999} or bottom-up merging-based one \cite{Pio2013}. 

In this study, we propose a statistical test method for the number of biclusters $K$ in a given observed matrix $A$, without the regular grid constraint of the latent bicluster structure. Specifically, we consider the notions of \textit{disjointness} and \textit{bi-disjointness}. A bicluster structure can be called a \textit{disjoint} structure iff each entry belongs to at most one bicluster (as shown in Figure \ref{fig:submatrix} (a) and (b)). \textit{Bi-disjointness} has a stricter condition: we call a bicluster structure a \textit{bi-disjoint} structure iff each row or column belongs to at most one bicluster (as shown in Figure \ref{fig:submatrix} (a)). We develop for the first time a statistical test on $K$ under the assumptions that the underlying bicluster structure is disjoint (but not necessarily bi-disjoint) and that the submatrix localization algorithm is consistent. 

\begin{figure}[t]
  \centering
  \includegraphics[width=\hsize]{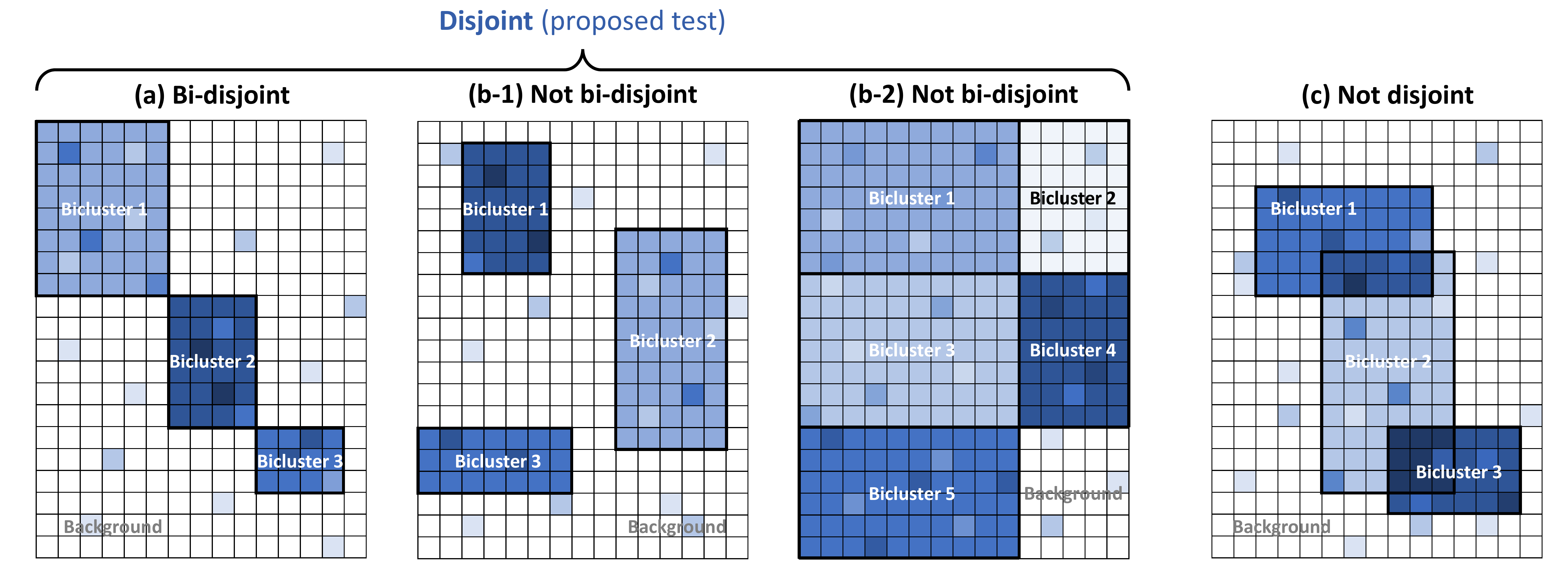}
  \caption{(a) Bi-disjoint, (b) disjoint but not bi-disjoint, and (c) not disjoint bicluster structures. In the proposed method, we assume that the underlying bicluster structure is disjoint, but not necessarily bi-disjoint. We assume that the observed matrix consists of one or multiple biclusters, in each of which the entries are generated in the i.i.d. sense. Note that in the cases of (b-1) and (c), it is not always possible to make all the rows and columns within the biclusters contiguous by sorting the rows and columns.}
  \label{fig:submatrix}
\end{figure}

To guarantee the asymptotic behavior of the proposed test statistic in the null case (i.e., $K = K_0$), which is given in Theorem \ref{th:realizable}, we use the properties of a random matrix with a sub-exponential decay \cite{Bloemendal2016,Pillai2014}. Moreover, we derive its behavior in the alternative case (i.e., $K > K_0$) such that it increases with the matrix size in high probability, as given in Theorem \ref{th:unrealizable}. Unlike a previous study \cite{Watanabe2021}, wherein the number of biclusters $K$ is assumed to be a fixed constant that does not depend on the matrix size $m$, we consider a case in which $K$ might increase with $m$ (the precise description of this assumption is given in (\ref{eq:KH_condition_KH_sum})). Additionally, since we consider more general bicluster structures (i.e., without the regular-grid assumption) than in the previous study \cite{Watanabe2021}, we use a different approach to complete the proof in the alternative case. Based on these results, in Sect.~\ref{sec:method}, we explain a method for estimating $K$ from an observed matrix $A$, by sequentially testing the hypothetical numbers of biclusters in an ascending order (i.e., $K_0 = 0, 1, 2, \dots$) until the null hypothesis is accepted. 

This paper is organized as follows: in Sect.~\ref{sec:method}, we describe the problem settings and the model of the underlying bicluster structure in an observed matrix. Next, in Sect.~\ref{sec:statistic}, we propose a statistical test on the number of biclusters $K$ in an observed matrix, and derive its theoretical guarantee in both null and alternative cases. In Sect.~\ref{sec:experiments}, we provide some experimental results that demonstrate the effectiveness of the proposed test. Finally, in Sect.~\ref{sec:discussion}, we discuss the obtained results and limitations of the proposed method, and conclude this paper in Sect.~\ref{sec:conclusion}. 


\section{Problem setting and statistical model for goodness-of-fit test for submatrix detection problem}
\label{sec:method}

Let $A \in \mathbb{R}^{n \times p}$ be an $n \times p$ observed matrix. Given such an observed matrix $A$, the goal of submatrix detection problem is to determine whether it contains one or multiple disjoint submatrices, say \textit{biclusters}, in each of which the entries are generated in the i.i.d. sense (Figure \ref{fig:submatrix} (a) and (b)). As in the previous studies \cite{Cai2017,Liu2018}, we distinguish the submatrix detection from \textit{localization} problems in that the goal of the latter is not only to detect the existence of biclusters in an observed matrix, but to estimate their precise locations. 

Let $K$ be the minimum number of such biclusters to represent the matrix $A$, which is unknown beforehand. Aside from the $K$ biclusters, we assume that there are ``background'' entries in matrix $A$ that do not belong to any bicluster. Note that the difference between a bicluster and the background is in that the former can be represented as a submatrix (i.e., $\{ (i, j): i \in I_k, j \in J_k \}$ for some sets $I_k \subset \{ 1, \dots, n \}$ and $J_k \subset \{ 1, \dots, p \}$), while the latter does not necessarily have such a submatrix structure, as shown in Figure \ref{fig:submatrix}\footnote{It must be noted that multiple bicluster assignments may exist that represent an equivalent bicluster structure. For instance, in a regular-grid bicluster structure (as shown in Figure \ref{fig:submatrix} (b-2)), any block can be defined as the background. In such cases, the consistency condition that we give later in \ref{asmp:consistency} requires that the probability converges to one with increasing matrix size that an estimated bicluster assignment is included in the set of correct bicluster assignments.}. In the proposed test, we do \textbf{not} assume the underlying bicluster structure to be \textit{bi-disjoint}, that is, each row or column is assigned to at most one bicluster (Figure \ref{fig:submatrix} (a)). 

We denote the bicluster index of the $(i, j)$th entry of matrix $A$ as $g_{ij} \in \{ 0, 1, \dots, K \}$, where $g_{ij} = k$ if the $(i, j)$th entry belongs to the $k$th bicluster for some $k \in \{ 1, \dots, K \}$ and $g_{ij} = 0$ if it belongs to the background. We define the set of group indices of all the entries as $g \equiv (g_{ij})_{1\leq i \leq n, 1 \leq j \leq p}$. We also define that $\mathcal{I}_k \equiv \{ (i, j): g_{ij} = k \}$, which represents the set of entries in the $k$th group. Specifically, we consider the following model: 
\begin{align}
\label{eq:SD}
&P = (P_{ij})_{1\leq i \leq n, 1 \leq j \leq p}, \ \ \ \ \ 
P_{ij} = b_{g_{ij}}.  \nonumber \\
&\sigma = (\sigma_{ij})_{1\leq i \leq n, 1 \leq j \leq p}, \ \ \ \ \ 
\sigma_{ij} = s_{g_{ij}}.  \nonumber \\
&A = (A_{ij})_{1\leq i \leq n, 1 \leq j \leq p}, \ \ \ \ \ 
\mathbb{E} [A_{ij}] = P_{ij}, \ \ \ \ \ 
\mathbb{E} [(A_{ij} - P_{ij})^2] = \sigma_{ij}^2, 
\end{align}
where $b_k$ and $s_k > 0$, respectively, are the mean and standard deviation of the $k$th bicluster ($k = 1, \dots, K$) or background ($k = 0$). 
This model is a generalized version of well-studied submatrix detection models, in which we assume that the mean of the background noise is zero (i.e., $b_0 = 0$) \cite{Butucea2013,Ma2015,Shabalin2009}. 
Let $Z \in \mathbb{R}^{n \times p}$ be a standardized noise matrix, which is given by
\begin{align}
\label{eq:Z_true}
Z = (Z_{ij})_{1\leq i \leq n, 1 \leq j \leq p}, \ \ \ \ \ \ \ \ \ \ 
Z_{ij} = \frac{A_{ij} - P_{ij}}{\sigma_{ij}}. 
\end{align}

In most cases, the number of biclusters $K$ is unknown in advance. This study aims to develop a statistical test on $K$, which is based on the following null (N) and alternative (A) hypotheses: 
\begin{align}
\label{eq:hypothesis}
\mathrm{(N):}\ K = K_0, \ \ \ \ \ \ \ \ \ \ 
\mathrm{(A):}\ K > K_0, 
\end{align}
where $K_0$ is a given hypothetical number of biclusters. In this study, we only consider the cases where $K_0 \leq K$. If $K = K_0$ (i.e., the null case), then we call it a \textit{realizable} case. Otherwise (i.e., in the alternative case), we call it an \textit{unrealizable} case. To select the number of biclusters for a given observed matrix $A$, we propose the sequential testing of the bicluster numbers $K_0 = 0, 1, 2, \dots$ until the null hypothesis (N) is accepted. Let $\hat{K}$ be the hypothetical number of biclusters when (N) is accepted. The proposed method outputs $\hat{K}$ as the selected number of biclusters in matrix $A$. 

\paragraph{Notations.}
Throughout this paper, we use the following notations: 
\begin{align}
&X = O_p \left[ f(m) \right]. \nonumber \\
&\Leftrightarrow \ 
\forall \epsilon>0, \ \exists C>0, M>0, \ \forall m \geq M, \ \mathrm{Pr} \left[C f(m) \geq X \right] \geq 1-\epsilon. \nonumber \\
&X = \Omega_p \left[ f(m) \right]. \nonumber \\
&\Leftrightarrow \ 
\forall \epsilon>0, \ \exists C>0, M>0, \ \forall m \geq M, \ \mathrm{Pr} \left[C f(m) \leq X \right] \geq 1-\epsilon. \nonumber \\
&X = \Theta_p \left[ f(m) \right]. \nonumber \\
&\Leftrightarrow \ 
\forall \epsilon>0, \ \exists C_1, C_2>0, M>0, \ \forall m \geq M, \nonumber \\
&\ \ \ \ \ \ \mathrm{Pr} \left[ C_1 f(m) \leq X \leq C_2 f(m) \right] \geq 1-\epsilon. \\
&\| A \|_{\mathrm{op}} = \sup_{\bm{u} \in \mathbb{R}^p} \frac{\| A \bm{u} \|}{\| \bm{u} \|}, \ \ \ \ \ 
\| A \|_{\mathrm{F}} = \sqrt{\sum_{i = 1}^n \sum_{j = 1}^p A_{ij}^2}. 
\end{align}

In the proofs in Sect.~\ref{sec:statistic}, we use the following \textbf{sample} mean matrix $\tilde{P}$ and standard deviation matrix $\tilde{\sigma}$ for the \textbf{correct} block structure, and matrix $\tilde{Z}$: 
\begin{align}
\label{eq:tilde_Z}
&\tilde{P} = (\tilde{P}_{ij})_{1\leq i \leq n, 1 \leq j \leq p}, \ \ \ \ \ \ \ \ \ \ \tilde{P}_{ij} = \tilde{b}_{g_{ij}}, \nonumber \\
&\tilde{\sigma} = (\tilde{\sigma}_{ij})_{1\leq i \leq n, 1 \leq j \leq p}, \ \ \ \ \ \ \ \ \ \ \tilde{\sigma}_{ij} = \tilde{s}_{g_{ij}}, \nonumber \\
&\tilde{Z} = (\tilde{Z}_{ij})_{1\leq i \leq n, 1 \leq j \leq p}, \ \ \ \ \ \ \ \ \ \ \tilde{Z}_{ij} = \frac{A_{ij} - \tilde{P}_{ij}}{\tilde{\sigma}_{ij}}. 
\end{align}
where $\tilde{b}_k$ and $\tilde{s}_k$, respectively, are the \textbf{sample} mean and standard deviation of the entries in the $k$th \textbf{null} group in observed matrix $A$. Let $\tilde{\lambda}_1$ and $\tilde{\bm{v}}_1$, respectively, be the maximum eigenvalue of matrix $\tilde{Z}^{\top} \tilde{Z}$ and the corresponding eigenvector whose Euclid norm is constrained to be one: 
\begin{align}
\label{eq:tilde_v_eigenvec}
\tilde{Z}^{\top} \tilde{Z} \tilde{\bm{v}}_1 = \tilde{\lambda}_1 \tilde{\bm{v}}_1, \ \ \ \| \tilde{\bm{v}}_1 \| = 1. 
\end{align}
Similarly, we denote the eigenvalues (in descending order) and the corresponding normalized eigenvectors of matrix $Z^{\top} Z$ as $\{ \lambda_j \}$ and $\{ \bm{v}_j \}$, respectively: 
\begin{align}
\label{eq:v_eigenvec}
&Z^{\top} Z \bm{v}_j = \lambda_j \bm{v}_j, \ \ \ \ \ \| \bm{v}_j \| = 1, \ \ \ \ \ j = 1, \dots, p. \nonumber \\
&\lambda_1 \geq \lambda_2 \geq \dots \geq \lambda_p. 
\end{align}

\paragraph{Assumptions.}
In Sect.~\ref{sec:statistic}, we develop a statistical test for the number of biclusters in an observed matrix based on the following assumptions: 
\begin{enumerate}
\renewcommand{\theenumi}{(\roman{enumi})}
\item We assume that the biclusters (and the background) are \textit{disjoint}\footnote{Note that in this paper, we use the term ``disjoint'' to indicate that each \textit{entry} belongs to a single group, not that the bicluster structure is bi-disjoint as in Figure \ref{fig:submatrix} (a) (i.e., each \textit{row} or \textit{column} is assigned to at most one bicluster).}, that is, there is no entry $(i, j)$ that is assigned to two or more biclusters (i.e., $\mathcal{I}_k \cap \mathcal{I}_{k'} = \emptyset$ if $k \neq k'$, $k, k' \in \{ 0, 1, \dots, K \}$). Moreover, we assume that each $(i, j)$th entry belongs to exactly one group. 
\item We assume that each entry $Z_{ij}$ of the standardized noise matrix has a sub-exponential decay (i.e., there exists some $\vartheta >0$ such that for $x>1$, $\mathrm{Pr} \left( \left| Z_{ij} \right| > x \right) \leq \vartheta^{-1}\exp (-x^{\vartheta})$). Furthermore, we assume the following conditions: 
\begin{itemize}
\item $\max_{k = 0, 1, \dots, K} s_k = O(1)$, and $\min_{k = 0, 1, \dots, K} s_k = \Omega (1)$. 
\item $\max_{k=1, \dots, K, k'=1, \dots, K} \left| b_k - b_{k'} \right| = O(1)$, and $\max_{k = 0, 1, \dots, K} |b_k| = O(K)$. 
We also assume that the minimum difference between a pair of different biclusters (including background) is lower bounded by some constant that does not depend on the matrix size: $\min_{k \neq k'} |b_k - b_{k'}| \geq C^{\bm{b}} > 0$. 
\item $C^{\mathrm{M}(4)} \equiv \max_{i=1, \dots, n, j=1, \dots, p} \mathbb{E} [Z_{ij}^4] = O(1)$. 
\end{itemize}
\label{asmp:Z_exp_S}
\item We assume that both the row and column sizes $n$ and $p$ of the observed matrix increase in proportion to an integer $m$ (i.e., $n, p \propto m$) and we consider an asymptotics of $m \to \infty$. 
\label{asmp:np_size}
\item As for the background, we assume that it can be divided into $H$ disjoint submatrices. Regarding the row and column sizes of each $k$th submatrix, ($|I_k|$ and $|J_k|$, respectively), we assume that they monotonically increase with $m$, where $k = 1, \dots, K+H$. 
\begin{itemize}
\item In Theorem \ref{th:realizable} for a \textbf{realizable} case, we assume that the minimum number of biclusters $K$ and that of background submatrices $H$ to represent the observed matrix $A$ satisfy the following conditions: 
  \begin{align}
  \label{eq:KH_condition_KH_sum}
  &K+H = O \left( m^{\frac{1}{42} - \epsilon_1} \right),\ \mathrm{for\ some}\ \epsilon_1 > 0. \\
  &n_{\mathrm{min}} \equiv \min_{k = 1, \dots, K+H} |I_k| = \Omega \left( m^{\frac{8}{21}} \right), \nonumber \\
  &p_{\mathrm{min}} \equiv \min_{k = 1, \dots, K+H} |J_k| = \Omega \left( m^{\frac{8}{21}} \right). 
  \end{align}
  Note that from these conditions, for some $\epsilon_1 > 0$, 
  \begin{align}
  \label{eq:KH_condition_KHminIk}
  &(K+H) \left( \min_{k = 1, \dots, K+H} |\mathcal{I}_k| \right)^{-\frac{1}{4}} 
  \leq \frac{K+H}{n_{\mathrm{min}}^{\frac{1}{4}} p_{\mathrm{min}}^{\frac{1}{4}}}
  = O \left( m^{-\frac{1}{6} - \epsilon_1} \right). 
  \end{align}
\item In Theorem \ref{th:unrealizable} for an \textbf{unrealizable} case, we assume the following stricter condition: 
  \begin{align}
  \label{eq:n_min_p_min_un}
  \frac{K+H}{\sqrt{n_{\mathrm{min}} p_{\mathrm{min}}}} = O \left( m^{-\frac{3}{4} - \epsilon_2} \right),\ \mathrm{for\ some}\ \epsilon_2 > 0. 
  \end{align}
\end{itemize}
\label{asmp:block_size}
\item In the realizable case, we assume that a submatrix localization algorithm for estimating the bicluster structure $g$ is \textit{consistent}, that is, $\mathrm{Pr} (\hat{g} = g) \to 1$ in the limit of $m \to \infty$, where $g$ and $\hat{g}$, respectively, are the null and estimated bicluster structures of the observed matrix $A$ (the precise definition of $\hat{g}$ is given in Sect.~\ref{sec:statistic})\footnote{Such consistency condition in submatrix localization has been considered by several previous studies, although most of them have assumed the existence of at most one bicluster in a given observed matrix \cite{Balakrishnan2011,Brennan2018,Butucea2015,Hajek2017,Hajek2018,Kolar2011,Luo2020}. Multiple biclusters may be localized by applying these methods to a given observed matrix multiple times; however, there is no guarantee for the consistency of such a heuristic approach. For bi-disjoint bicluster structures, some submatrix localization algorithms based on the maximum likelihood estimator for a known model parameter \cite{Chen2016} and singular value decomposition \cite{Cai2017} have been shown to be consistent. As another method, by imposing the regular grid constraint to the underlying bicluster structure, we can consider a special case of non-bi-disjoint bicluster structures, which can be represented as a result of row-column clustering (as shown in Figure \ref{fig:submatrix} (b-2)). As for the biclustering problem with such a regular grid structure, Flynn and Perry \cite{Flynn2020} have proposed a consistent algorithm based on the criterion of (generalized) profile likelihood. However, these algorithms cannot be directly applied to our case, where the localization problem cannot be formulated as row-column (hard) clustering \cite{Shabalin2009}. Although the proposed test itself can be applied without the bi-disjoint assumption, currently there is no way to consistently estimate non-bi-disjoint bicluster structures. Instead, we propose a heuristic submatrix localization algorithm in Appendix \ref{sec:consistency_approx} and use it the experiments. To develop a consistent submatrix localization algorithm that can be applied without the bi-disjoint assumption is beyond the scope of this paper.}. 
\label{asmp:consistency}
\end{enumerate}


\section{A test statistic for determining the number of biclusters}
\label{sec:statistic}

We develop the test statistic $T$ of the proposed test based on the estimated version of the standardized noise matrix $Z$ in (\ref{eq:Z_true}), given a hypothetical number of biclusters $K_0$. 
We denote the \textbf{estimated} group index of the $(i, j)$th entry of matrix $A$ as $\hat{g}_{ij} \in \{ 0, 1, \dots, K_0 \}$, where $\hat{g}_{ij} = k$ if the $(i, j)$th entry is estimated to be a member of the $k$th bicluster for some $k$ and $\hat{g}_{ij} = 0$ otherwise (i.e., the $(i, j)$th entry is estimated to be a member of background). 
We define the set of \textbf{estimated} group indices of all the entries as $\hat{g} \equiv (\hat{g}_{ij})_{1\leq i \leq n, 1 \leq j \leq p}$. We also define that $\hat{\mathcal{I}}_k \equiv \{ (i, j): \hat{g}_{ij} = k \}$, which represents the \textbf{estimated} set of entries in the $k$th group. 

Based on the above notations, the \textbf{estimated} mean, standard deviation and noise matrices $\hat{P}$, $\hat{\sigma}$ and $\hat{Z}$ are given by
\begin{align}
\label{eq:BPS_hat}
&\hat{\bm{b}} = (\hat{b}_k)_{1 \leq k \leq K_0}, \ \ \ \ \ \ \ \ \ \ 
\hat{b}_k = \frac{1}{|\hat{\mathcal{I}}_k|} \sum_{(i, j) \in \hat{\mathcal{I}}_k} A_{ij}, \nonumber \\
&\hat{P} = (\hat{P}_{ij})_{1\leq i \leq n, 1 \leq j \leq p}, \ \ \ \ \ \ \ \ \ \ 
\hat{P}_{ij} = \hat{b}_{\hat{g}_{ij}}, \nonumber \\
&\hat{\bm{s}} = (\hat{s}_k)_{1 \leq k \leq K_0}, \ \ \ \ \ \ \ \ \ \ 
\hat{s}_k = \sqrt{ \frac{1}{|\hat{\mathcal{I}}_k|} \sum_{(i, j) \in \hat{\mathcal{I}}_k} \left( A_{ij} - \hat{P}_{ij} \right)^2}, \nonumber \\
&\hat{\sigma} = (\hat{\sigma}_{ij})_{1\leq i \leq n, 1 \leq j \leq p}, \ \ \ \ \ \ \ \ \ \ 
\hat{\sigma}_{ij} = \hat{s}_{\hat{g}_{ij}}. \\
\label{eq:Z_hat}
&\hat{Z} = (\hat{Z}_{ij})_{1\leq i \leq n, 1 \leq j \leq p}, \ \ \ \ \ \ \ \ \ \ 
\hat{Z}_{ij} = \frac{A_{ij} - \hat{P}_{ij}}{\hat{\sigma}_{ij}}. 
\end{align}

To construct a statistical test on the number of biclusters $K$, we use the following result from \cite{Pillai2014}, which shows that the scaled maximum eigenvalue $T^*$ of sample covariance matrix $Z^\top Z$ converges in law to the Tracy-Widom distribution with index $1$ ($TW_1$) in the limit of $m \to \infty$: 
\begin{align}
\label{eq:T_true}
T^* = \frac{\lambda_1 - a^{\mathrm{TW}}}{b^{\mathrm{TW}}}, \ \ \ \ \ 
T^* \rightsquigarrow TW_1 \ \mathrm{(Convergence\ in\ law)}, 
\end{align}
where $\lambda_1$ is the maximum eigenvalue of matrix $Z^\top Z$ and 
\begin{align}
\label{eq:ab}
a^{\mathrm{TW}} = (\sqrt{n} + \sqrt{p})^2, \ \ \ \ \ 
b^{\mathrm{TW}} = (\sqrt{n} + \sqrt{p}) \left( \frac{1}{\sqrt{n}} + \frac{1}{\sqrt{p}} \right)^{\frac{1}{3}}. 
\end{align}

Based on the above fact, we define the test statistic $T$, which is an estimator of $T^*$ in (\ref{eq:T_true}), from the maximum eigenvalue $\hat{\lambda}_1$ of matrix $\hat{Z}^\top \hat{Z}$: 
\begin{align}
\label{eq:T_statistic}
T = \frac{\hat{\lambda}_1 - a^{\mathrm{TW}}}{b^{\mathrm{TW}}}. 
\end{align}
By using the test statistic $T$, we define the rule of the proposed test at the significance level of $\alpha$ as follows: 
\begin{align}
\label{eq:rejection}
\mathrm{Reject\ null\ hypothesis}\ (K = K_0),\ \ \ \ \ \mathrm{if}\ T \geq t(\alpha), 
\end{align}
where $t(\alpha)$ is the $\alpha$ upper quantile of the $TW_1$ distribution. We give the theoretical guarantees for the above test in both the null and alternative cases later in Theorems \ref{th:realizable} and \ref{th:unrealizable}, respectively. 

\begin{theorem}[Realizable case]
Under the assumptions in Sect.~\ref{sec:method}, if $K = K_0$, 
\begin{align}
\label{eq:T_to_TW1}
T \rightsquigarrow TW_1 \ \mathrm{(Convergence\ in\ law)},
\end{align}
in the limit of $m \to \infty$, where $T$ is defined as in (\ref{eq:T_statistic}). 
\label{th:realizable}
\end{theorem}
\begin{proof}
To apply the result reported in \cite{Pillai2014}, we consider the difference between $T^*$ and $T$. By definitions of (\ref{eq:T_true}) and (\ref{eq:T_statistic}), we have
\begin{align}
\label{eq:diff_T_Tstar}
|T - T^*| = \frac{| \hat{\lambda}_1 - \lambda_1 |}{b^{\mathrm{TW}}}. 
\end{align}
Next, we prove that the right side of (\ref{eq:diff_T_Tstar}) can be bounded by $\frac{| \lambda_1 - \hat{\lambda}_1 |}{b^{\mathrm{TW}}} = O_p \left( m^{-\epsilon} \right)$ for some $\epsilon>0$, which is given in Lemma \ref{lm:zop2_hat} later. If this bound holds, from Slutsky's theorem, (\ref{eq:T_to_TW1}) also holds. To show Lemma \ref{lm:zop2_hat}, we first state the following Lemmas \ref{lm:zop2_eq1} and \ref{lm:zop2_eq2}, which give the lower and upper bounds for the maximum eigenvalue $\tilde{\lambda}_1$ of matrix $\tilde{Z}^{\top} \tilde{Z}$. The proofs of the following lemmas are mainly based on those given in \cite{Watanabe2021}. The main differences between them are as follows: first, we assume a regular-grid bicluster structure in the previous study, whereas we consider a more generalized disjoint one. Second, unlike the previous study, where we assume that the null number of biclusters $K$ is a fixed constant that does not depend on the matrix size $m$, we consider a case in which $K$ might increase with $m$. 

\begin{lemma}
\label{lm:zop2_eq1}
Under the assumptions noted in Sect.~\ref{sec:method}, if $K = K_0$, 
\begin{align}
\lambda_1 \leq \tilde{\lambda}_1 + O_p \left( m^{\frac{1}{3} - \epsilon} \right), \ \ \ \mathrm{for\ some}\ \epsilon > 0. 
\end{align}
\end{lemma}
\begin{proof}
A proof is given in Appendix \ref{sec:lmd_1_upper}. 
\end{proof}

\begin{lemma}
\label{lm:zop2_eq2}
Under the assumptions in Sect.~\ref{sec:method}, if $K = K_0$, 
\begin{align}
\label{eq:lemma_zop2_1_4}
\tilde{\lambda}_1 \leq \lambda_1 + O_p \left( m^{\frac{1}{3} - \epsilon} \right), \ \ \ \mathrm{for\ some}\ \epsilon > 0. 
\end{align}
\end{lemma}
\begin{proof}
A proof is given in Appendix \ref{sec:lmd_1_lower}. 
\end{proof}

\begin{lemma}
\label{lm:zop2_hat}
Under the assumptions in Sect.~\ref{sec:method}, if $K = K_0$, 
\begin{align}
\label{eq:lmd_hatlmd_diff}
\frac{| \lambda_1 - \hat{\lambda}_1 |}{b^{\mathrm{TW}}} = O_p \left( m^{-\epsilon} \right), \ \ \ \mathrm{for\ some}\ \epsilon>0. 
\end{align}
\end{lemma}
\begin{proof}
By combining Lemmas \ref{lm:zop2_eq1}, \ref{lm:zop2_eq2}, and the definition of $b^{\mathrm{TW}}$ in (\ref{eq:ab}), we have
\begin{align}
\label{eq:lmd_tldlmd_diff}
\frac{| \lambda_1 - \tilde{\lambda}_1 |}{b^{\mathrm{TW}}} = O_p \left( m^{-\epsilon} \right), \ \ \ \mathrm{for\ some}\ \epsilon>0. 
\end{align}
We consider the following three events: 
\begin{itemize}
\item $\mathcal{E}^{(1)}_m$ represents the event that $\tilde{Z} = \hat{Z}$ holds. 
\item $\mathcal{E}^{(2)}_m$ represents the event that the solution given by the submatrix localization algorithm is correct (i.e., $\hat{g} = g$). 
\item $\mathcal{E}^{(3)}_{m, C}$ represents the event that $\frac{| \lambda_1 - \tilde{\lambda}_1 |}{b^{\mathrm{TW}}} \leq Cm^{-\epsilon}$ holds. 
\end{itemize}
The joint probability of events $\mathcal{E}^{(1)}_m$ and $\mathcal{E}^{(3)}_{m, C}$ can be lower bounded by
\begin{align}
\label{eq:event_ineq_tilde}
\mathrm{Pr} \left( \mathcal{E}^{(1)}_m \cap \mathcal{E}^{(3)}_{m, C} \right) 
&\geq \mathrm{Pr} \left( \mathcal{E}^{(2)}_m \cap \mathcal{E}^{(3)}_{m, C} \right) \nonumber \\
&\geq 1 - \mathrm{Pr} \left[ \left( \mathcal{E}^{(2)}_m \right)^{\mathrm{C}} \right] - \mathrm{Pr} \left[ \left( \mathcal{E}^{(3)}_{m, C} \right)^{\mathrm{C}} \right], 
\end{align}
where $\mathcal{E}^{\mathrm{C}}$ is the complement of event $\mathcal{E}$. From the consistency assumption \ref{asmp:consistency} in Sect.~\ref{sec:method}, if $K = K_0$, the second term on the right side of (\ref{eq:event_ineq_tilde}) satisfies that $\mathrm{Pr} \left[ \left( \mathcal{E}^{(2)}_m \right)^{\mathrm{C}} \right] \to 0$ in the limit of $m \to \infty$. As for the third term $\mathrm{Pr} \left[ \left( \mathcal{E}^{(3)}_{m, C} \right)^{\mathrm{C}} \right]$, we already have (\ref{eq:lmd_tldlmd_diff}). By combining these facts, we have
\begin{align}
\forall \tilde{\epsilon}>0, \ \exists C>0, M>0, \ \forall m \geq M, \ 
\mathrm{Pr} \left( \mathcal{E}^{(1)}_m \cap \mathcal{E}^{(3)}_{m, C} \right) \geq 1 - \tilde{\epsilon}, 
\end{align}
which results in (\ref{eq:lmd_hatlmd_diff}). 
\end{proof}
From Lemma \ref{lm:zop2_hat}, we finally obtain
\begin{align}
|T - T^*| = \frac{| \hat{\lambda}_1 - \lambda_1 |}{b^{\mathrm{TW}}} = O_p \left( m^{-\epsilon} \right), \ \mathrm{for\ some}\ \epsilon>0. 
\end{align}
By combining this fact with Slutsky's theorem, the convergence of the test statistic $T$ in law to $TW_1$ distribution in (\ref{eq:T_to_TW1}) holds. 
\end{proof} 

\begin{theorem}[Unrealizable case]
Under the assumptions in Sect.~\ref{sec:method}, if $K > K_0$, 
\begin{align}
\label{eq:T_un_up}
T = O_p \left( m^{\frac{5}{3}} \right), 
\end{align}
and 
\begin{align}
\label{eq:T_un_low}
T = \Omega_p \left( m^{\frac{2}{3}} \right), 
\end{align}
where $T$ is defined as in (\ref{eq:T_statistic}). 
\label{th:unrealizable}
\end{theorem}
\begin{proof}
We first prove the upper bound in (\ref{eq:T_un_up}). Let $\underline{X}^{\mathrm{E} (k)}$ be an $n \times p$ matrix whose entries in the $k$th \textbf{estimated} bicluster (including background) are the same as matrix $X$ and all the other entries are zero. Since the Frobenius norm upper bounds the operator norm, 
\begin{align}
\label{eq:unre_Z_op_up}
&\| \hat{Z} \|_{\mathrm{op}} \leq \| \hat{Z} \|_{\mathrm{F}} 
= \sqrt{\sum_{k = 0}^{K_0} \| \underline{\hat{Z}}^{\mathrm{E} (k)} \|_{\mathrm{F}}^2} 
= \sqrt{\sum_{k = 0}^{K_0} \frac{1}{\hat{s}_k^2} \| \underline{A}^{\mathrm{E} (k)} - \underline{\hat{P}}^{\mathrm{E} (k)} \|_{\mathrm{F}}^2} \nonumber \\
&= \sqrt{\sum_{k = 0}^{K_0} \frac{|\hat{\mathcal{I}}_k|}{\| \underline{A}^{\mathrm{E} (k)} - \underline{\hat{P}}^{\mathrm{E} (k)} \|_{\mathrm{F}}^2} \| \underline{A}^{\mathrm{E} (k)} - \underline{\hat{P}}^{\mathrm{E} (k)} \|_{\mathrm{F}}^2} 
= \sqrt{\sum_{k = 0}^{K_0} |\hat{\mathcal{I}}_k|} 
= \sqrt{np}. 
\end{align}

From the assumption \ref{asmp:np_size}, let $n = C_n m$ and $p = C_p m$, where $C_n$ and $C_p$ are positive constants. 
According to the definition in (\ref{eq:ab}), we have
\begin{align}
\label{eq:ab_m}
a^{\mathrm{TW}} = (\sqrt{C_n} + \sqrt{C_p})^2 m, \ \ \ \ \ 
b^{\mathrm{TW}} = (\sqrt{C_n} + \sqrt{C_p}) \left( \frac{1}{\sqrt{C_n}} + \frac{1}{\sqrt{C_p}} \right)^{\frac{1}{3}} m^{\frac{1}{3}}. 
\end{align}

By substituting (\ref{eq:unre_Z_op_up}) and (\ref{eq:ab_m}) into (\ref{eq:T_statistic}), we have
\begin{align}
T &= \frac{\| \hat{Z} \|_{\mathrm{op}}^2 - a^{\mathrm{TW}}}{b^{\mathrm{TW}}} \leq \frac{np - a^{\mathrm{TW}}}{b^{\mathrm{TW}}} 
= \frac{C_n C_p m^{\frac{5}{3}} - (\sqrt{C_n} + \sqrt{C_p})^2 m^{\frac{2}{3}}}{(\sqrt{C_n} + \sqrt{C_p}) \left( \frac{1}{\sqrt{C_n}} + \frac{1}{\sqrt{C_p}} \right)^{\frac{1}{3}}} \nonumber \\
&= O_p \left( m^{\frac{5}{3}} \right), 
\end{align}
which concludes the proof of (\ref{eq:T_un_up}). 

We next show the lower bound in (\ref{eq:T_un_low}). As shown in Figure \ref{fig:unrealizable}, we define that $\bar{P}$ is a matrix that consists of the \textbf{estimated} bicluster structure and entries of the \textbf{population} means. For instance, if the $(i, j)$th entry of observed matrix $A$ belongs to an estimated bicluster that consists of $n^{(1)}$ entries of the $k_1$th null bicluster and $n^{(2)}$ entries of the $k_2$th null bicluster, its value in matrix $\bar{P}$ is given by $\bar{P}_{ij} = \left( n^{(1)} b_{k_1} + n^{(2)} b_{k_2} \right) / \left( n^{(1)} + n^{(2)} \right)$. 
\begin{figure}[t]
  \centering
  \includegraphics[width=\hsize]{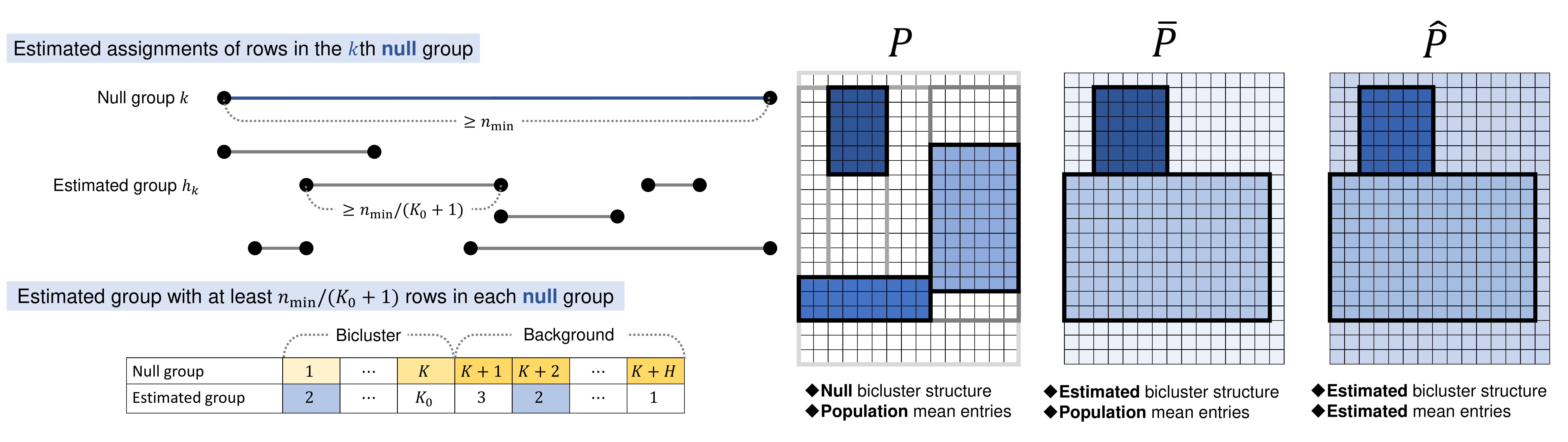} 
  \caption{Definition of matrices $P$, $\bar{P}$, and $\hat{P}$ in an unrealizable case. }
  \label{fig:unrealizable}
\end{figure}

From the assumption \ref{asmp:block_size}, for all $k \in \{ 1, \dots, K+H \}$, the $k$th \textbf{null} submatrix (i.e., bicluster or background submatrix) has a size of at least $n_{\mathrm{min}} \times p_{\mathrm{min}}$. Therefore, for all \textbf{null} group index $k \in \{ 1, \dots, K+H \}$, there exists at least one \textbf{estimated} group $h_k \in \{ 0, 1, \dots, K_0 \}$ that contains a submatrix of the $k$th \textbf{null} submatrix with the size of $\frac{n_{\mathrm{min}}}{K_0 + 1} \times \frac{p_{\mathrm{min}}}{K_0 + 1}$ or more (Figure \ref{fig:unrealizable}). Since $K_0 < K$, there exists at least one set of \textbf{null} group indices $(k_1, k_2)$, $k_1, k_2 \in \{ 1, \dots, K+H \}$ that satisfies $h_{k_1} = h_{k_2}$, $k_1 \neq k_2$, and $k_1 \in \{ 1, \dots, K \}$ (i.e., a pair of mutually different \textbf{null} groups that belong to the same group in the \textbf{estimated} bicluster structure). In other words, there exists at least one \textbf{estimated} group $\bar{k}$ ($= h_{k_1} = h_{k_2}$) that satisfies the following conditions: 
\begin{itemize}
\item The $\bar{k}$th \textbf{estimated} group contains two submatrices. We denote the sets of entries in these two submatrices as $\mathcal{I}^{(1)}$ and $\mathcal{I}^{(2)}$. 
\item The set of entries $\mathcal{I}^{(1)}$ forms a submatrix of the $k_1$th \textbf{null} group ($k_1 \in \{ 1, \dots, K \}$). The size of submatrix $\mathcal{I}^{(1)}$ is at least $\frac{n_{\mathrm{min}}}{K_0 + 1} \times \frac{p_{\mathrm{min}}}{K_0 + 1}$. 
\item The set of entries $\mathcal{I}^{(2)}$ forms a submatrix of the $k_2$th \textbf{null} group ($k_2 \in \{ 0, 1, \dots, K \}$). The size of submatrix $\mathcal{I}^{(2)}$ is at least $\frac{n_{\mathrm{min}}}{K_0 + 1} \times \frac{p_{\mathrm{min}}}{K_0 + 1}$. 
\item The \textbf{null} groups of $\mathcal{I}^{(1)}$ and $\mathcal{I}^{(2)}$ are mutually different (i.e., $k_1 \neq k_2$). 
\end{itemize}

The population means of submatrices $\mathcal{I}^{(1)}$ and $\mathcal{I}^{(2)}$, respectively, are $b_{k_1}$ and $b_{k_2}$. We assume $b_{k_1} > b_{k_2}$ without loss of generality. Let $\bar{b}$ be the constant value of the submatrices $\mathcal{I}^{(1)}$ and $\mathcal{I}^{(2)}$ in matrix $\bar{P}$. Here, we consider the following two patterns: 
\begin{itemize}
\item If $\bar{b} \geq (b_{k_1} + b_{k_2})/2$, we have $|\bar{b} - b_{k_2}| \geq (b_{k_1} - b_{k_2})/2$. 
\item If $\bar{b} < (b_{k_1} + b_{k_2})/2$, we have $|\bar{b} - b_{k_1}| > (b_{k_1} - b_{k_2})/2$. 
\end{itemize}
Therefore, for any case, there exists a submatrix $\mathcal{I}$ that satisfies the following two conditions (note that $\mathcal{I} = \mathcal{I}^{(2)}$ in the former case and $\mathcal{I} = \mathcal{I}^{(1)}$ in the latter case): 
\begin{itemize}
\item The size of submatrix $\mathcal{I}$ is at least $\frac{n_{\mathrm{min}}}{K_0 + 1} \times \frac{p_{\mathrm{min}}}{K_0 + 1}$. 
\item Let $b_{\mathcal{I}}$ and $\bar{b}_{\mathcal{I}}$, respectively, be the constant values of submatrix $\mathcal{I}$ in matrices $P$ and $\bar{P}$. Note that we have $\bar{b}_{\mathcal{I}} = \bar{b}_{\bar{k}}$. From the assumption \ref{asmp:Z_exp_S}, the following inequality holds: 
\begin{align}
\label{eq:bar_B_B_diff}
|\bar{b}_{\mathcal{I}} - b_{\mathcal{I}}| \geq \min_{k \neq k'} |b_k - b_{k'}|/2 \geq \frac{C^{\bm{b}}}{2}. 
\end{align}
\end{itemize}

The difference between $\hat{b}_{\bar{k}}$ and $\bar{b}_{\bar{k}}$ is given by
\begin{align}
\label{eq:o_block_un}
&|\hat{b}_{\bar{k}} - \bar{b}_{\bar{k}}| 
= \frac{1}{|\hat{\mathcal{I}}_{\bar{k}}|} \left| \sum_{(i, j) \in \hat{\mathcal{I}}_{\bar{k}}} \left( \hat{P}_{ij} - \bar{P}_{ij} \right) \right| 
= \frac{1}{|\hat{\mathcal{I}}_{\bar{k}}|} \left| \sum_{(i, j) \in \hat{\mathcal{I}}_{\bar{k}}} \left( A_{ij} - P_{ij} \right) \right| \nonumber \\
&\leq \frac{\max_{k = 0, 1, \dots, K} s_k}{|\hat{\mathcal{I}}_{\bar{k}}|} \left| \sum_{(i, j) \in \hat{\mathcal{I}}_{\bar{k}}} Z_{ij} \right|. 
\end{align}
To derive the upper bound of the right side of (\ref{eq:o_block_un}), we cannot take the same strategy as in the previous study \cite{Watanabe2021}, since it uses the assumption that the $\bar{k}$th estimated group can always be represented as a submatrix and it does not consider a general background structure. Therefore, we adopt an alternative approach to use the Lyapunov variant of the central limit theorem. 
Here, $Z_{ij}$ independently follows a distribution with zero mean and unit variance. From the sub-exponential condition \ref{asmp:Z_exp_S}, for any $n^{\mathrm{M}} \in \mathbb{N}$, we have $\mathbb{E} [Z_{ij}^{n^{\mathrm{M}}}] < \infty$. Let $\mathcal{I}^{\mathrm{CLT}}$ be a subset of entries in a $n \times p$ matrix. By defining $\delta \equiv 2$ and from the assumption \ref{asmp:Z_exp_S}, the following Lyapunov’s condition holds: 
\begin{align}
&\lim_{|\mathcal{I}^{\mathrm{CLT}}| \to \infty} \frac{1}{|\mathcal{I}^{\mathrm{CLT}}|^{1+\frac{1}{2}\delta}} \sum_{(i, j) \in \mathcal{I}^{\mathrm{CLT}}} \mathbb{E} \left[ |Z_{ij}|^{2+\delta} \right] 
\leq \lim_{|\mathcal{I}^{\mathrm{CLT}}| \to \infty} \frac{C^{\mathrm{M}(4)} |\mathcal{I}^{\mathrm{CLT}}|}{|\mathcal{I}^{\mathrm{CLT}}|^2} \nonumber \\
&= \lim_{|\mathcal{I}^{\mathrm{CLT}}| \to \infty} \frac{C^{\mathrm{M}(4)}}{|\mathcal{I}^{\mathrm{CLT}}|} = 0. 
\end{align}
Therefore, from the Lyapunov variant of the central limit theorem, 
\begin{align}
\frac{1}{\sqrt{|\mathcal{I}^{\mathrm{CLT}}|}} \sum_{(i, j) \in \mathcal{I}^{\mathrm{CLT}}} Z_{ij} \rightsquigarrow N(0, 1). 
\end{align}
From (\ref{eq:o_block_un}), Prokhorov's theorem \cite{Vaart1998}, and the fact that $|\hat{\mathcal{I}}_{\bar{k}}| \geq |\mathcal{I}| \geq n_{\mathrm{min}} p_{\mathrm{min}} / (K_0 + 1)^2$ (note that the right side monotonically increases with the matrix size $m$ from the assumption \ref{asmp:block_size}), we have
\begin{align}
\label{eq:o_block_un2}
&|\hat{b}_{\bar{k}} - \bar{b}_{\bar{k}}| \leq \frac{\max_{k = 0, 1, \dots, K} s_k}{\sqrt{|\hat{\mathcal{I}}_{\bar{k}}|}} O_p (1) 
\leq \left( \max_{k = 0, 1, \dots, K} s_k \right) \frac{K_0 + 1}{\sqrt{n_{\mathrm{min}} p_{\mathrm{min}}}} O_p (1) \nonumber \\
&\leq O_p \left( \frac{K+H}{\sqrt{n_{\mathrm{min}} p_{\mathrm{min}}}} \right) \ \ \ (\because K_0 + 1 < K + H\ \mathrm{and\ assumption\ \ref{asmp:Z_exp_S}}). 
\end{align}

From (\ref{eq:o_block_un2}), we have
\begin{align}
\left| |b_{\mathcal{I}} - \bar{b}_{\mathcal{I}}| - |b_{\mathcal{I}} - \hat{b}_{\bar{k}}| \right| 
\leq |\hat{b}_{\bar{k}} - \bar{b}_{\bar{k}}| 
\leq O_p \left( \frac{K+H}{\sqrt{n_{\mathrm{min}} p_{\mathrm{min}}}} \right). 
\end{align}
By combining this result with (\ref{eq:bar_B_B_diff}), 
\begin{align}
\frac{C^{\bm{b}}}{2} \leq |b_{\mathcal{I}} - \hat{b}_{\bar{k}}| + O_p \left( \frac{K+H}{\sqrt{n_{\mathrm{min}} p_{\mathrm{min}}}} \right). 
\end{align}
Therefore, from (\ref{eq:n_min_p_min_un}) in the assumption \ref{asmp:block_size}, we have
\begin{align}
\label{eq:B_I_B_hat_I}
|b_{\mathcal{I}} - \hat{b}_{\bar{k}}| = \Omega_p (1). 
\end{align}

Let $X^{\mathcal{I}}$ be a submatrix of $X$ with the set of entries $\mathcal{I}$. Since the operator norm of a submatrix is not larger than that of the original matrix, 
\begin{align}
\label{eq:Z_hat_op_lower}
\| \hat{Z} \|_{\mathrm{op}} &\geq \| \hat{Z}^{\mathcal{I}} \|_{\mathrm{op}} 
= \frac{1}{\hat{s}_{\bar{k}}} \| A^{\mathcal{I}} - \hat{P}^{\mathcal{I}} \|_{\mathrm{op}} 
\geq \frac{1}{\hat{s}_{\bar{k}}} \left| \| A^{\mathcal{I}} - P^{\mathcal{I}} \|_{\mathrm{op}} - \| P^{\mathcal{I}} - \hat{P}^{\mathcal{I}} \|_{\mathrm{op}} \right|. 
\end{align}

Let $\bar{k}^{\mathrm{N}}$ be the \textbf{null} bicluster index (including background) of submatrix $\mathcal{I}$. Note that $b_{\mathcal{I}} = b_{\bar{k}^{\mathrm{N}}}$. As for the first term in (\ref{eq:Z_hat_op_lower}), from the assumption \ref{asmp:Z_exp_S}, we have
\begin{align}
\label{eq:A_E_k_P_E_k}
\| A^{\mathcal{I}} - P^{\mathcal{I}} \|_{\mathrm{op}} &= s_{\bar{k}^{\mathrm{N}}} \| Z^{\mathcal{I}} \|_{\mathrm{op}} 
\leq s_{\bar{k}^{\mathrm{N}}} \| Z \|_{\mathrm{op}} 
\leq \left( \max_{k = 0, 1, \dots, K} s_k \right) O_p (\sqrt{m}) \nonumber \\
&= O_p (\sqrt{m}). 
\end{align}

In regard to the second term in (\ref{eq:Z_hat_op_lower}), since all the entries in matrix $(P^{\mathcal{I}} - \hat{P}^{\mathcal{I}})$ is $(b_{\mathcal{I}} - \hat{b}_{\bar{k}})$ and thus its rank is one, we have
\begin{align}
\label{eq:P_E_P_hat_E}
\| P^{\mathcal{I}} - \hat{P}^{\mathcal{I}} \|_{\mathrm{op}} &= \| P^{\mathcal{I}} - \hat{P}^{\mathcal{I}} \|_{\mathrm{F}} = \sqrt{|\mathcal{I}| (b_{\mathcal{I}} - \hat{b}_{\bar{k}})^2} 
\geq \frac{\sqrt{n_{\mathrm{min}} p_{\mathrm{min}}}}{K_0 + 1} |b_{\mathcal{I}} - \hat{b}_{\bar{k}}| \nonumber \\
&\geq \frac{\sqrt{n_{\mathrm{min}} p_{\mathrm{min}}}}{K+H} |b_{\mathcal{I}} - \hat{b}_{\bar{k}}| 
= \Omega_p \left( \frac{\sqrt{n_{\mathrm{min}} p_{\mathrm{min}}}}{K+H} \right). 
\end{align}
To derive the last equation, we used the assumption \ref{asmp:block_size} and (\ref{eq:B_I_B_hat_I}). 

Finally, we can derive an upper bound of $\hat{s}_{\bar{k}}$ by
\begin{align}
\label{eq:S_hat_bar_k_up}
\hat{s}_{\bar{k}} &= \frac{1}{\sqrt{|\hat{\mathcal{I}}_{\bar{k}}|}} \| \underline{A}^{\mathrm{E} (\bar{k})} - \underline{\hat{P}}^{\mathrm{E} (\bar{k})} \|_{\mathrm{F}} \nonumber \\
&\leq \frac{1}{\sqrt{|\hat{\mathcal{I}}_{\bar{k}}|}} \left( \| \underline{A}^{\mathrm{E} (\bar{k})} - \underline{P}^{\mathrm{E} (\bar{k})} \|_{\mathrm{F}} + \| \underline{P}^{\mathrm{E} (\bar{k})} - \underline{\hat{P}}^{\mathrm{E} (\bar{k})} \|_{\mathrm{F}} \right) \nonumber \\
&\leq \frac{1}{\sqrt{|\hat{\mathcal{I}}_{\bar{k}}|}} \left( \| A - P \|_{\mathrm{F}} + \| \underline{P}^{\mathrm{E} (\bar{k})} - \underline{\hat{P}}^{\mathrm{E} (\bar{k})} \|_{\mathrm{F}} \right) \nonumber \\
&= \frac{1}{\sqrt{|\hat{\mathcal{I}}_{\bar{k}}|}} \left( \sqrt{ \sum_{k=1}^{K+H} s_k^2 \| Z^{\mathrm{N} (k)} \|_{\mathrm{F}}^2 } + \| \underline{P}^{\mathrm{E} (\bar{k})} - \underline{\hat{P}}^{\mathrm{E} (\bar{k})} \|_{\mathrm{F}} \right) \nonumber \\
&\leq \frac{1}{\sqrt{|\hat{\mathcal{I}}_{\bar{k}}|}} \left[ \left( \max_{k = 0, 1, \dots, K} s_k \right) \| Z \|_{\mathrm{F}} + \| \underline{P}^{\mathrm{E} (\bar{k})} - \underline{\hat{P}}^{\mathrm{E} (\bar{k})} \|_{\mathrm{F}} \right] \nonumber \\
&= \frac{1}{\sqrt{|\hat{\mathcal{I}}_{\bar{k}}|}} \left[ \left( \max_{k = 0, 1, \dots, K} s_k \right) \| Z \|_{\mathrm{F}} + \sqrt{ \sum_{(i, j) \in \hat{\mathcal{I}}_{\bar{k}}} \left( P_{ij} - \hat{b}_{\bar{k}} \right)^2 } \right] \nonumber \\
&\leq \frac{1}{\sqrt{|\hat{\mathcal{I}}_{\bar{k}}|}} \left( \max_{k = 0, 1, \dots, K} s_k \right) \| Z \|_{\mathrm{F}} + \max_{k = 0, 1, \dots, K} \left| b_k - \hat{b}_{\bar{k}} \right| \nonumber \\
&\leq \frac{K_0 + 1}{\sqrt{n_{\mathrm{min}} p_{\mathrm{min}}}} \left( \max_{k = 0, 1, \dots, K} s_k \right) \| Z \|_{\mathrm{F}} + \max_{k = 0, 1, \dots, K} \left| b_k - \hat{b}_{\bar{k}} \right| \nonumber \\
&\leq \frac{K+H}{\sqrt{n_{\mathrm{min}} p_{\mathrm{min}}}} \left( \max_{k = 0, 1, \dots, K} s_k \right) \| Z \|_{\mathrm{F}} + \max_{k = 0, 1, \dots, K} \left| b_k - \hat{b}_{\bar{k}} \right|. 
\end{align}

The second term in (\ref{eq:S_hat_bar_k_up}) can be upper bounded as follows: 
\begin{align}
\label{eq:max_B_k_B_hat_k}
&\max_{k = 0, 1, \dots, K} \left| b_k - \hat{b}_{\bar{k}} \right| \nonumber \\
&\leq |\hat{b}_{\bar{k}}| + \max_{k = 0, 1, \dots, K} |b_k| 
= \frac{1}{|\hat{\mathcal{I}}_{\bar{k}}|} \left| \sum_{(i, j) \in \hat{\mathcal{I}}_{\bar{k}}} (\sigma_{ij} Z_{ij} + P_{ij}) \right| + \max_{k = 0, 1, \dots, K} |b_k| \nonumber \\
&\leq \frac{1}{|\hat{\mathcal{I}}_{\bar{k}}|} \sum_{(i, j) \in \hat{\mathcal{I}}_{\bar{k}}} \left( \sigma_{ij} \left| Z_{ij} \right| + \left| P_{ij} \right| \right) + \max_{k = 0, 1, \dots, K} |b_k| \nonumber \\
&\leq \frac{1}{|\hat{\mathcal{I}}_{\bar{k}}|} \sum_{(i, j) \in \hat{\mathcal{I}}_{\bar{k}}} \sigma_{ij} \left| Z_{ij} \right| + 2 \max_{k = 0, 1, \dots, K} |b_k| \nonumber \\
&\leq \frac{1}{|\hat{\mathcal{I}}_{\bar{k}}|} \left( \max_{k = 0, 1, \dots, K} s_k \right) \sum_{(i, j) \in \hat{\mathcal{I}}_{\bar{k}}} \left| Z_{ij} \right| + 2 \max_{k = 0, 1, \dots, K} |b_k| \nonumber \\
&\leq \frac{1}{|\hat{\mathcal{I}}_{\bar{k}}|} \left( \max_{k = 0, 1, \dots, K} s_k \right) \sqrt{\sum_{(i, j) \in \hat{\mathcal{I}}_{\bar{k}}} \left| Z_{ij} \right|^2} \sqrt{|\hat{\mathcal{I}}_{\bar{k}}|} + 2 \max_{k = 0, 1, \dots, K} |b_k| \nonumber \\
&\leq \frac{1}{\sqrt{|\hat{\mathcal{I}}_{\bar{k}}|}} \left( \max_{k = 0, 1, \dots, K} s_k \right) \| Z \|_{\mathrm{F}} + 2 \max_{k = 0, 1, \dots, K} |b_k| \nonumber \\
&\leq \frac{K_0 + 1}{\sqrt{n_{\mathrm{min}} p_{\mathrm{min}}}} \left( \max_{k = 0, 1, \dots, K} s_k \right) \| Z \|_{\mathrm{F}} + 2 \max_{k = 0, 1, \dots, K} |b_k| \nonumber \\
&\leq \frac{K+H}{\sqrt{n_{\mathrm{min}} p_{\mathrm{min}}}} \left( \max_{k = 0, 1, \dots, K} s_k \right) \| Z \|_{\mathrm{F}} + 2 \max_{k = 0, 1, \dots, K} |b_k|. 
\end{align}

Since $\mathbb{E} [Z_{ij}^2] = \mathbb{V} [Z_{ij}] + \mathbb{E} [Z_{ij}]^2 = 1$ and $\mathbb{V} [Z_{ij}^2] = \mathbb{E} [Z_{ij}^4] - \mathbb{E} [Z_{ij}^2]^2 = \mathbb{E} [Z_{ij}^4] - 1 < \infty$ from the assumption \ref{asmp:Z_exp_S}, from the central limit theorem and Prokhorov's theorem \cite{Vaart1998}, we have
\begin{align}
\frac{1}{\sqrt{np}} \sum_{i=1}^n \sum_{j=1}^p (Z_{ij}^2 - 1) = O_p (1) 
\iff \sum_{i=1}^n \sum_{j=1}^p Z_{ij}^2 = \| Z \|_{\mathrm{F}}^2 = np + O_p (m), 
\end{align}
which results in that 
\begin{align}
\label{eq:Z_F}
\| Z \|_{\mathrm{F}} = O_p (m). 
\end{align}

By substituting (\ref{eq:max_B_k_B_hat_k}) and (\ref{eq:Z_F}) into (\ref{eq:S_hat_bar_k_up}), and using the assumption \ref{asmp:Z_exp_S}, 
\begin{align}
\label{eq:S_hat_bar_k_up2}
\hat{s}_{\bar{k}} &\leq \frac{K+H}{\sqrt{n_{\mathrm{min}} p_{\mathrm{min}}}} O_p (m) + O (K) = O_p \left[ \frac{(K+H)m}{\sqrt{n_{\mathrm{min}} p_{\mathrm{min}}}} \right]. 
\end{align}

By substituting (\ref{eq:A_E_k_P_E_k}), (\ref{eq:P_E_P_hat_E}), and (\ref{eq:S_hat_bar_k_up2}) into (\ref{eq:Z_hat_op_lower}), and using (\ref{eq:n_min_p_min_un}) in the assumption \ref{asmp:block_size}, we finally have
\begin{align}
\label{eq:Z_hat_op_lower2}
&\| \hat{Z} \|_{\mathrm{op}} \geq \Omega_p \left[ \frac{\sqrt{n_{\mathrm{min}} p_{\mathrm{min}}}}{(K+H)m} \right] \left| \Omega_p \left( \frac{\sqrt{n_{\mathrm{min}} p_{\mathrm{min}}}}{K+H} \right) - O_p (\sqrt{m}) \right| \nonumber \\
&= \Omega_p \left[ \frac{n_{\mathrm{min}} p_{\mathrm{min}}}{(K+H)^2 m} \right] \left| \Omega_p (1) - O_p \left[ \frac{(K+H)\sqrt{m}}{\sqrt{n_{\mathrm{min}} p_{\mathrm{min}}}} \right] \right| \nonumber \\
&= \Omega_p \left[ \frac{n_{\mathrm{min}} p_{\mathrm{min}}}{(K+H)^2 m} \right] 
= \Omega_p \left( m^{\frac{1}{2} + 2 \epsilon_2} \right) \nonumber \\
\iff& \| \hat{Z} \|_{\mathrm{op}}^2 = \Omega_p \left( m^{1 + 4 \epsilon_2} \right),\ \mathrm{for\ some}\ \epsilon_2 > 0. 
\end{align}

By substituting (\ref{eq:Z_hat_op_lower2}) and (\ref{eq:ab_m}) into (\ref{eq:T_statistic}), we have
\begin{align}
T = \frac{\| \hat{Z} \|_{\mathrm{op}}^2 - a^{\mathrm{TW}}}{b^{\mathrm{TW}}} = \Omega_p \left( m^{\frac{2}{3} + 4 \epsilon_2} \right) 
\geq \Omega_p \left( m^{\frac{2}{3} } \right), 
\end{align}
which concludes the proof. 
\end{proof}



\section{Experiments}
\label{sec:experiments}

As we explained in the footnote in Sect.~\ref{sec:method}, currently, we do not have any consistent submatrix localization method that can be applied to general disjoint block structure. Instead, in all the following experiments, we used Algorithm \ref{algo:max_PL_SA2} in Appendix \ref{sec:consistency_approx} for estimating the bicluster structure of a given matrix, although it is not guaranteed to be consistent. 

\subsection{The convergence of test statistic $T$ in law to $TW_1$ distribution in the realizable case}
\label{sec:exp_null}

We first checked the asymptotic behavior of the proposed test statistic $T$ in the null case (i.e., $K = K_0$) by using synthetic data matrices, which were generated from the Gaussian, Bernoulli, and Poisson distributions. In this case, from the Theorem \ref{th:realizable}, $T$ converges in law to the $TW_1$ distribution in the limit of $m \to \infty$. 

We set the null number of biclusters at $K = 3$ in all the settings of distributions, and tried $10$ sets of matrix sizes: $(n, p) = (500 \times i, 375 \times i)$, $i = 1, \dots, 10$. For each distribution, we defined the null set of parameters and the relative entropy function $f$ of the generalized profile likelihood in (\ref{eq:gen_PL_def})\footnote{These experimental settings of the relative entropy function $f$ follow those in \cite{Flynn2020}. } as follows. 
\begin{itemize}
\item \textbf{Gaussian case}: Each entry in the $k$th group ($k = 0, 1, \dots, K$) of observed matrix $A$ was generated independently from the Gaussian distribution $\mathcal{N} (b_k, s_k)$, where 
\begin{align}
\label{eq:b_gauss}
&\bm{b} = 
\begin{pmatrix}
0.2 & 0.5 & 0.6 & 0.7
\end{pmatrix}^{\top}, \ 
\bm{s} = 
\begin{pmatrix}
0.03 & 0.04 & 0.06 & 0.07
\end{pmatrix}^{\top}. \\
&f(x) \equiv x^2/2. 
\end{align}
\item \textbf{Bernoulli case}: Each entry in the $k$th group of observed matrix $A$ was generated independently from the Bernoulli distribution $\mathrm{Bernoulli} (b_k)$, where 
\begin{align}
\label{eq:b_bernoulli}
&\bm{b} = 
\begin{pmatrix}
0.2 & 0.5 & 0.6 & 0.7
\end{pmatrix}^{\top}. \\
\label{eq:f_bernoulli}
&f(x) \equiv x \log \left( \max \{x, 10^{-5}\} \right) + (1-x) \log \left( \max \{1 - x, 10^{-5}\} \right). 
\end{align}
\item \textbf{Poisson case}: Each entry in the $k$th group of observed matrix $A$ was generated independently from the Poisson distribution $\mathrm{Pois} (b_k)$, where 
\begin{align}
\label{eq:b_poisson}
&\bm{b} = 
\begin{pmatrix}
2 & 5 & 6 & 7
\end{pmatrix}^{\top}. \\
&f(x) \equiv x \log \left( \max \{x, 10^{-5}\} \right) - x. 
\end{align}
\end{itemize}
For each combination of the distribution and matrix size settings, we randomly generated $5,000$ data matrices $A$ based on the \textbf{null} (non-bi-disjoint) bicluster structure, which was defined as follows. Let $K_1 \equiv (3K + 4 + K \bmod 2)/2$, $K_2 \equiv (3K + 4 - K \bmod 2)/2$, $n_1 \equiv \lfloor n / K_1 \rfloor$, and $p_1 \equiv \lfloor p / K_2 \rfloor$. 
For each $k$th bicluster ($k = 1, \dots, K$), we also define $k_1 \equiv (3k -2 - k \bmod 2) / 2$ and $k_2 \equiv (3k - 4 + k \bmod 2) / 2$. Based on these variables, the set of rows and columns of matrix $A$ belonging to the $k$th bicluster is given by $I_k = \{ k_1 n_1 + 1, \dots, (k_1 + 2) n_1 \}$ and $J_k = \{ k_2 p_1 + 1, \dots, (k_2 + 2) p_1 \}$, respectively. Figure \ref{fig:A_example_pre} shows the examples of Gaussian, Bernoulli, and Poisson data matrices. 

\begin{figure}[t]
  \centering
  \includegraphics[width=0.25\hsize]{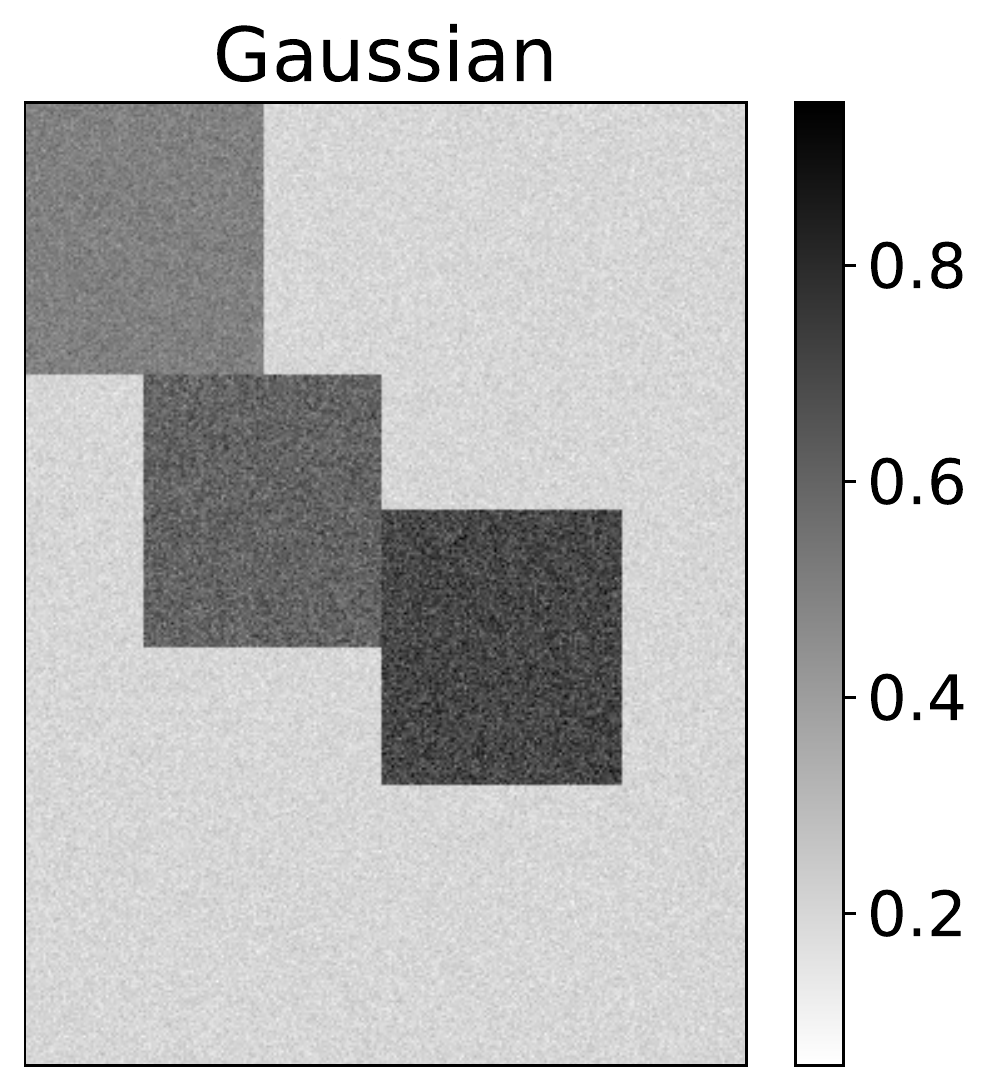}
  \includegraphics[width=0.25\hsize]{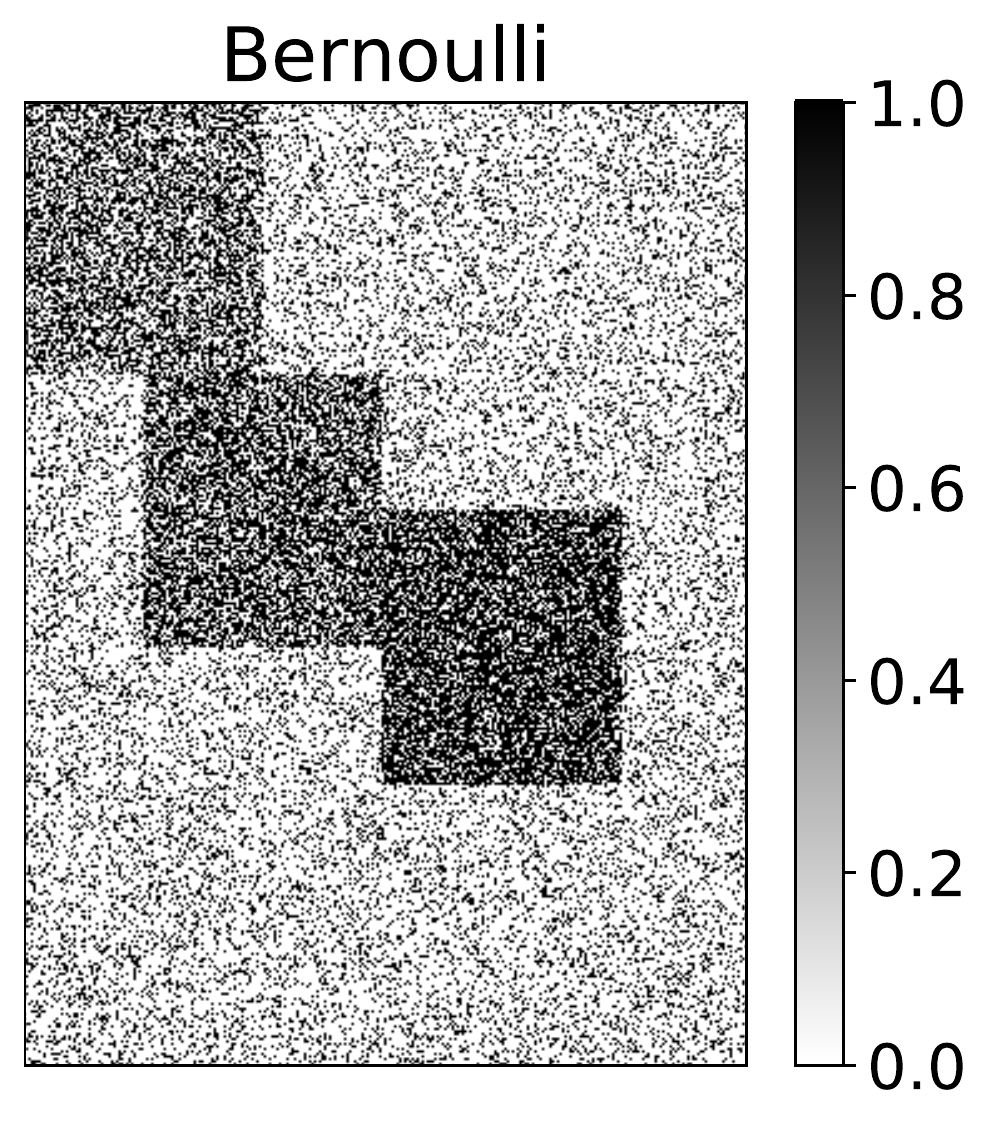}
  \includegraphics[width=0.25\hsize]{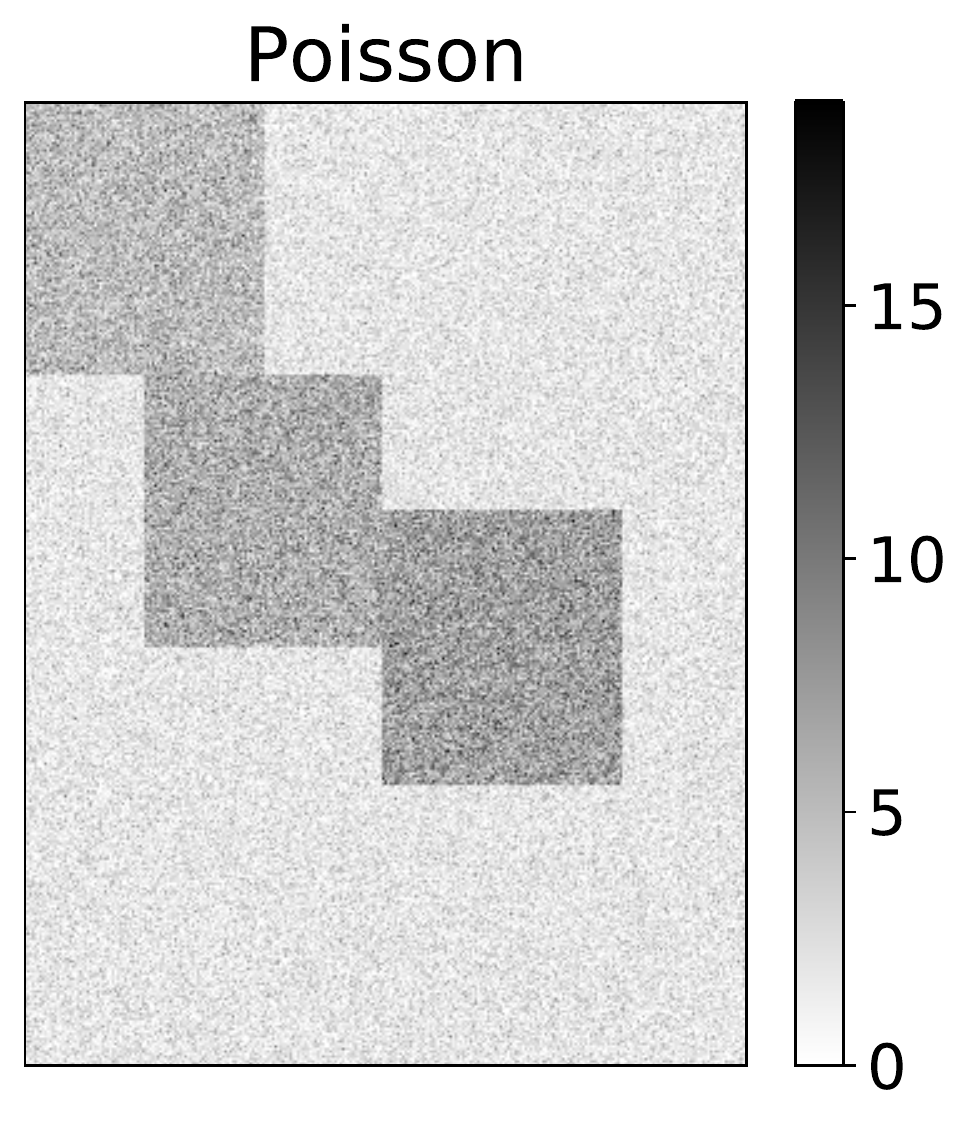}
  \caption{Examples of the observed data matrices. The left, center, and right figures show the Gaussian, Bernoulli, and Poisson cases, respectively.}
  \label{fig:A_example_pre}
\end{figure}

After generating the observed data matrices, we estimated their bicluster structures by the proposed submatrix localization algorithm (i.e., Algorithm \ref{algo:max_PL_SA2} in Sect.~\ref{sec:compress_algo}). To compress the original data matrix $A$, we applied Ward's hierarchical clustering method \cite{Ward1963} to the rows and columns of matrix $A$ with the number of clusters $L_1 = \min \{2^K, n\}$ and $L_2 = \min \{2^K, p\}$, respectively. Initial bicluster structures in Algorithm \ref{algo:max_PL_SA2} were given as follows: for each $k$th bicluster, it contains a (uniformly randomly chosen) single entry $A_{i_k j_k}$ in $A$, where $(i_k, j_k) \neq (i_{k'}, j_{k'})$ for $k \neq k'$. 
In Sect.~\ref{sec:naive_sa}, we describe a sufficient condition regarding the cooling schedule $\{ T_t \}$ for the SA algorithm to converge in probability to the global optimal solution. However, such a setting requires too many iterations to converge. In our experiments, we used the following cooling schedule instead: $T_t = 0.999^t$ for all $t \geq 0$. We also defined a threshold of temperature at $\epsilon^{\mathrm{SA}} = 10^{-5}$. Since this setting no longer guaranteed a convergence in probability to the global optimal solution, we applied the Algorithm \ref{algo:max_PL_SA2} to each observed matrix five times and adopted the best solution that achieved the maximum profile likelihood in the last step of the algorithm (this procedure was also performed in all the subsequent experiments in Sections \ref{sec:exp_alternative}, \ref{sec:accuracy}, and \ref{sec:exp_practical}). Based on the estimated bicluster structure, we finally applied the proposed statistical test by setting the hypothetical number of biclusters at $K$. 

Figures \ref{fig:preliminary_hist_normal}, \ref{fig:preliminary_hist_bernoulli}, and \ref{fig:preliminary_hist_poisson}, respectively, show the histograms of the proposed test statistic $T$ with different matrix sizes under the Gaussian, Bernoulli, and Poisson settings. 
Figure \ref{fig:preliminaryT} illustrates the empirical tail probabilities of the proposed test statistic $T$ (i.e., the ratios of the trials where $T \geq t(0.01)$, $T \geq t(0.05)$, and $T \geq t(0.1)$, where $t(\alpha)$ is the $\alpha$ upper quantile of the $TW_1$ distribution) under the three settings of distributions. As in the previous study \cite{Watanabe2021}, we used the approximated values $t(0.01) \approx 2.02345$, $t(0.05) \approx 0.97931$, and $t(0.1) \approx 0.45014$, based on Table $2$ in \cite{Tracy2009}. To check the convergence of $T$ to the $TW_1$ distribution, we also applied the Kolmogorov-Smirnov (KS) test \cite{Conover1999} to the test statistics $T$ of the $5,000$ trials, and plotted the results in Figure \ref{fig:KStest}. Let $D$ be the maximum absolute difference between the empirical distribution function of $T$ and the cumulative distribution function of the $TW_1$ distribution. The test statistic of the KS test is $D\sqrt{r}$, where $r$ is the number of trials (i.e., $5,000$ in this case). 

From Figures \ref{fig:preliminary_hist_normal}, \ref{fig:preliminary_hist_bernoulli}, \ref{fig:preliminary_hist_poisson}, \ref{fig:preliminaryT} and \ref{fig:KStest}, we see that the proposed test statistic $T$ converges in law to the $TW_1$ distribution in each setting of distributions. In the Bernoulli case, however, the convergence of $T$ in law to the $TW_1$ distribution is slow, compared to the other two cases (i.e., Gaussian and Poisson). To investigate the cause of this, we also computed $\tilde{T}$ with the \textbf{null} bicluster structure [i.e., $\tilde{T} \equiv (\tilde{\lambda}_1 - a^{\mathrm{TW}})/b^{\mathrm{TW}}$] and plotted its empirical tail probabilities and the test statistic of the KS test in Figures \ref{fig:preliminaryT_null} and \ref{fig:KStest_null}, respectively. From these figures, we see that the convergence of $\tilde{T}$ in law to the $TW_1$ distribution is still slow in the Bernoulli case. Therefore, the slow convergence of $T$ would not have been caused by the low accuracy of the biclustering algorithm, but it would have been a problem specific to a Bernoulli random matrix. 

\begin{figure}[p]
  \centering
  \includegraphics[width=0.2\hsize]{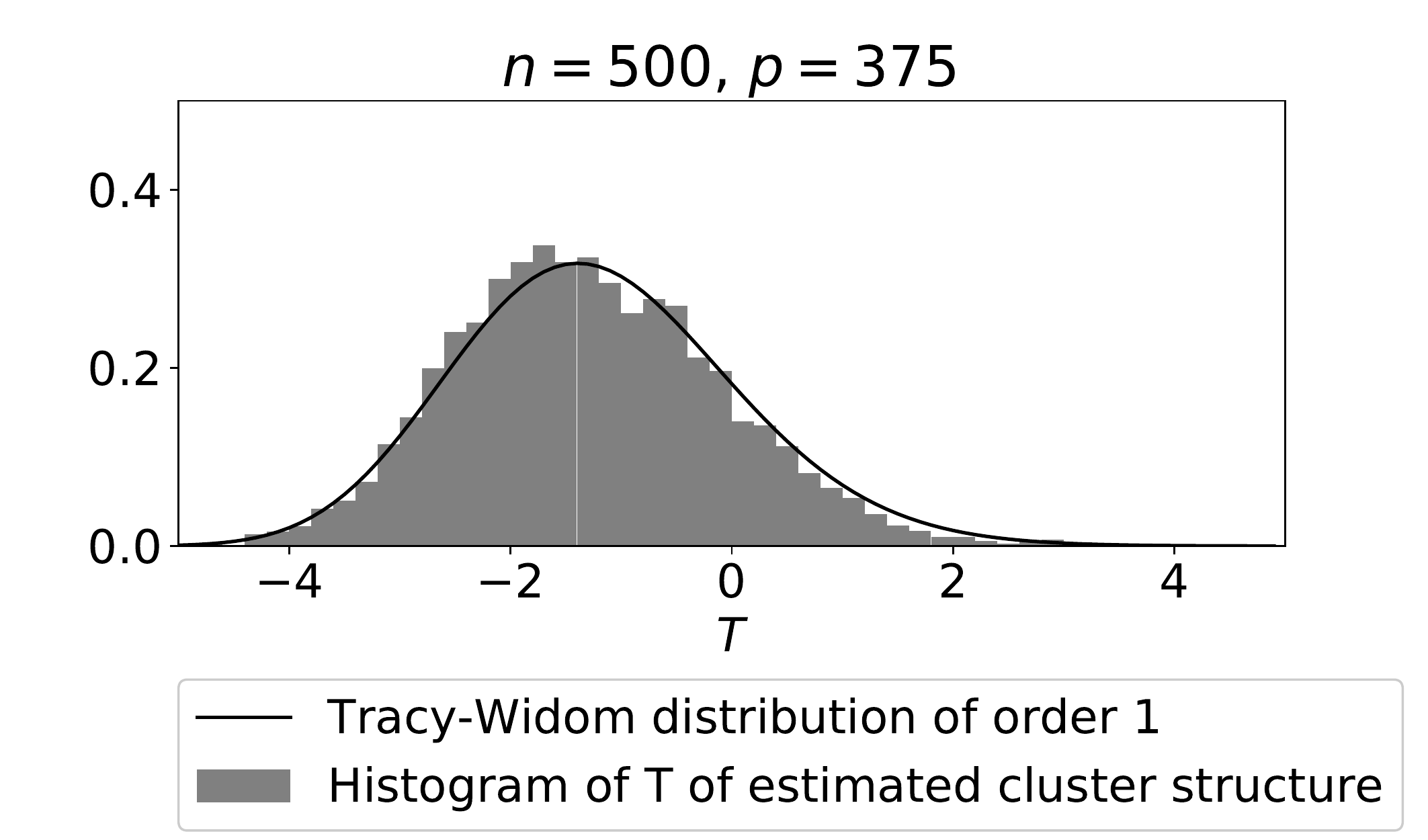}
  \includegraphics[width=0.2\hsize]{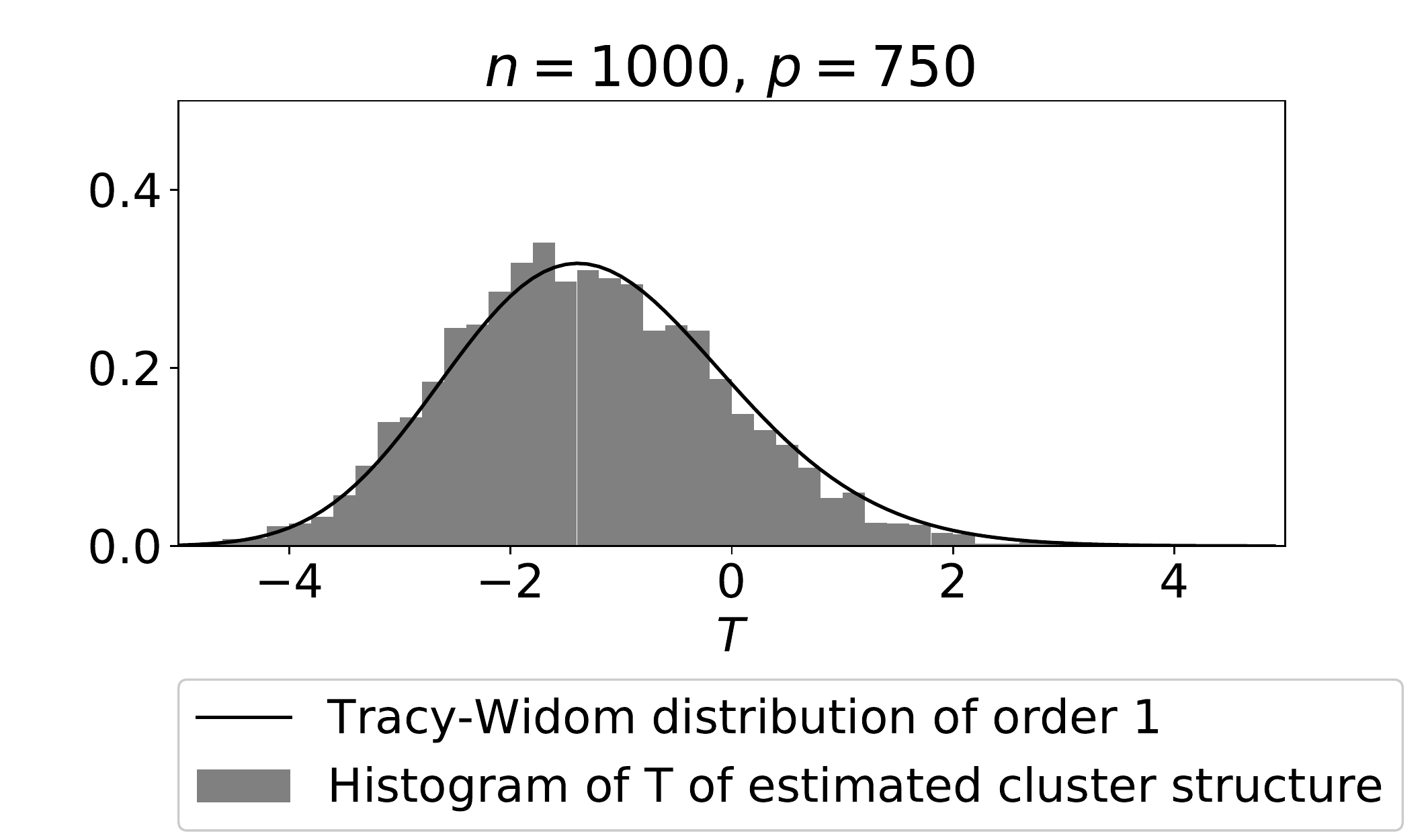}
  \includegraphics[width=0.2\hsize]{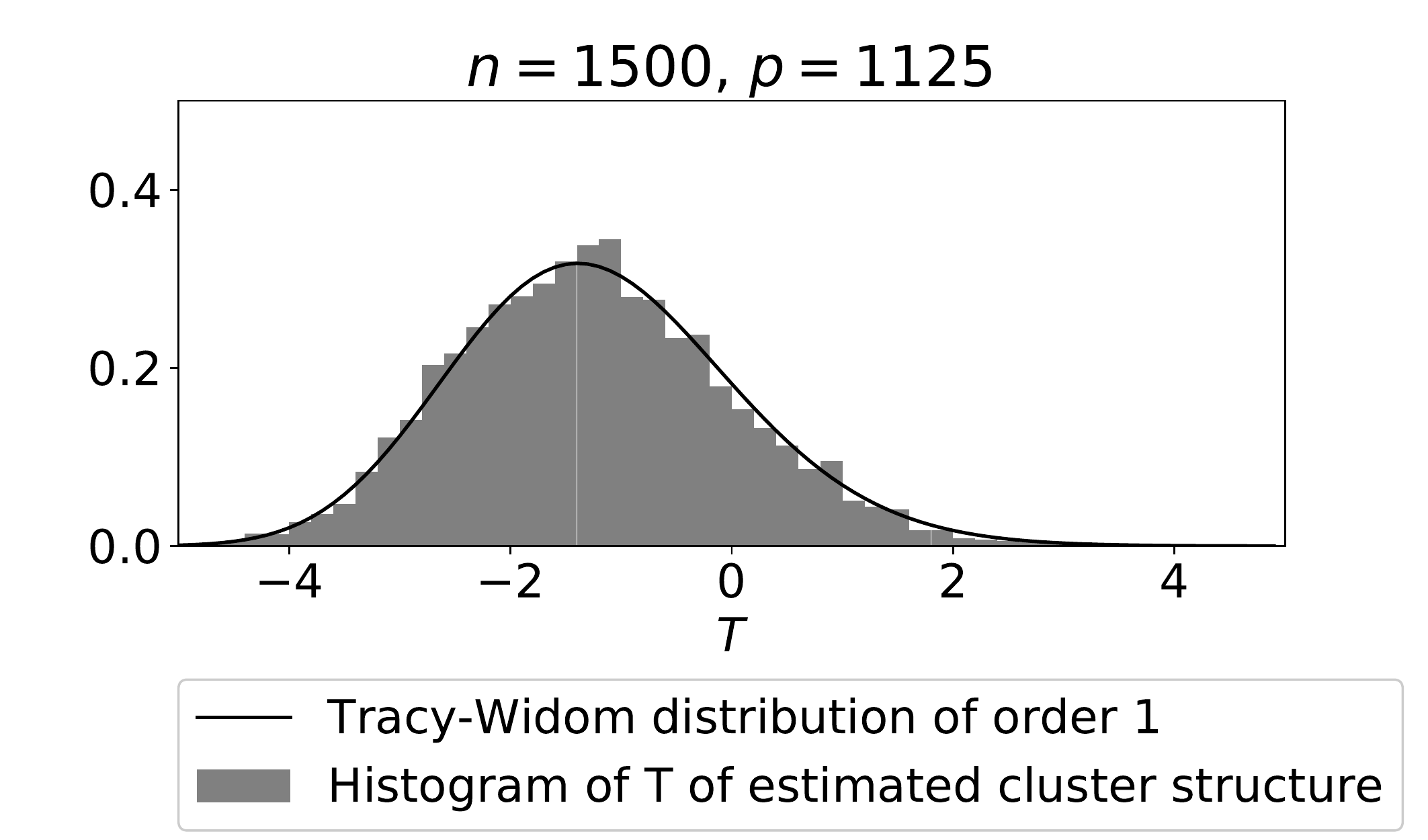}
  \includegraphics[width=0.2\hsize]{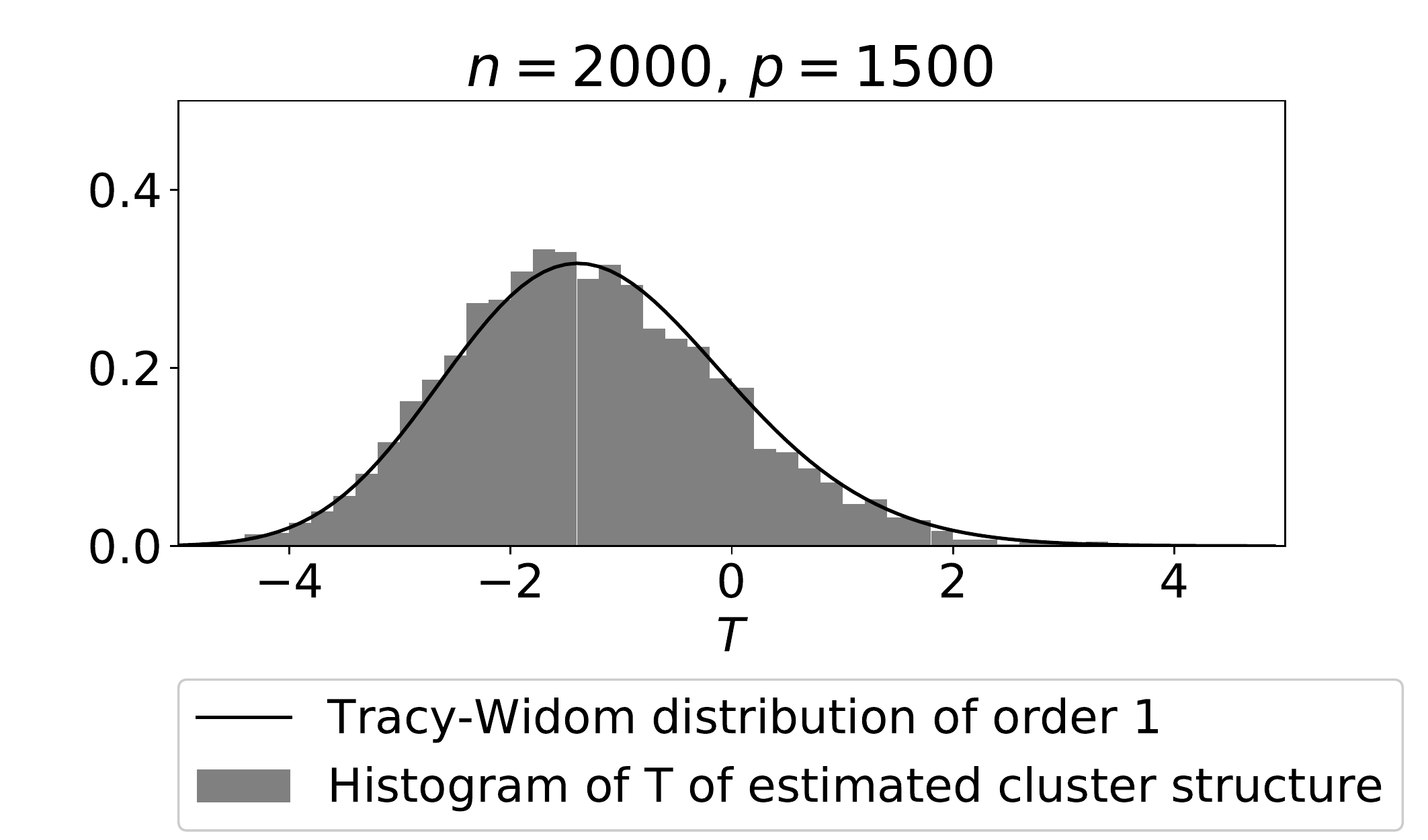}
  \includegraphics[width=0.2\hsize]{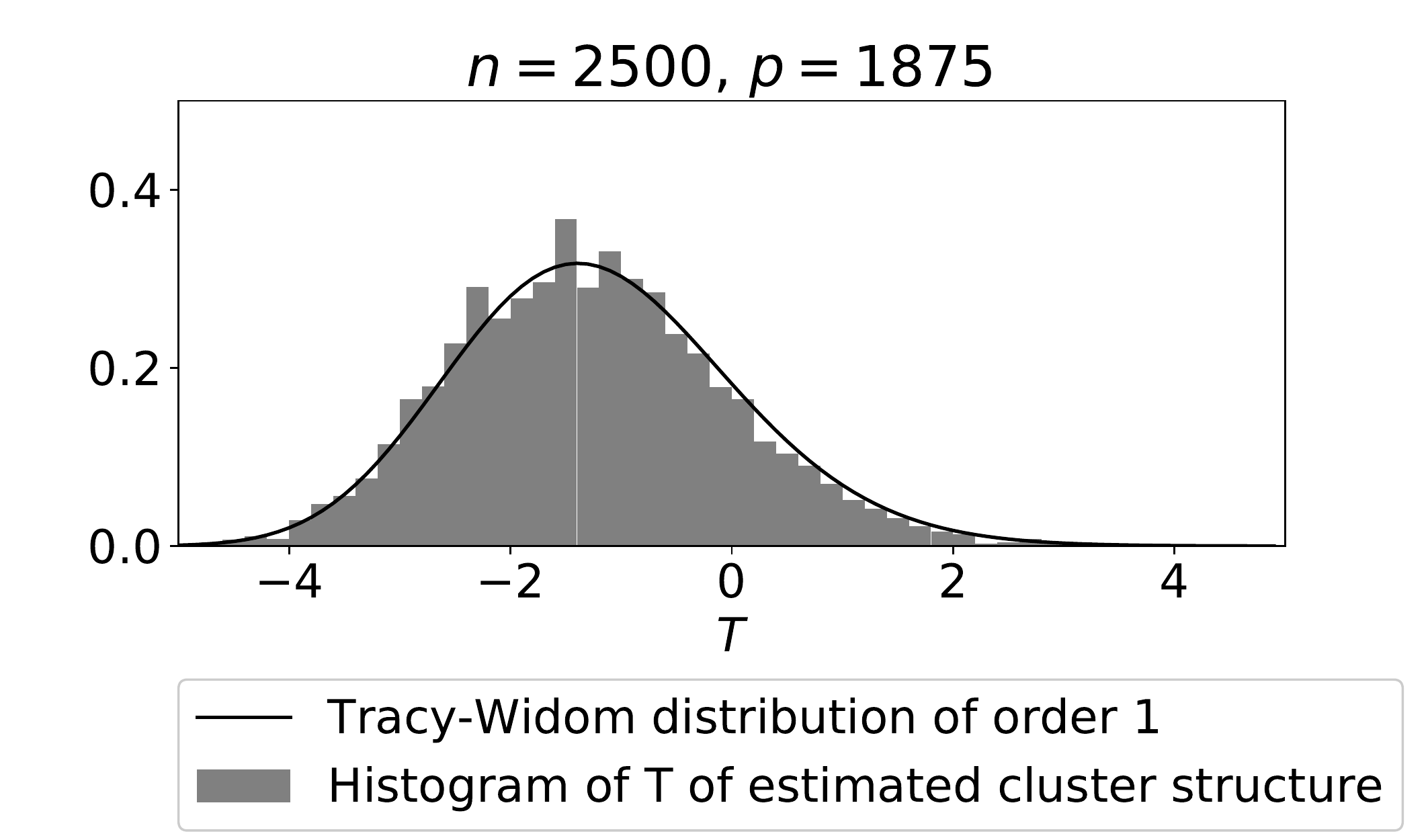}
  \includegraphics[width=0.2\hsize]{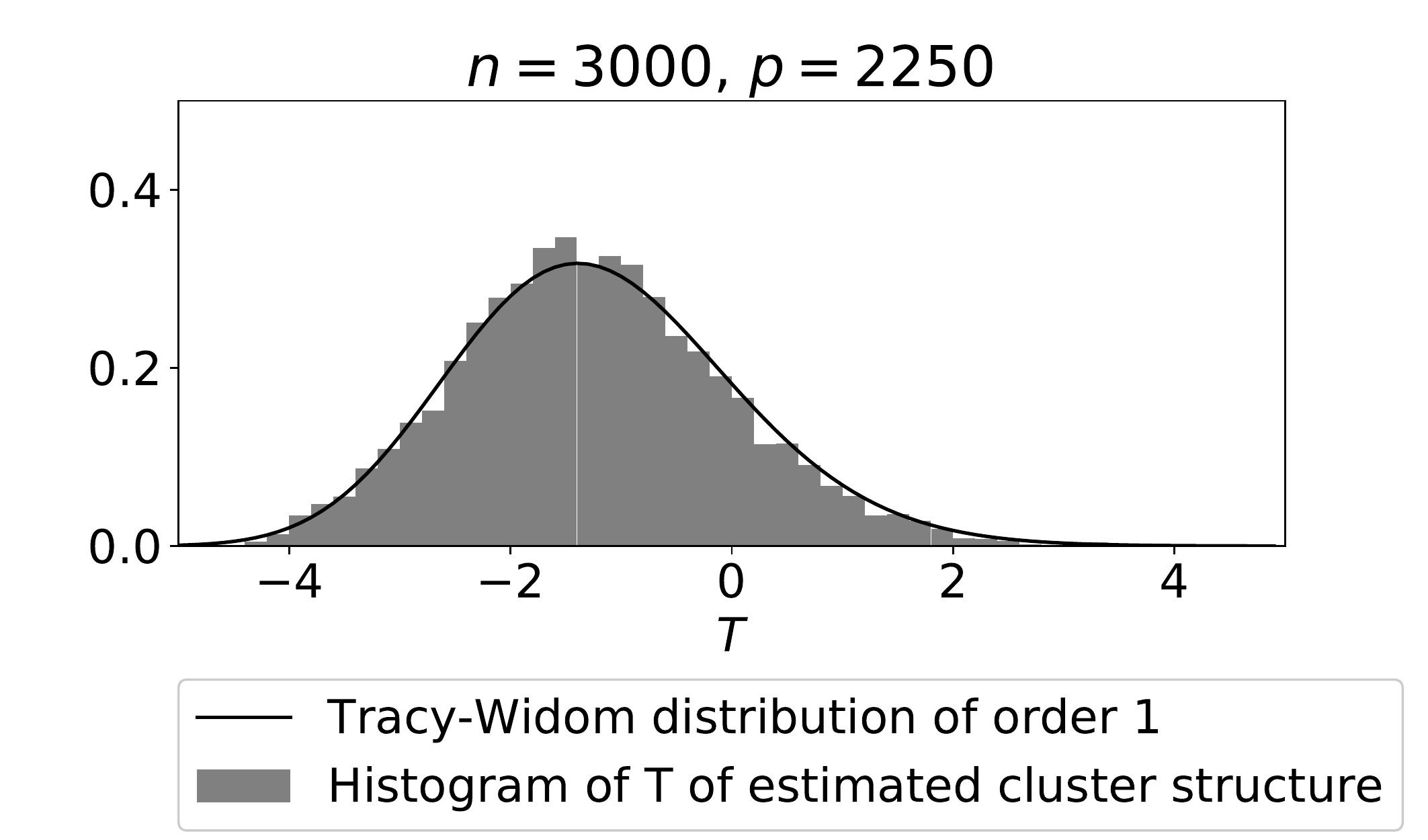}
  \includegraphics[width=0.2\hsize]{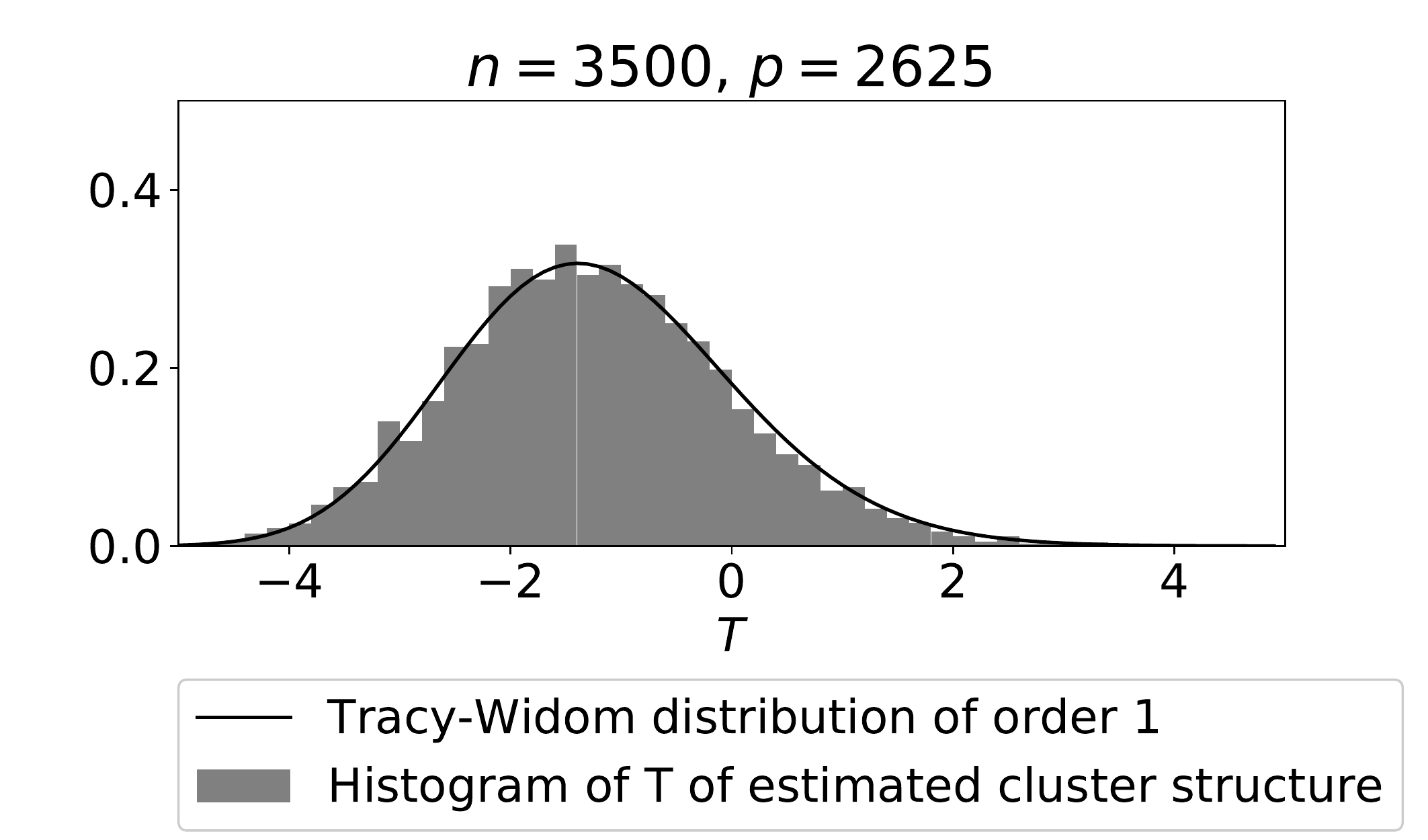}
  \includegraphics[width=0.2\hsize]{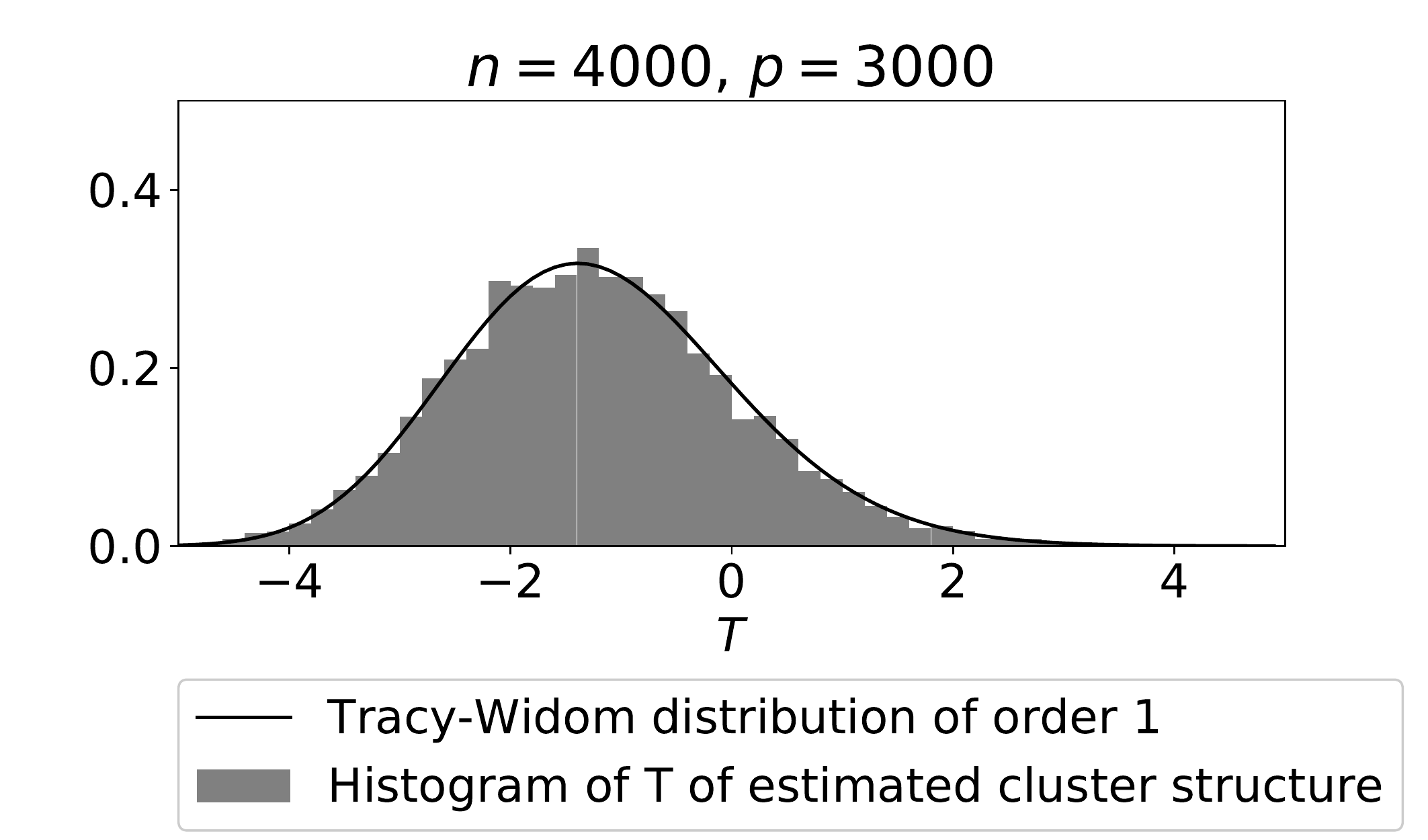}
  \includegraphics[width=0.2\hsize]{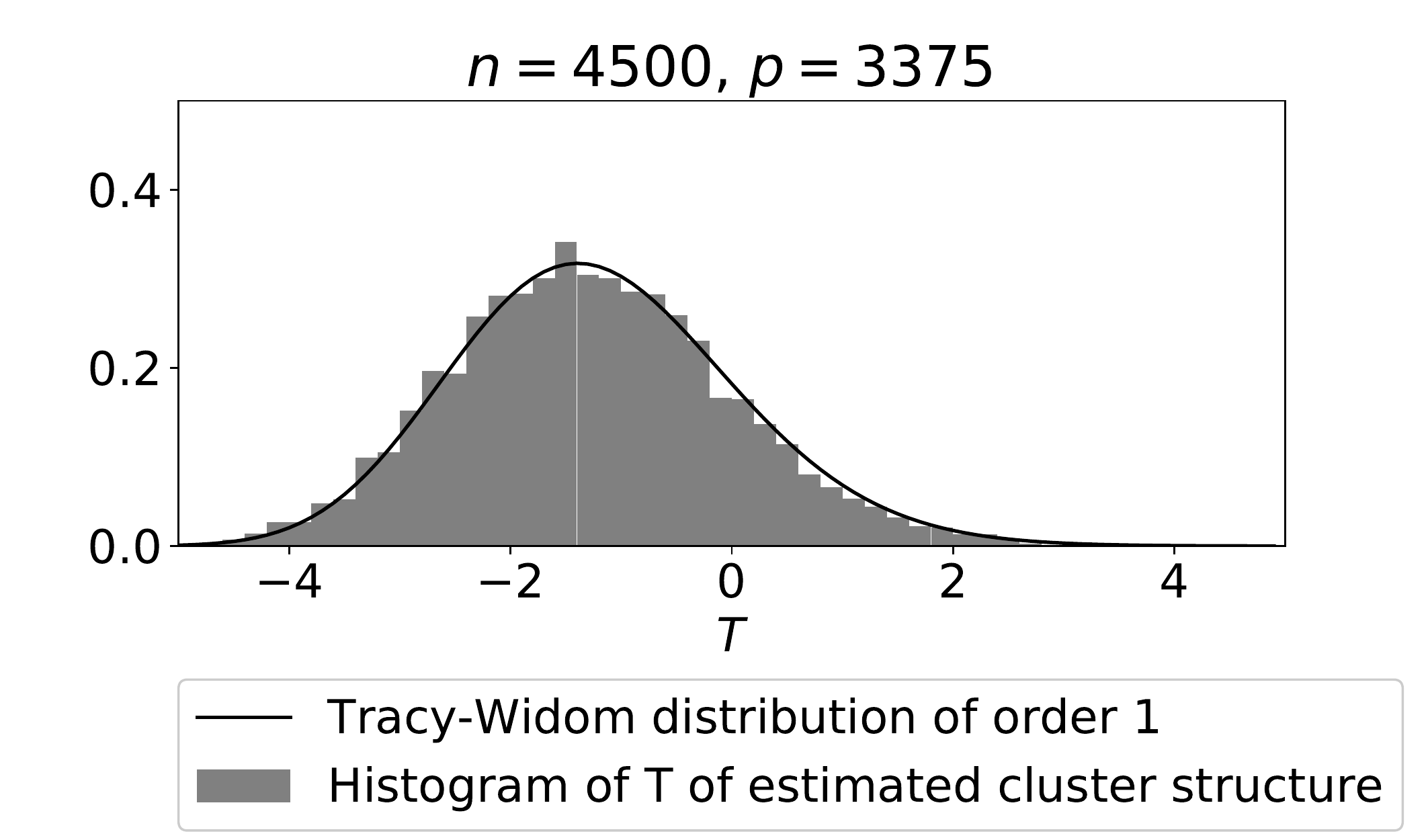}
  \includegraphics[width=0.2\hsize]{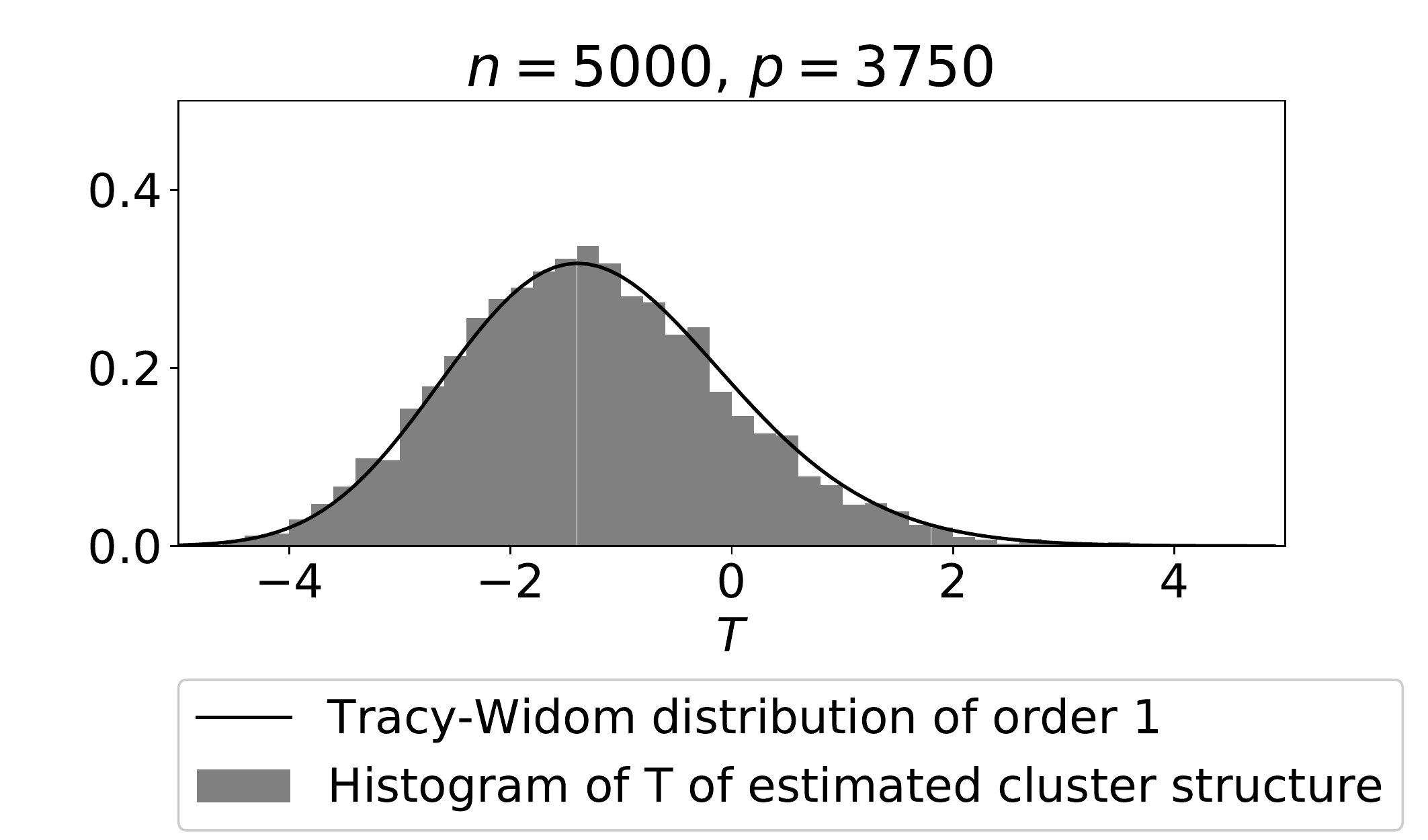}
  \caption{Histogram of the proposed test statistic $T$, which was computed with \textbf{estimated} bicluster structure (\textbf{Gaussian case}). The titles of the figures show the different matrix sizes.}\vspace{3mm}
  \label{fig:preliminary_hist_normal}
  \includegraphics[width=0.2\hsize]{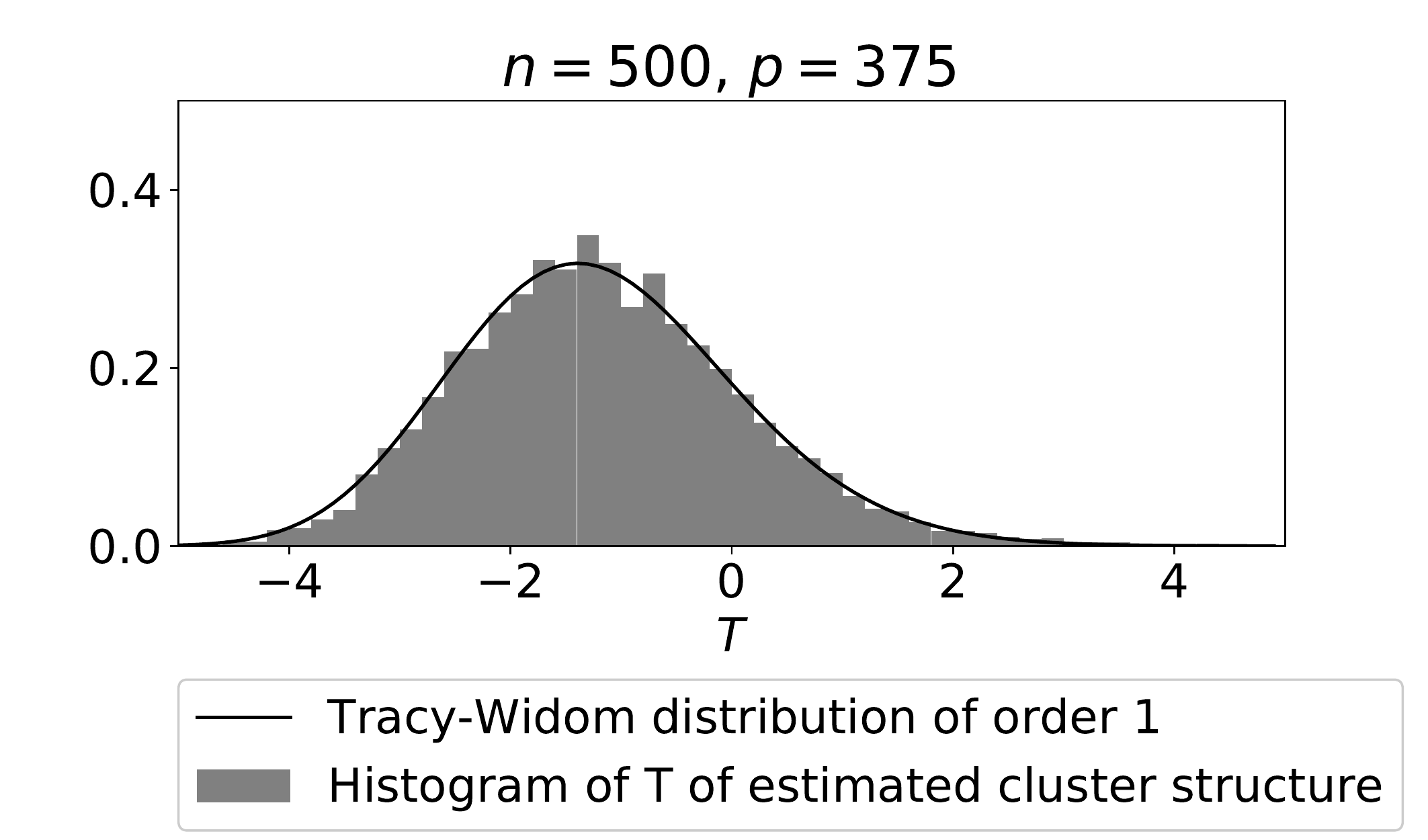}
  \includegraphics[width=0.2\hsize]{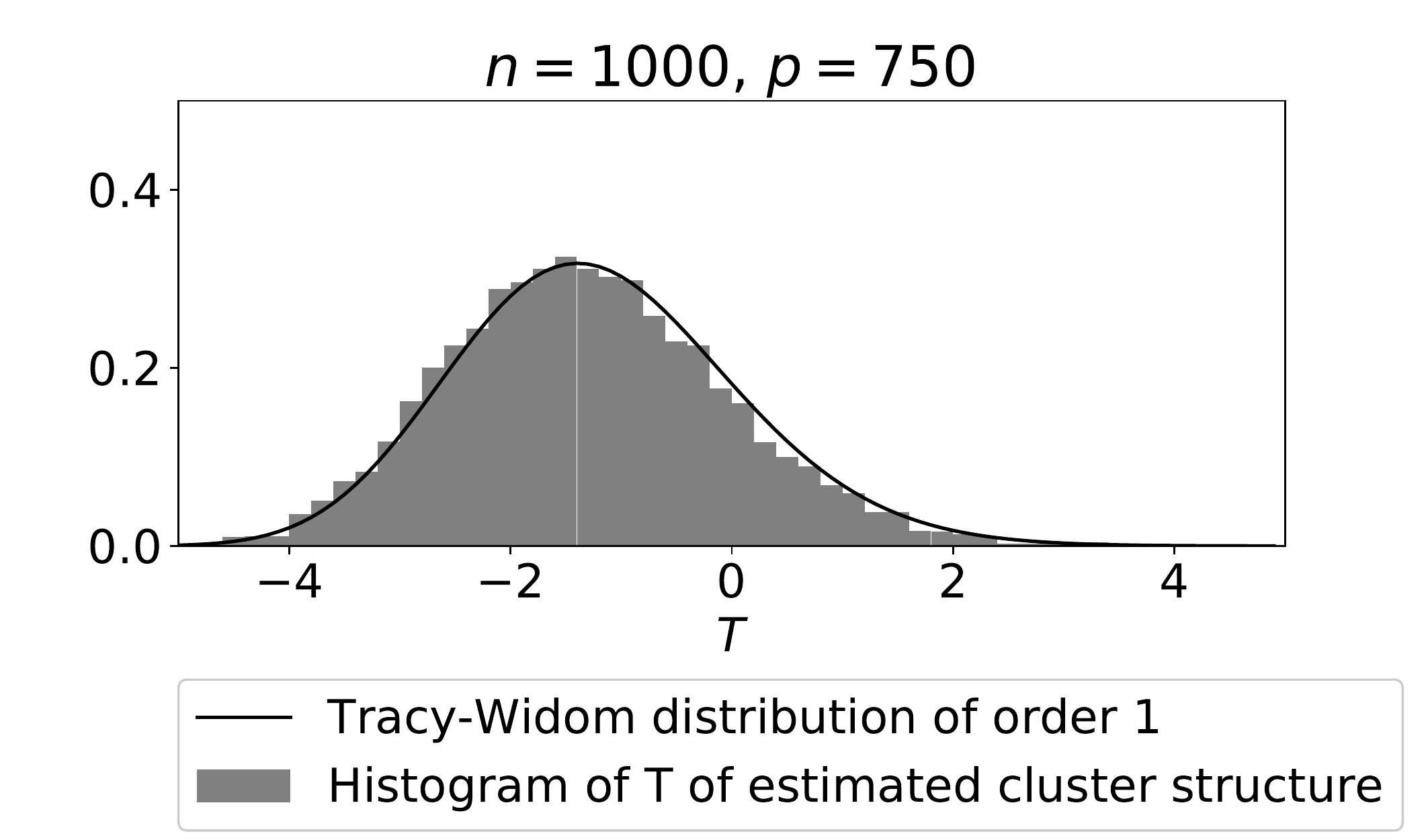}
  \includegraphics[width=0.2\hsize]{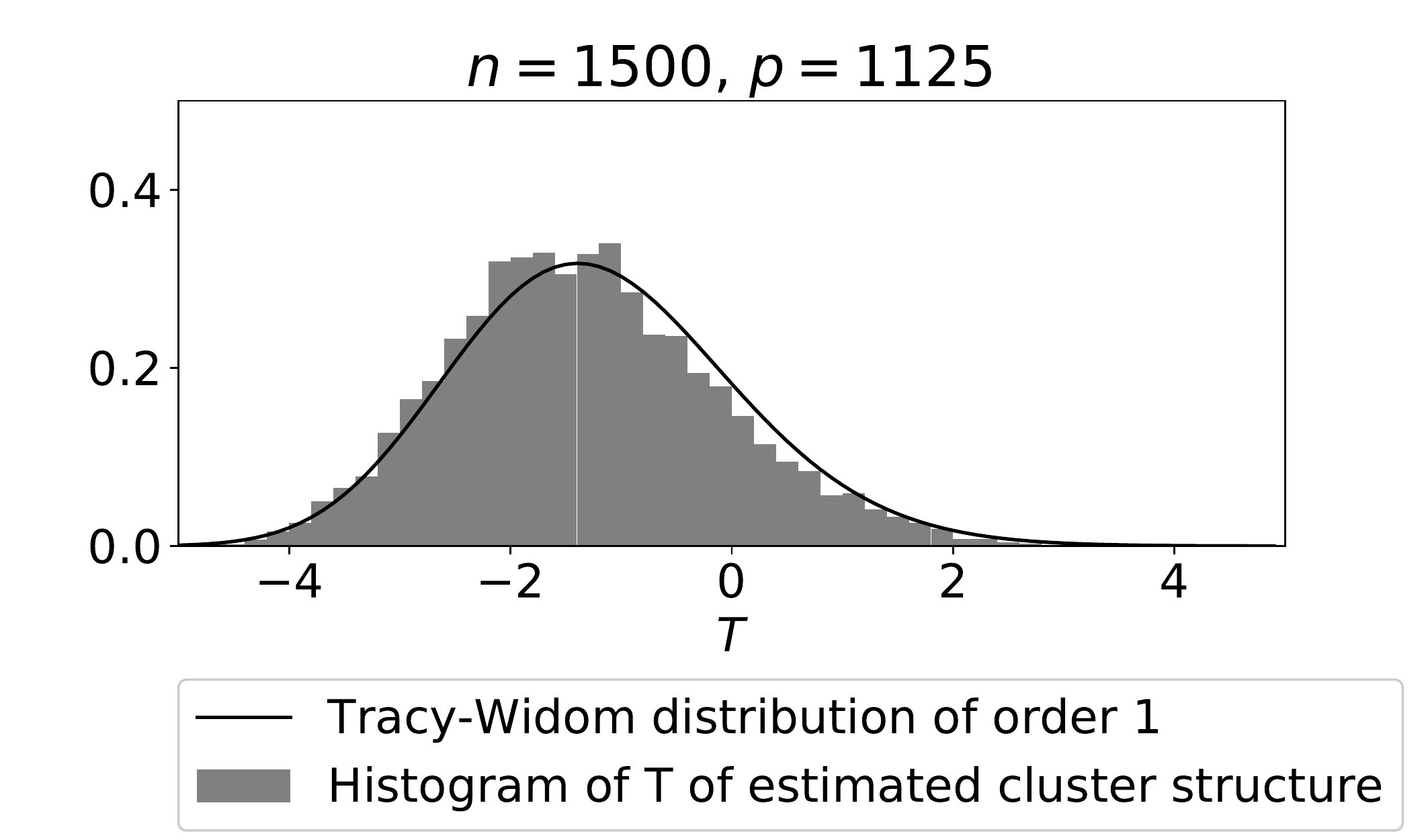}
  \includegraphics[width=0.2\hsize]{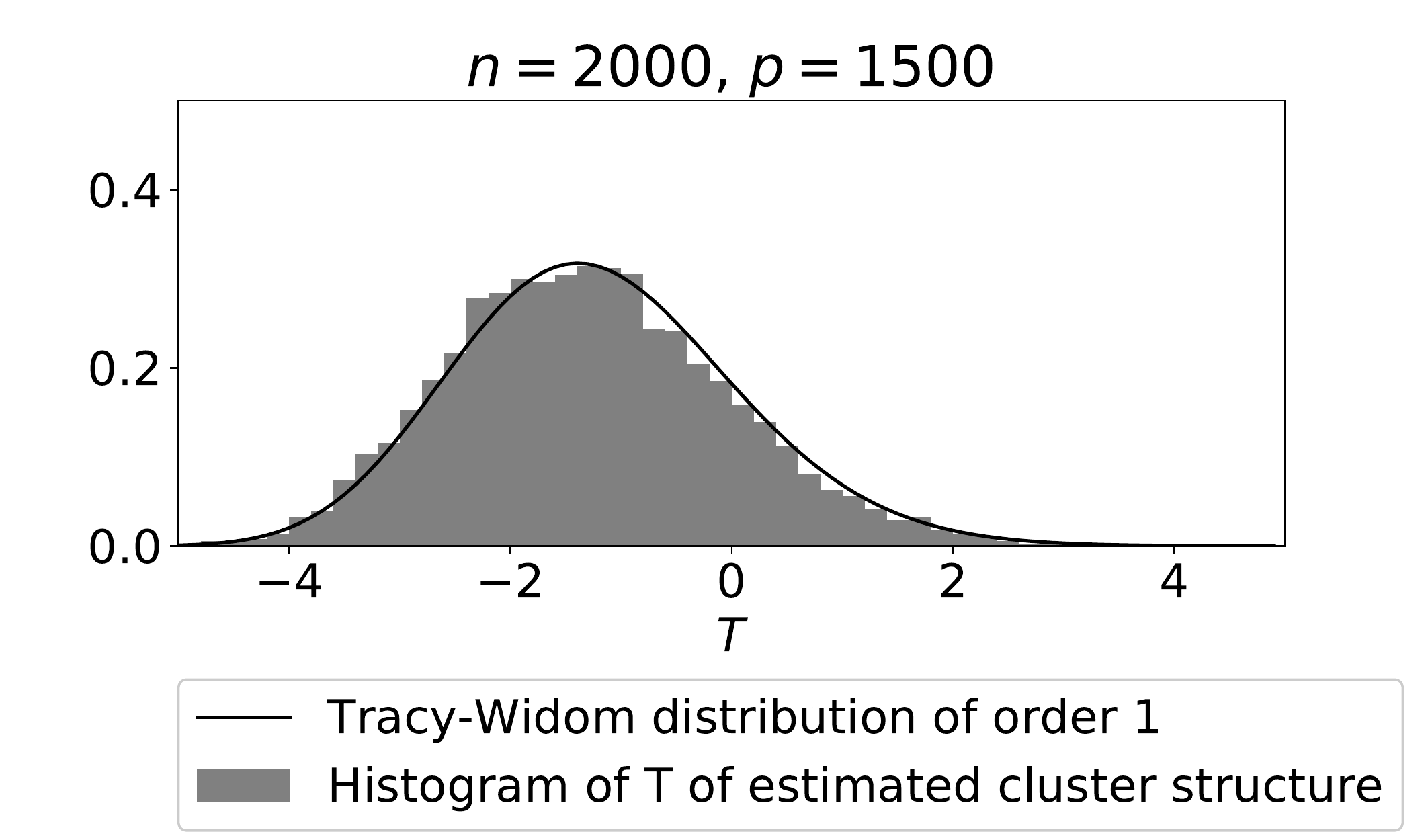}
  \includegraphics[width=0.2\hsize]{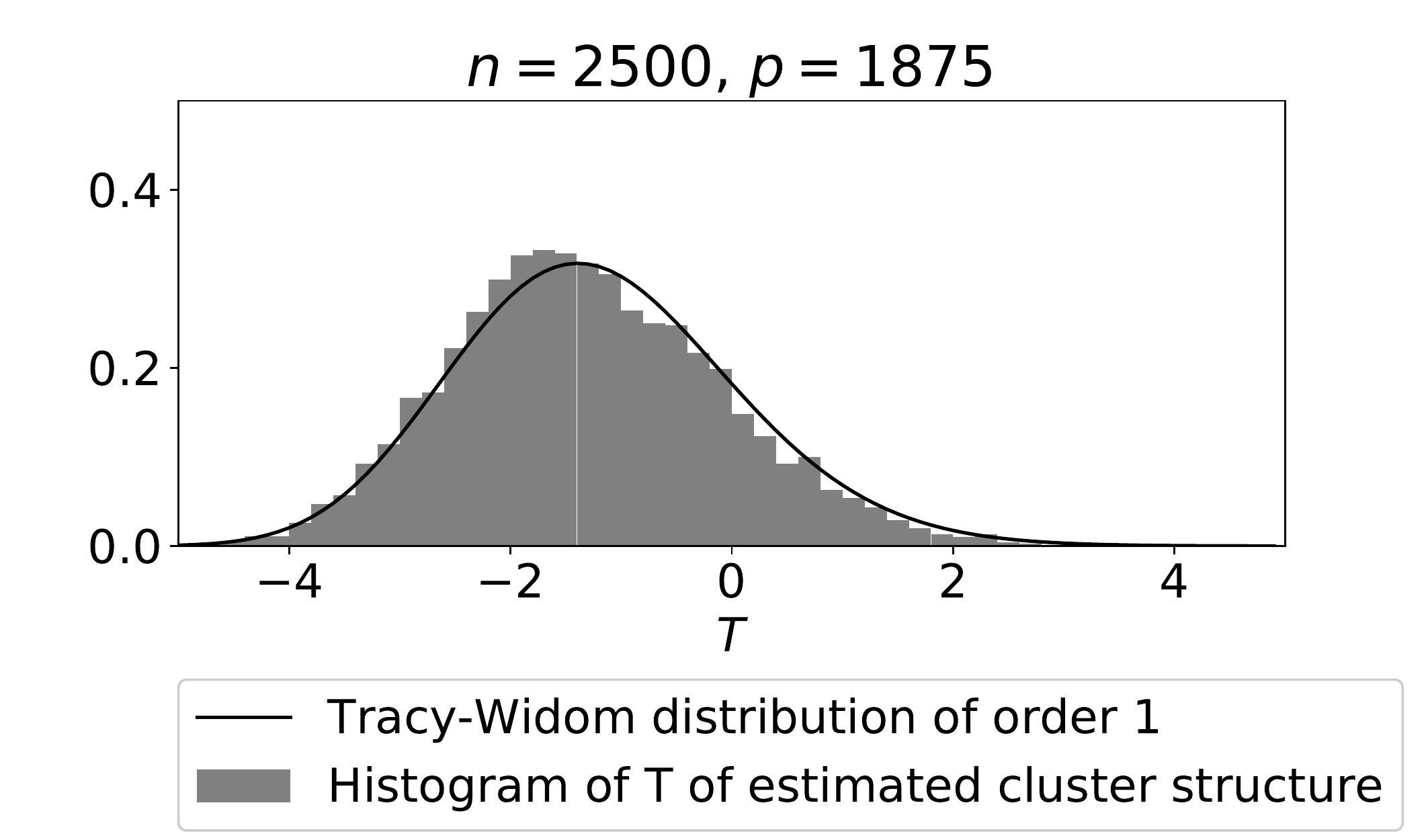}
  \includegraphics[width=0.2\hsize]{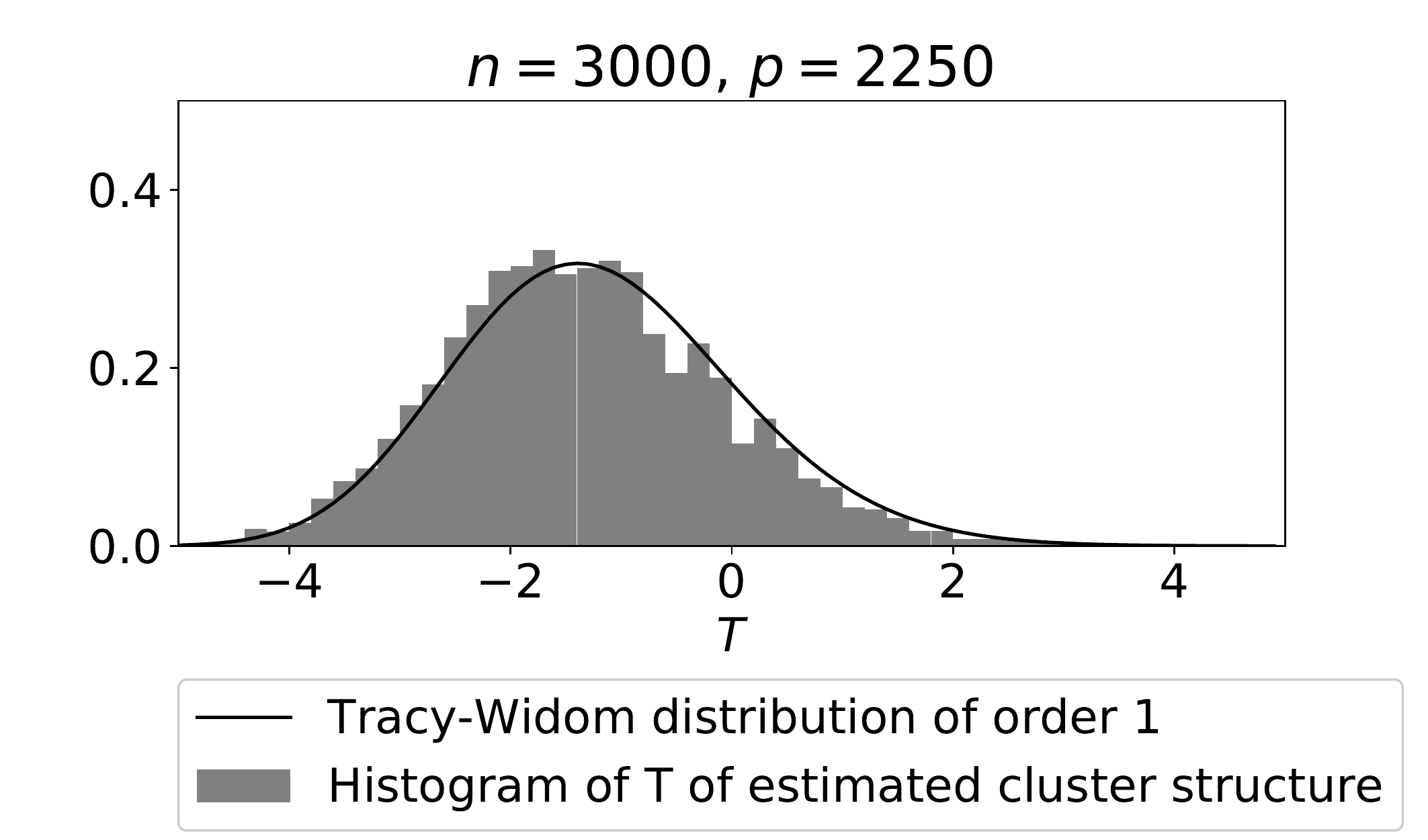}
  \includegraphics[width=0.2\hsize]{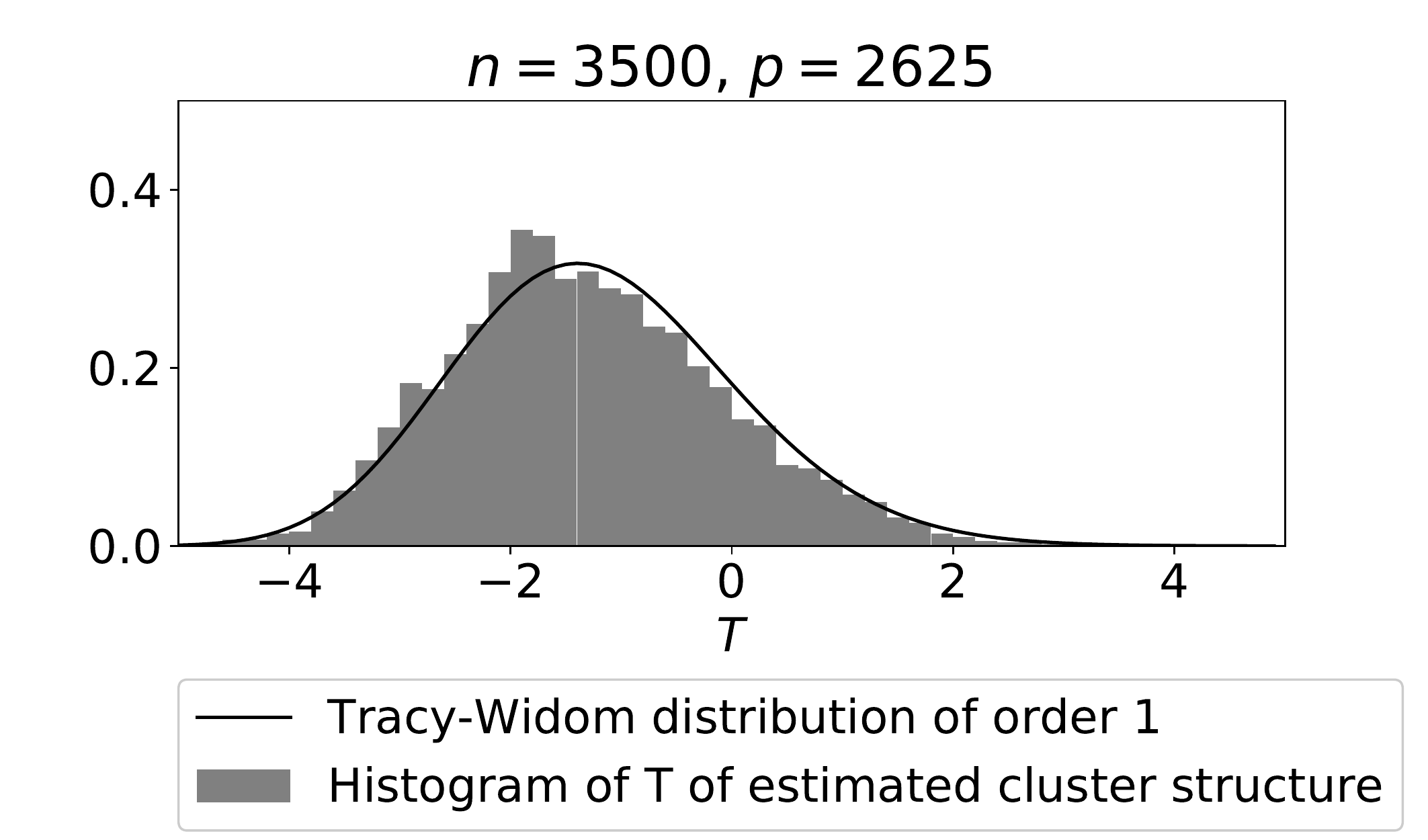}
  \includegraphics[width=0.2\hsize]{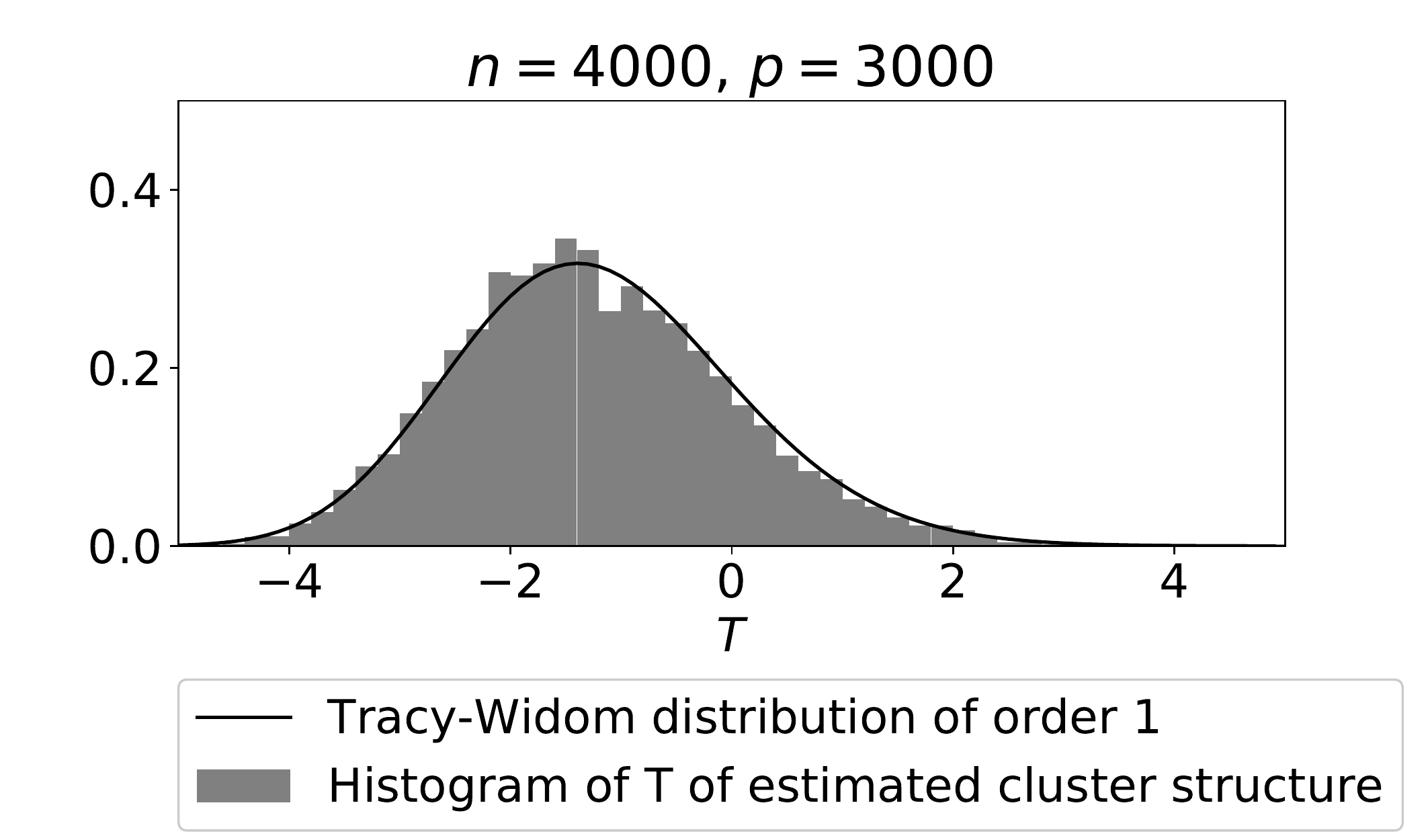}
  \includegraphics[width=0.2\hsize]{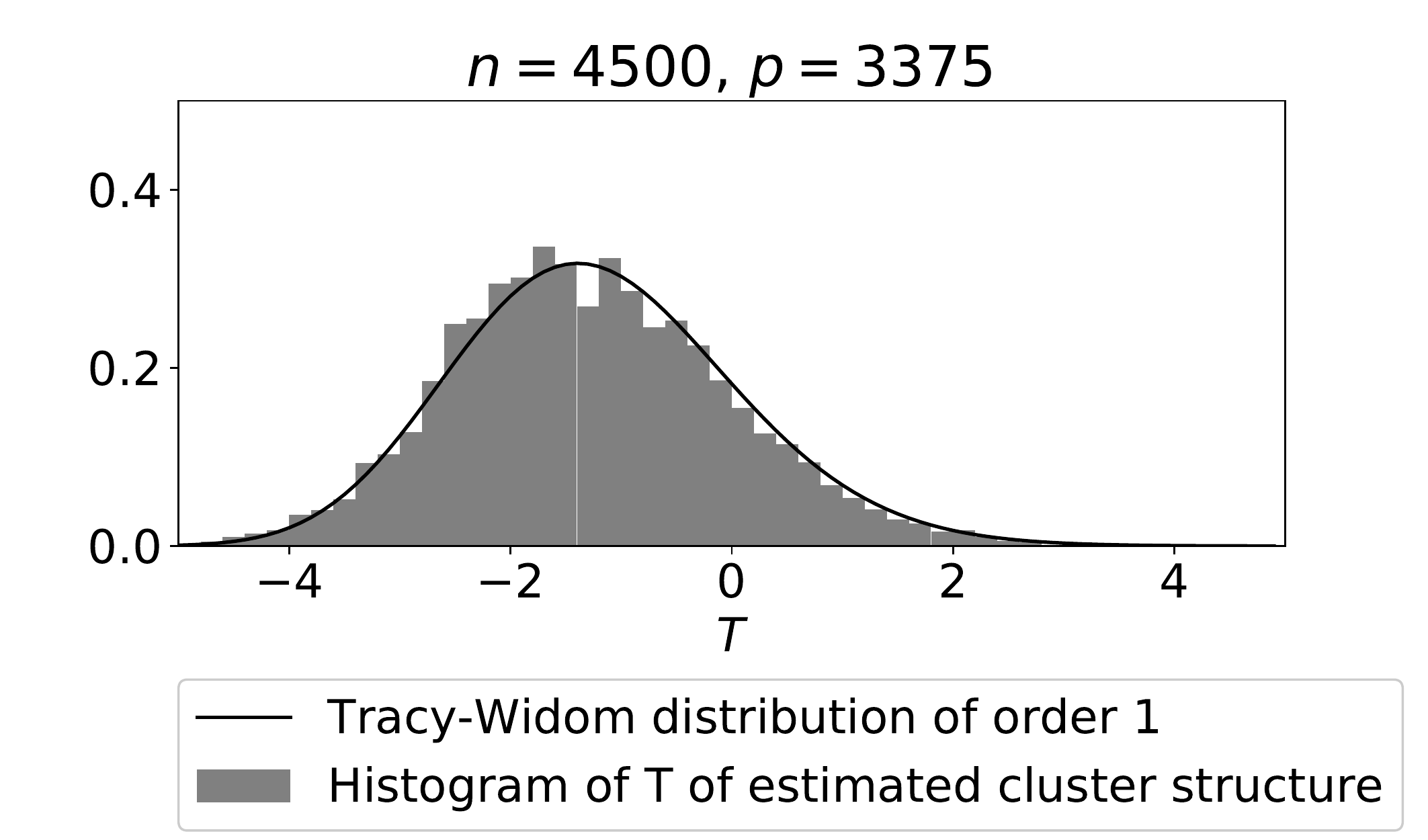}
  \includegraphics[width=0.2\hsize]{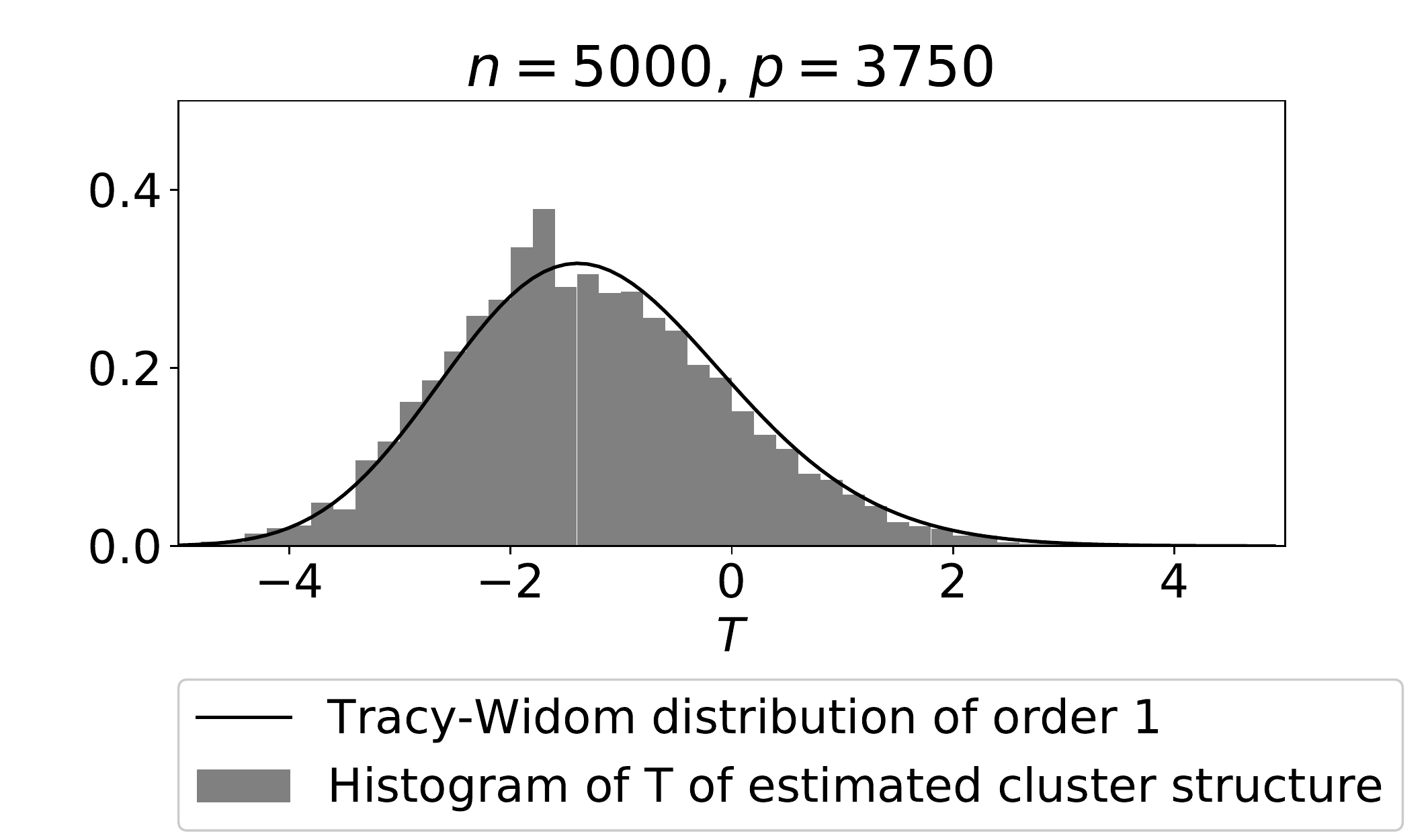}
  \caption{Histogram of the proposed test statistic $T$, which was computed with \textbf{estimated} bicluster structure (\textbf{Bernoulli case}). The titles of the figures show the different matrix sizes.}\vspace{3mm}
  \label{fig:preliminary_hist_bernoulli}
  \centering
  \includegraphics[width=0.2\hsize]{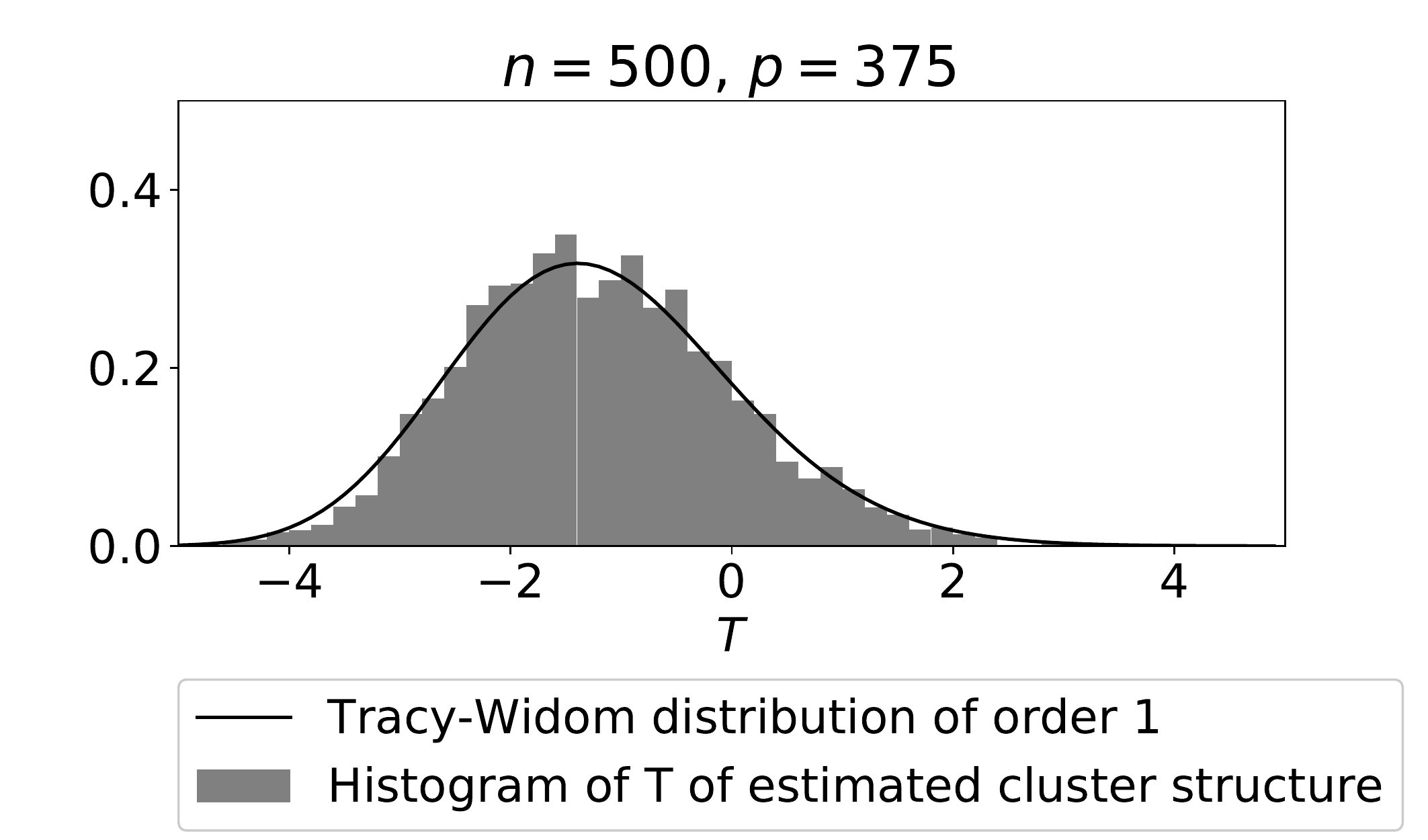}
  \includegraphics[width=0.2\hsize]{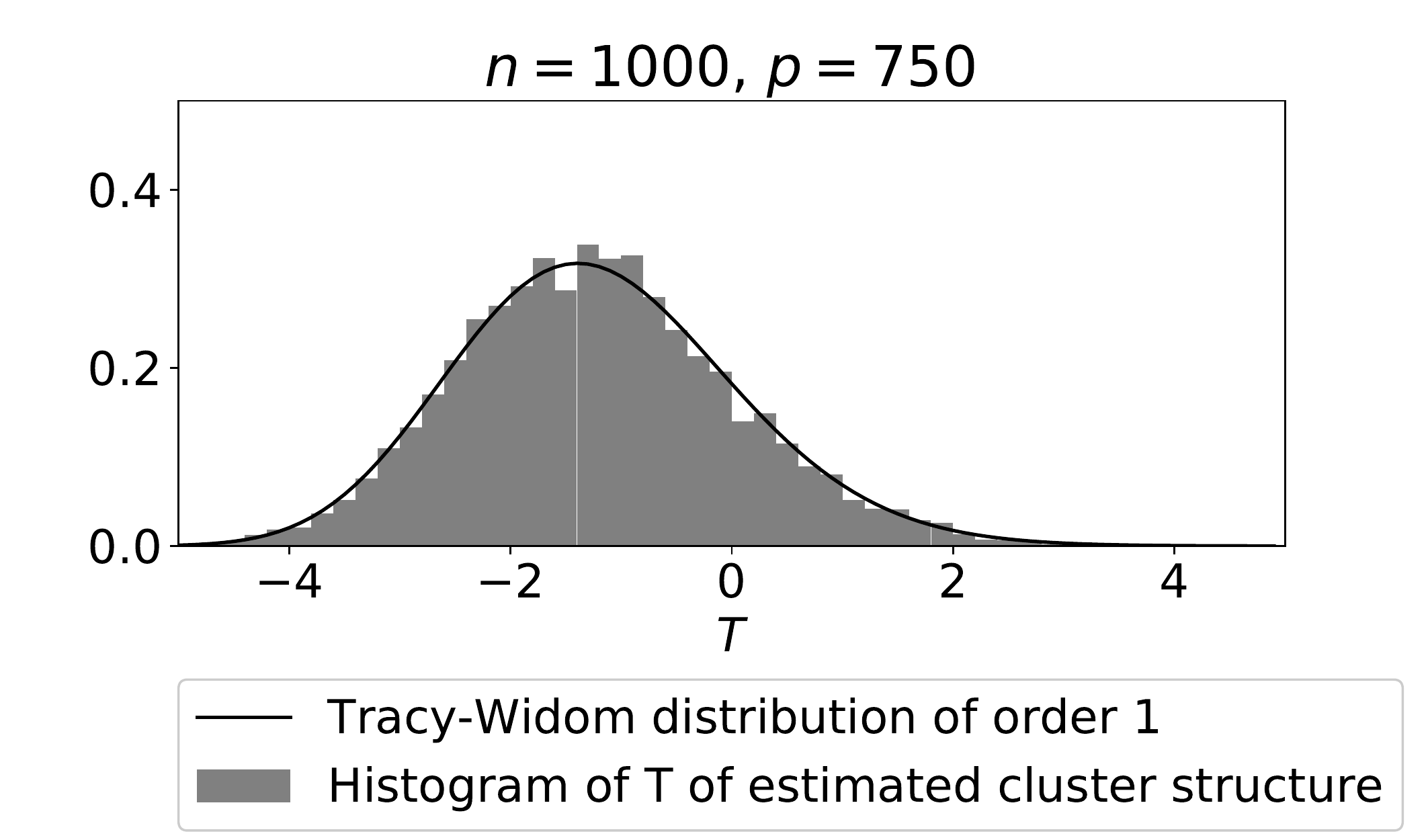}
  \includegraphics[width=0.2\hsize]{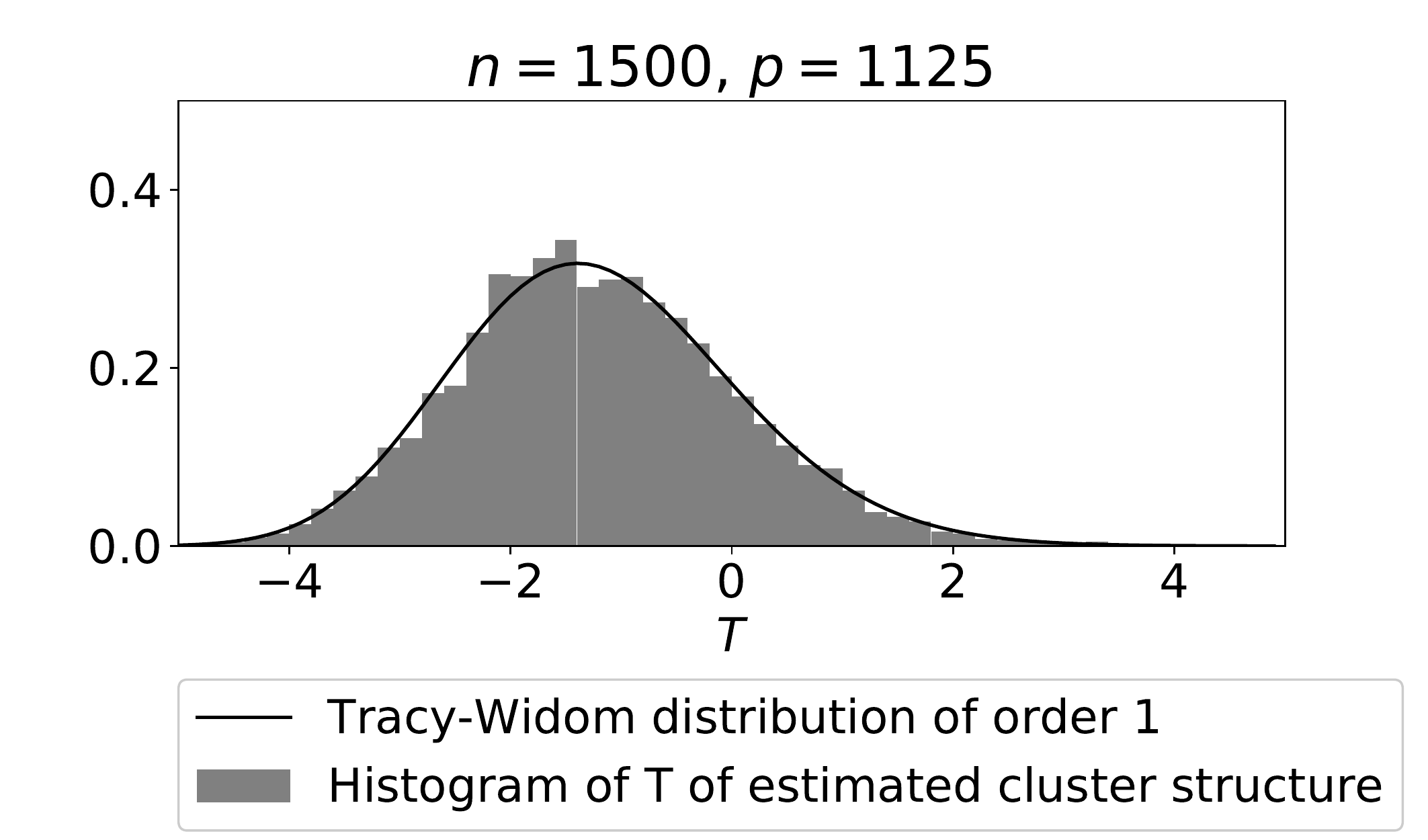}
  \includegraphics[width=0.2\hsize]{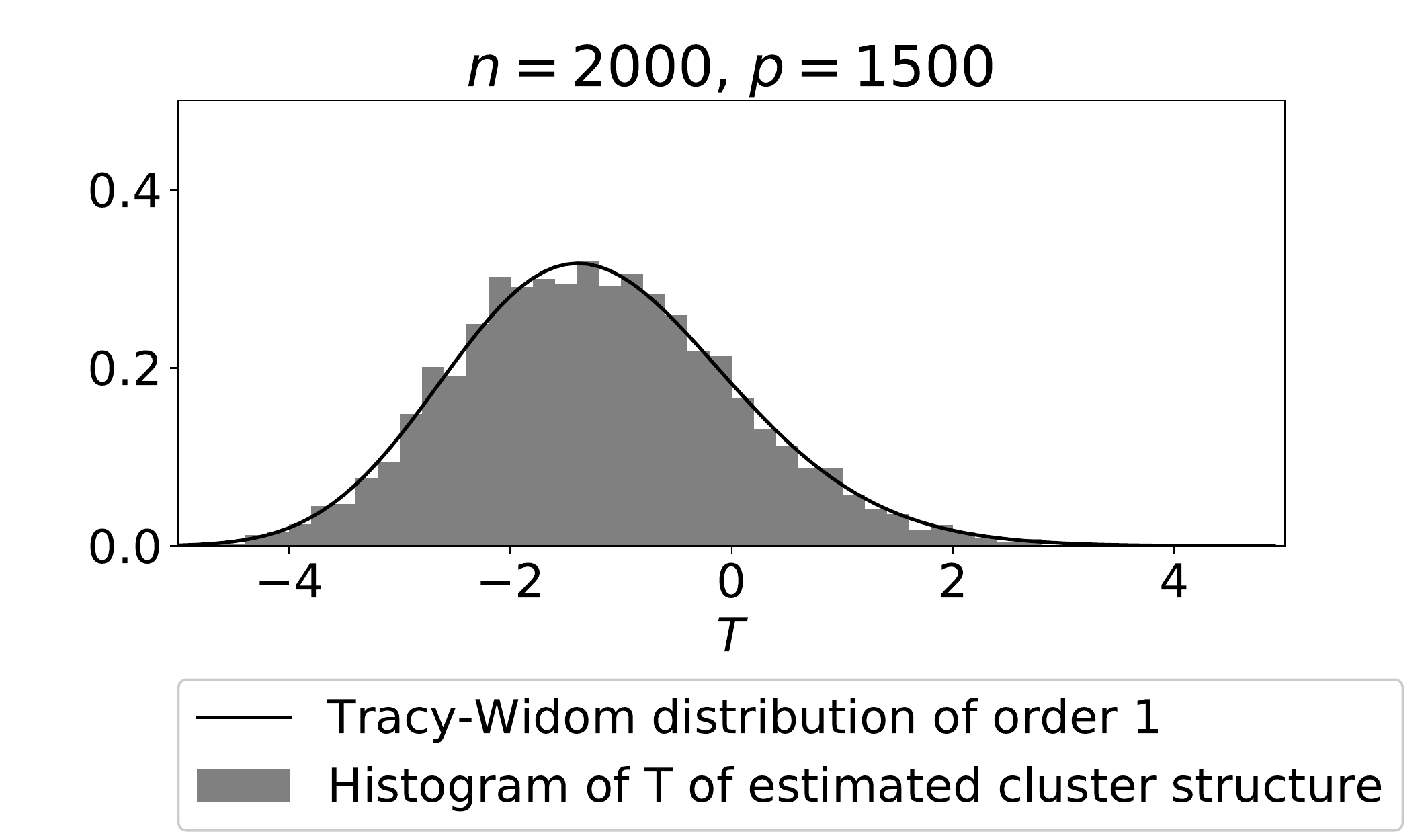}
  \includegraphics[width=0.2\hsize]{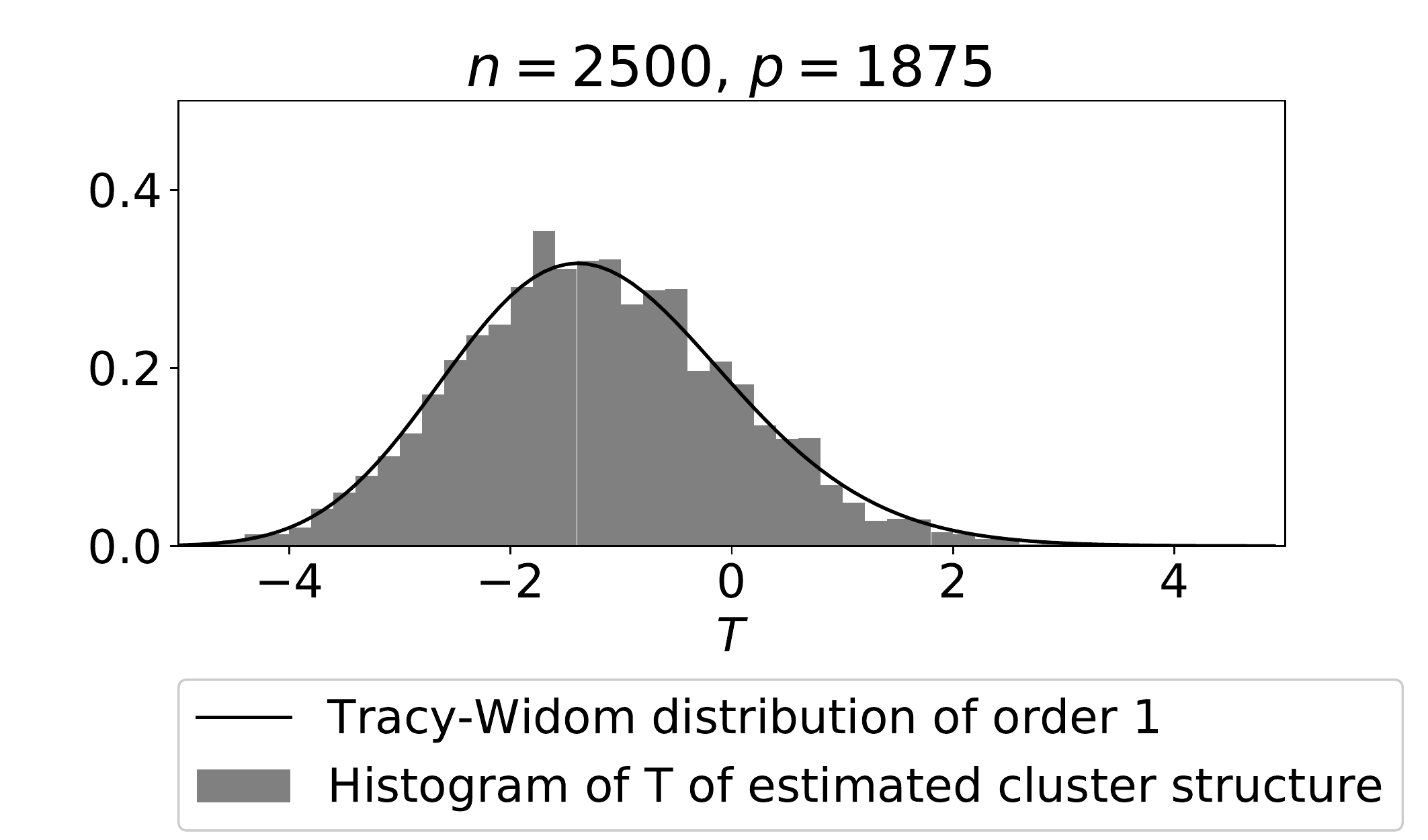}
  \includegraphics[width=0.2\hsize]{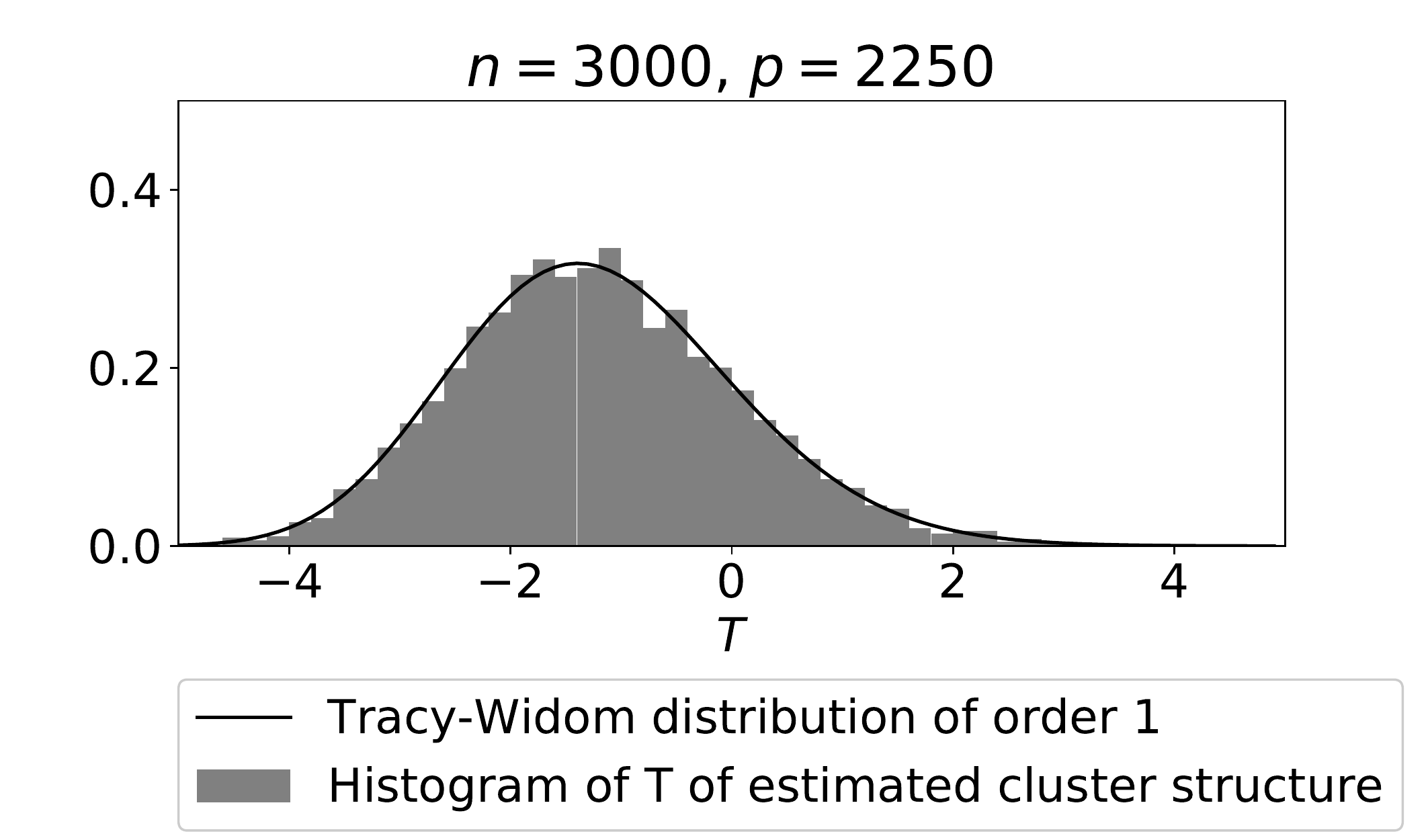}
  \includegraphics[width=0.2\hsize]{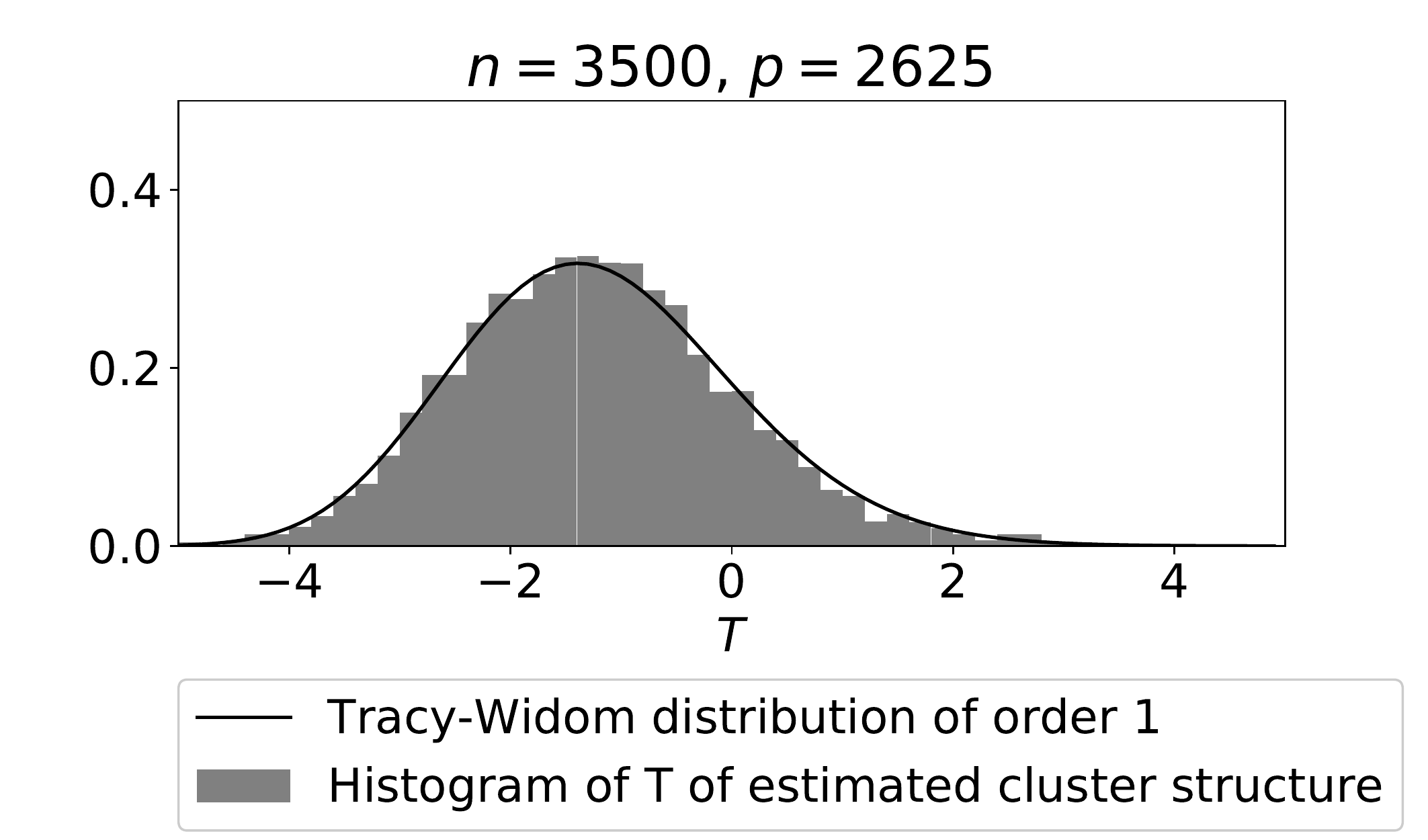}
  \includegraphics[width=0.2\hsize]{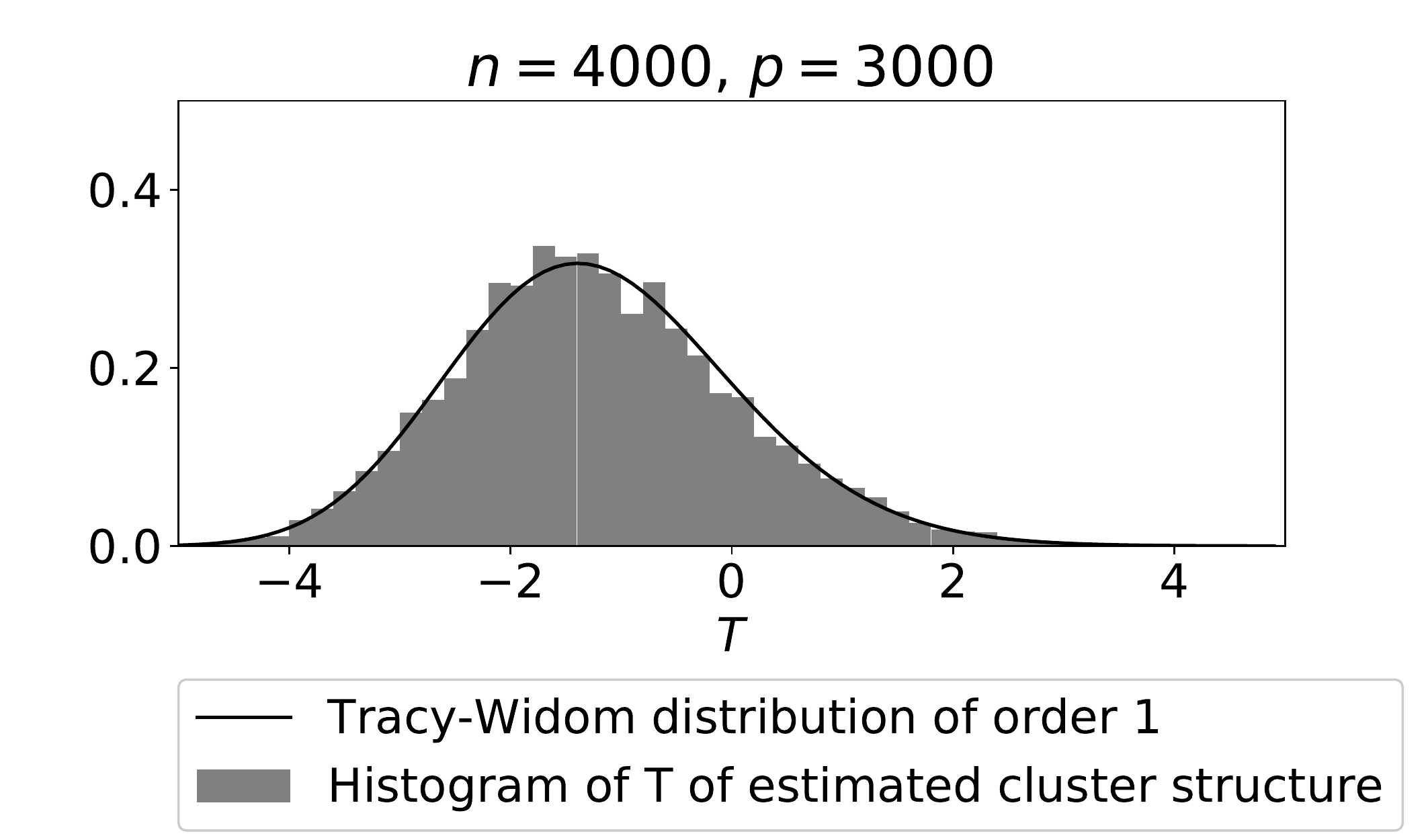}
  \includegraphics[width=0.2\hsize]{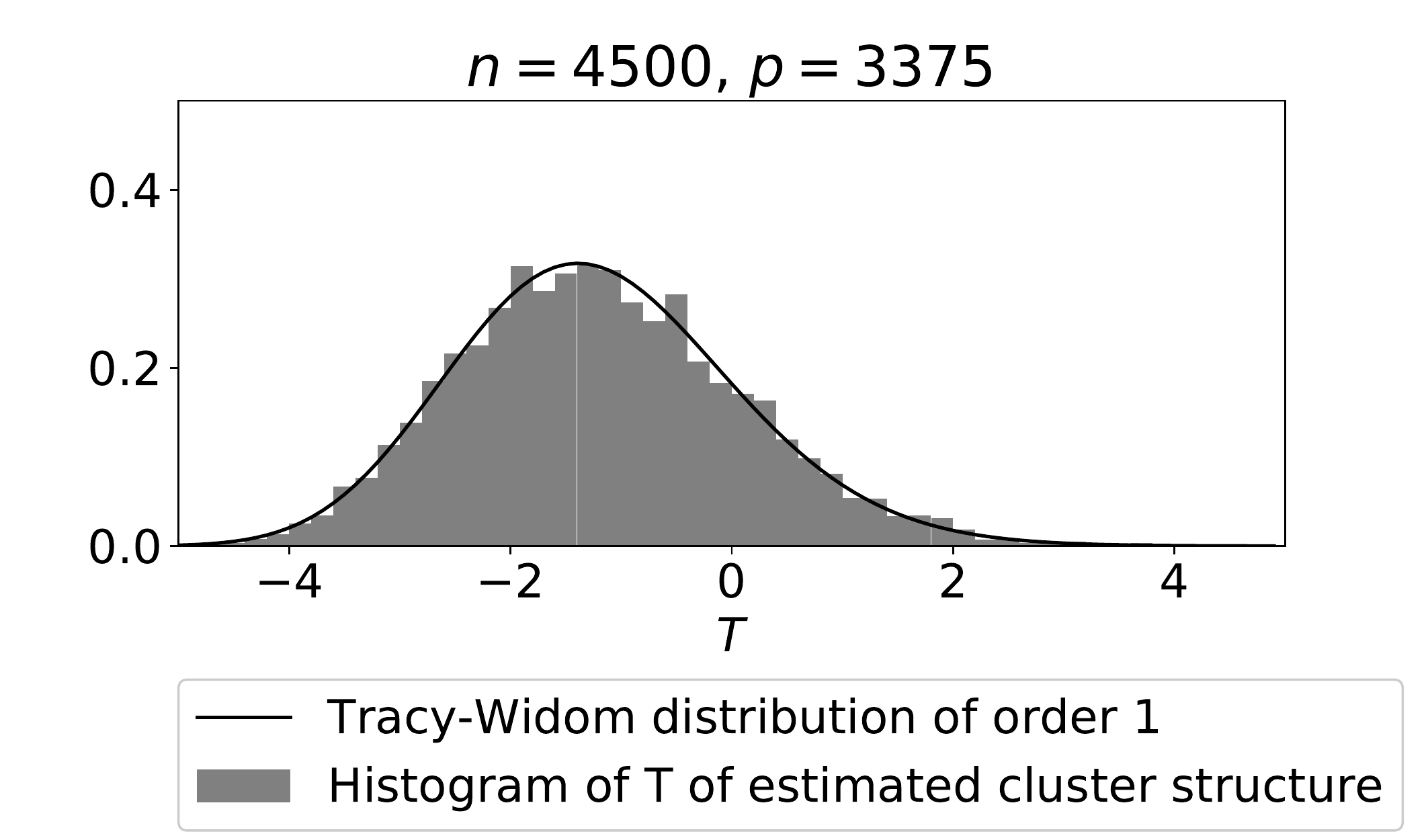}
  \includegraphics[width=0.2\hsize]{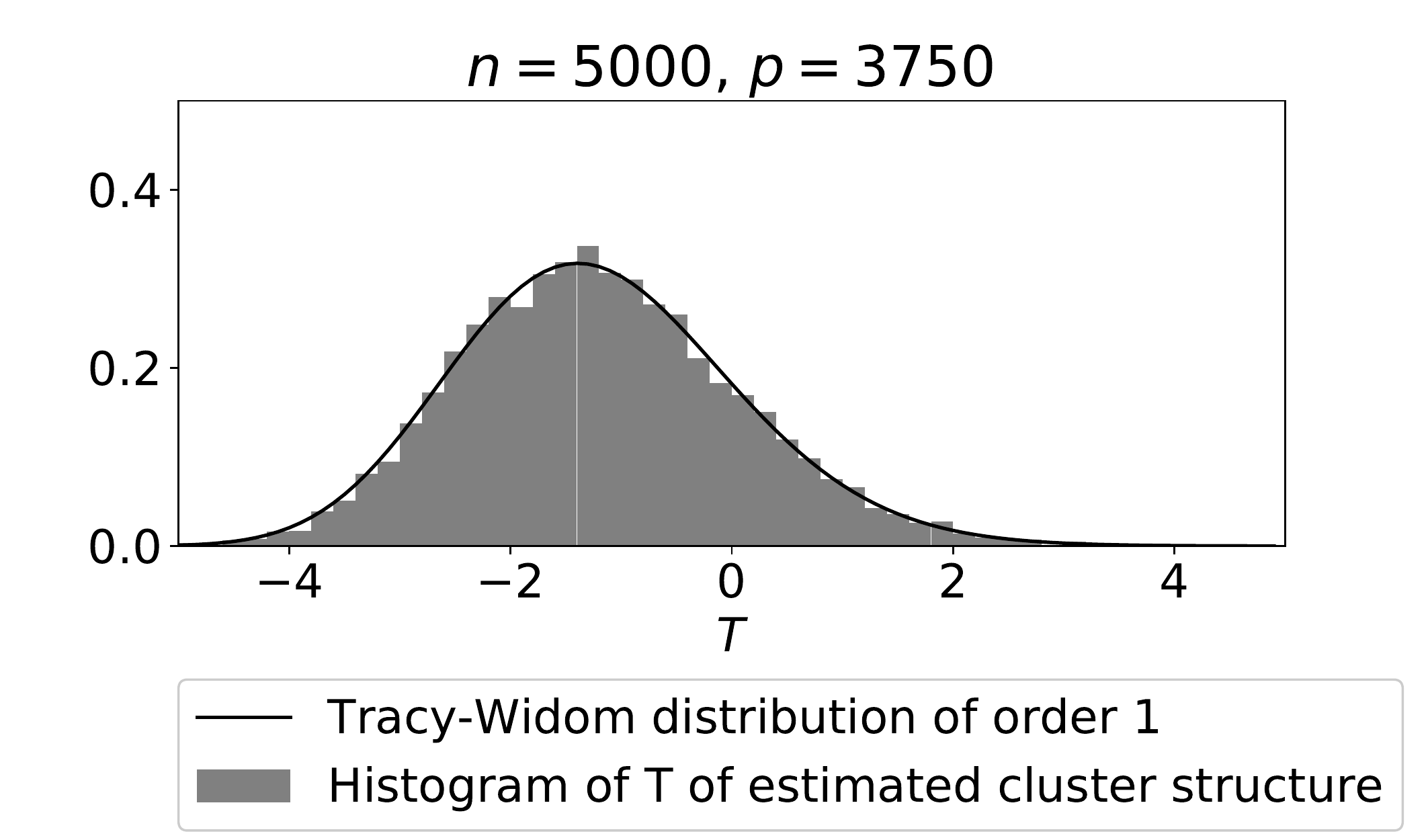}
  \caption{Histogram of the proposed test statistic $T$, which was computed with \textbf{estimated} bicluster structure (\textbf{Poisson case}). The titles of the figures show the different matrix sizes.}
  \label{fig:preliminary_hist_poisson}
\end{figure}
\begin{figure}[p]
  \centering
  \includegraphics[width=0.31\hsize]{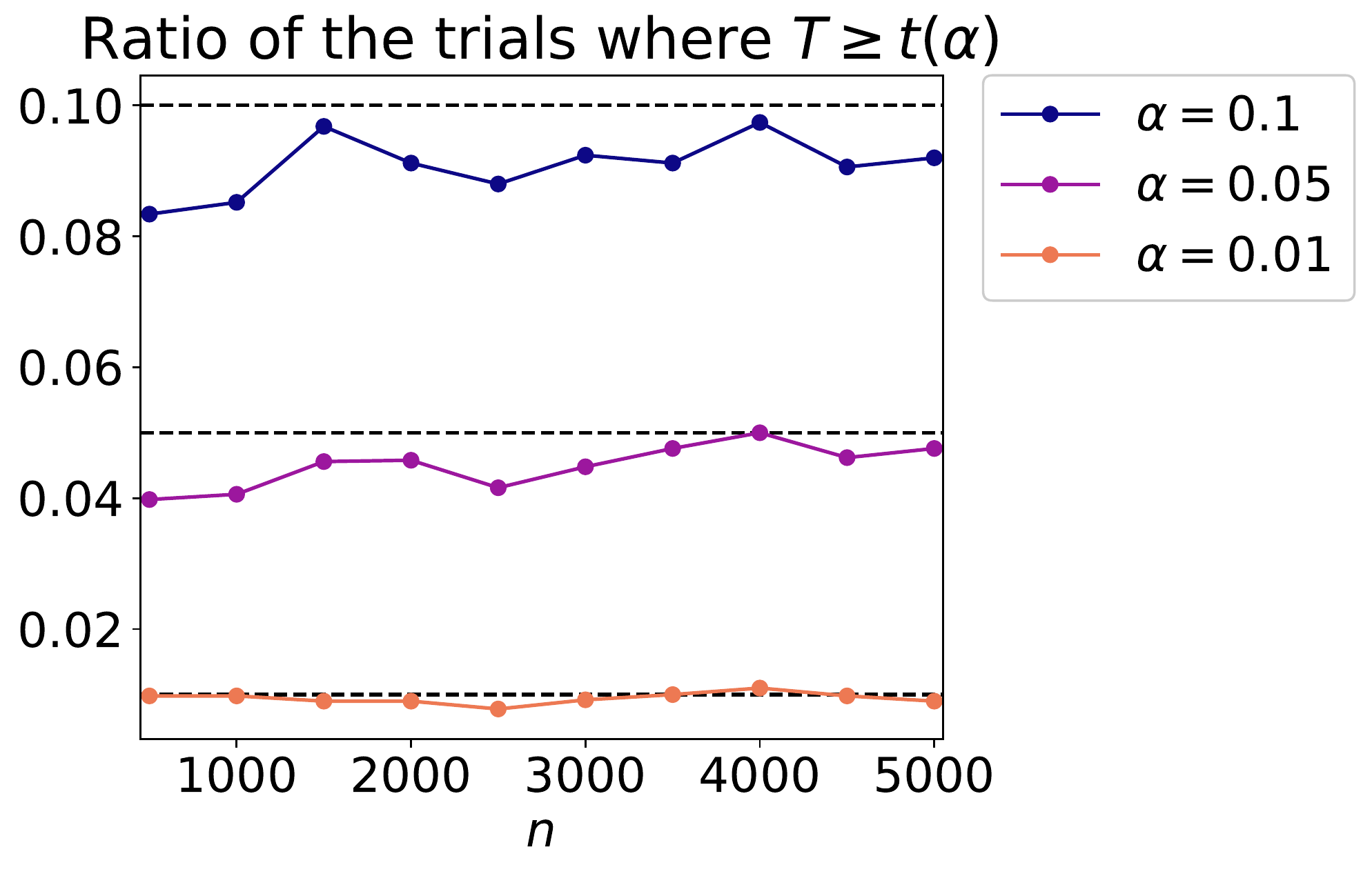}
  \includegraphics[width=0.31\hsize]{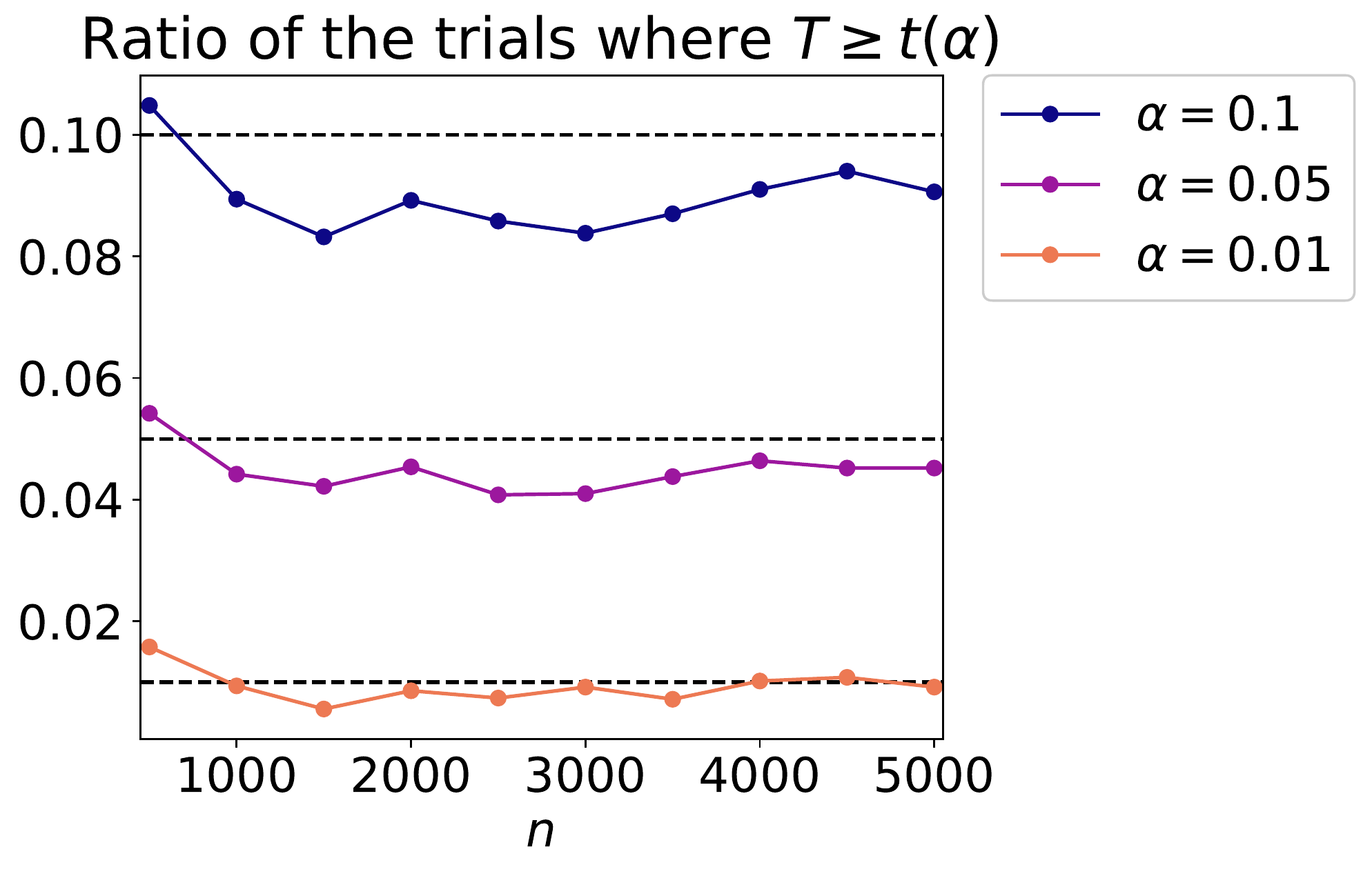}
  \includegraphics[width=0.31\hsize]{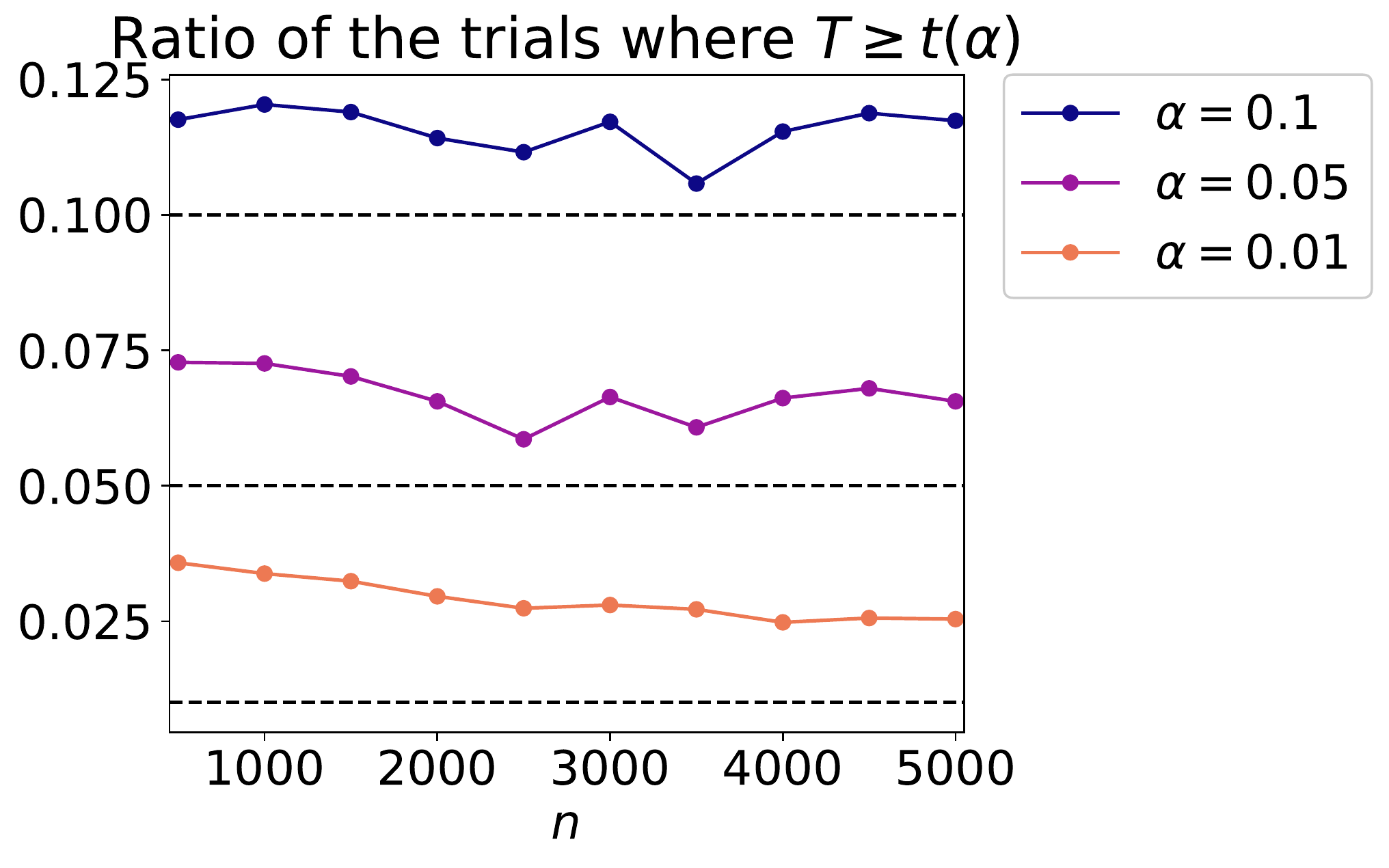}
  \caption{The empirical tail probabilities of the proposed test statistic $T$ under the three settings of distributions, which was computed with \textbf{estimated bicluster structure}. The left, center, and right figures, respectively, show the results where each entry of the observed matrix $A$ was generated using Gaussian, Bernoulli, and Poisson distributions. The horizontal line indicates the row size $n$ of matrix $A$, and the dashed lines indicate the three significance levels.}\vspace{3mm}
  \label{fig:preliminaryT}
  \includegraphics[width=0.31\hsize]{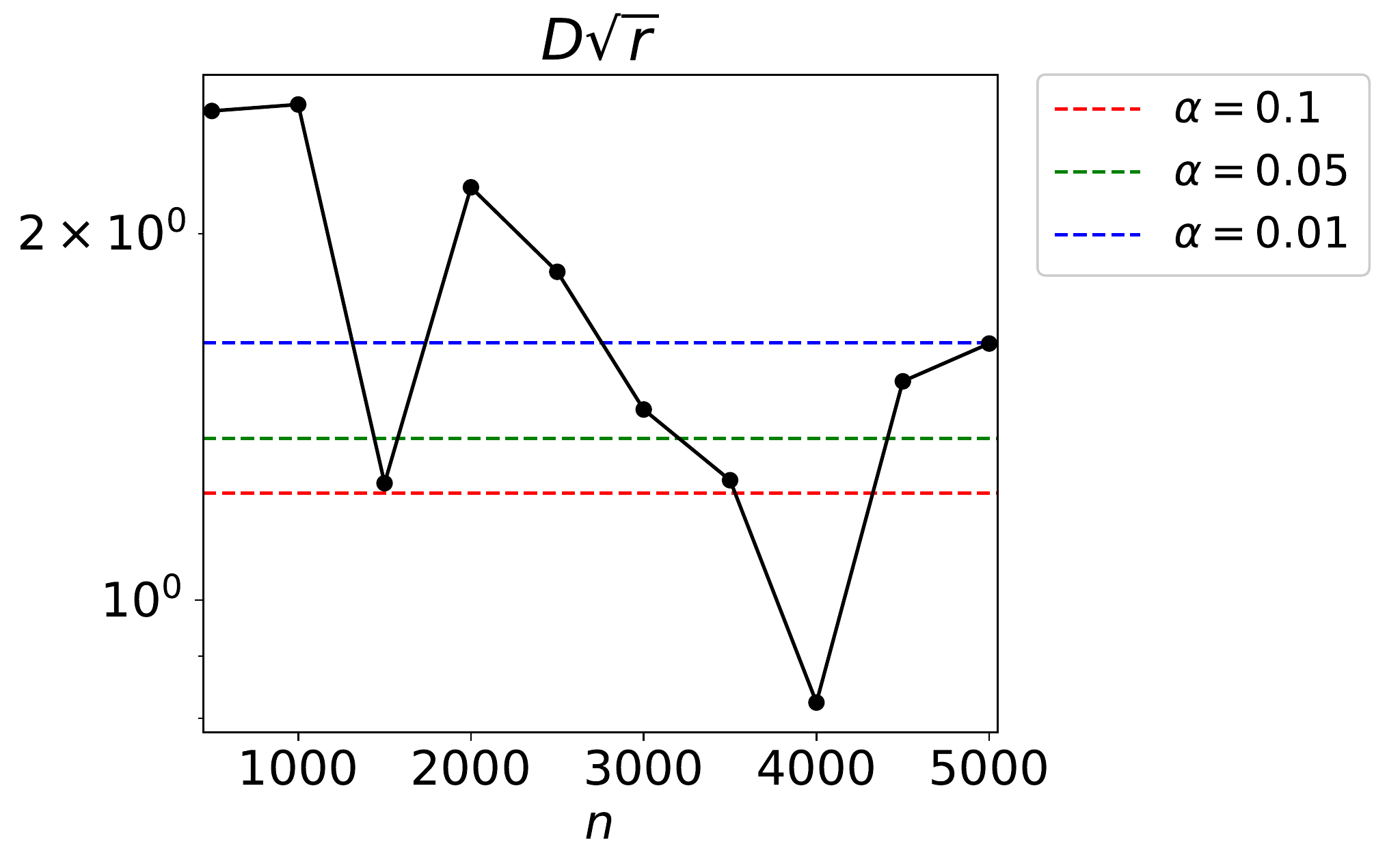}
  \includegraphics[width=0.31\hsize]{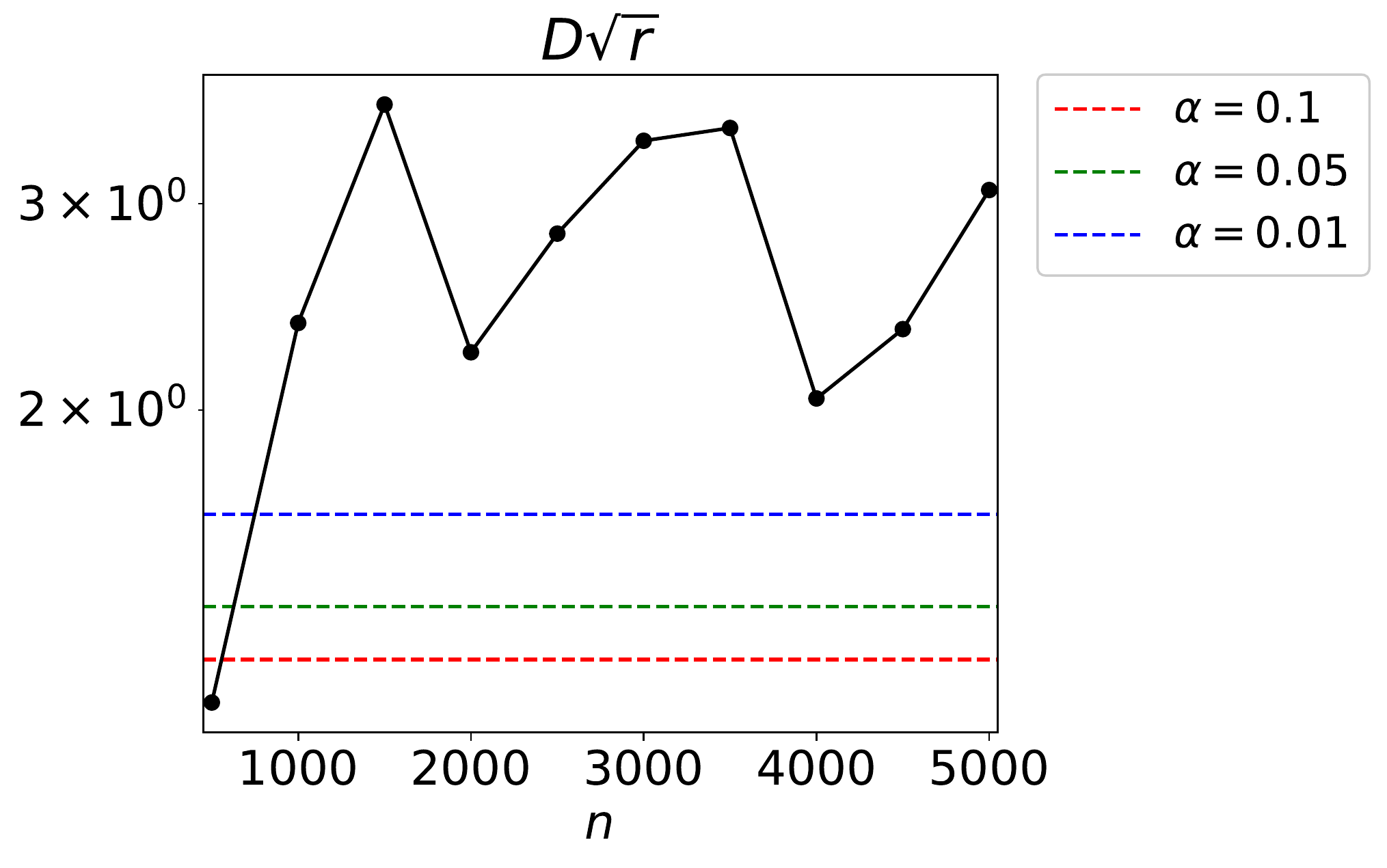}
  \includegraphics[width=0.31\hsize]{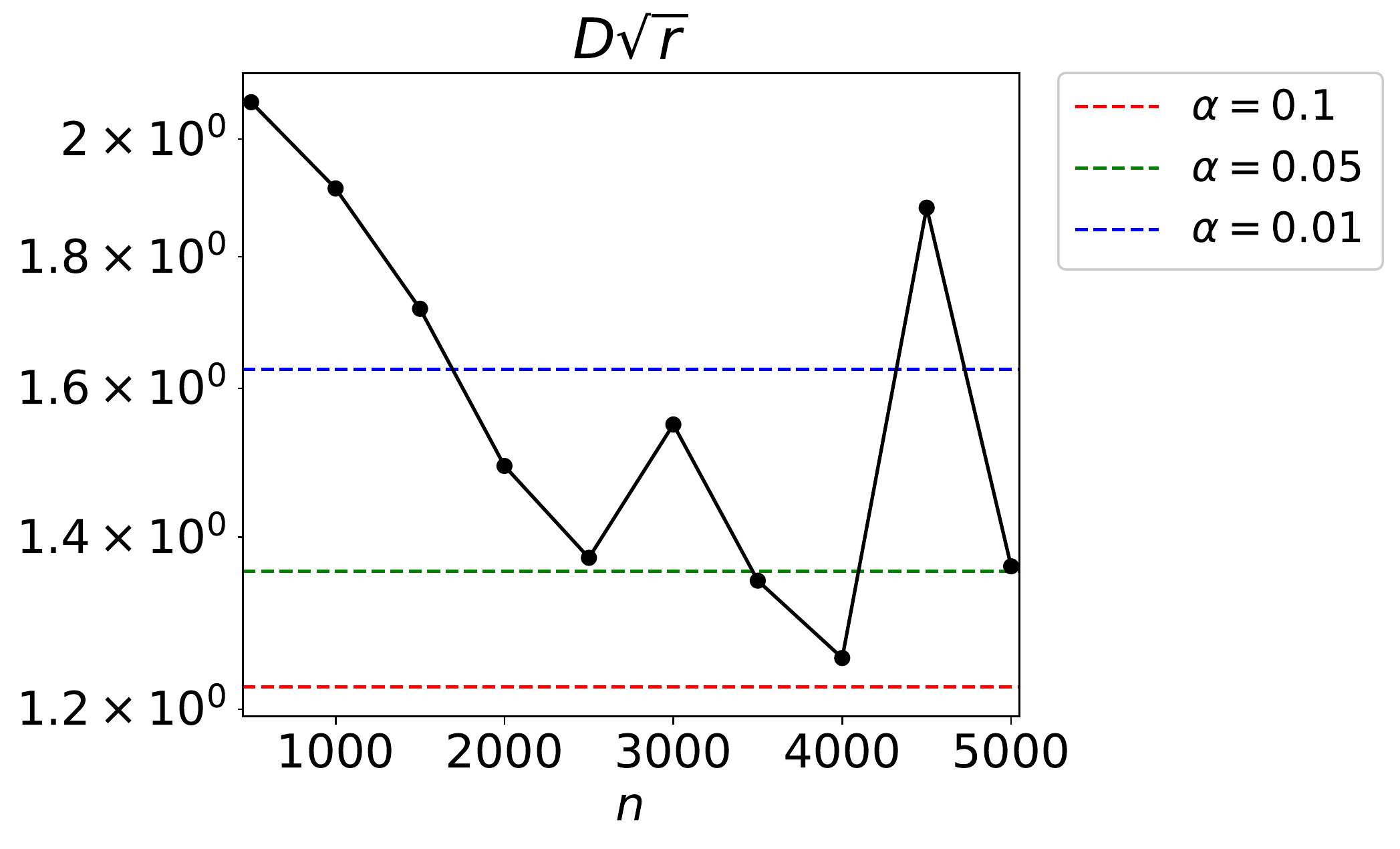}
  \caption{Test statistics $D\sqrt{r}$ of the KS test \cite{Conover1999}, which was computed using an \textbf{estimated bicluster structure}. The left, center, and right figures, respectively, depict the results where each entry of the observed matrix $A$ was generated using Gaussian, Bernoulli, and Poisson distributions. Given a significance level $\alpha^{\mathrm{KS}}$ for the KS test, iff $D\sqrt{r} > \alpha^{\mathrm{KS}}$, then the null hypothesis that $T$ follows the $TW_1$ distribution is rejected.}\vspace{3mm}
  \label{fig:KStest}
  \includegraphics[width=0.31\hsize]{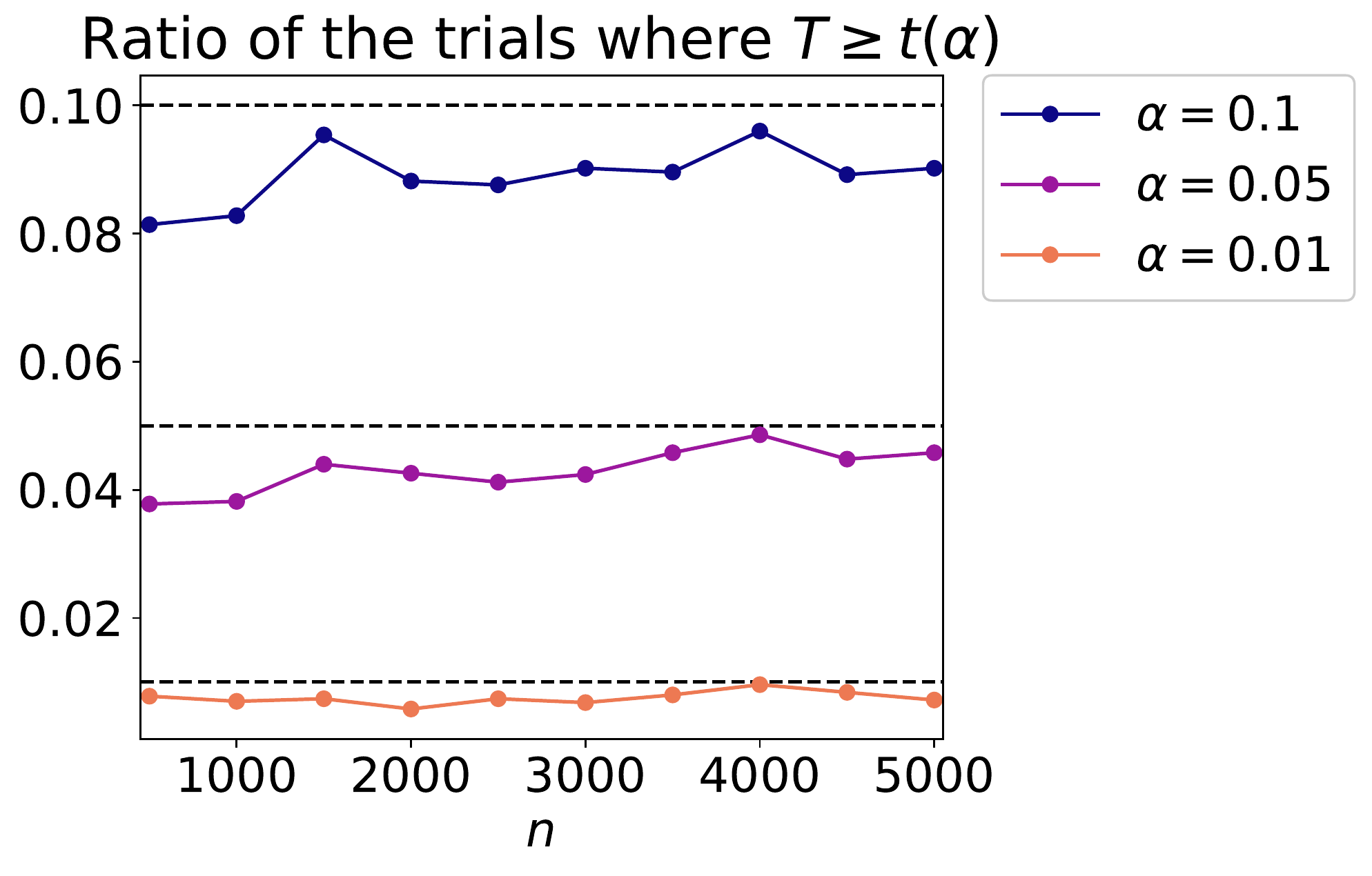}
  \includegraphics[width=0.31\hsize]{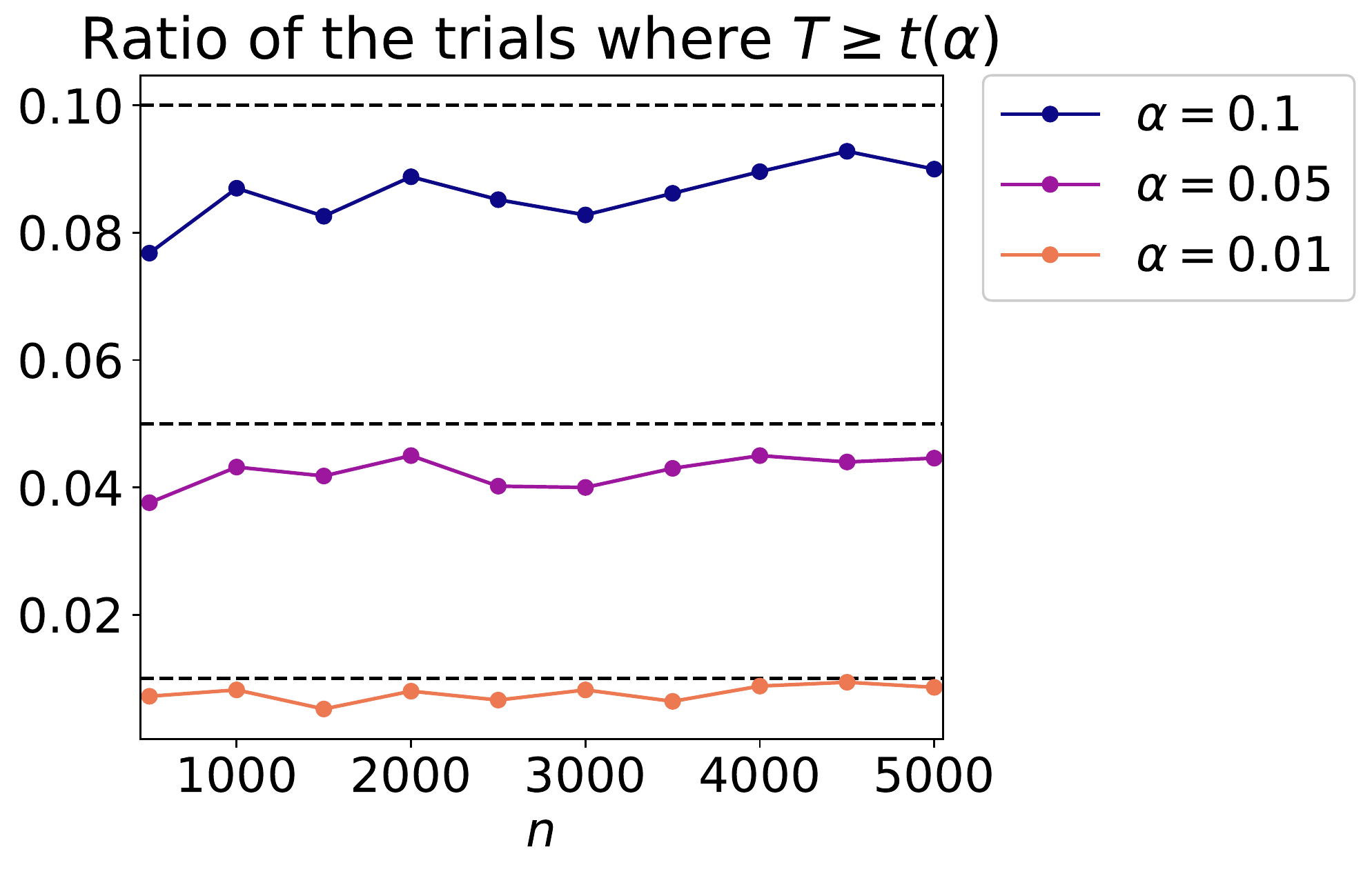}
  \includegraphics[width=0.31\hsize]{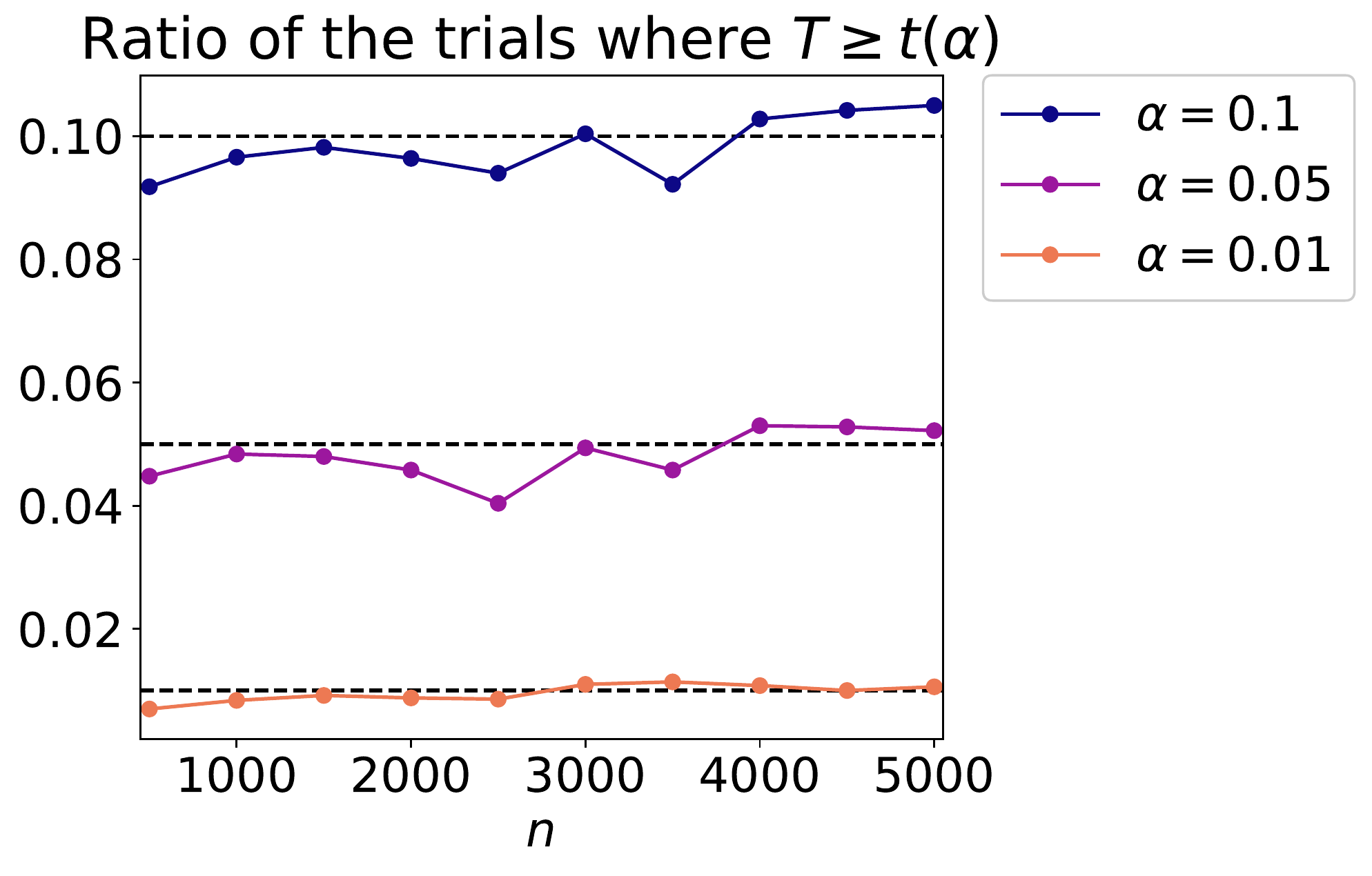}
  \caption{The empirical tail probabilities of the proposed test statistic $T$ under the three settings of distributions, which was computed with \textbf{null bicluster structure}. The left, center, and right figures, respectively, represent the results where each entry of the observed matrix $A$ was generated using Gaussian, Bernoulli, and Poisson distributions.}\vspace{3mm}
  \label{fig:preliminaryT_null}
  \includegraphics[width=0.31\hsize]{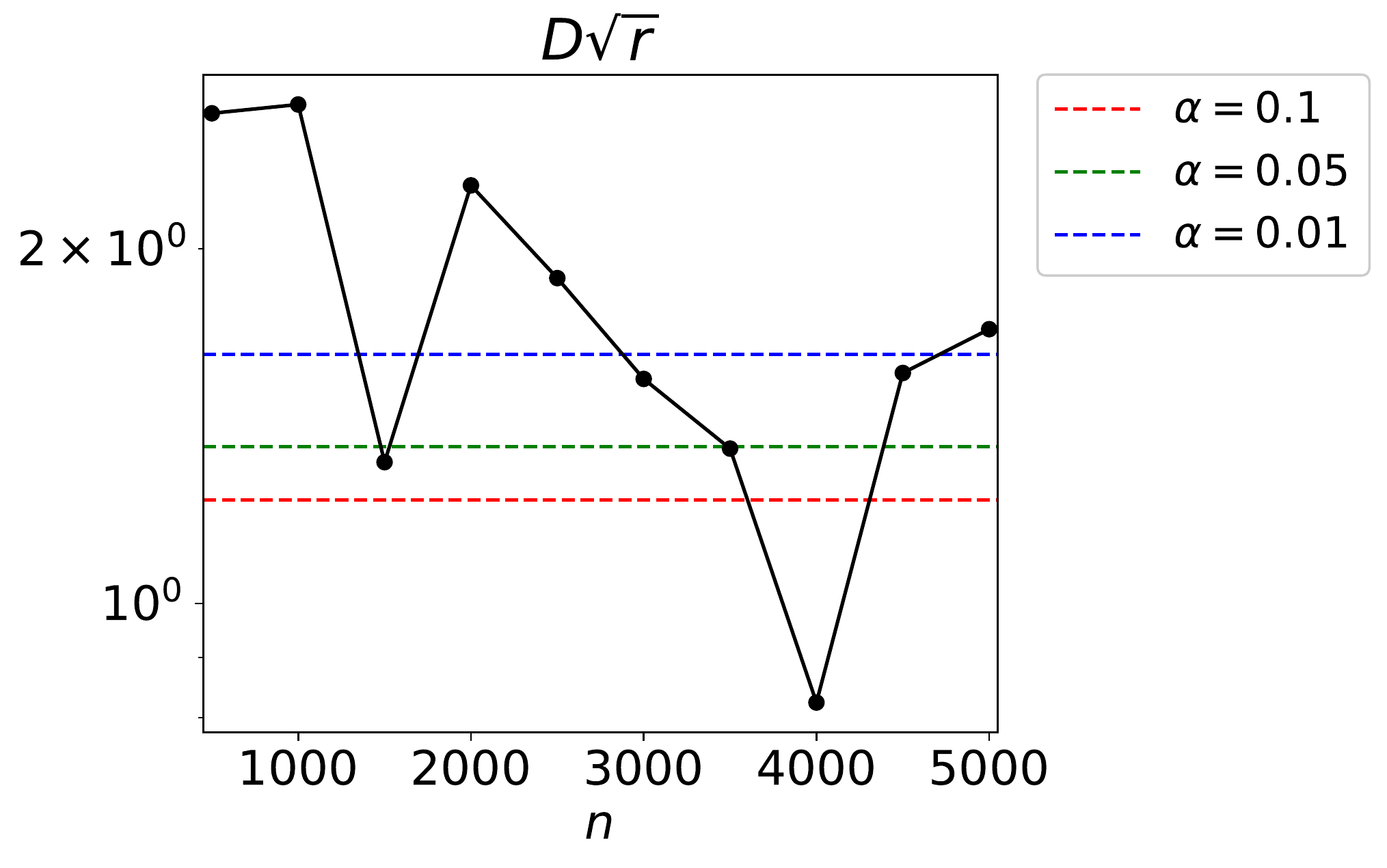}
  \includegraphics[width=0.31\hsize]{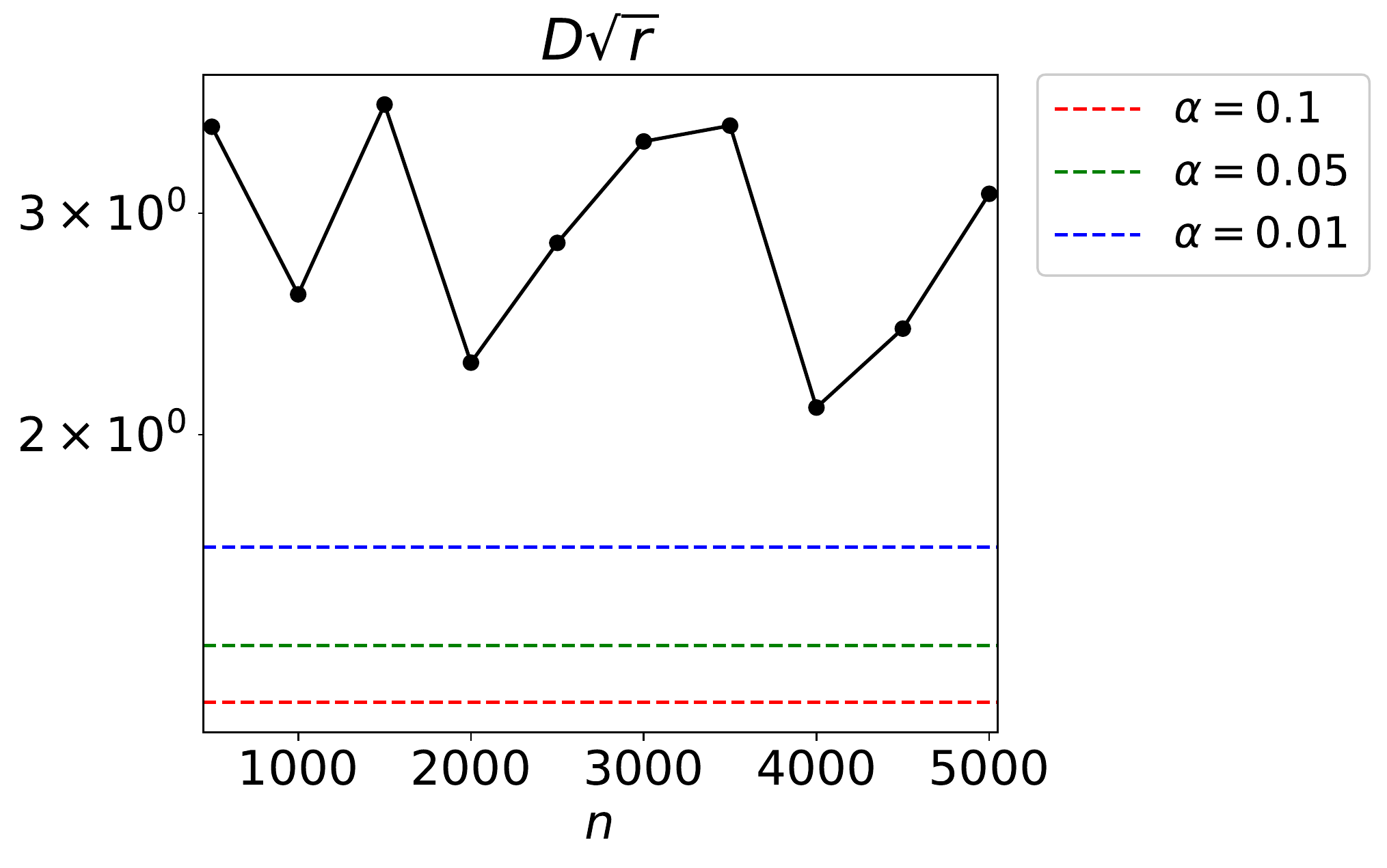}
  \includegraphics[width=0.31\hsize]{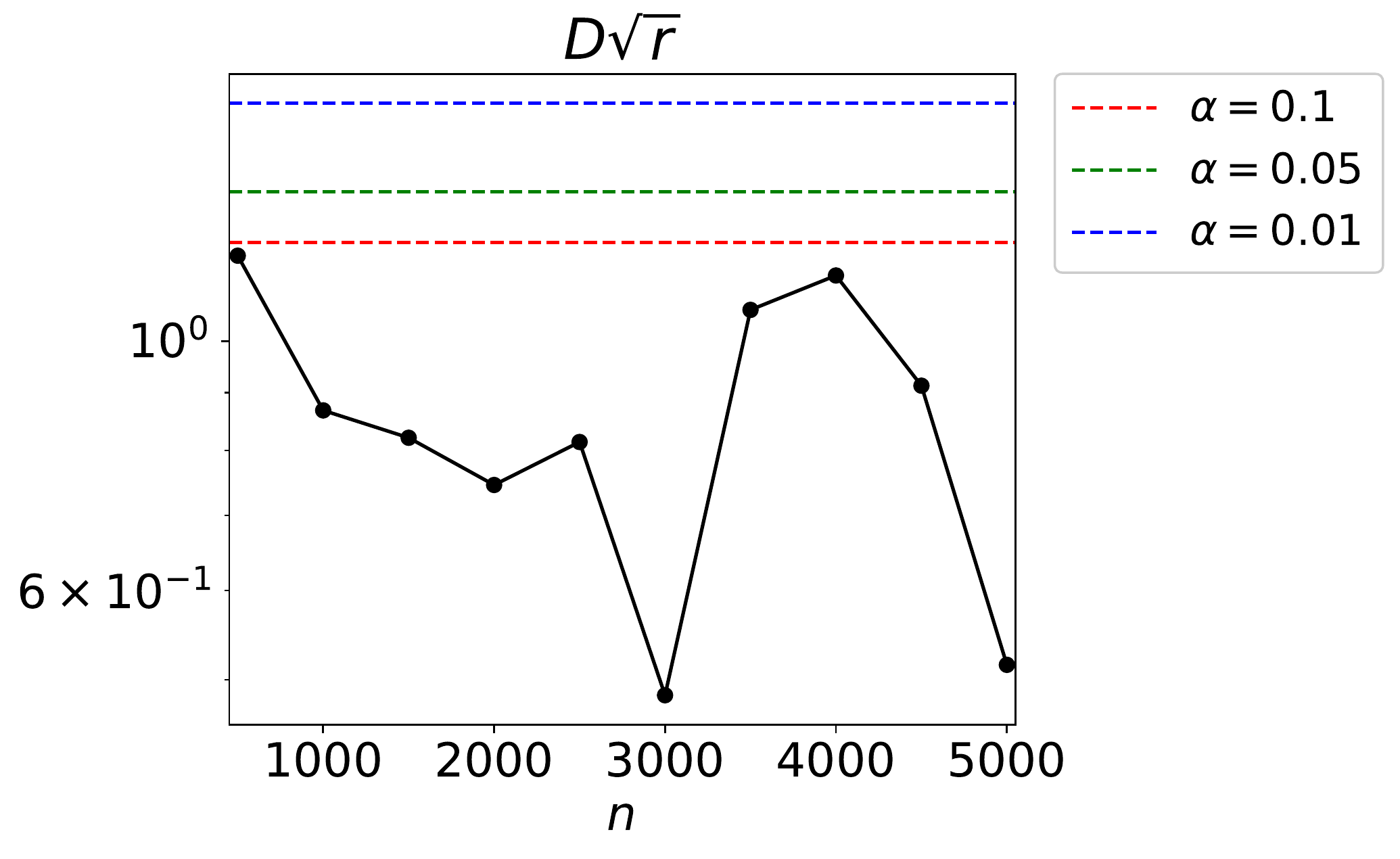}
  \caption{Test statistics $D\sqrt{r}$ of the KS test \cite{Conover1999}, which was computed with \textbf{null bicluster structure}. The left, center, and right figures, respectively, show the results where each entry of the observed matrix $A$ was generated using Gaussian, Bernoulli, and Poisson distributions.}
  \label{fig:KStest_null}
\end{figure}

\subsection{The asymptotic behavior of test statistic $T$ in the unrealizable case}
\label{sec:exp_alternative}

Second, we consider the unrealizable cases (i.e., $K > K_0$). Specifically, under the assumptions firstly that the total number of biclusters and background submatrices $(K+H)$ is a fixed constant that does not depend on the matrix size and secondly that the minimum row and column sizes of these submatrices ($n_{\mathrm{min}}$ and $p_{\mathrm{min}}$, respectively) satisfy $n_{\mathrm{min}} = \Omega_p (m)$ and $p_{\mathrm{min}} = \Omega_p (m)$, from (\ref{eq:Z_hat_op_lower2}) and (\ref{eq:T_un_up}), we have $T = \Theta_p \left( m^{\frac{5}{3}} \right)$. 

Based on the same procedure outlined in Sect.~\ref{sec:exp_null}, we generated Gaussian, Bernoulli, and Poisson random matrices with three biclusters (i.e., $K=3$), estimated their bicluster structures, and computed the test statistics. We used the same settings as in Sect.~\ref{sec:exp_null} for (1) the null parameters of three distributions (\ref{eq:b_gauss}), (\ref{eq:b_bernoulli}), and (\ref{eq:b_poisson}), (2) the procedure to generate the observed matrices, and (3) the SA-based submatrix localization algorithm. In this experiment, we tried the following $10$ sets of matrix sizes: $(n, p) = (200 \times i, 150 \times i)$, $i = 1, \dots, 10$. For each combination of the distribution and matrix size settings, we randomly generated $100$ data matrices $A$, estimated their bicluster structures with $K_0 = 0, 1, \dots, K - 1$, and checked the average behavior of test statistic $T$. 

Figure \ref{fig:unrealizableT2} represents the asymptotic behavior of the mean of the proposed test statistic $T$ divided by $n^{\frac{5}{3}}$ under unrealizable settings. This figure illustrates that $T$ increases in proportion to $m^{\frac{5}{3}}$ in all the settings of distributions, as shown in the first paragraph of this section. 

\begin{figure}[t]
  \centering
  \includegraphics[width=0.31\hsize]{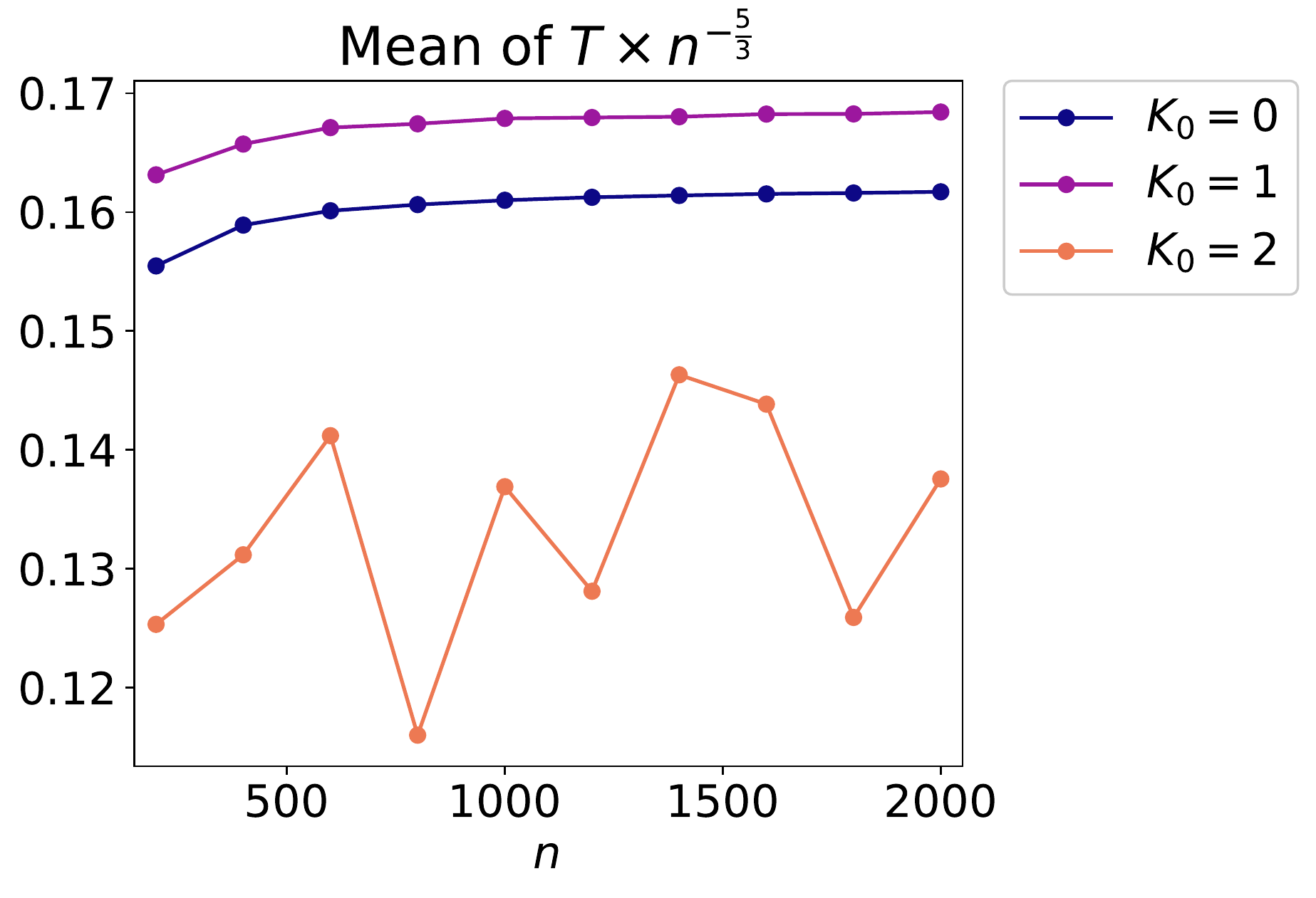}
  \includegraphics[width=0.31\hsize]{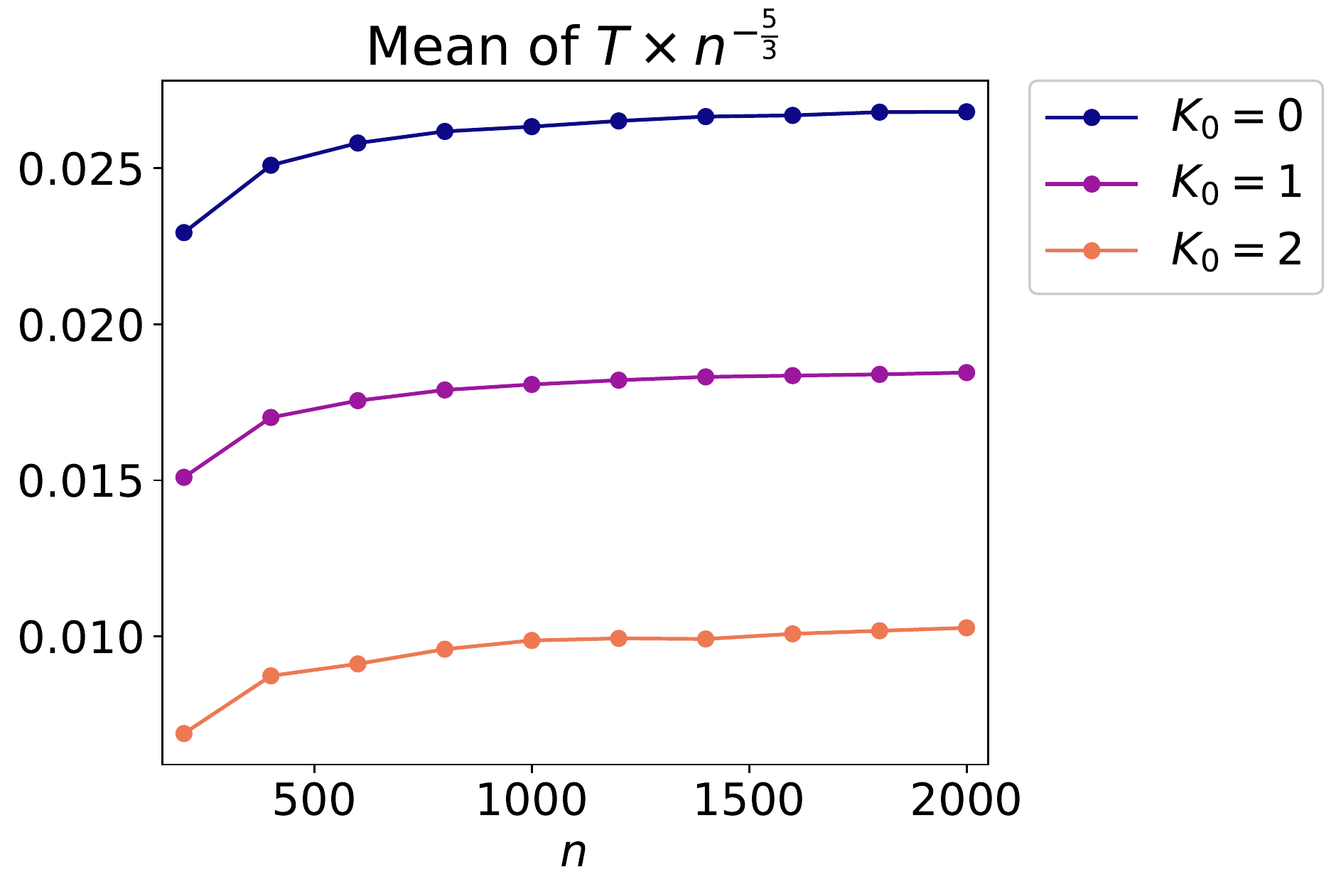}
  \includegraphics[width=0.31\hsize]{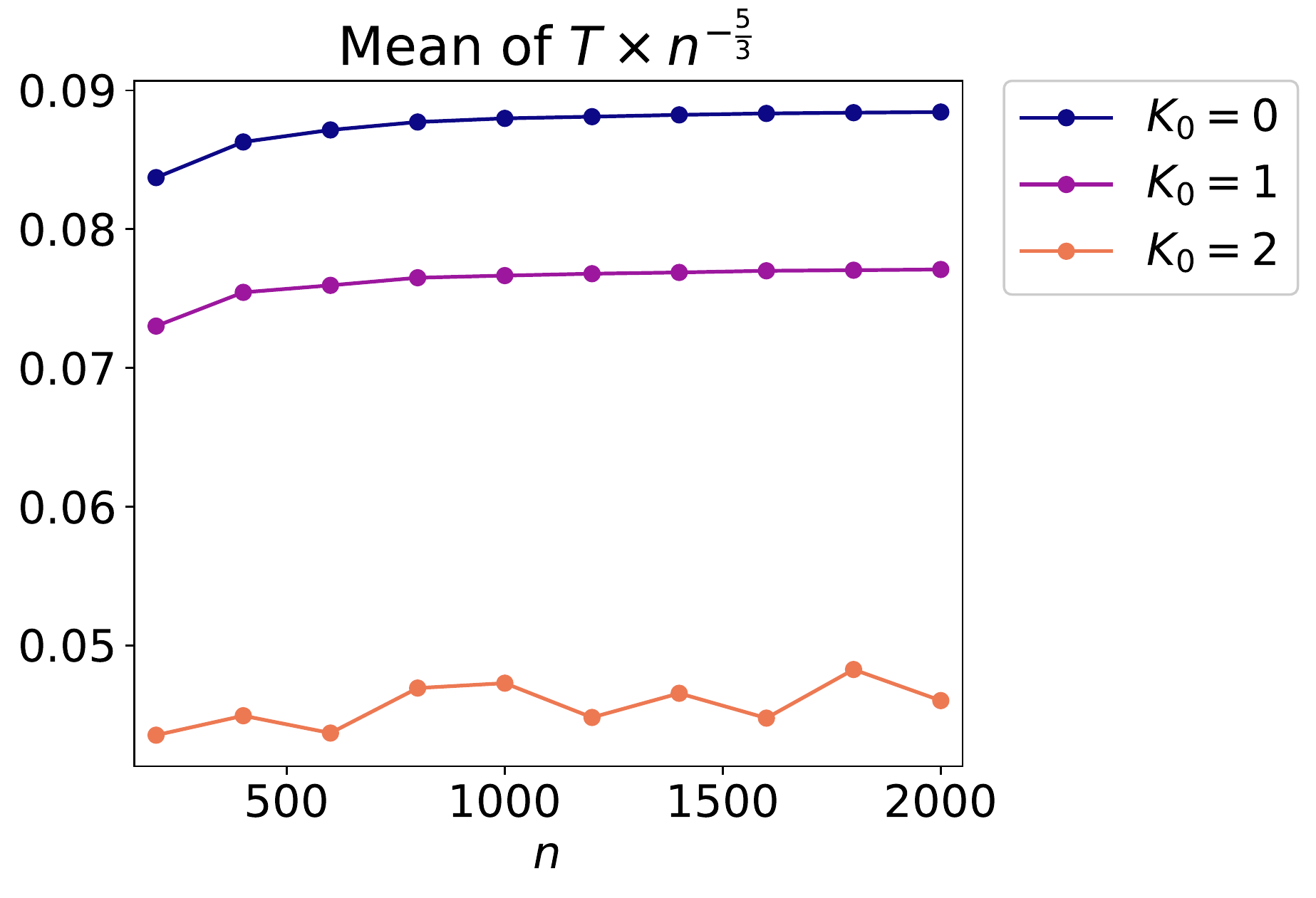}
  \caption{Mean of the proposed test statistics $T$ divided by $n^{\frac{5}{3}}$ in the unrealizable case for $100$ trials. The null number of biclusters was set at $K=3$. The left, center, and right figures, respectively, represent the results where each entry of observed matrix $A$ was generated using Gaussian, Bernoulli, and Poisson distributions. The horizontal line represents the row size $n$ of the observed matrix. Each plotted line represents a result for a given hypothetical number of biclusters $K_0$.}
  \label{fig:unrealizableT2}
\end{figure}

\subsection{The accuracy of the proposed test in selecting the number of biclusters $K$}
\label{sec:accuracy}

Third, we checked the accuracy of the proposed test in selecting the number of biclusters $K$, by using the synthetic Gaussian, Bernoulli, and Poisson data matrices that were generated using the procedure outlined in Sect.~\ref{sec:exp_null}. We set the null number of biclusters at $K=3$. 
As for the null mean parameters of the three distributions, we tried the ten settings $\{\bm{b}^{(1)}, \dots, \bm{b}^{(10)}\}$, where 
\begin{align}
&\bm{b}^{(t)} \equiv \left( 1 - \frac{t}{10} \right) \left( \bm{b} - 0.5 \begin{bmatrix}
1 & \dots & 1
\end{bmatrix}^{\top} \right) + 0.5 \begin{bmatrix}
1 & \dots & 1
\end{bmatrix}^{\top}\ \mathrm{(Gaussian\ and\ Bernoulli\ cases)}, \\
&\bm{b}^{(t)} \equiv \left( 1 - \frac{t}{10} \right) \left( \bm{b} - 5 \begin{bmatrix}
1 & \dots & 1
\end{bmatrix}^{\top} \right) + 5 \begin{bmatrix}
1 & \dots & 1
\end{bmatrix}^{\top}\ \mathrm{(Poisson\ case)}, 
\end{align}
for all $t = 1, \dots, 10$. In the above settings, we set $\bm{b}$ at the same vector as given in (\ref{eq:b_gauss}), (\ref{eq:b_bernoulli}), and (\ref{eq:b_poisson}) for each setting of distributions. For the standard deviation parameter $\bm{s}$, we used the same setting as in (\ref{eq:b_gauss}). 
Aside from these model parameters, we used the same settings as in Sect.~\ref{sec:exp_null} for (1) the procedure to generate the observed matrices, and (2) the SA-based submatrix localization algorithm. Figures \ref{fig:A_example_acc_g}, \ref{fig:A_example_acc_b}, and \ref{fig:A_example_acc_p}, respectively, depict the examples of generated data matrices in Gaussian, Bernoulli, and Poisson cases. In this experiment, we tried the following $10$ sets of matrix sizes: $(n, p) = (40 \times i, 30 \times i)$, $i = 1, \dots, 10$. For each combination of the distribution and matrix size settings, we randomly generated $1,000$ data matrices $A$ and applied the proposed sequential test with a significance level of $\alpha = 0.01$ and the hypothetical number of biclusters $K_0 = 0, 1, 2, \dots$. 

\begin{figure}[t]
  \centering
  \includegraphics[width=0.17\hsize]{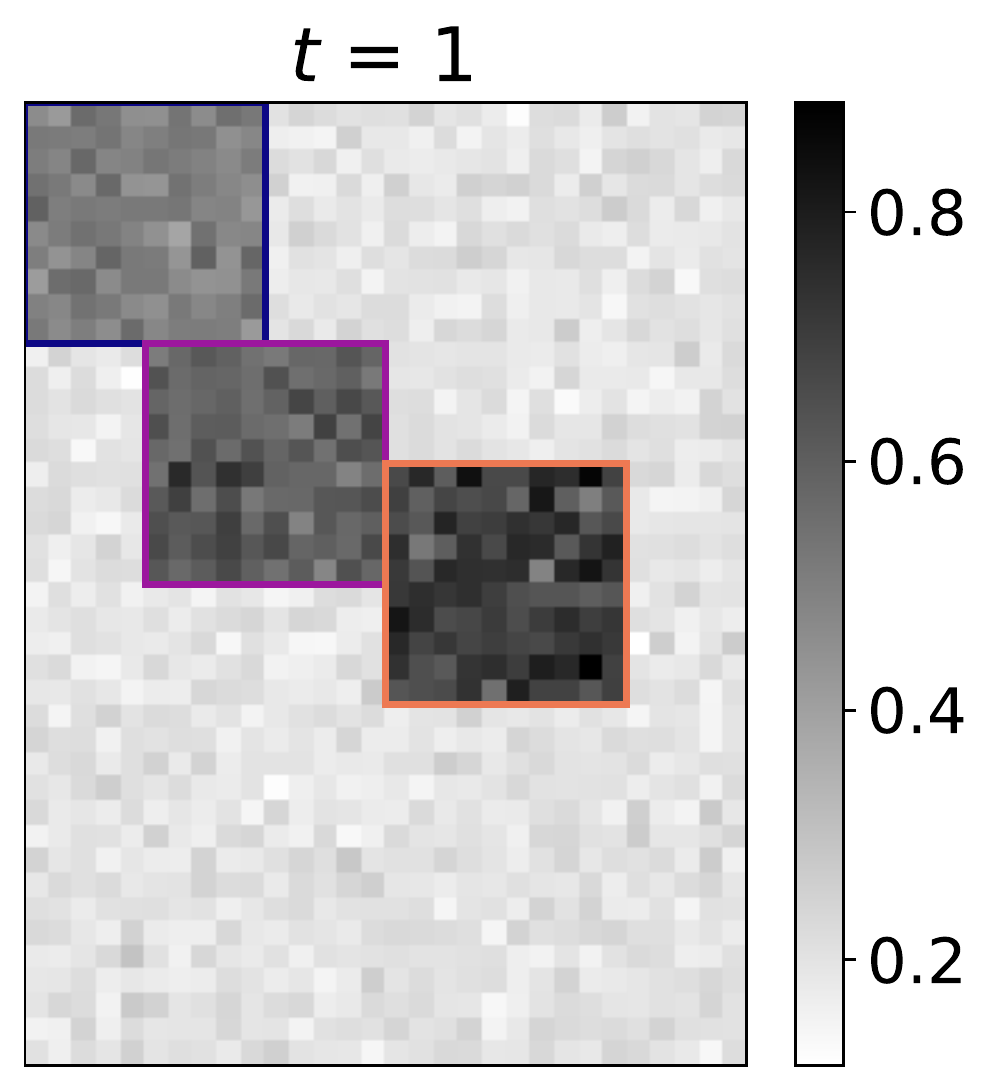}
  \includegraphics[width=0.17\hsize]{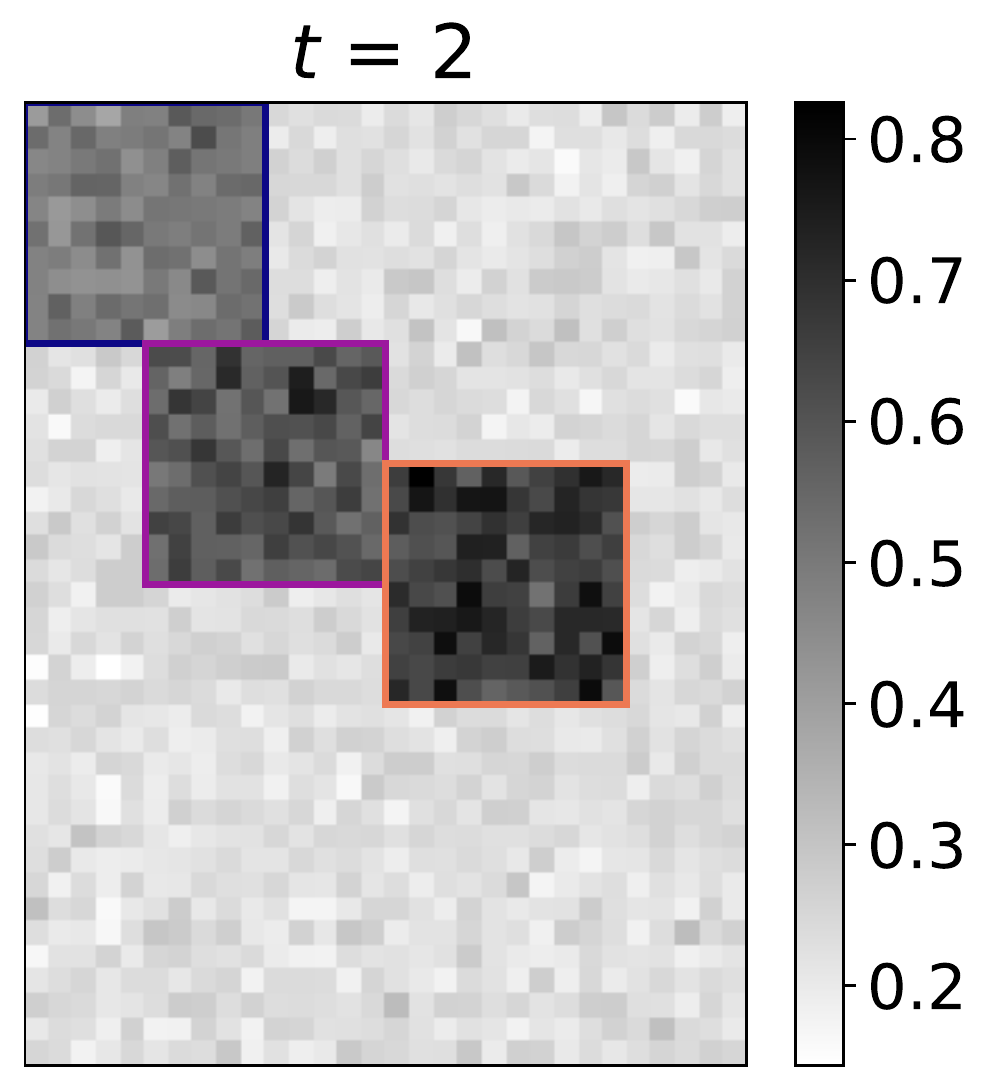}
  \includegraphics[width=0.17\hsize]{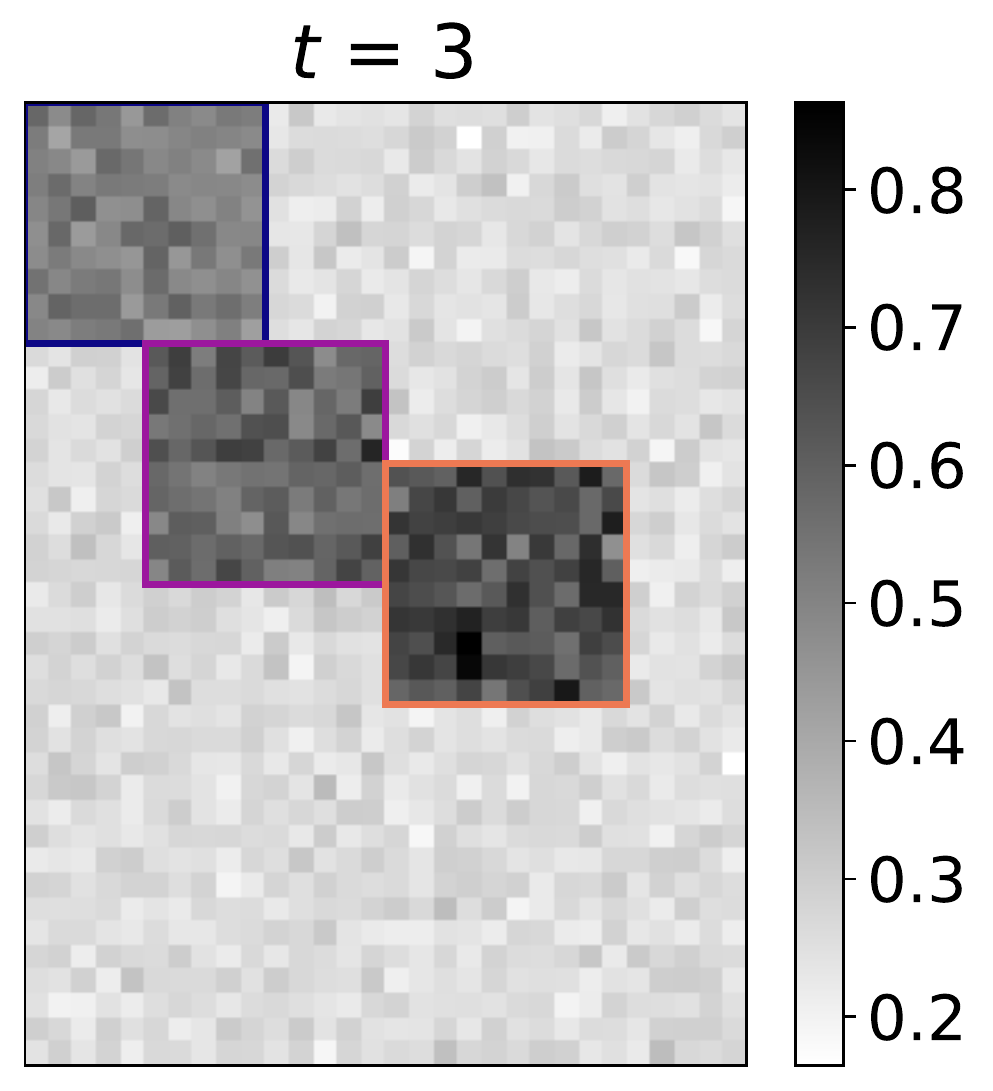}
  \includegraphics[width=0.17\hsize]{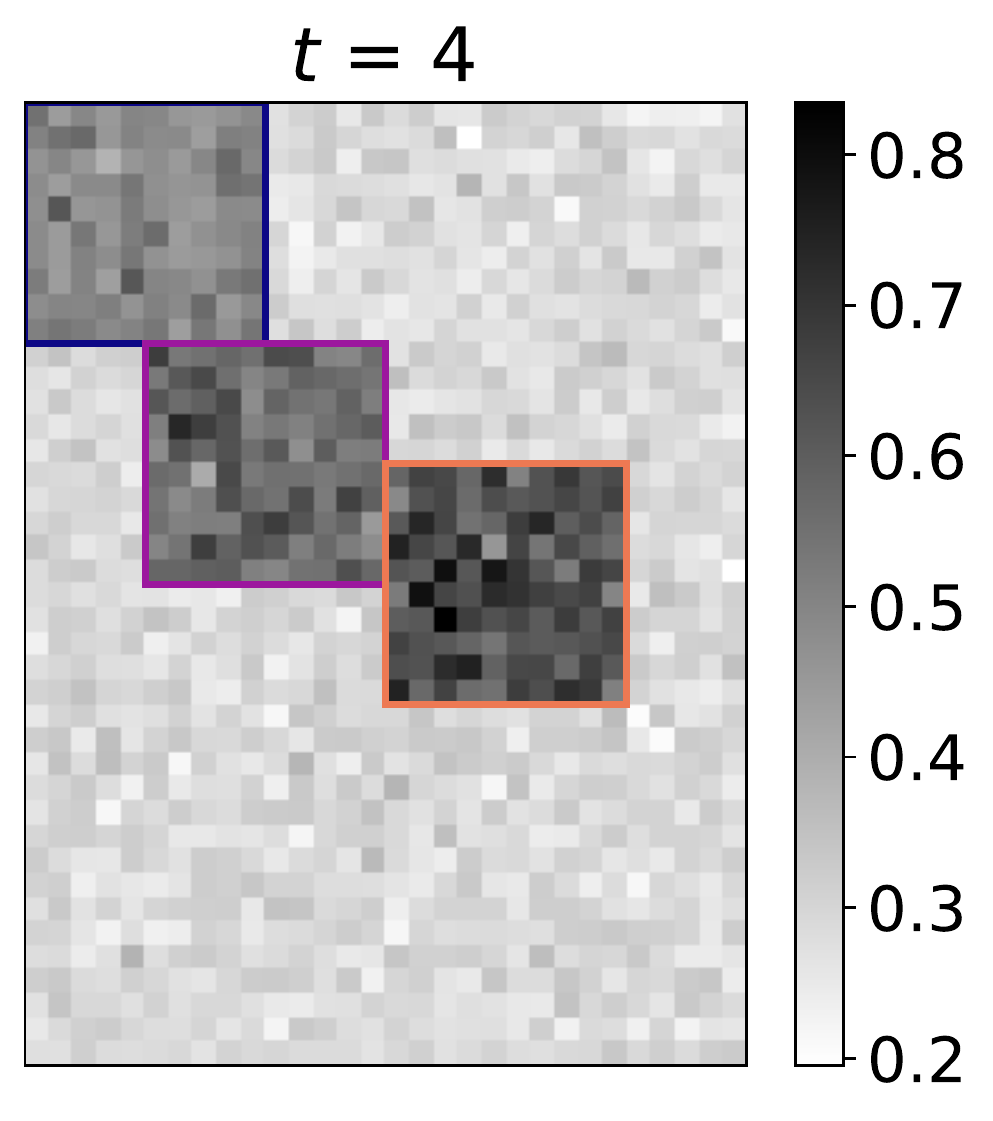}
  \includegraphics[width=0.17\hsize]{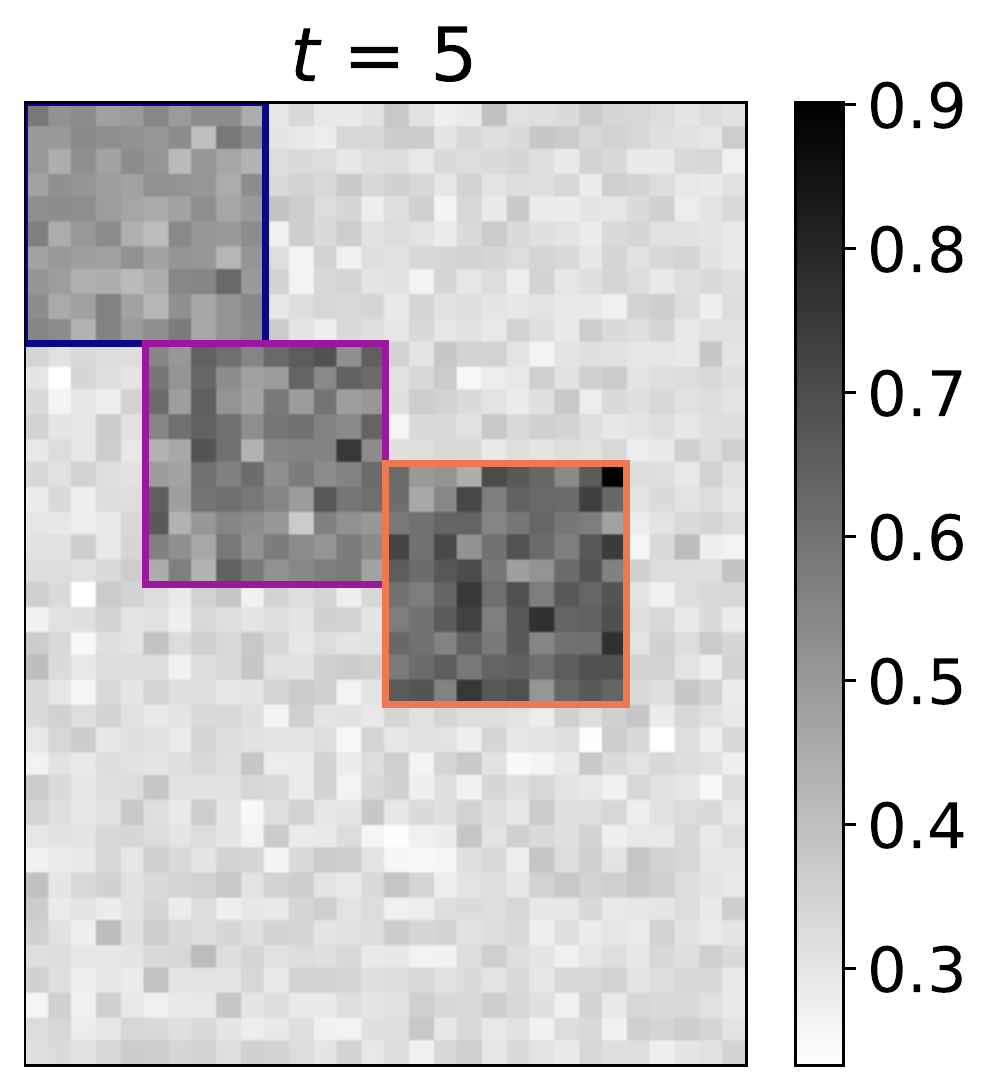}\\
  \includegraphics[width=0.17\hsize]{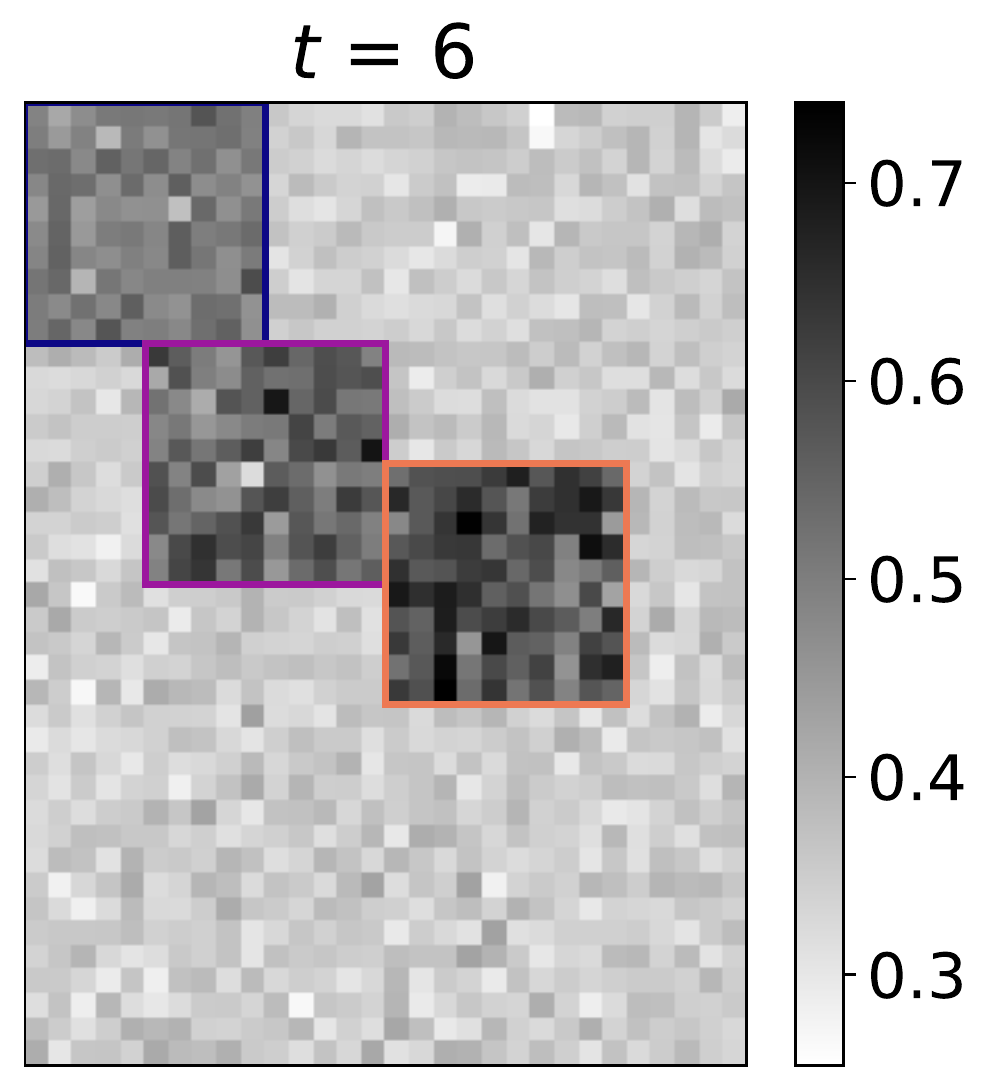}
  \includegraphics[width=0.17\hsize]{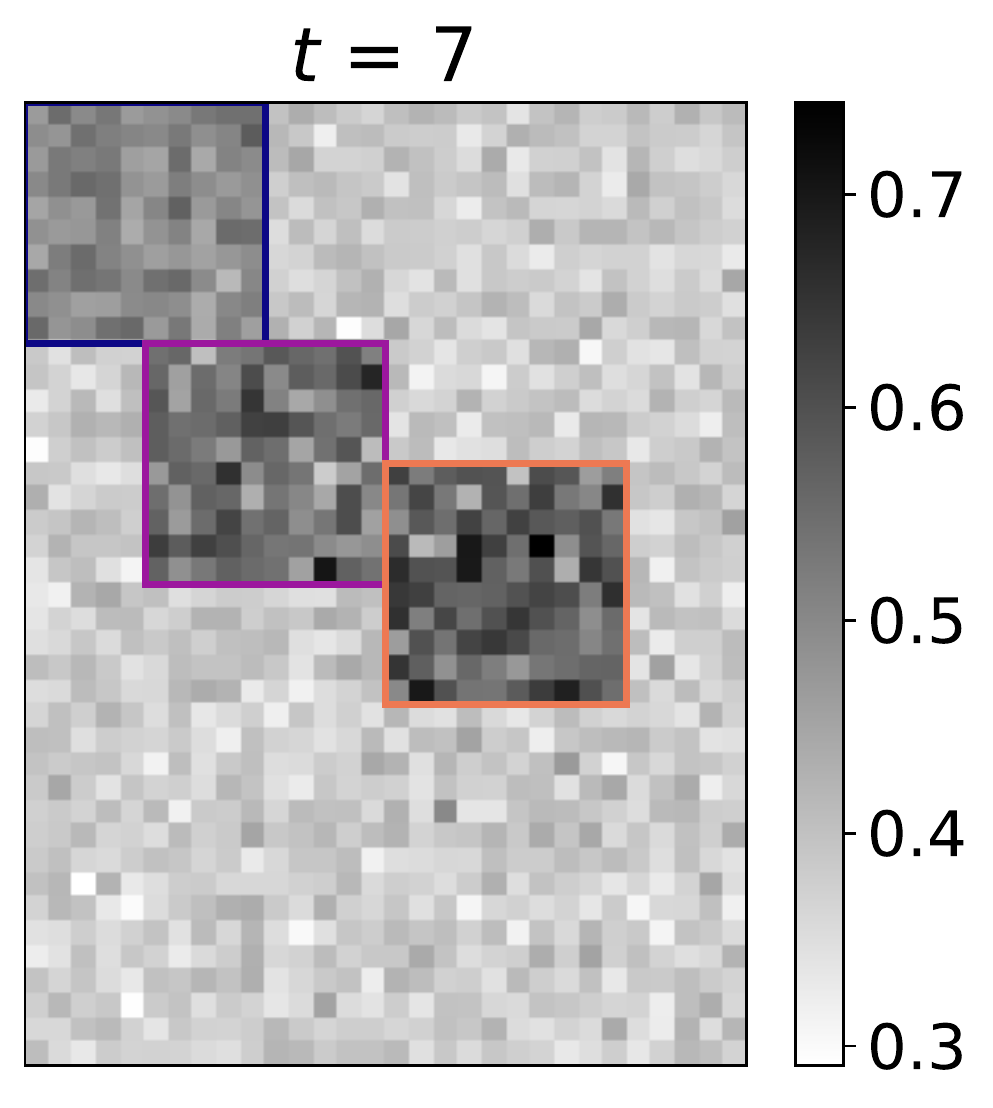}
  \includegraphics[width=0.17\hsize]{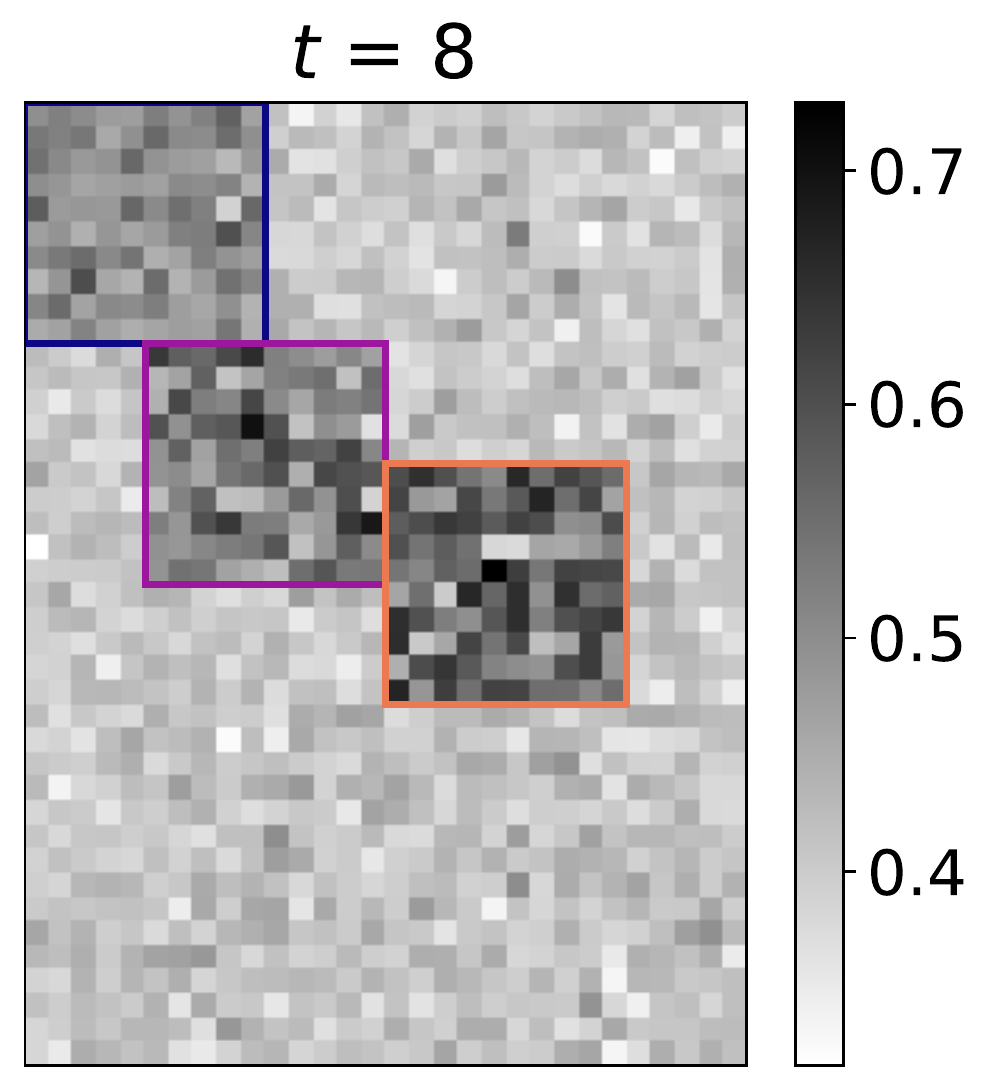}
  \includegraphics[width=0.17\hsize]{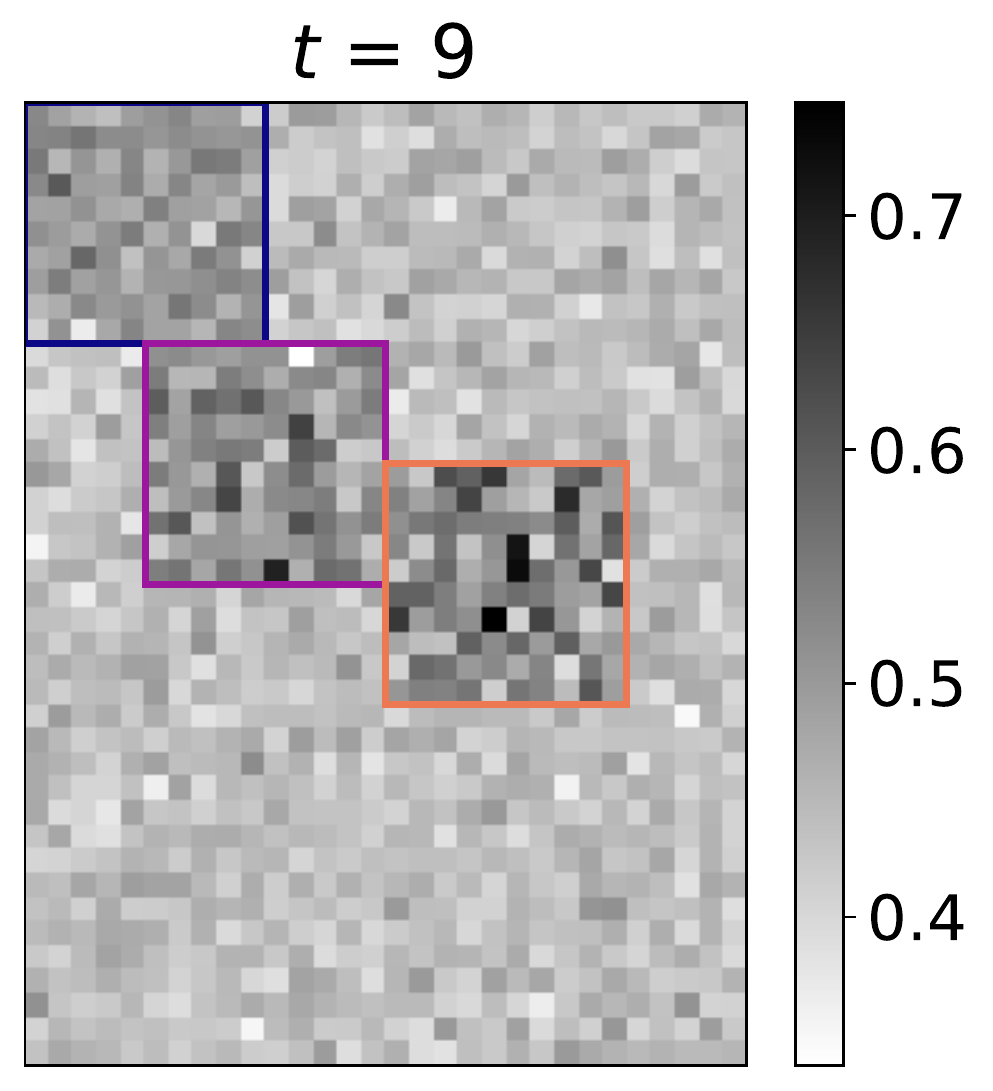}
  \includegraphics[width=0.17\hsize]{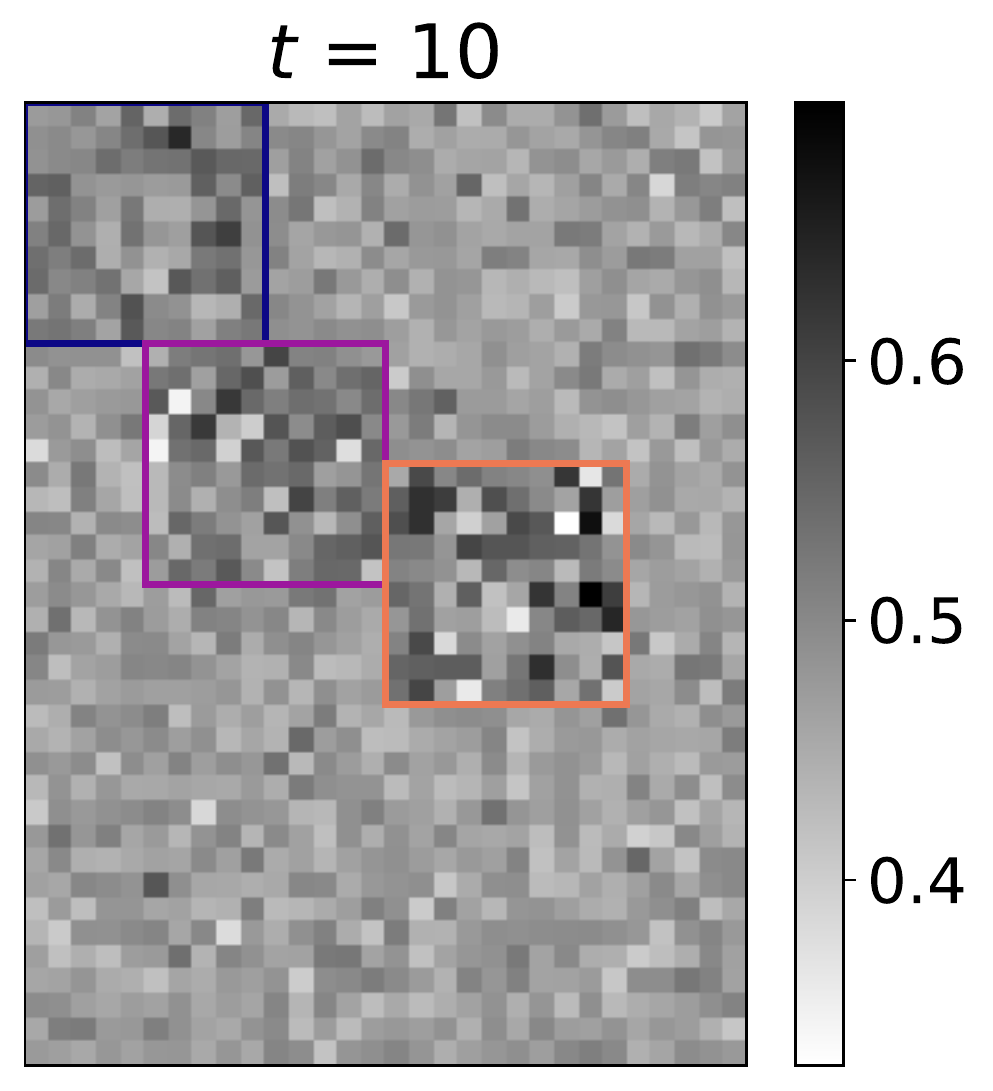}
  \caption{Examples of the observed data matrices for $t = 1, \dots, 10$ (\textbf{Gaussian case}). The colored boxes show the null bicluster structure.}\vspace{3mm}
  \label{fig:A_example_acc_g}
  \includegraphics[width=0.17\hsize]{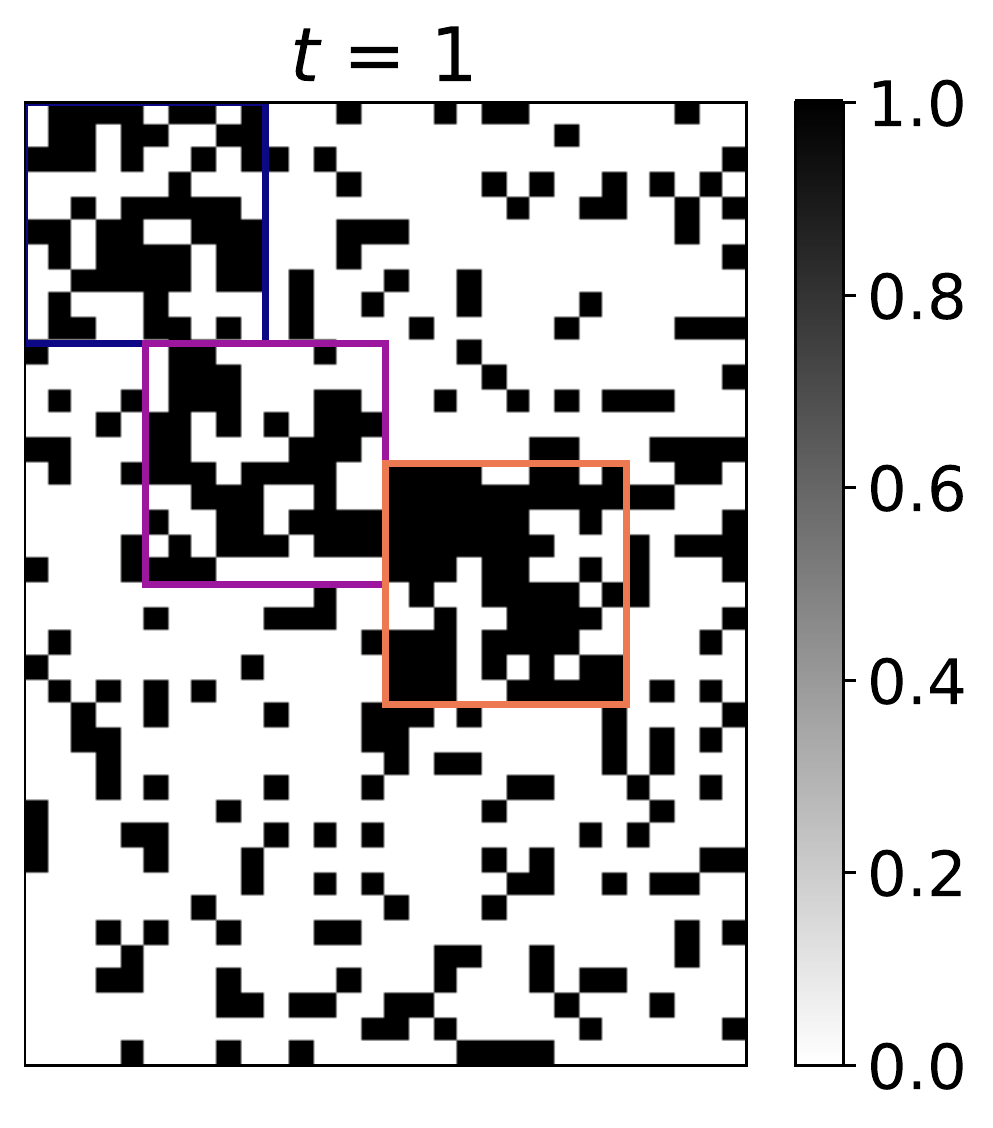}
  \includegraphics[width=0.17\hsize]{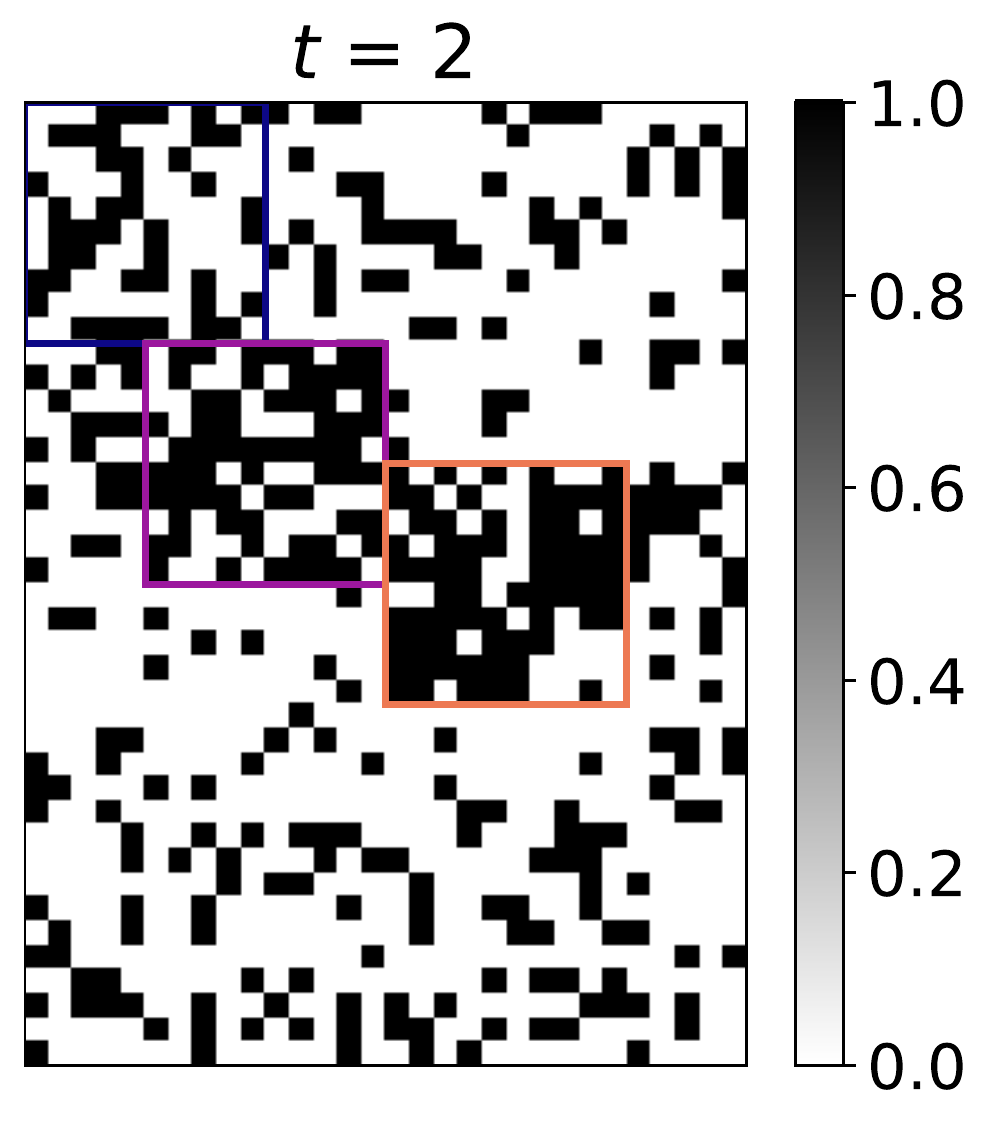}
  \includegraphics[width=0.17\hsize]{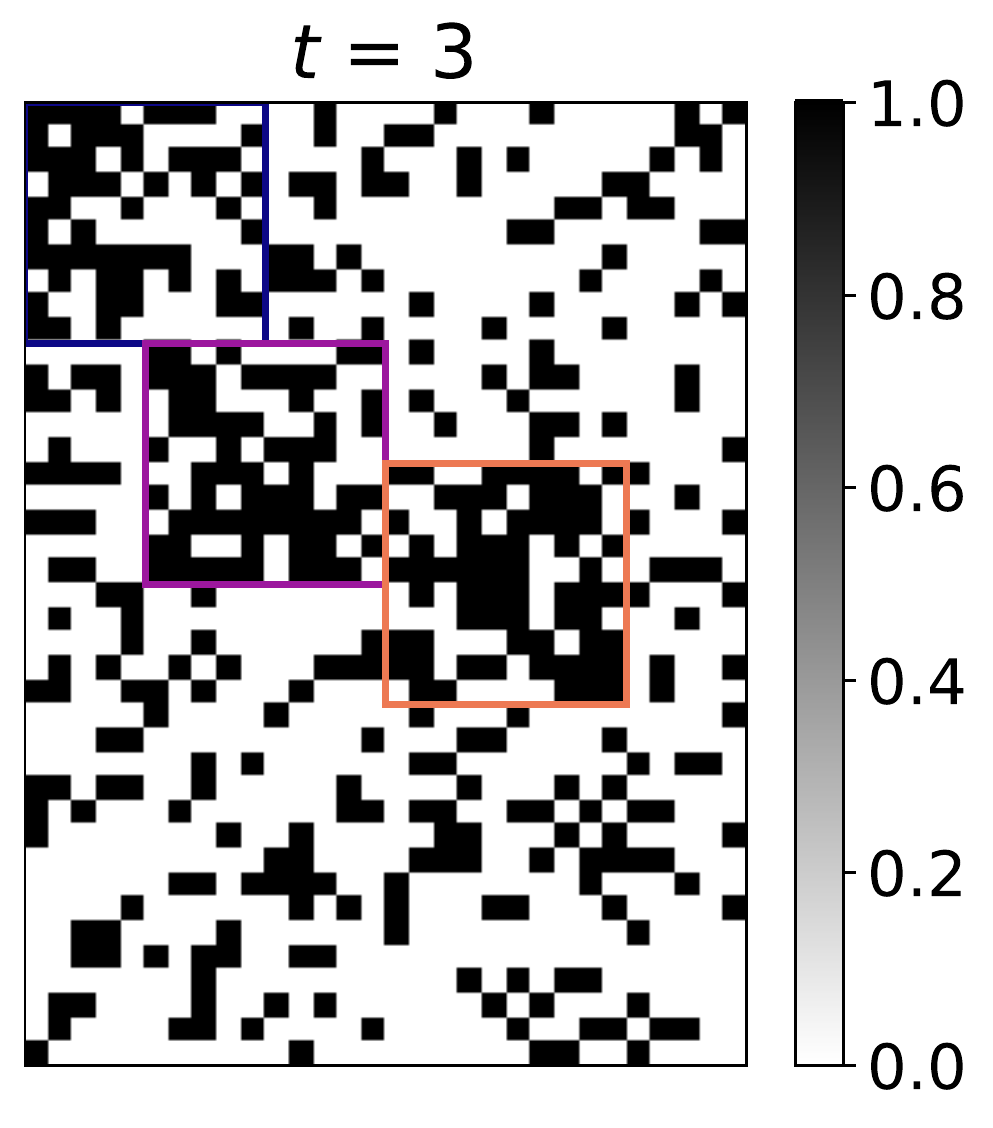}
  \includegraphics[width=0.17\hsize]{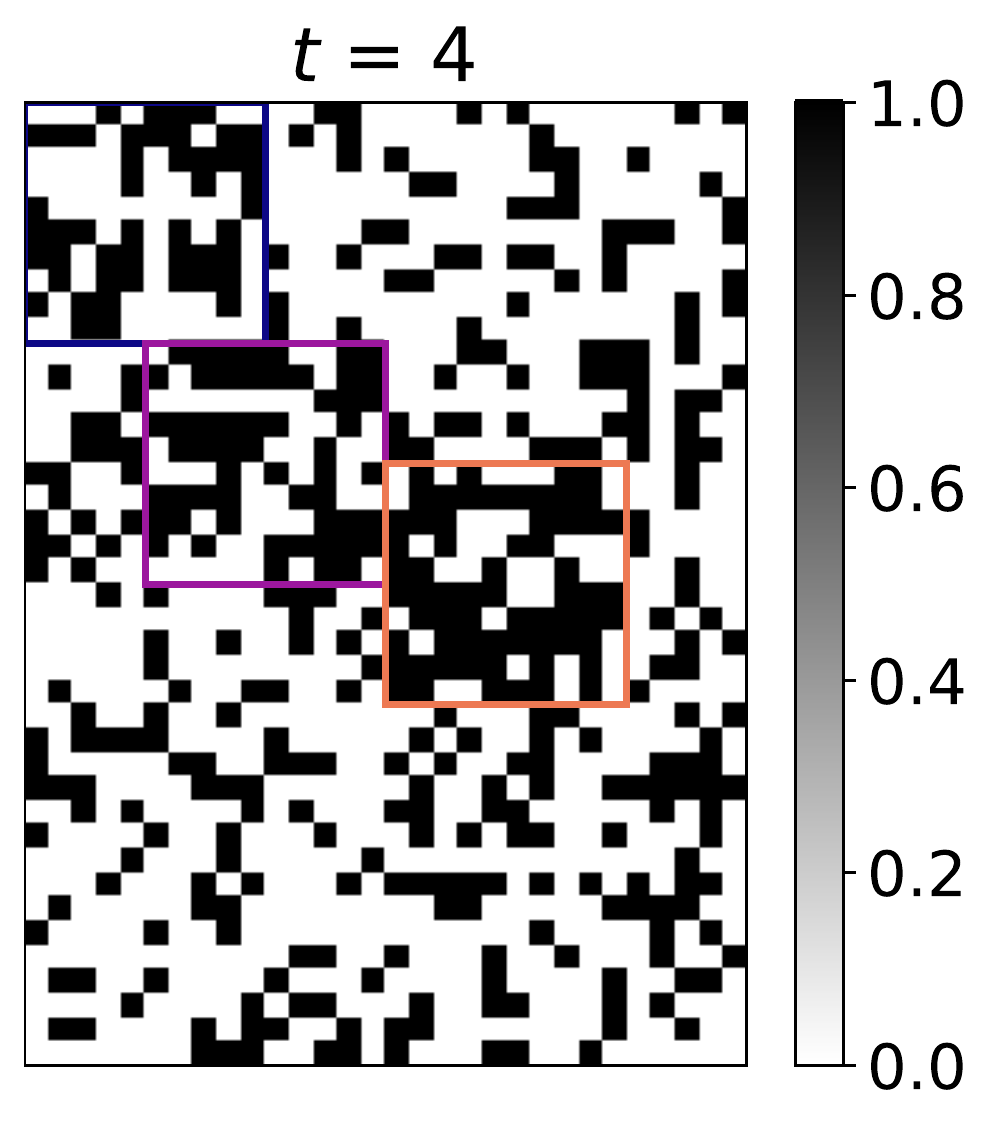}
  \includegraphics[width=0.17\hsize]{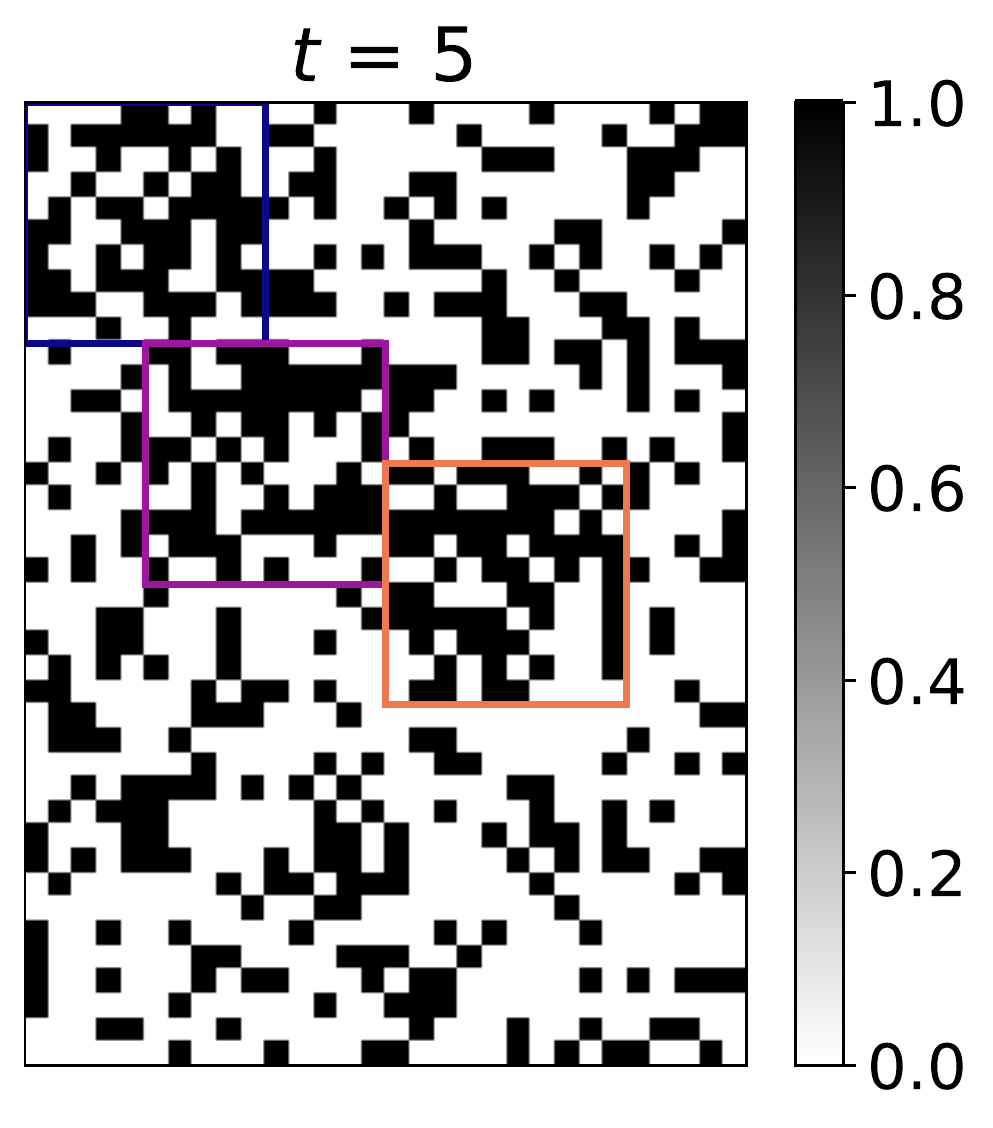}\\
  \includegraphics[width=0.17\hsize]{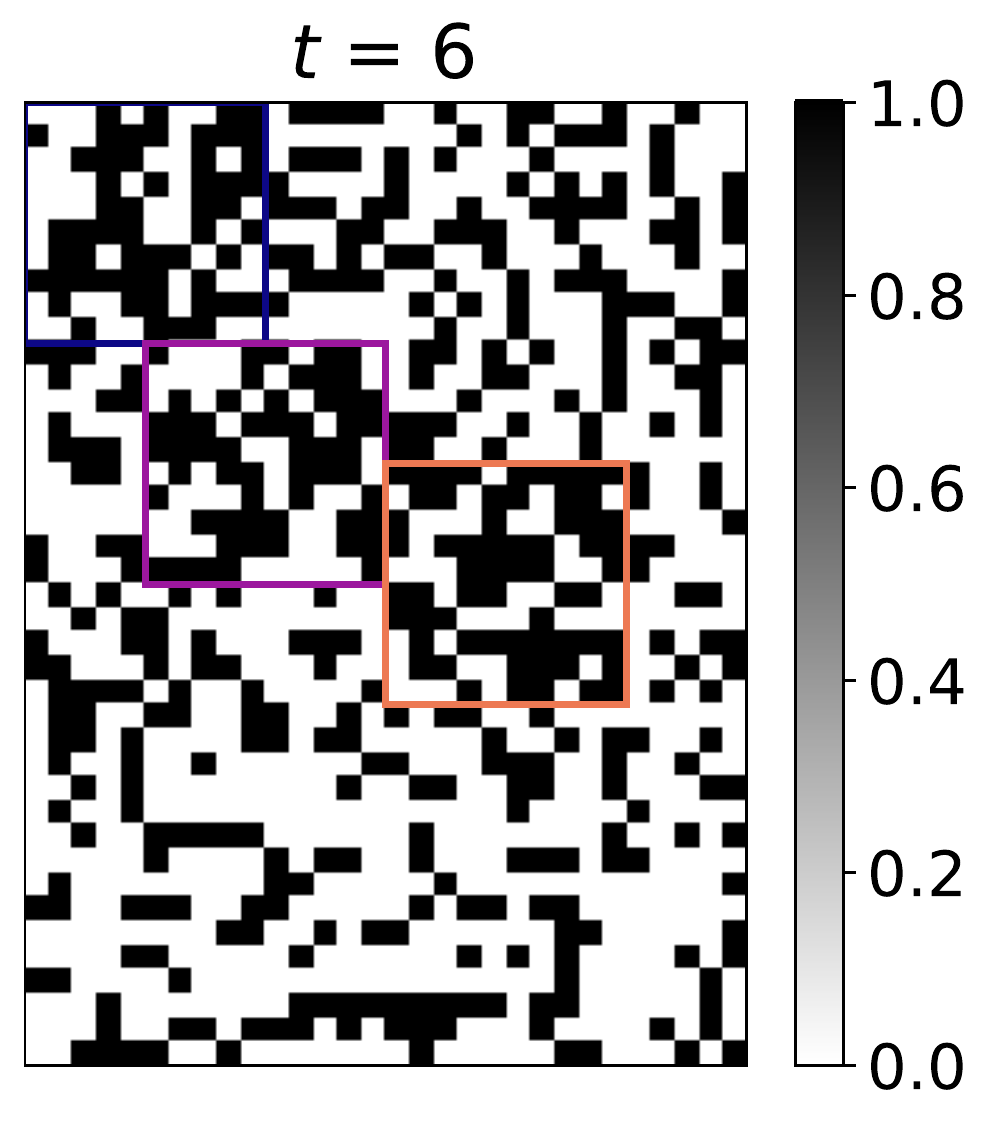}
  \includegraphics[width=0.17\hsize]{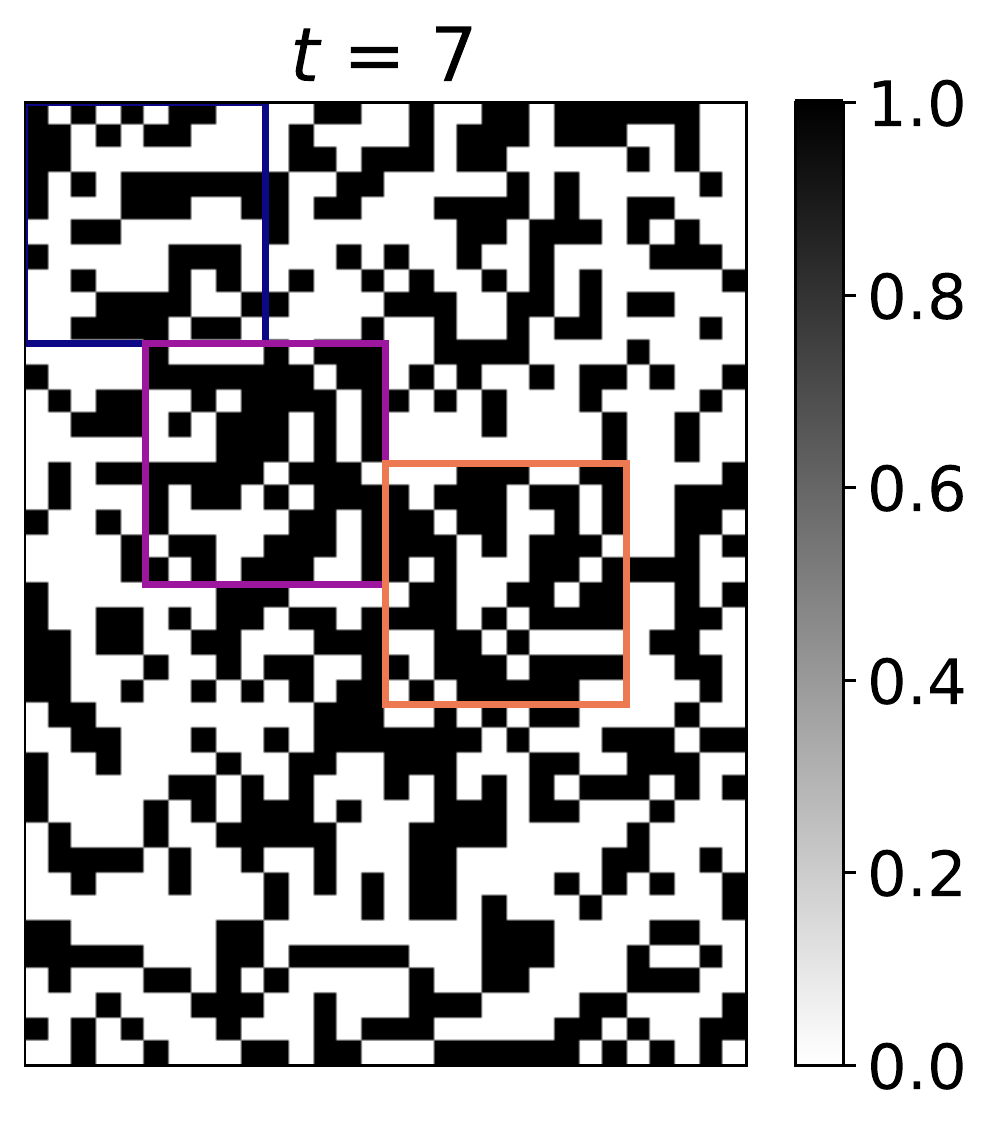}
  \includegraphics[width=0.17\hsize]{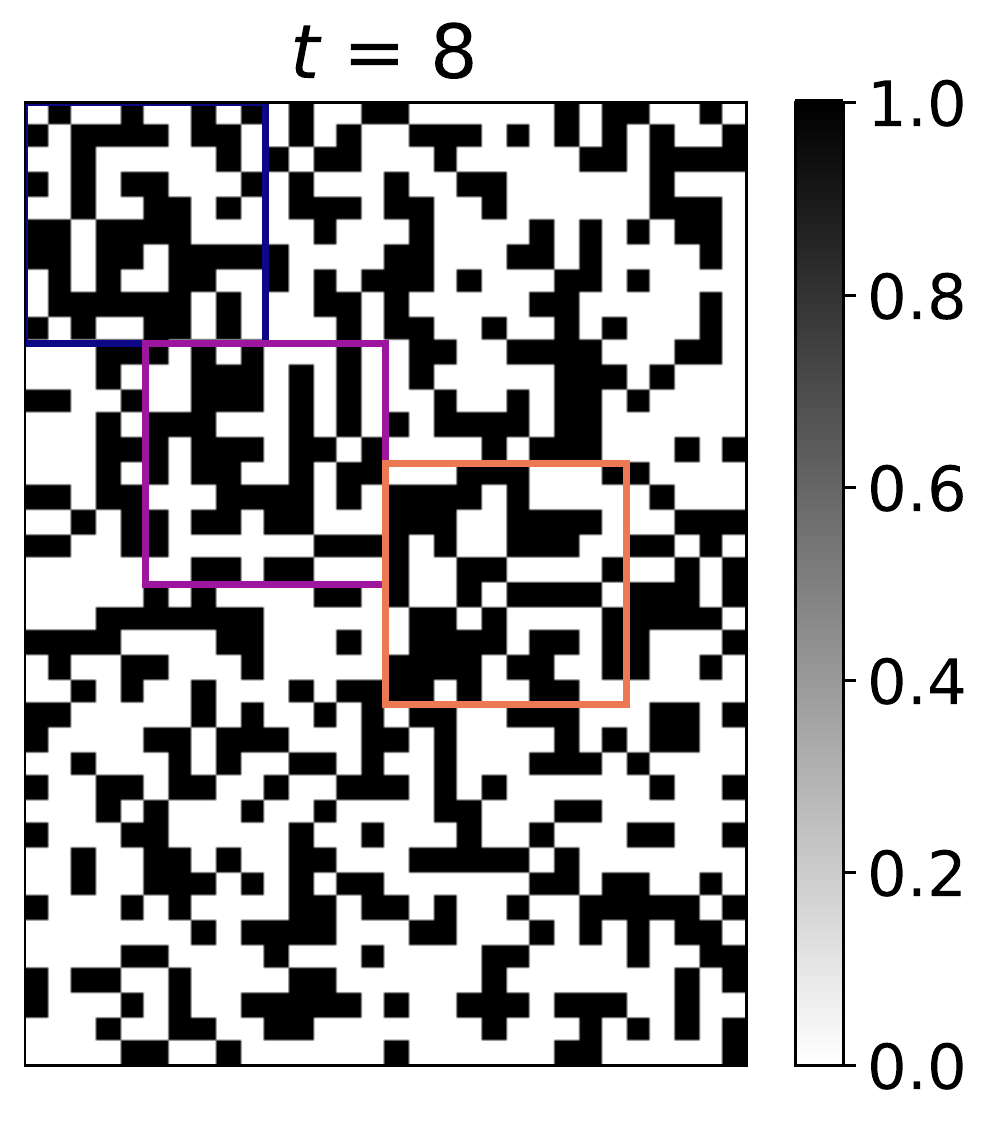}
  \includegraphics[width=0.17\hsize]{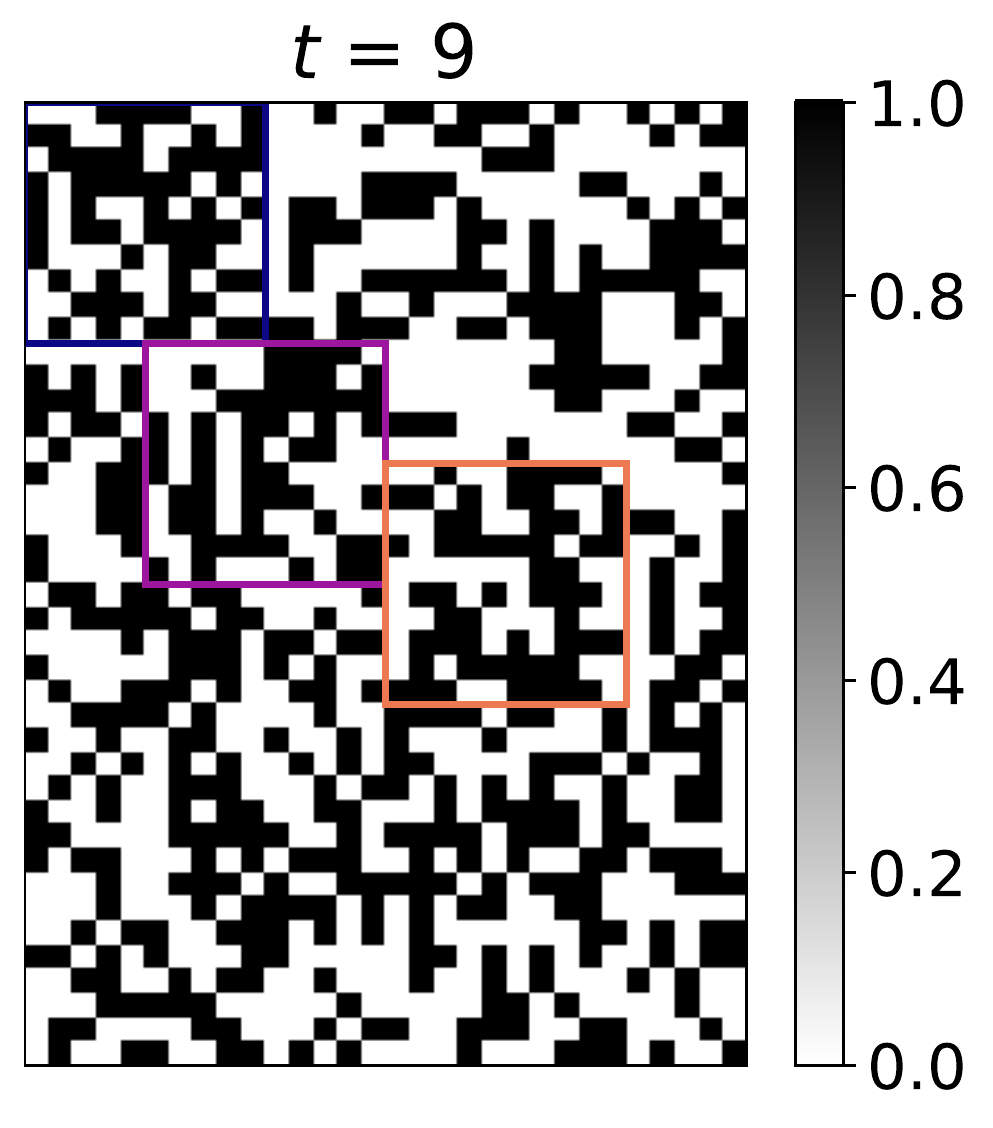}
  \includegraphics[width=0.17\hsize]{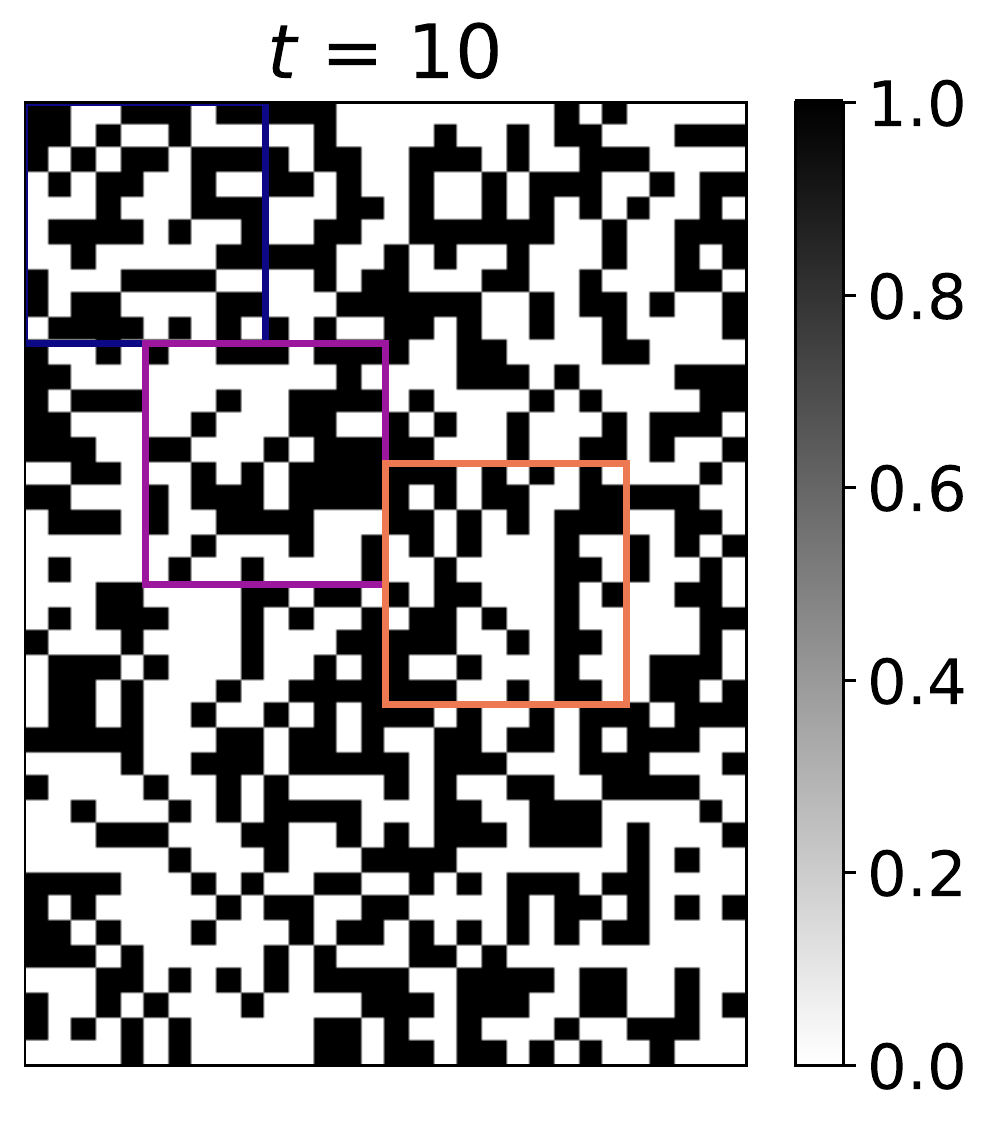}
  \caption{Examples of the observed data matrices for $t = 1, \dots, 10$ (\textbf{Bernoulli case}).}\vspace{3mm}
  \label{fig:A_example_acc_b}
  \includegraphics[width=0.17\hsize]{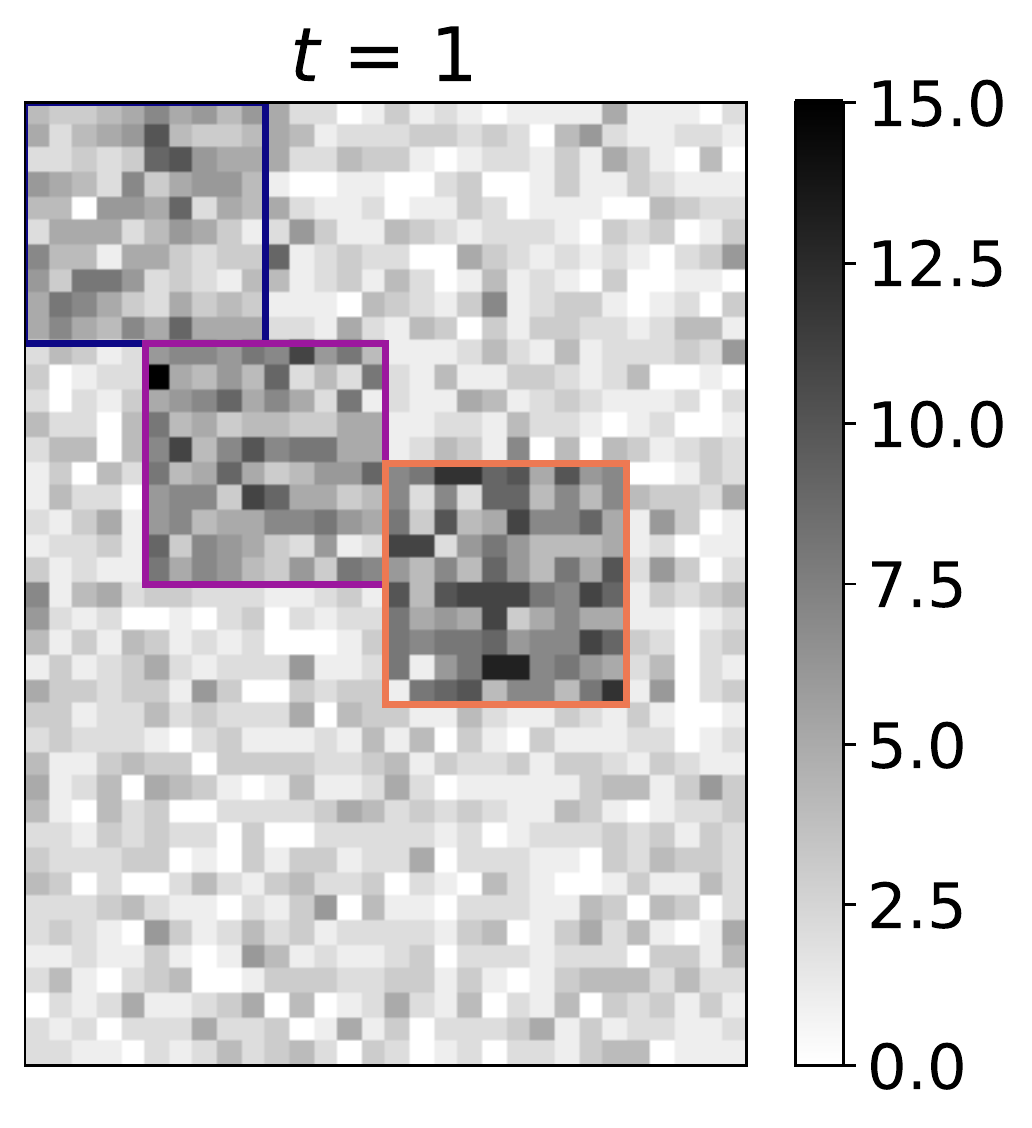}
  \includegraphics[width=0.17\hsize]{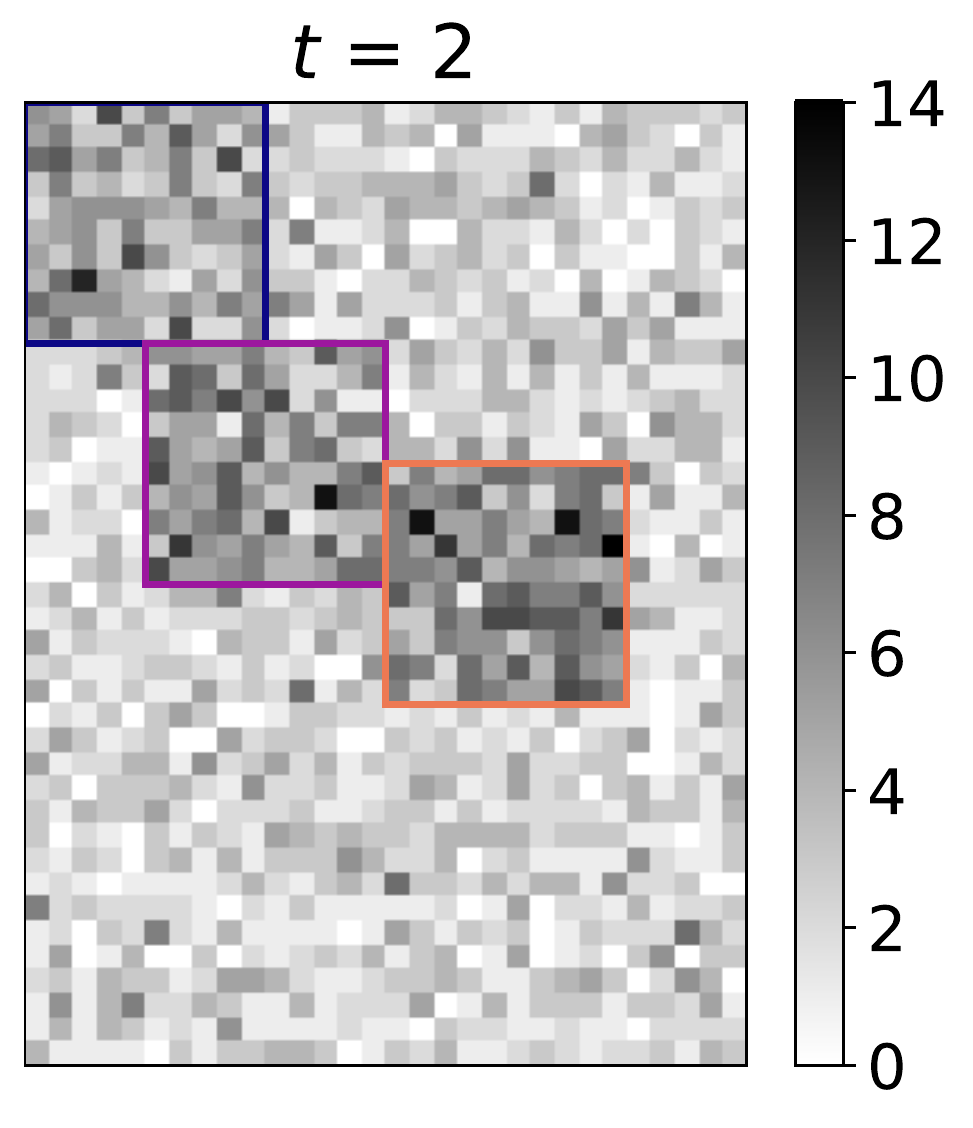}
  \includegraphics[width=0.17\hsize]{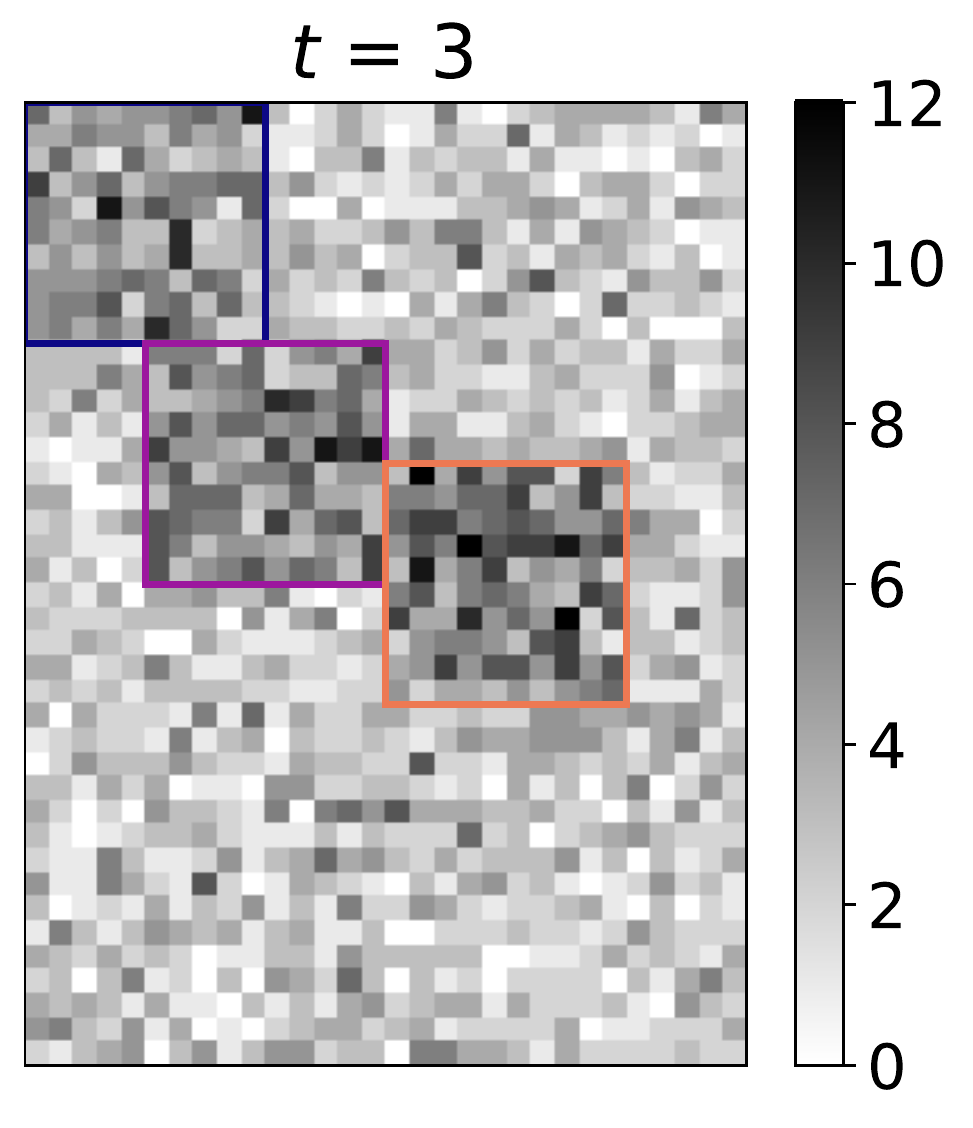}
  \includegraphics[width=0.17\hsize]{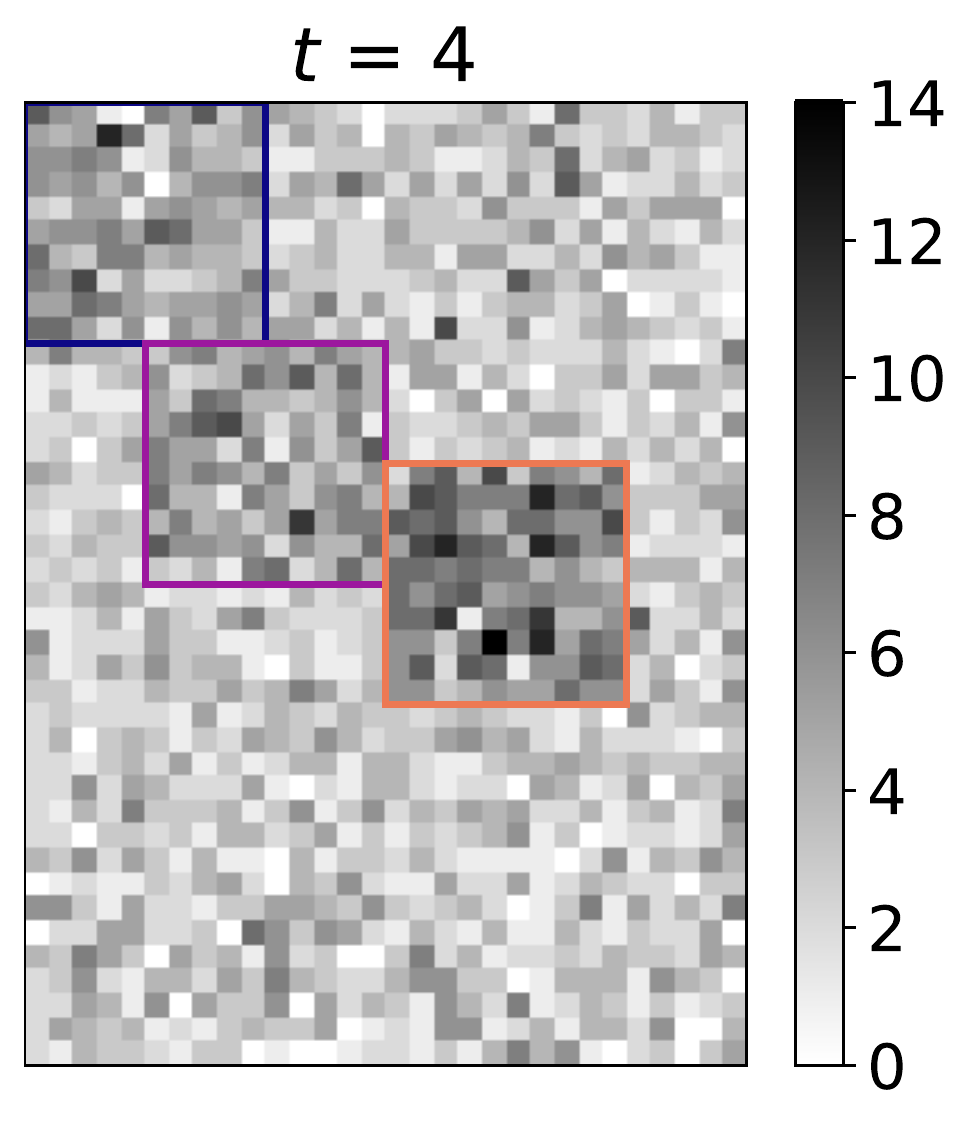}
  \includegraphics[width=0.17\hsize]{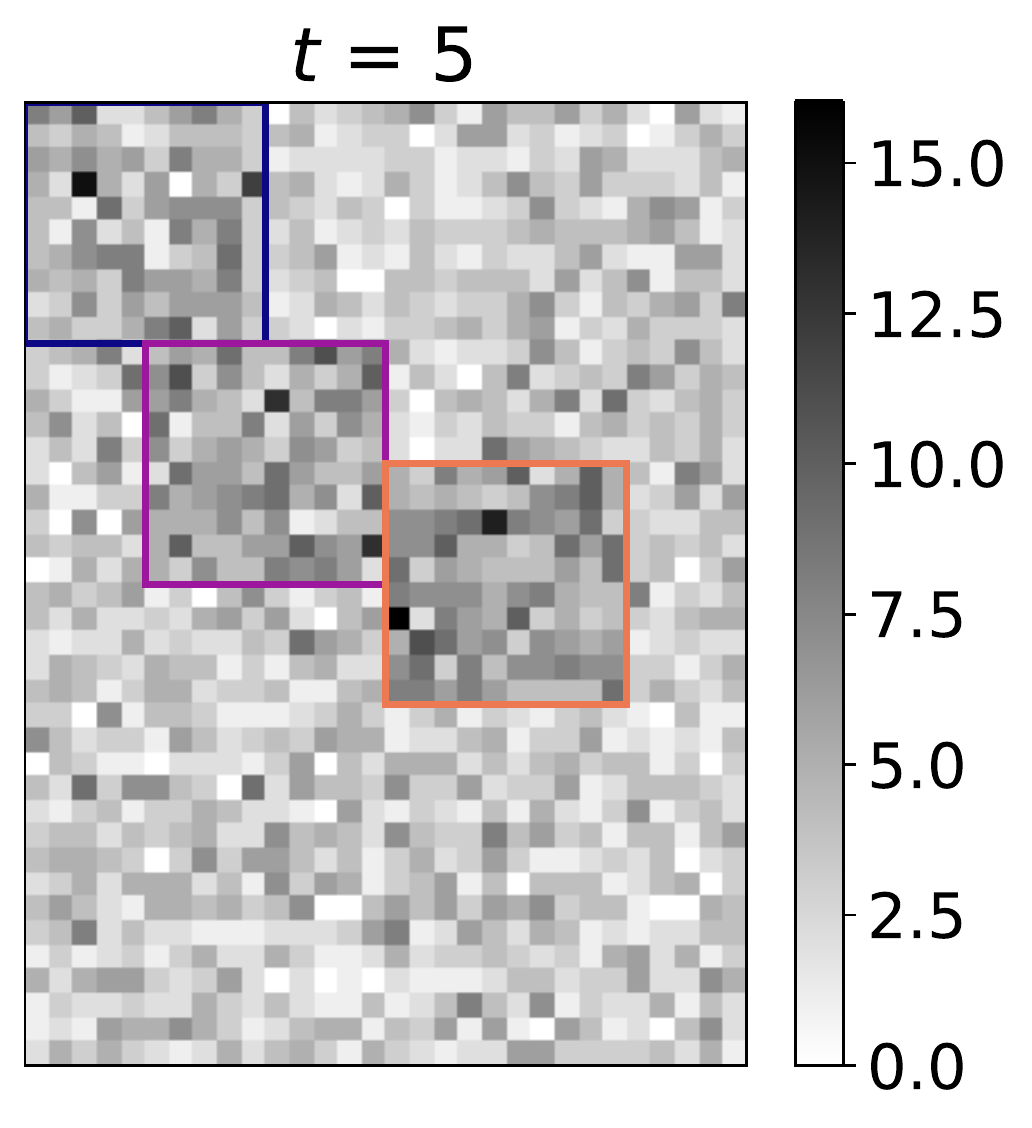}\\
  \includegraphics[width=0.17\hsize]{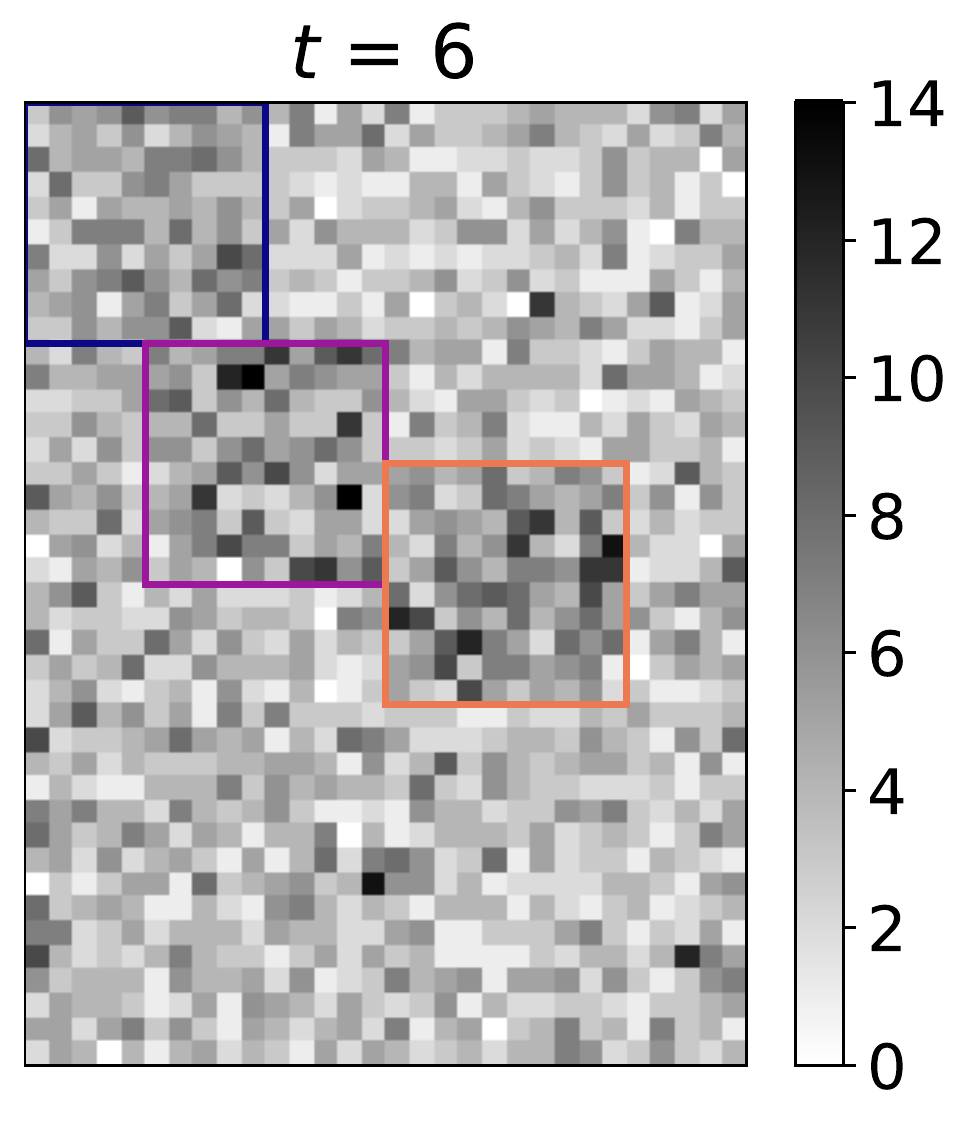}
  \includegraphics[width=0.17\hsize]{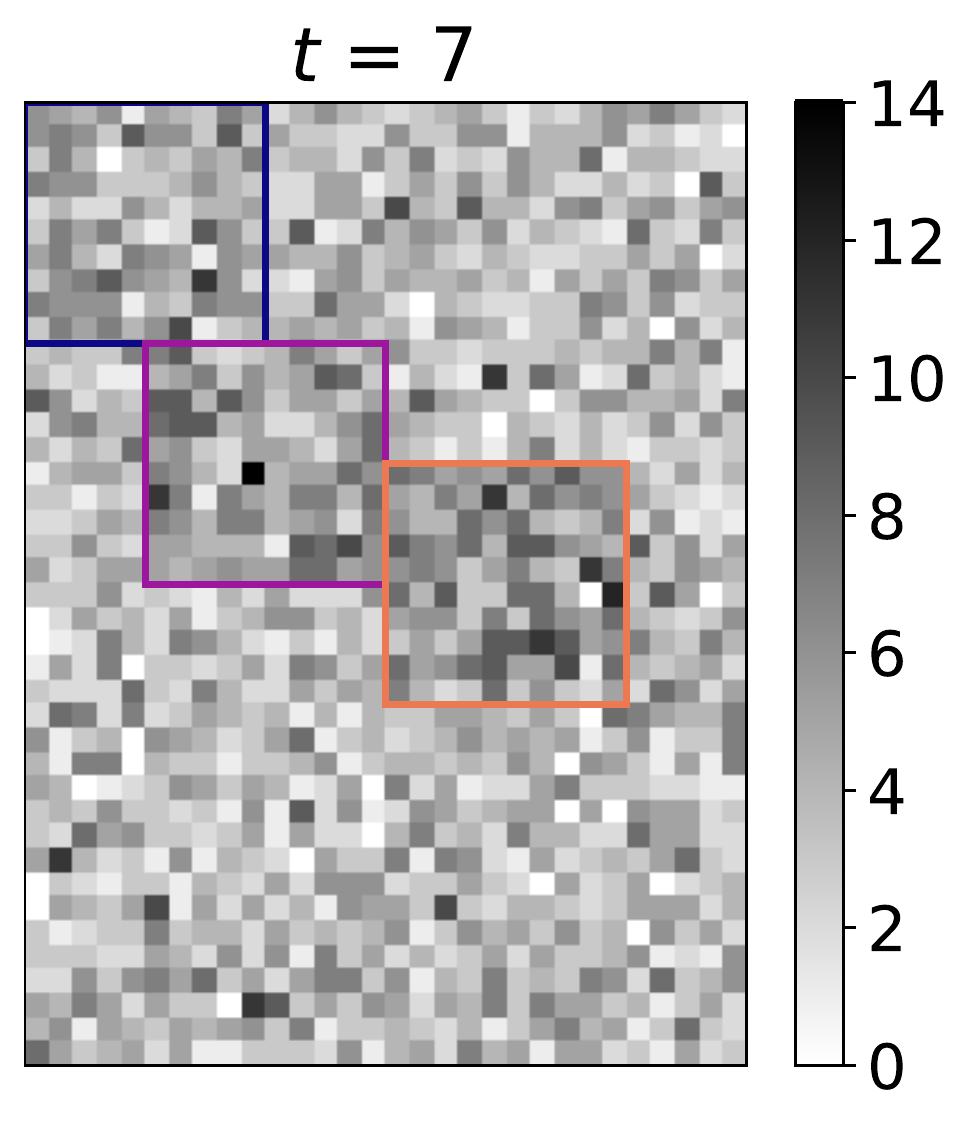}
  \includegraphics[width=0.17\hsize]{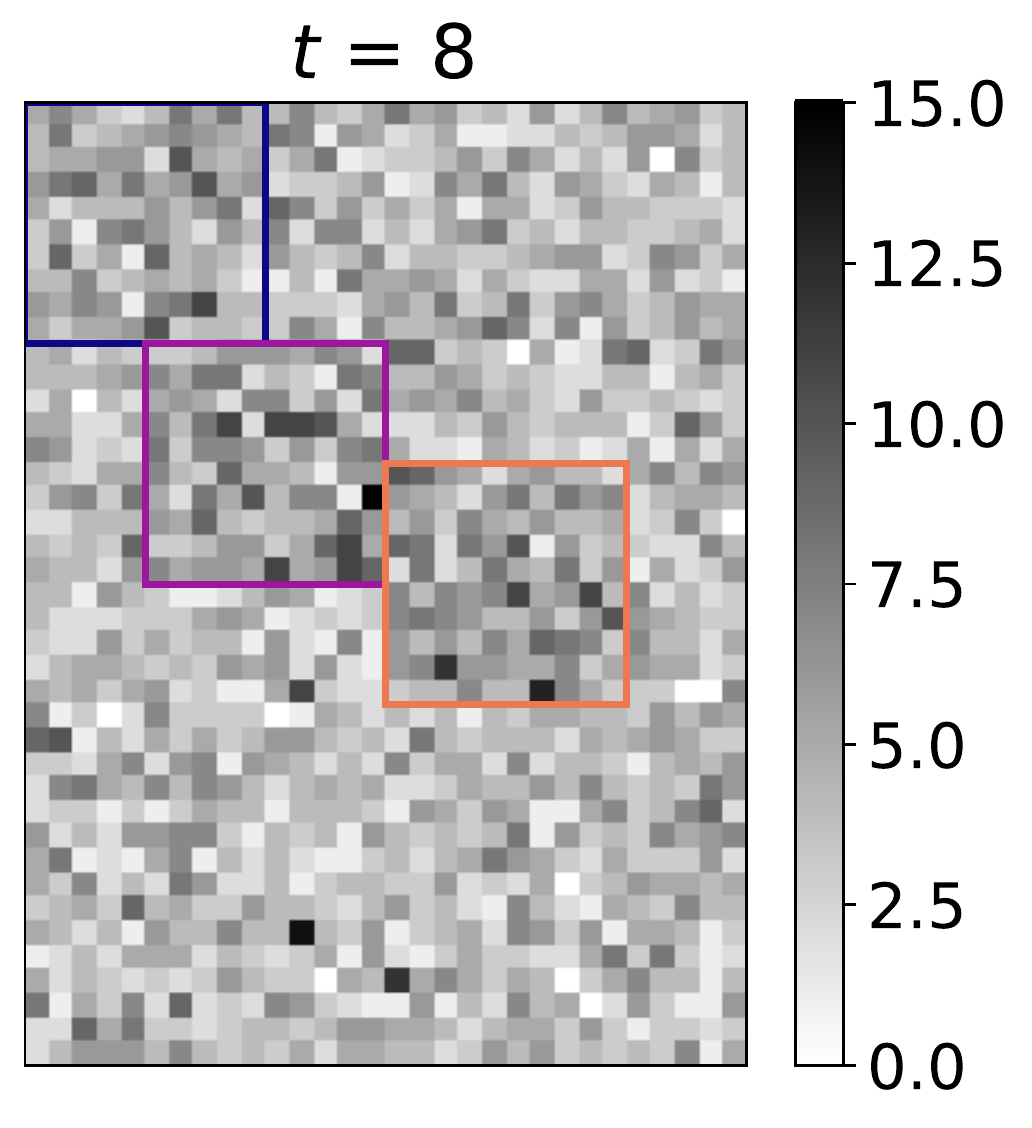}
  \includegraphics[width=0.17\hsize]{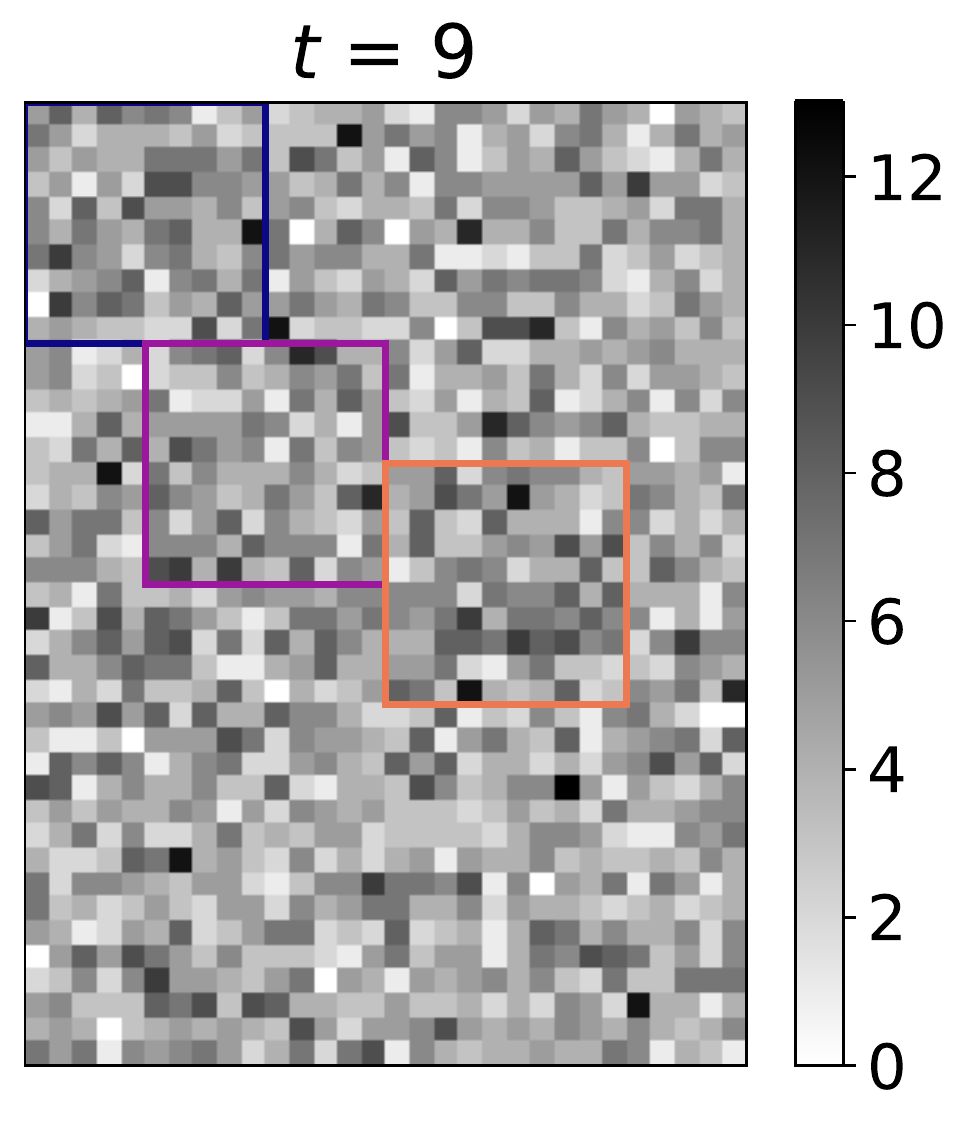}
  \includegraphics[width=0.17\hsize]{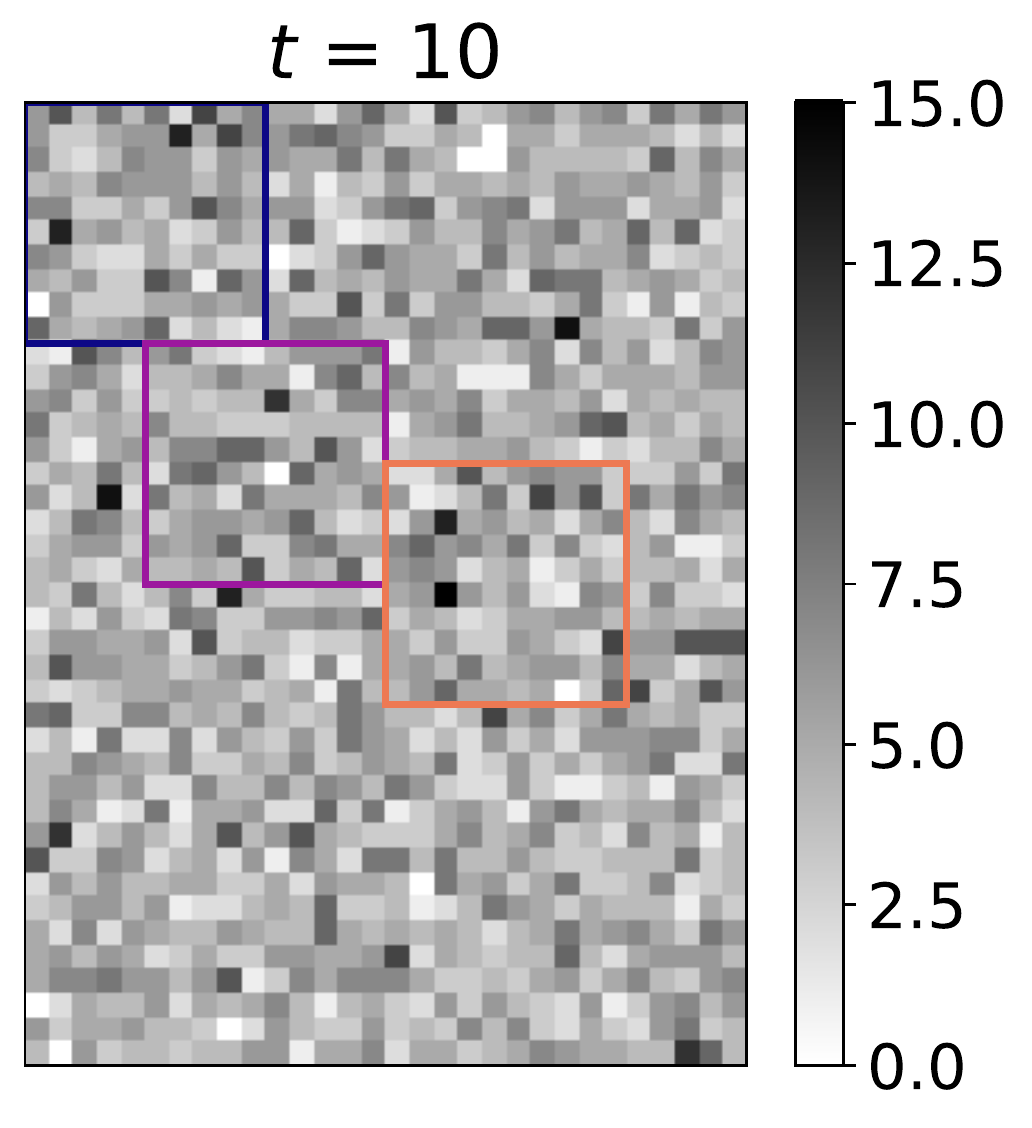}
  \caption{Examples of the observed data matrices for $t = 1, \dots, 10$ (\textbf{Poisson case}).}
  \label{fig:A_example_acc_p}
\end{figure}

Figure \ref{fig:accuracy} shows the accuracy of the proposed test, that is, the ratio of trials where the selected number of biclusters $\hat{K}$ was equal to the null one $K$. From Figure \ref{fig:accuracy}, it is clear that the proposed test achieved higher accuracy with the larger matrix sizes and with the smaller differences between the group-wise means. This result is consistent with our intuition, since larger matrix sizes and smaller differences between the elements in mean vector tend to make it more difficult to correctly estimate the underlying bicluster structure of matrix $A$, based on which we computed the test statistic $T$. 

\begin{figure}[t]
  \centering
  \includegraphics[width=0.325\hsize]{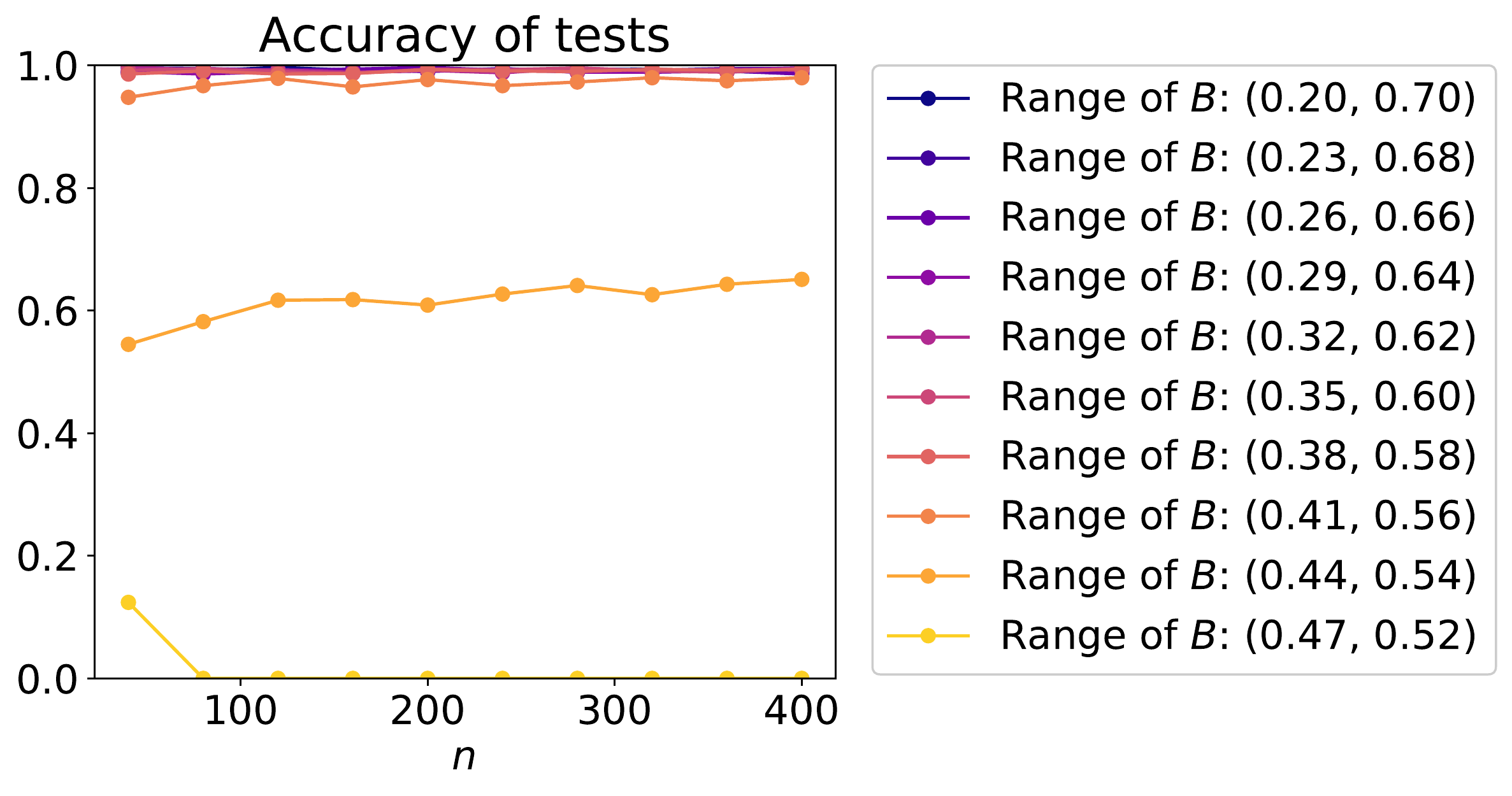}
  \includegraphics[width=0.325\hsize]{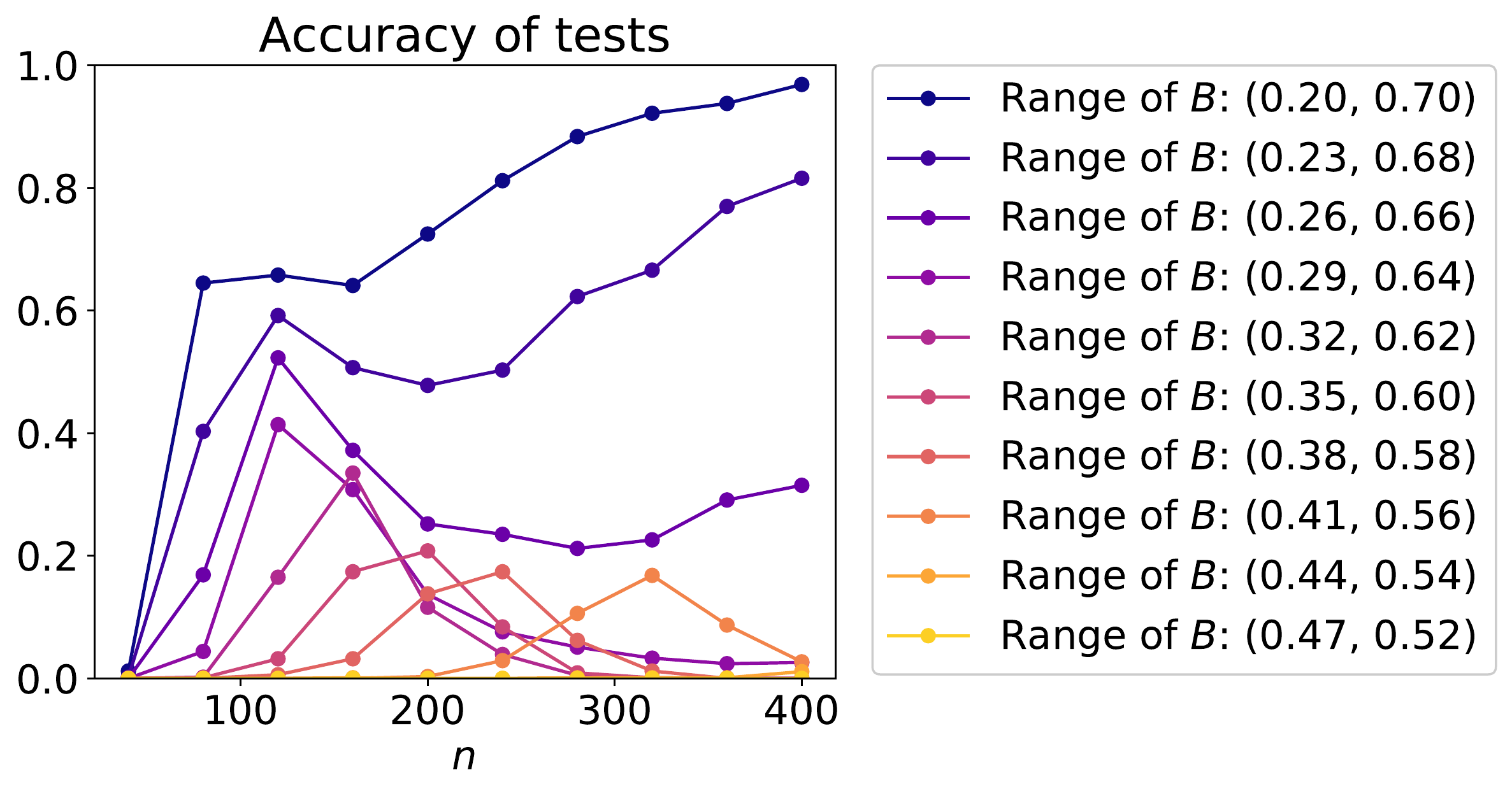}
  \includegraphics[width=0.325\hsize]{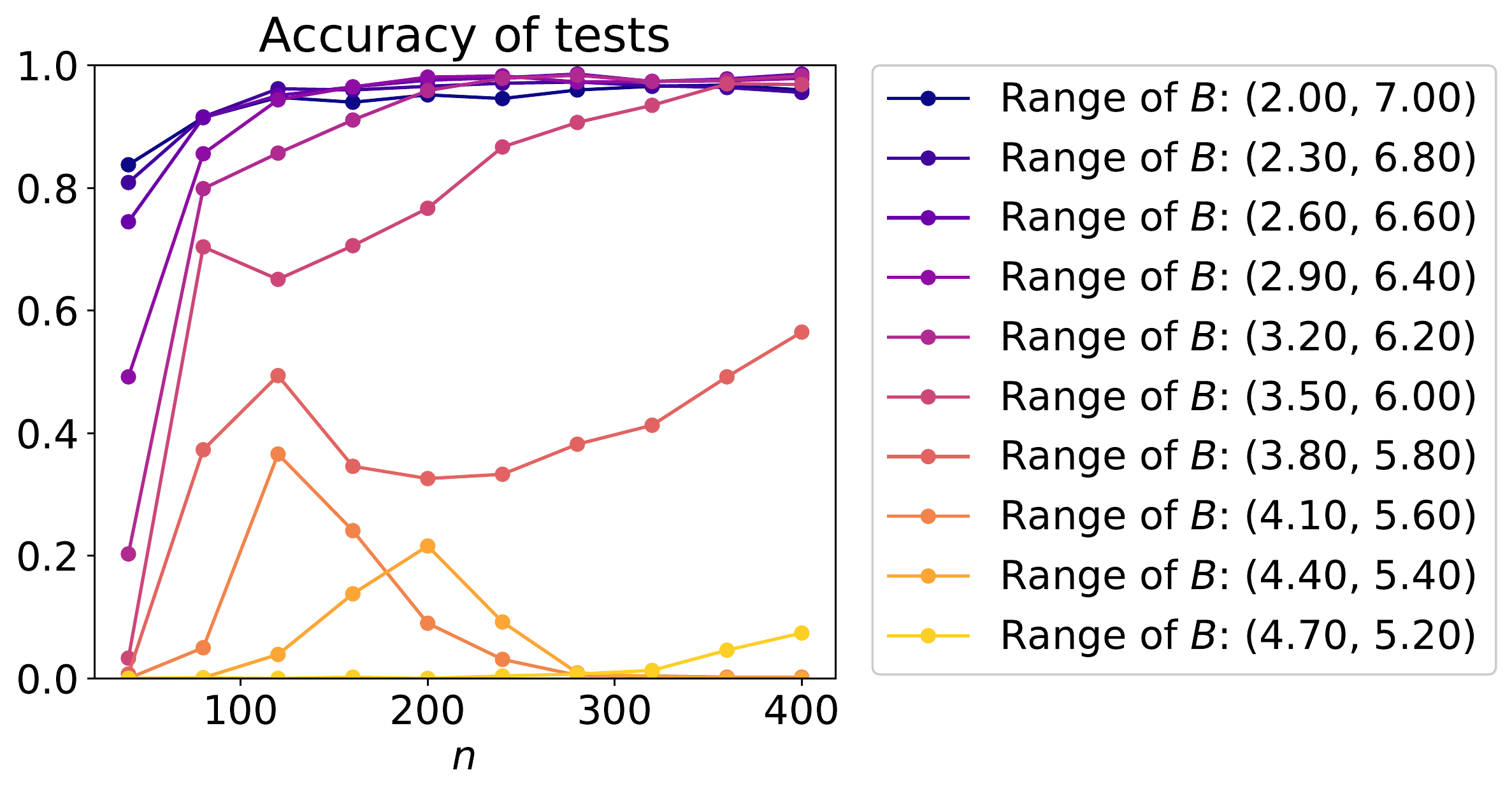}
  \caption{The accuracy of the proposed test in selecting the number of biclusters $K$, under $10$ different mean parameter settings $\{\bm{b}^{(1)}, \dots, \bm{b}^{(10)}\}$. The left, center, and right figures, respectively, illustrate the results where each entry of observed matrix $A$ was generated using Gaussian, Bernoulli, and Poisson distributions.}
  \label{fig:accuracy}
\end{figure}

\subsection{Goodness-of-fit test and model selection with practical data set}
\label{sec:exp_practical}

Finally, we applied the proposed test and the conventional LBM-based one \cite{Watanabe2021} to the Divorce Predictors data set \cite{Yontem2019} from the UCI Machine Learning Repository \cite{Dua2017}, and compared the results. The rows and columns of the original observed matrix $\check{A} \in \mathbb{R}^{170 \times 54}$ represent the $170$ participants and $54$ attributes, respectively, and each $(i, j)$th entry shows the Divorce Predictors Scale (DPS), which takes values of $0, 1, \dots, 4$. According to \cite{Yontem2017}, the original questionnaire was done based on the following five-factor scale: $0$: ``Never,'' $1$: ``Rarely,'' $2$: ``Occasionally,'' $3$: ``Often,'' and $4$: ``Always,'' which was used as a score for Attributes $31$ to $54$. As for Attributes $1$ to $30$, this scale was reversed (i.e., $0$ meant ``Always'' and $4$ meant ``Never'') so that higher values indicated a higher divorce risk in all the attributes. 
Based on the original matrix $\check{A}$, we defined a binary data matrix $A$ by setting $A_{ij} = 1$ if $\check{A}_{ij} \geq 2$ for the pair of $i$th participant and the $j$th attribute, and $A_{ij} = 0$ otherwise. The upper left section of Figure \ref{fig:practical_A} depicts the observed data matrix, where the meaning of each attribute index is shown in Table \ref{tb:divorce_attribute}. 

As for the proposed test, we applied it sequentially as in Sect.~\ref{sec:accuracy} with a significance level of $\alpha = 0.01$ until some hypothetical number of biclusters was accepted. In the SA algorithm, we used the relative entropy function $f$ in (\ref{eq:f_bernoulli}) and the cooling schedule of $T_t = 0.9999^t$ for all $t \geq 0$. For each hypothetical number of biclusters $K_0$, we set the threshold at $\epsilon^{\mathrm{SA}} = 10^{-K_0/2.5 - 2}$. 
Based on these settings, we applied the Algorithm \ref{algo:max_PL_SA2} $30$ times and adopted the best solution that achieved the maximum profile likelihood in the last step. Based on the estimated bicluster structure, we applied the proposed statistical test. 

Regarding the conventional LBM-based test, we used the same settings as those employed by Watanabe and Suzuki \cite{Watanabe2021}. That is, for each hypothetical set of row and column cluster numbers $(K_0, H_0)$, we estimated the regular-grid bicluster structure by applying Ward's hierarchical clustering method \cite{Ward1963} to the rows and columns of observed matrix. Based on the estimated row and column cluster assignments, we applied the test in \cite{Watanabe2021} with a significance level of $\alpha = 0.01$. We tried multiple combinations of $K_0$ and $H_0$ in the following order: 
\begin{align}
(K_0, H_0) = (1, 1), (1, 2), (2, 1), (1, 3), (2, 2), (3, 1), \dots, 
\end{align}
until the null hypothesis was accepted. 

Based on the above settings, the estimated number of biclusters by the proposed test was $30$, while the estimated set of row and column cluster numbers by the conventional one \cite{Watanabe2021} was $(14, 46)$ (i.e., the estimated number of biclusters was $644$, aside from the background). The upper right and bottom sections of Figure \ref{fig:practical_A} show, respectively, the estimated bicluster structure when the null hypotheses were accepted by the proposed and previous LBM-based tests. From these results, we see that that the proposed test could capture the bicluster structure more flexibly than the previous regular-grid based one, and thereby accepted the smaller hypothetical number of biclusters. 

More specifically, Figure \ref{fig:practical_bic_proposed} shows each estimated bicluster when the null hypothesis was accepted by the proposed test. Most of the estimated biclusters contained constant values (i.e., $0$ or $1$), except for Biclusters $0$, $3$, $18$, $19$, $23$, and $28$, and many biclusters consisted of a small number of attributes (i.e., one or two). Except Biclusters $6$ and $27$, most of the rows (i.e., participants) in each bicluster belonged to the same class (i.e., divorced or married). Each estimated bicluster was composed of some homogeneous sets of rows and columns: for example, Bicluster $9$ shows that there exists a (mostly) married group of participants who gave small DPS (i.e., $\check{A}_{ij} \leq 1$) to the attributes $24$ and $27$, both of which were related to knowledge about the stress of their spouse. From Bicluster $16$, we also see that there existed a divorced group of participants who gave large DPS (i.e., $\check{A}_{ij} \geq 2$) to many attributes, including their similarity to the spouse (e.g., attributes from $12$ to $20$) and their awareness of the spouse (e.g., attributes from $21$ to $30$). 

\begin{figure}[t]
  \centering
  \includegraphics[width=0.495\hsize]{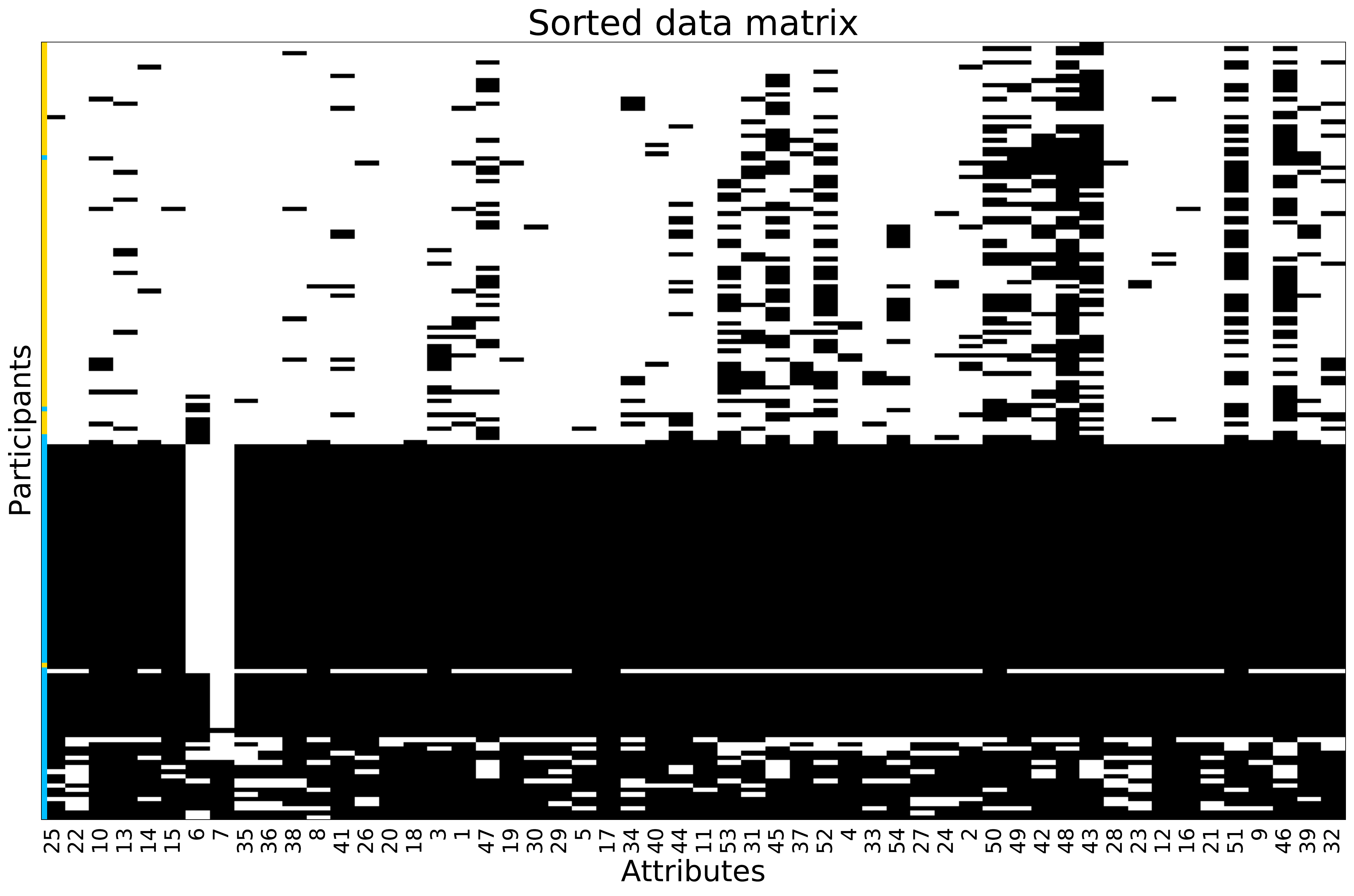}
  \includegraphics[width=0.495\hsize]{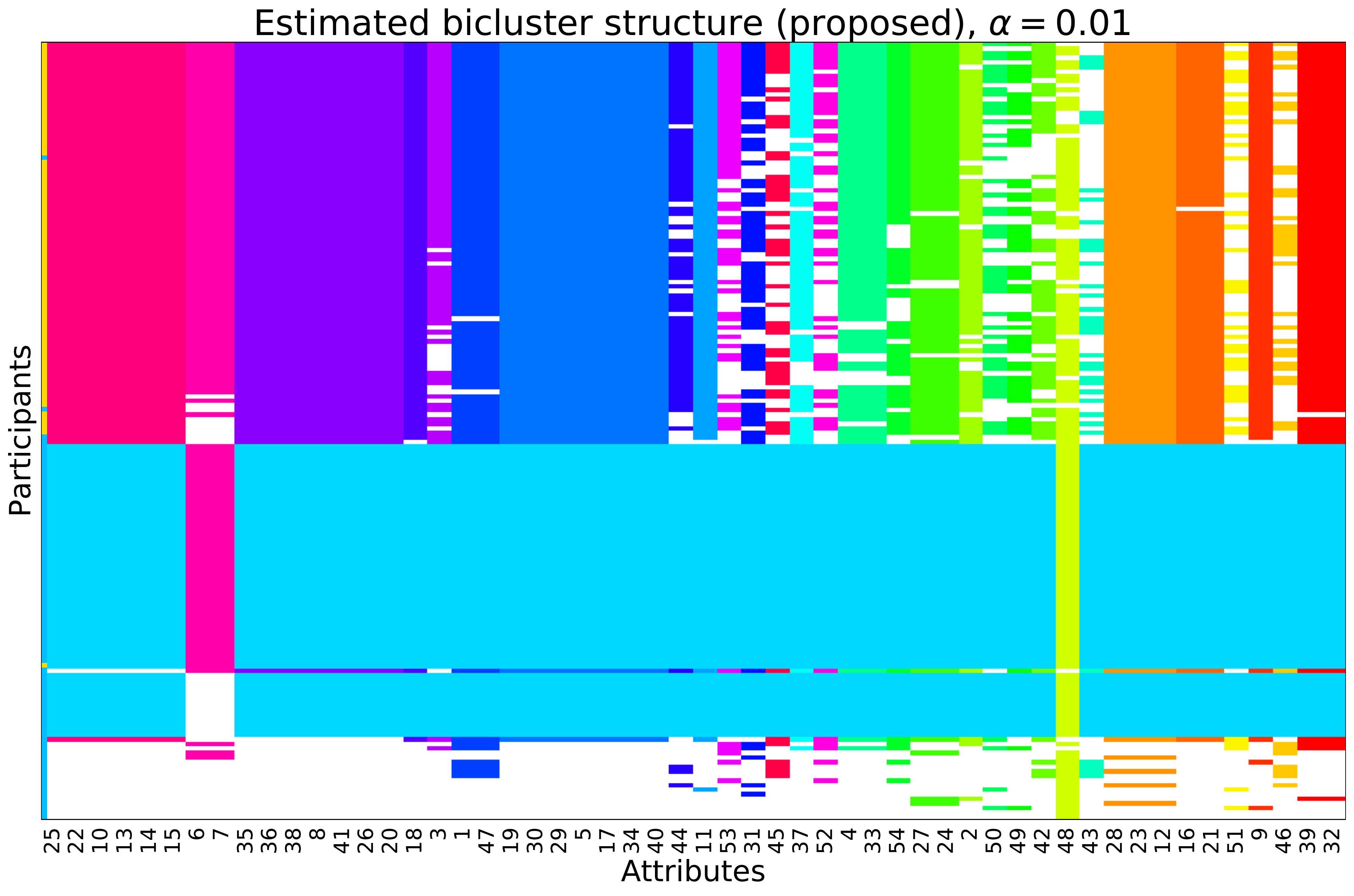}
  \includegraphics[width=0.495\hsize]{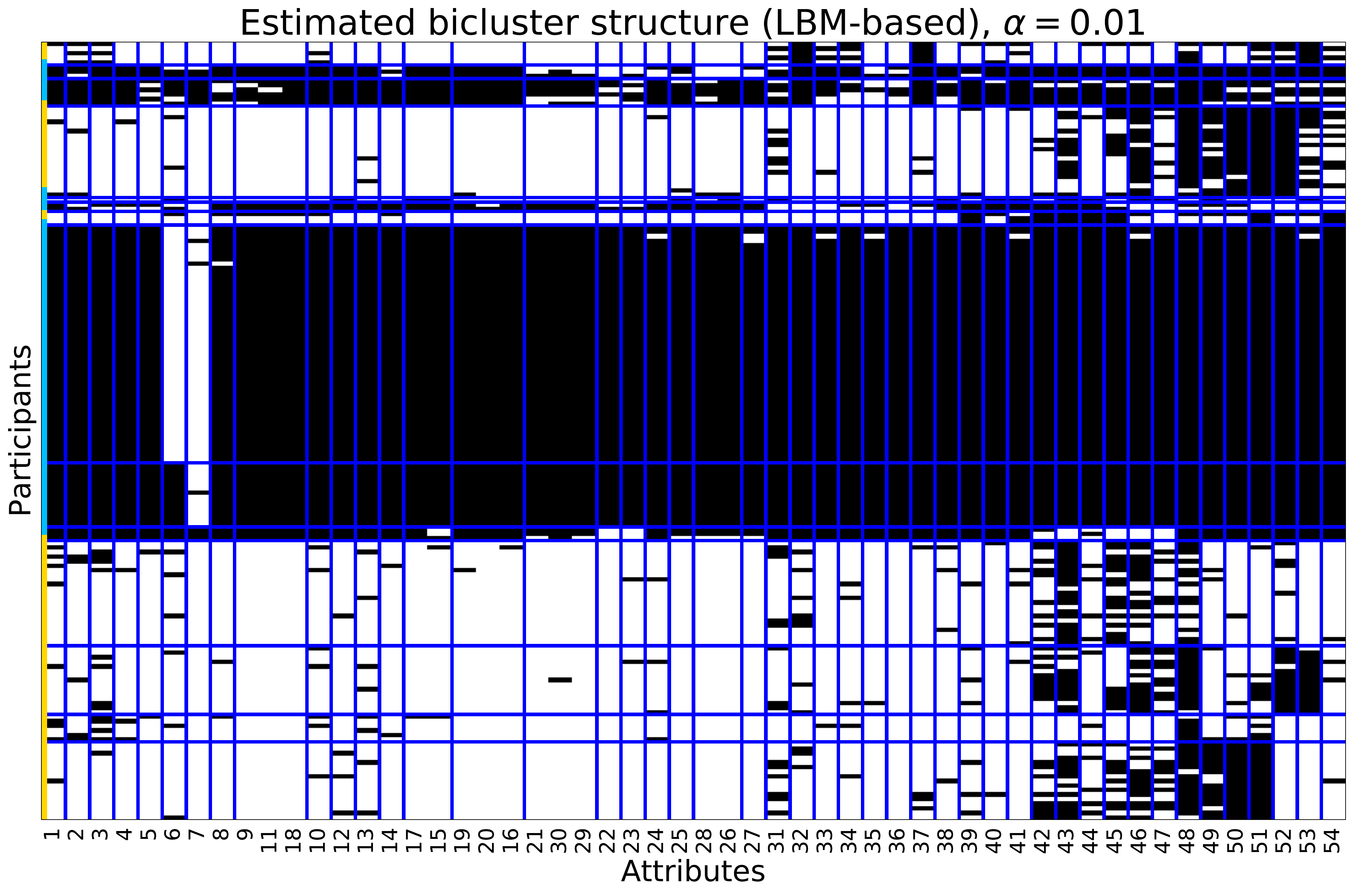}
  \includegraphics[width=0.2\hsize]{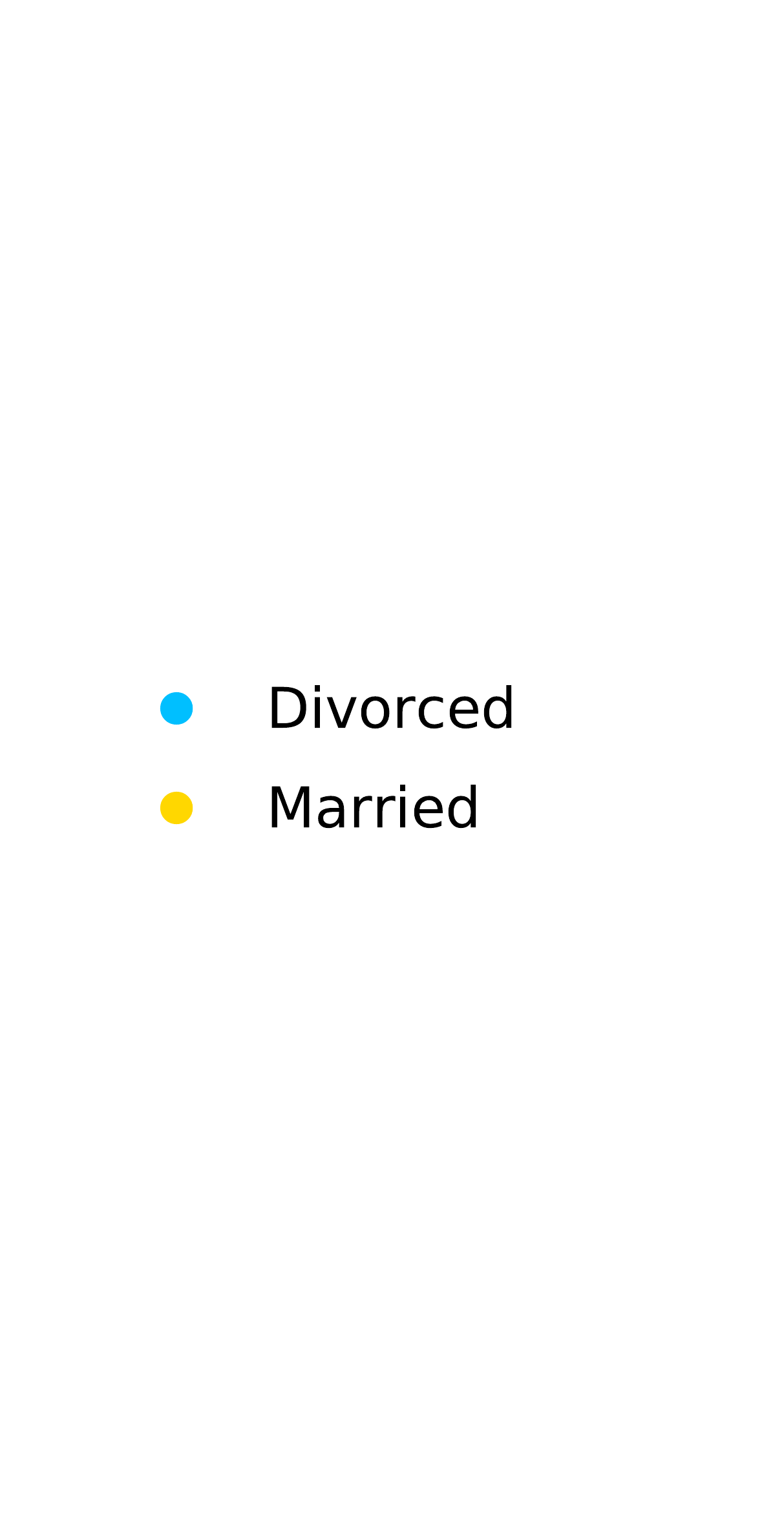}
  \caption{The sorted observed data matrix of the Divorce Predictors data set \cite{Yontem2019} (upper left) and its estimated bicluster structures when the null hypotheses were accepted by the proposed test (upper right) and the previous one (bottom). In the upper left and bottom figures, the black and white elements represent one and zero, respectively. In the upper right figure, the sorting orders of the rows and columns are the same as in the upper left figure, and the color of each element indicates its group index (the white elements were estimated as background), regardless of its value. In the bottom figure, the blue lines represent the regular grid bicluster structure (note that in this figure, the sorting orders of the rows and columns are different from those of the upper left figure. The meaning of each attribute index is shown in Table \ref{tb:divorce_attribute}.}
  \label{fig:practical_A}
\end{figure}
\begin{figure}[t]
  \centering
  \includegraphics[height=0.2\hsize]{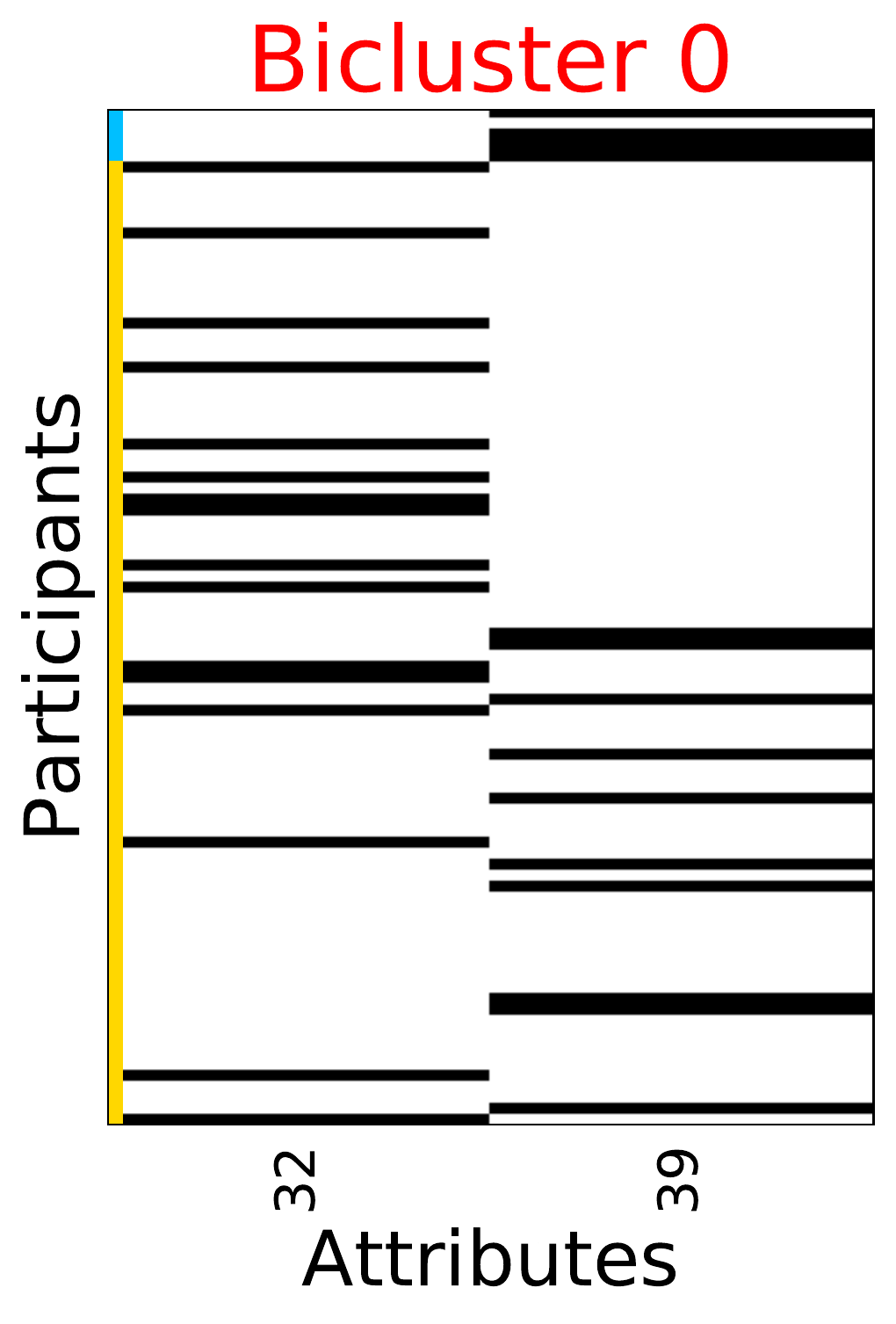}
  \includegraphics[height=0.2\hsize]{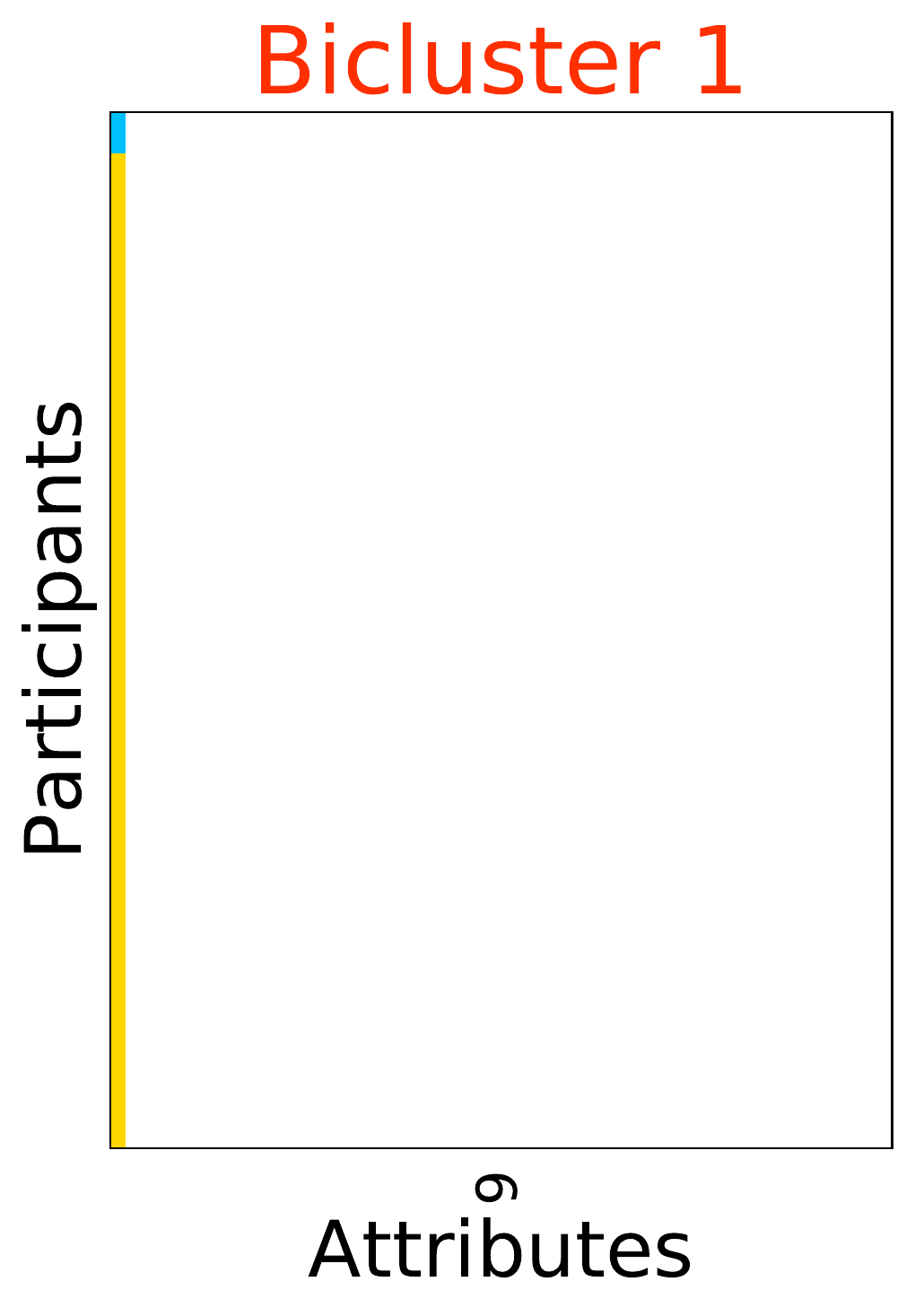}
  \includegraphics[height=0.2\hsize]{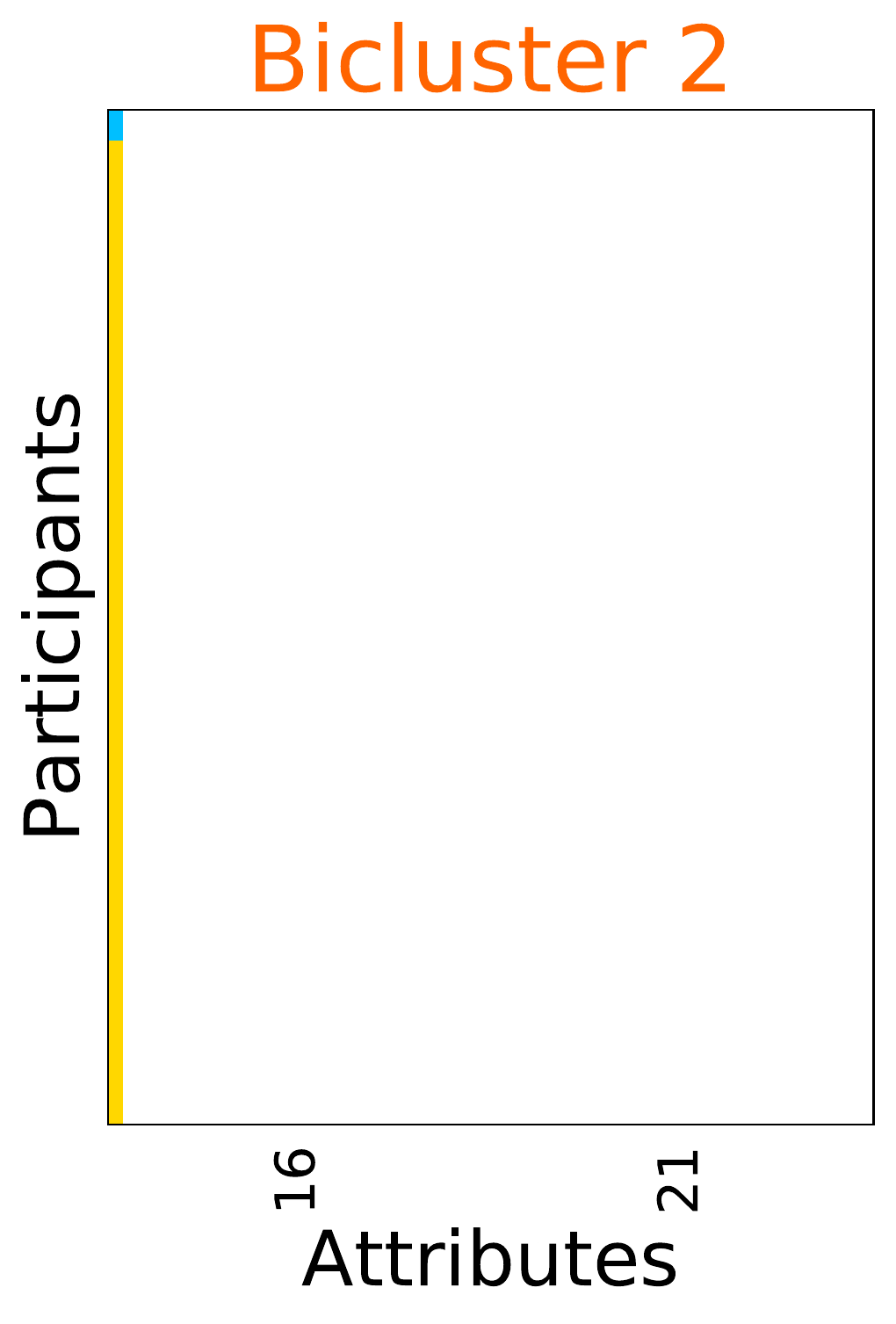}
  \includegraphics[height=0.2\hsize]{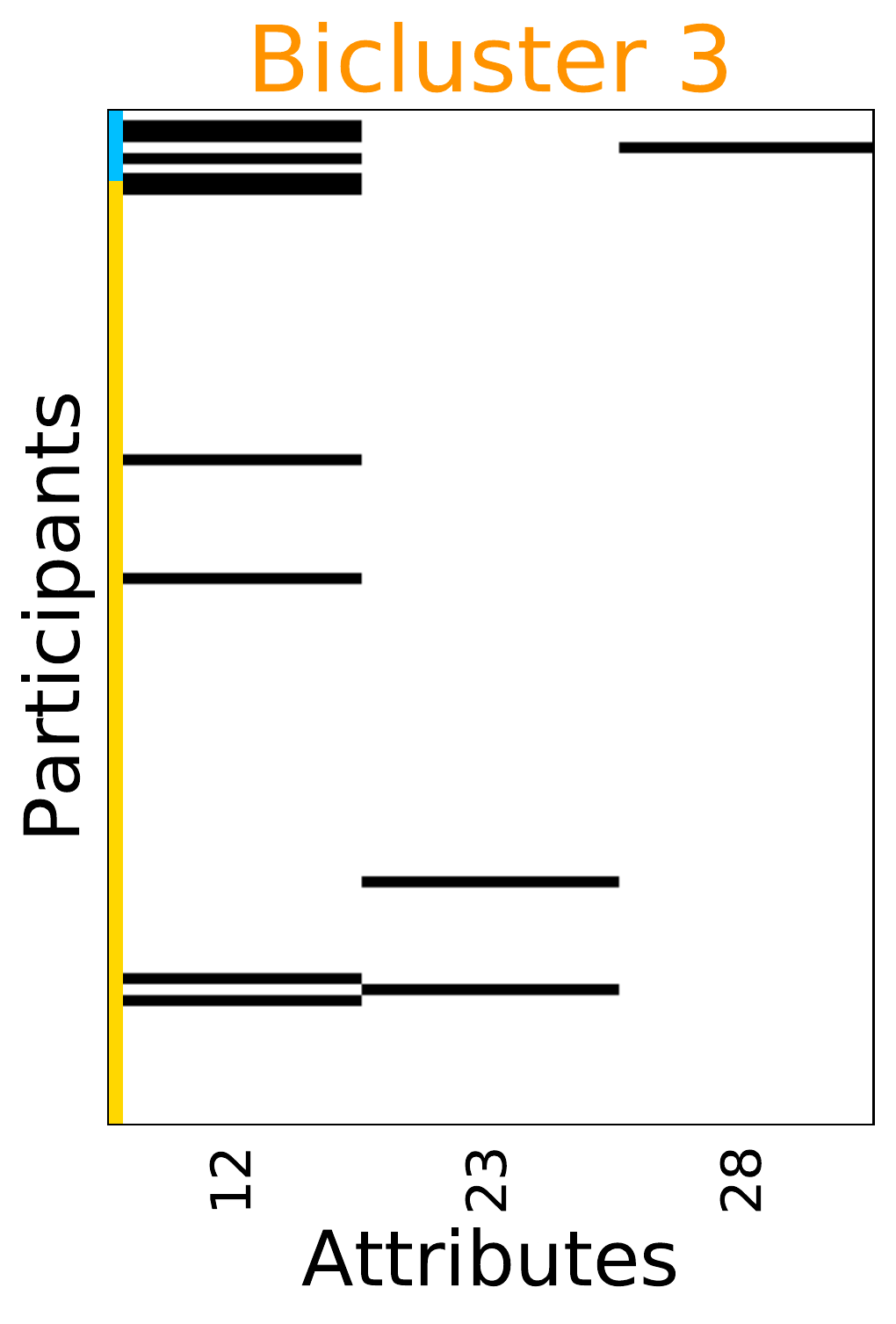}
  \includegraphics[height=0.2\hsize]{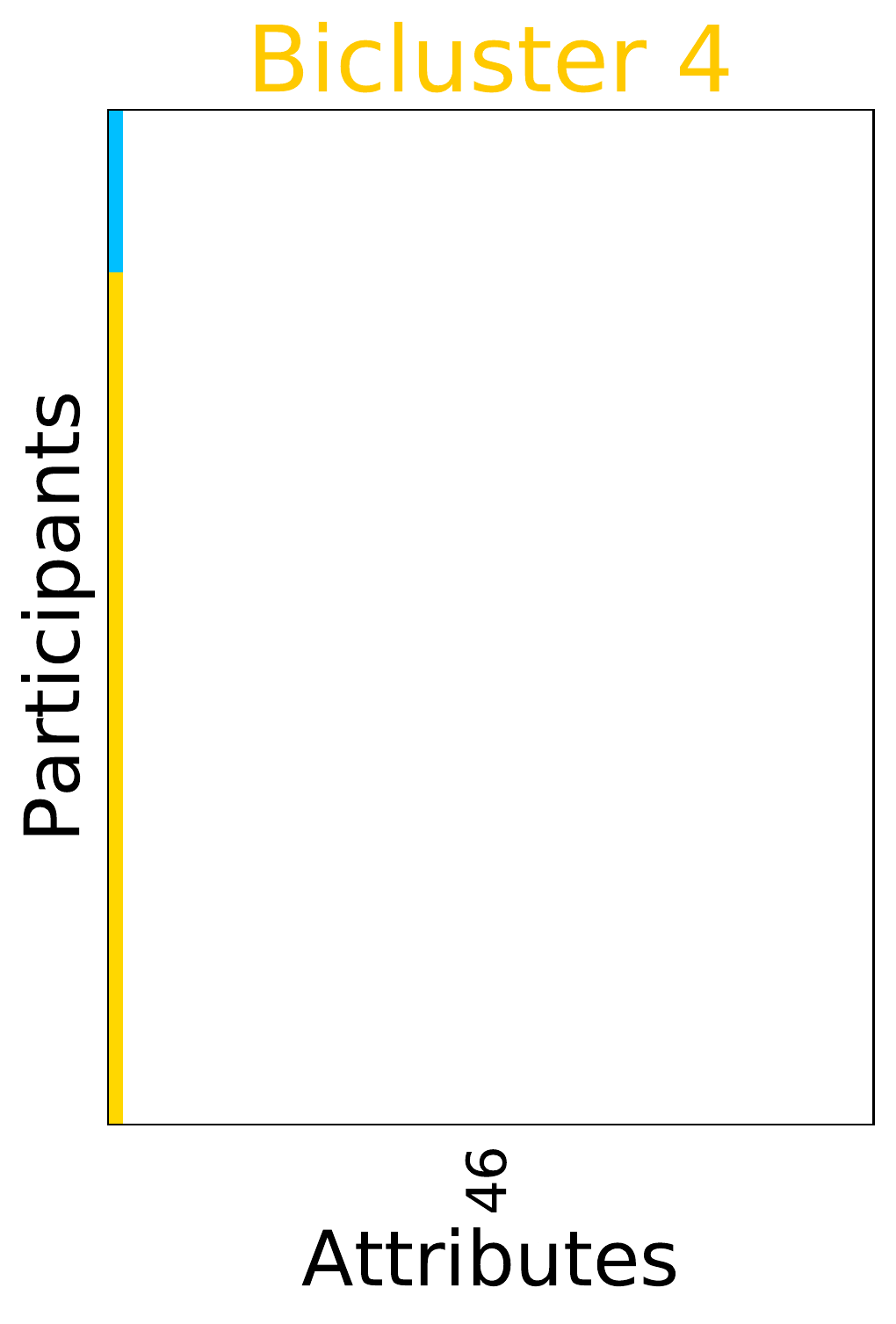}
  \includegraphics[height=0.2\hsize]{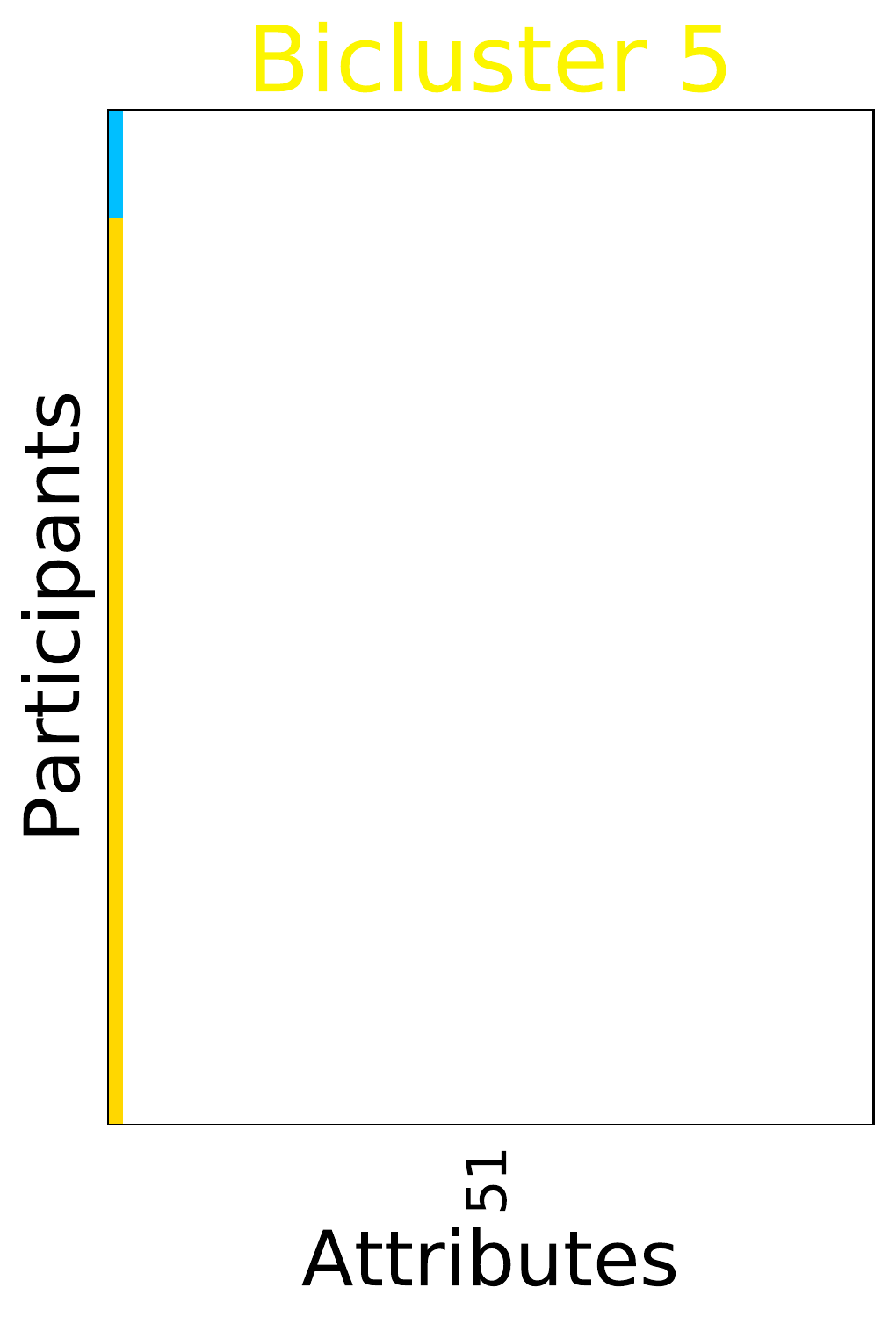}
  \includegraphics[height=0.2\hsize]{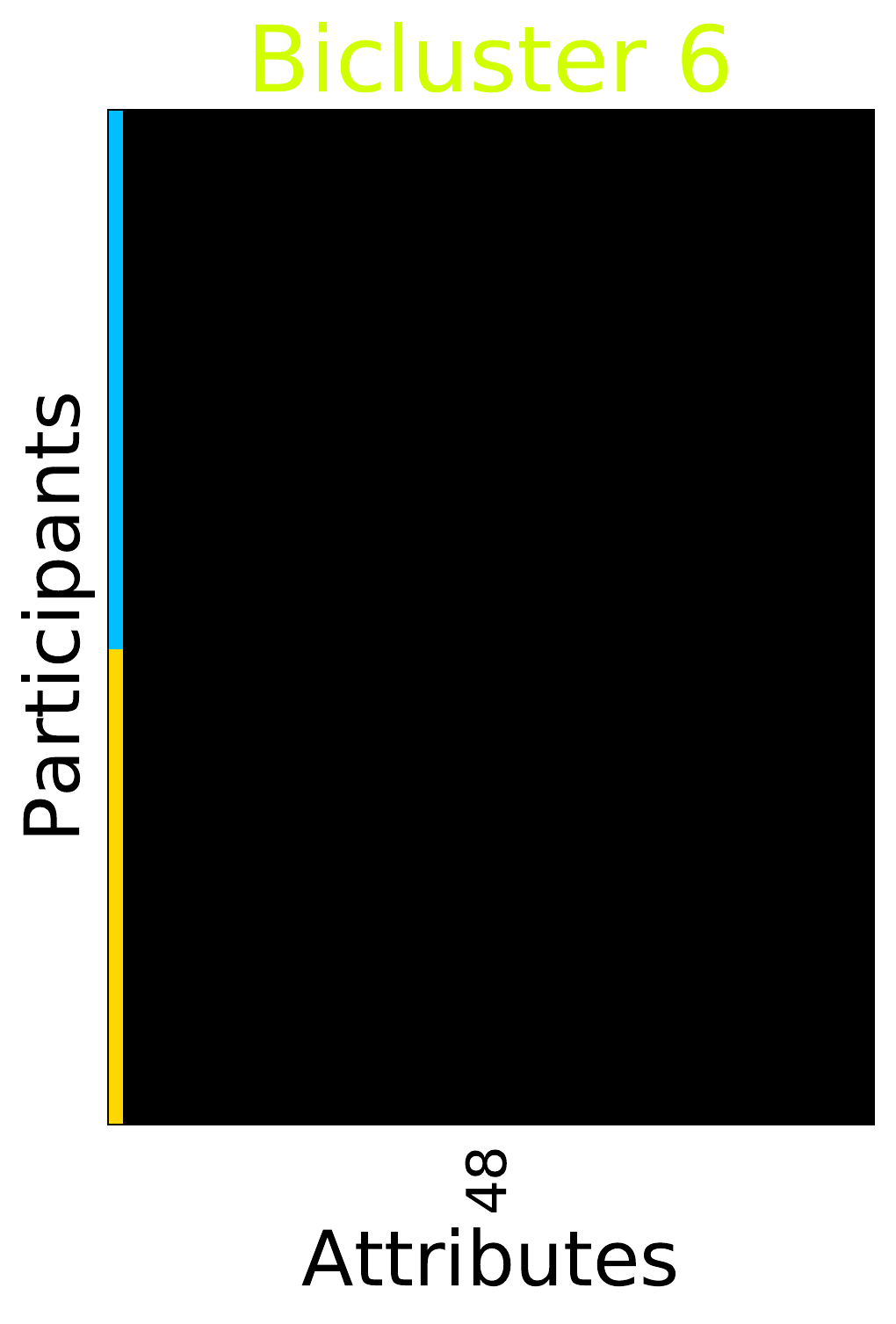}
  \includegraphics[height=0.2\hsize]{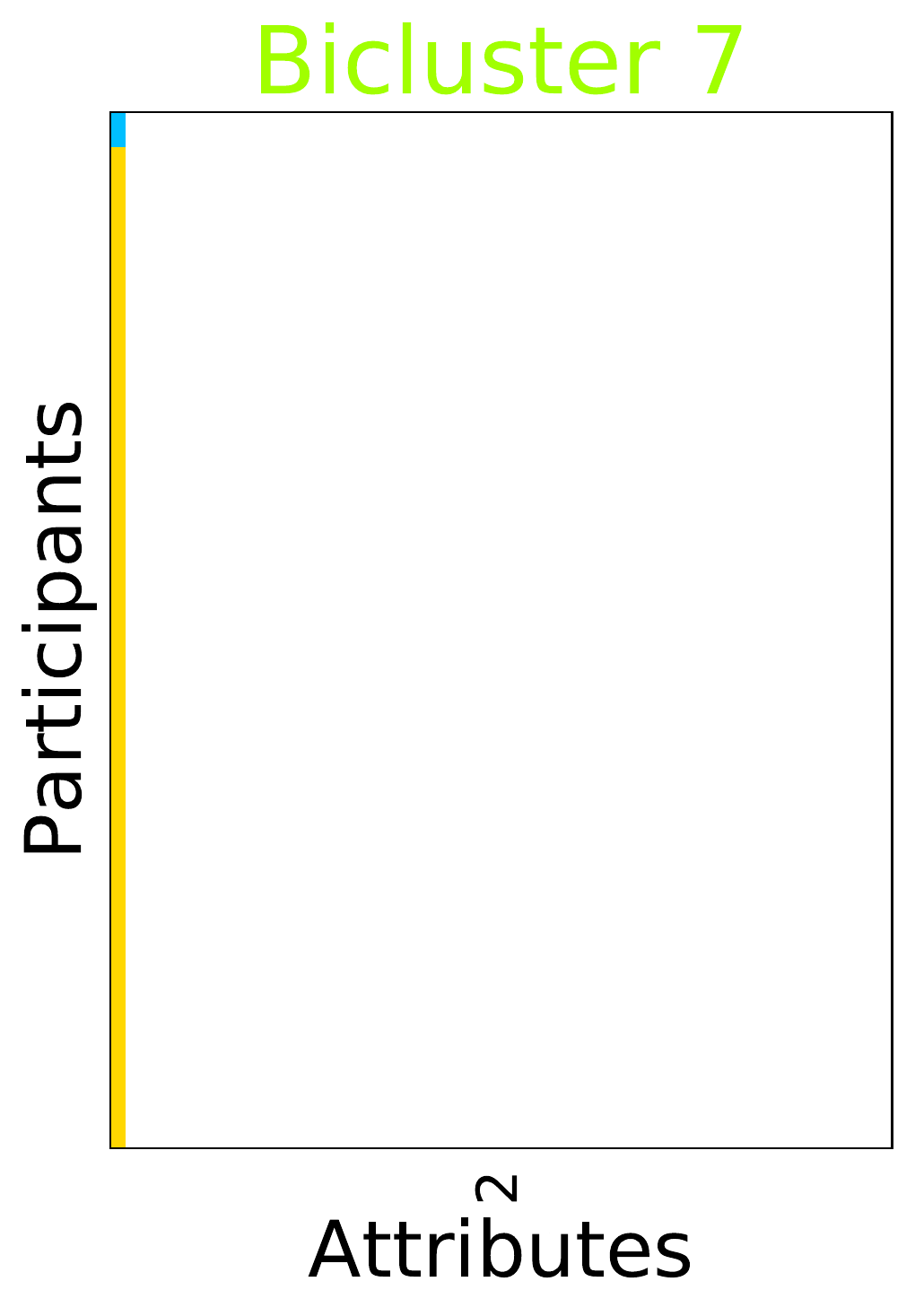}
  \includegraphics[height=0.2\hsize]{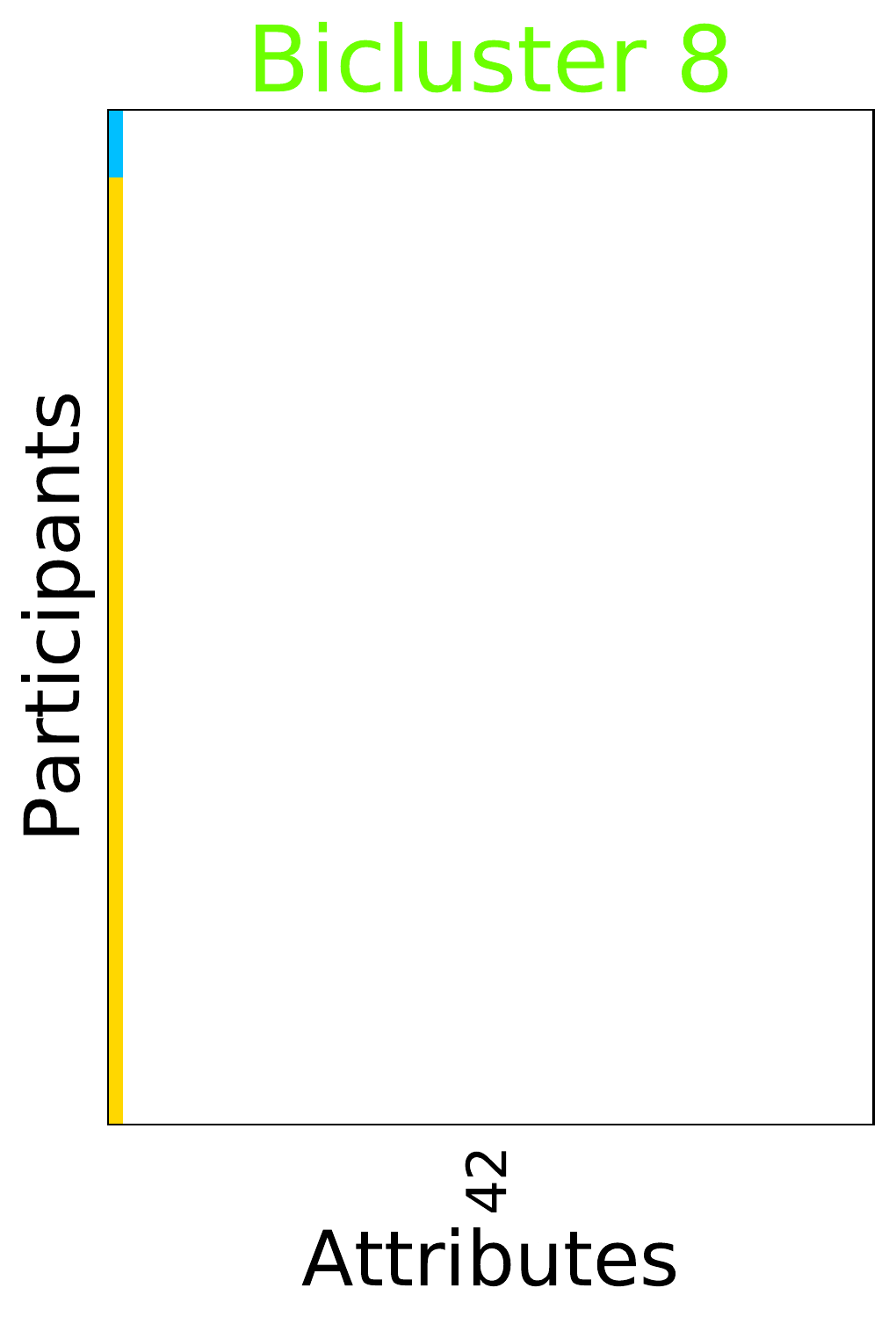}
  \includegraphics[height=0.2\hsize]{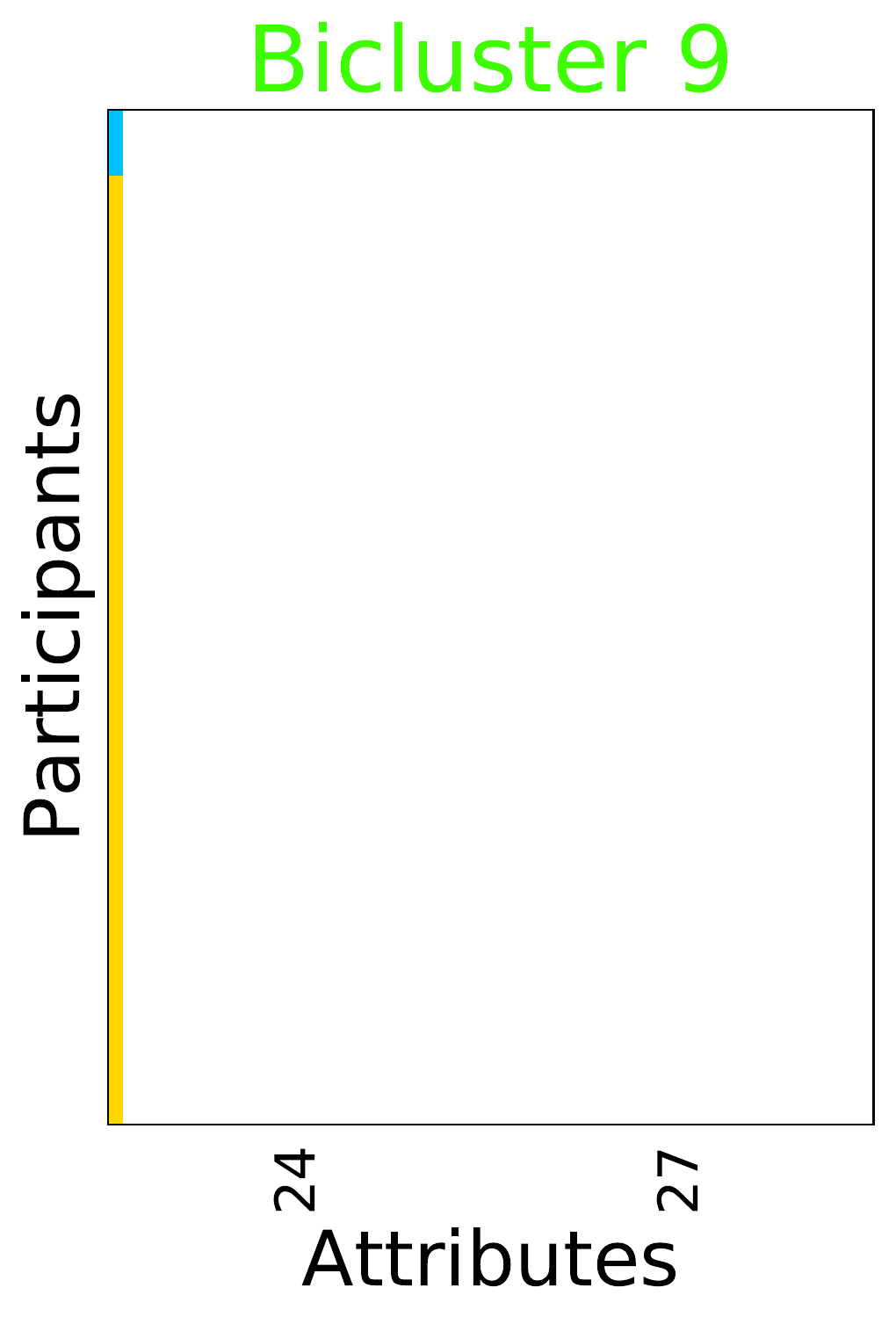}
  \includegraphics[height=0.2\hsize]{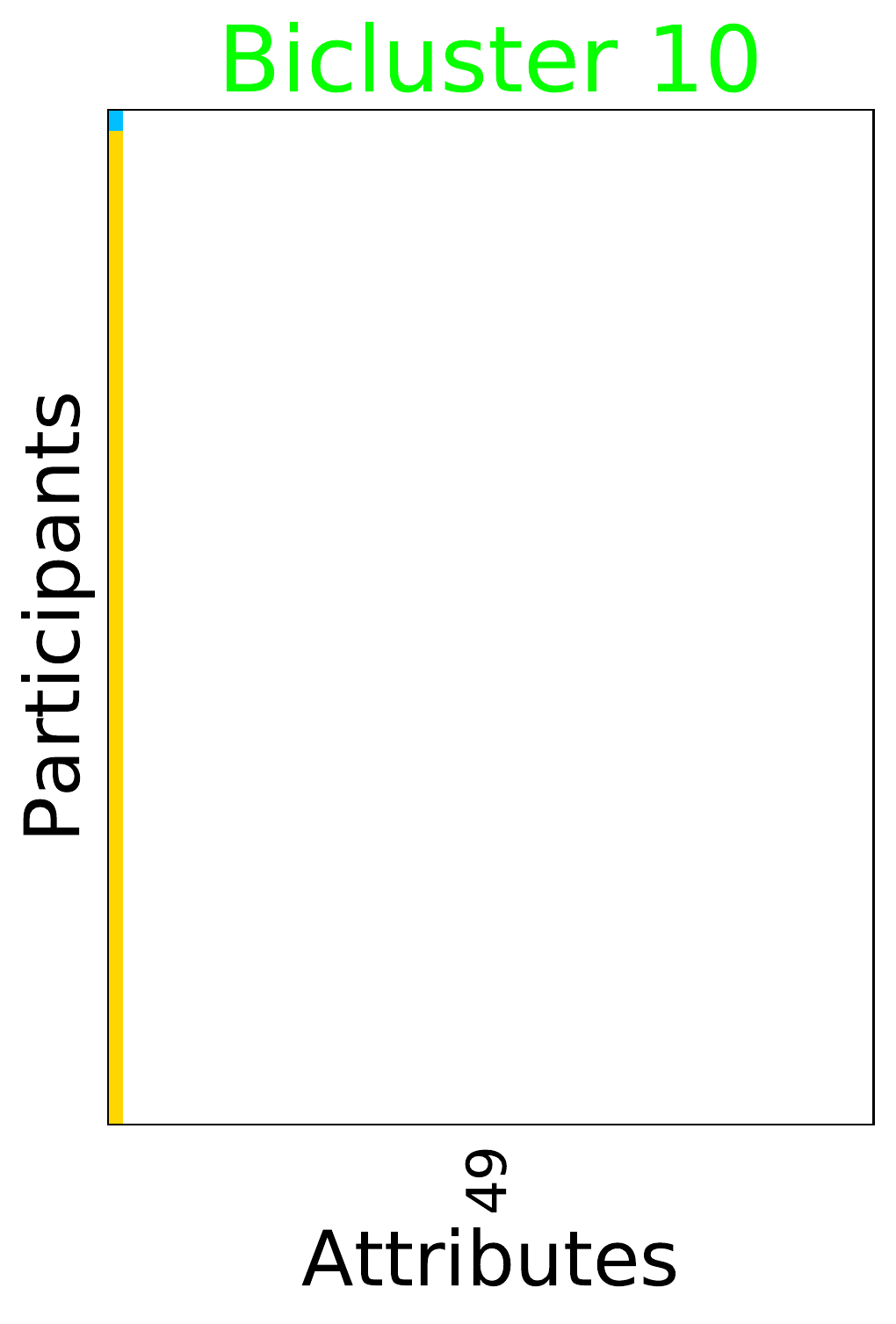}
  \includegraphics[height=0.2\hsize]{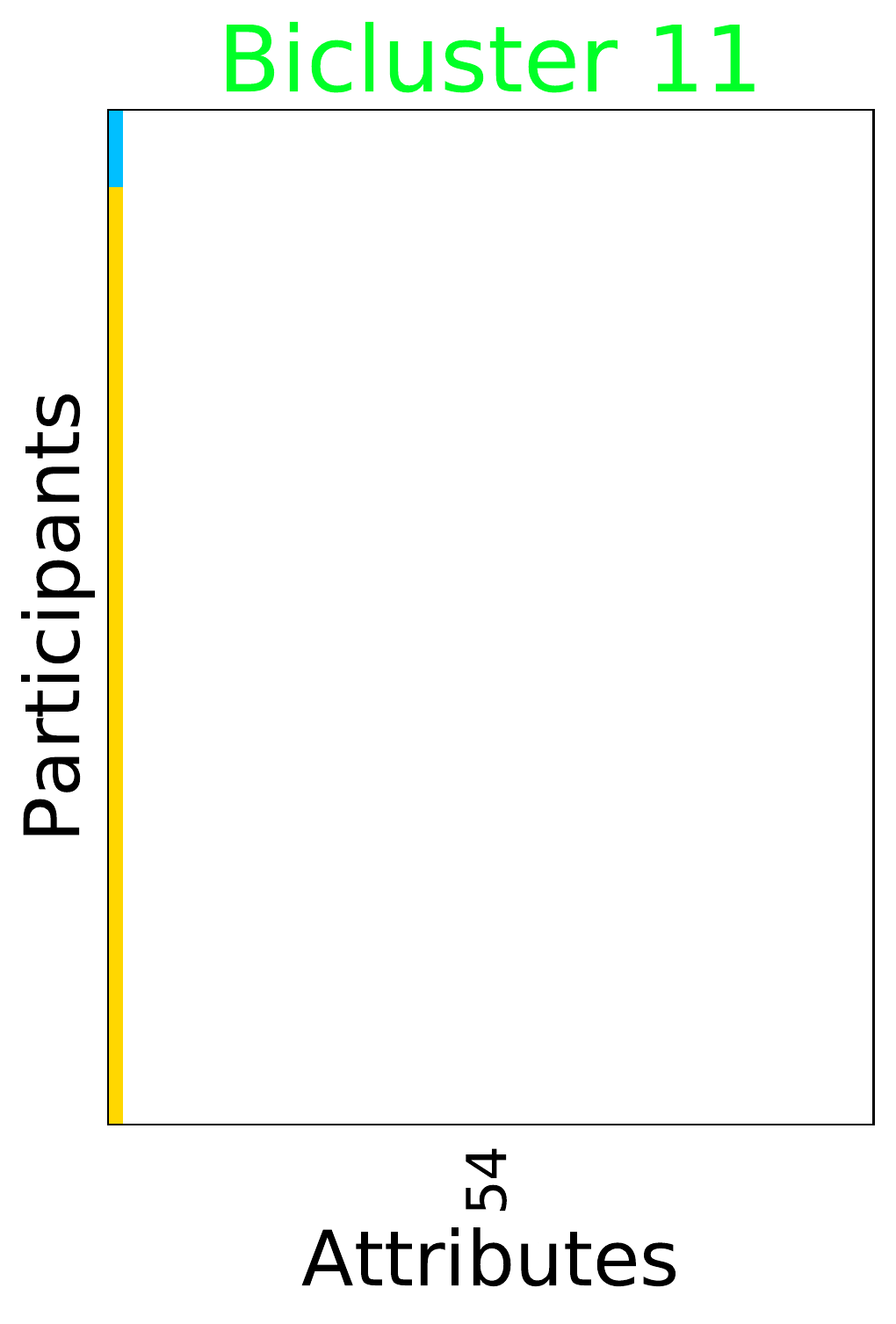}
  \includegraphics[height=0.2\hsize]{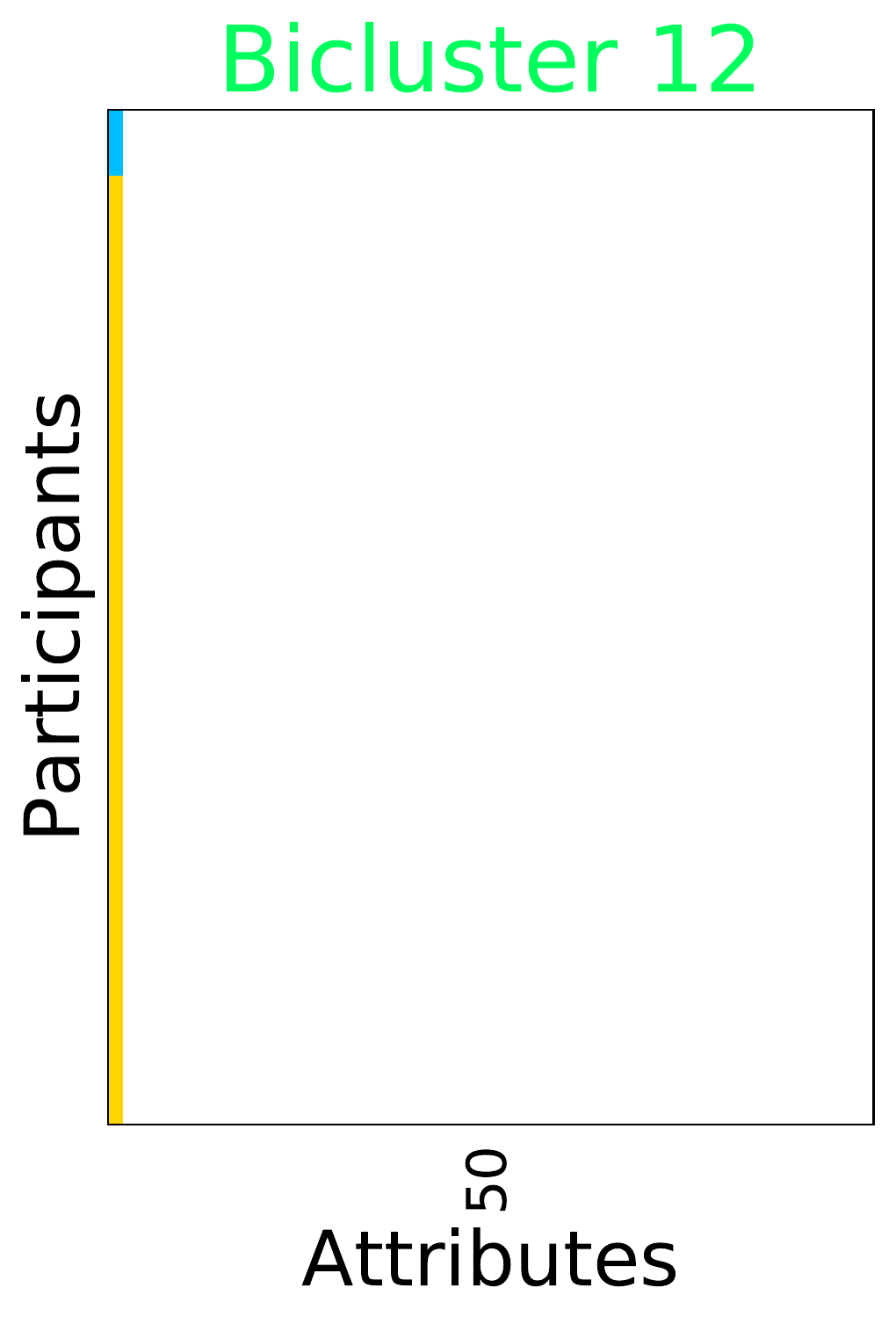}
  \includegraphics[height=0.2\hsize]{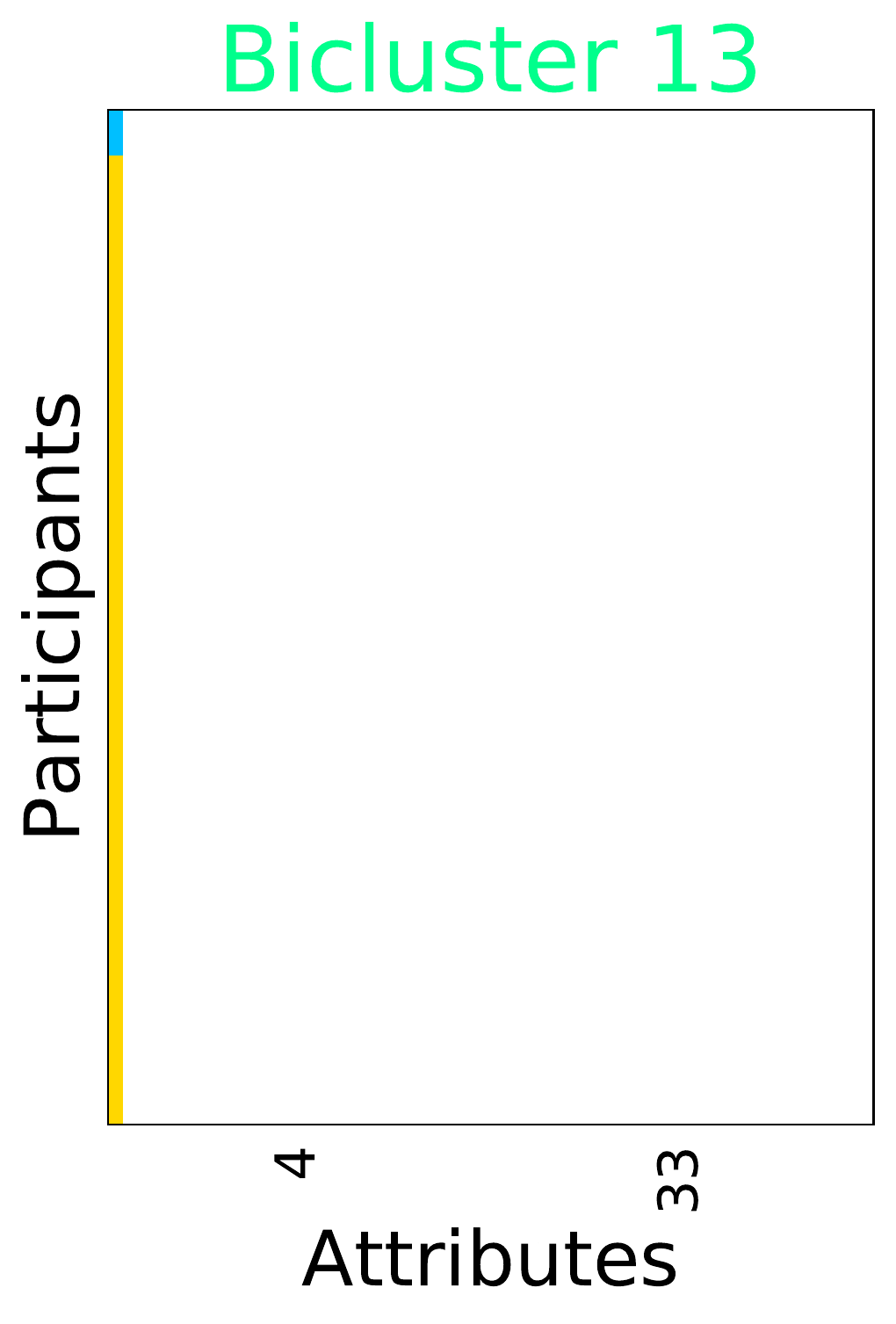}
  \includegraphics[height=0.2\hsize]{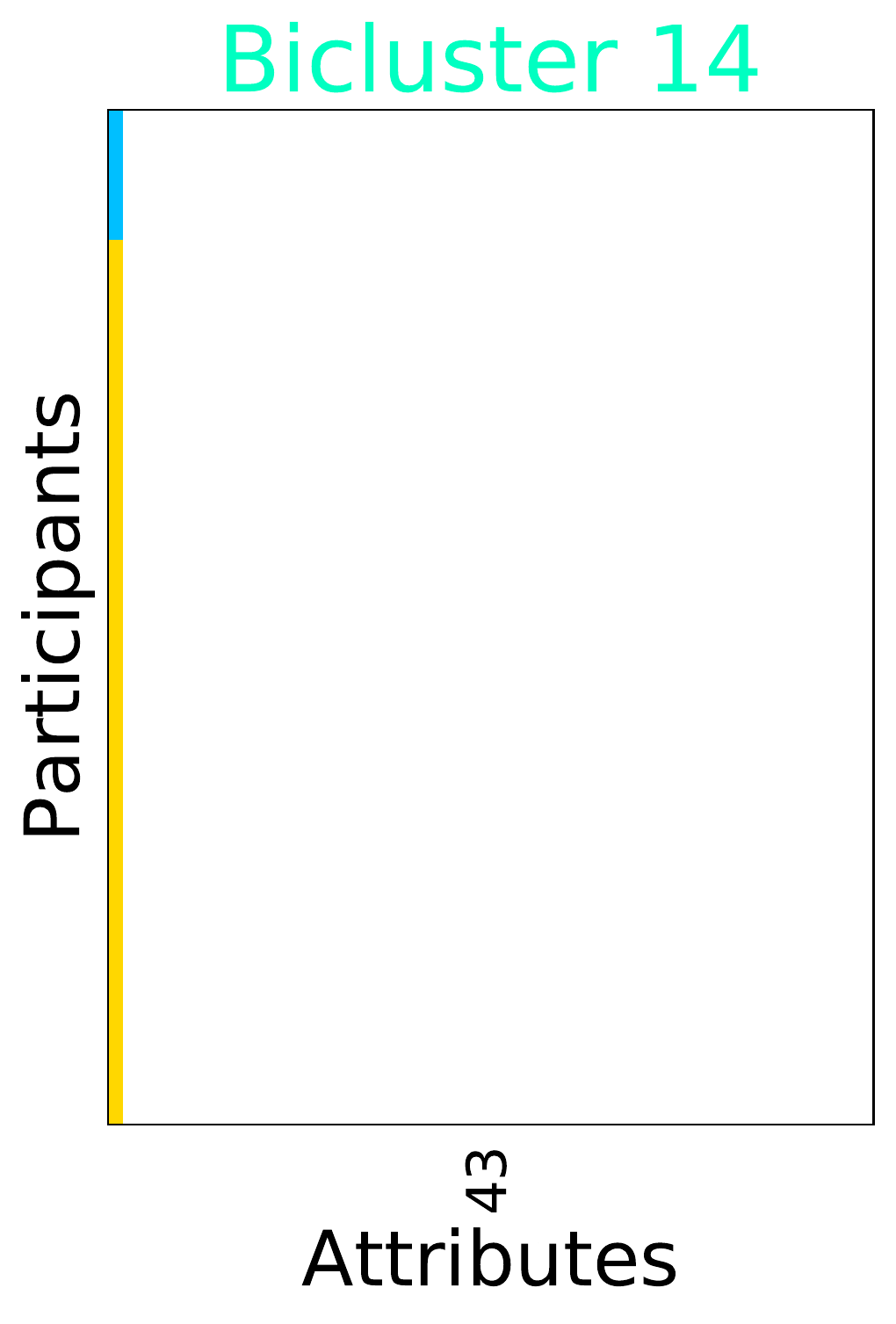}
  \includegraphics[height=0.2\hsize]{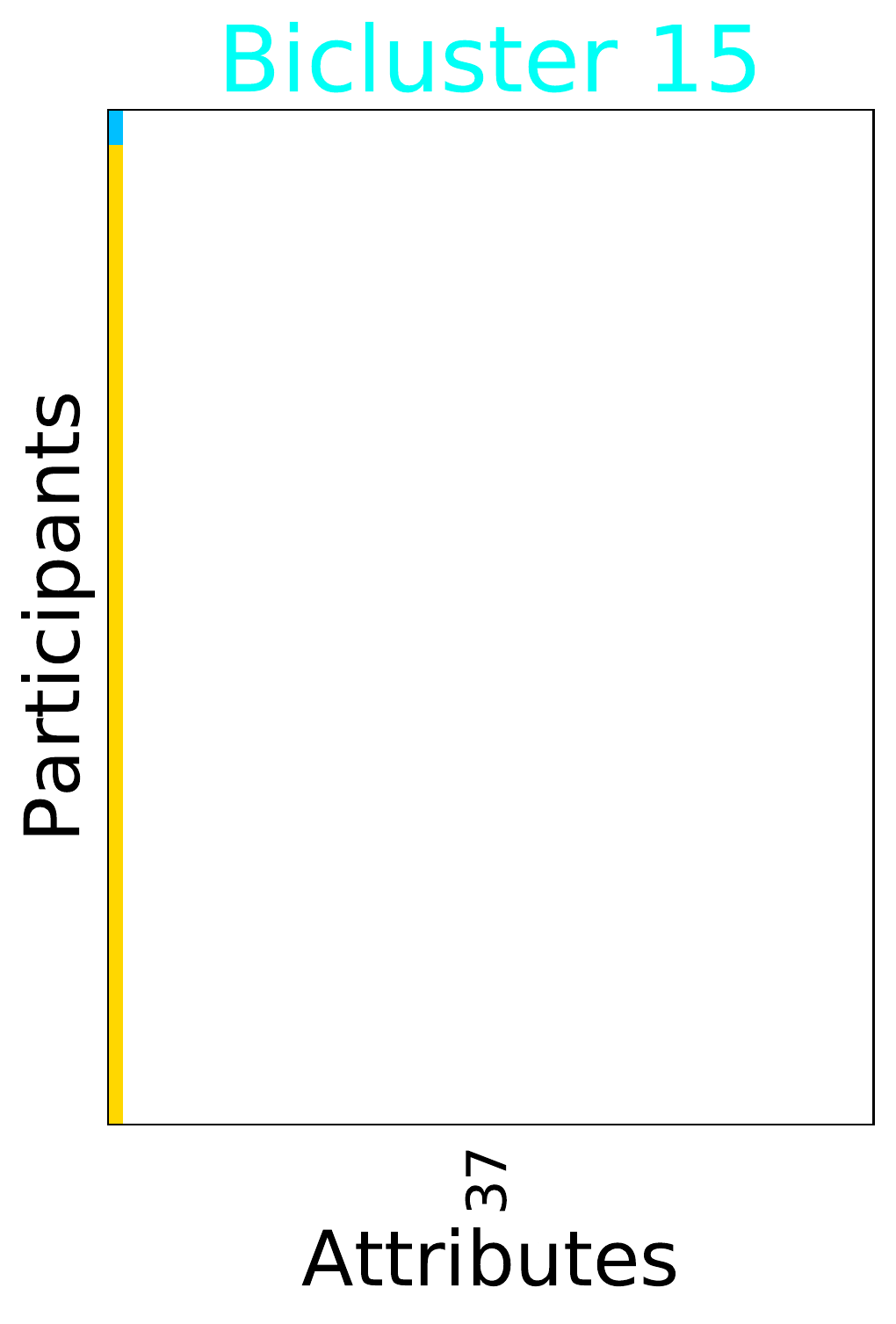}
  \includegraphics[height=0.2\hsize]{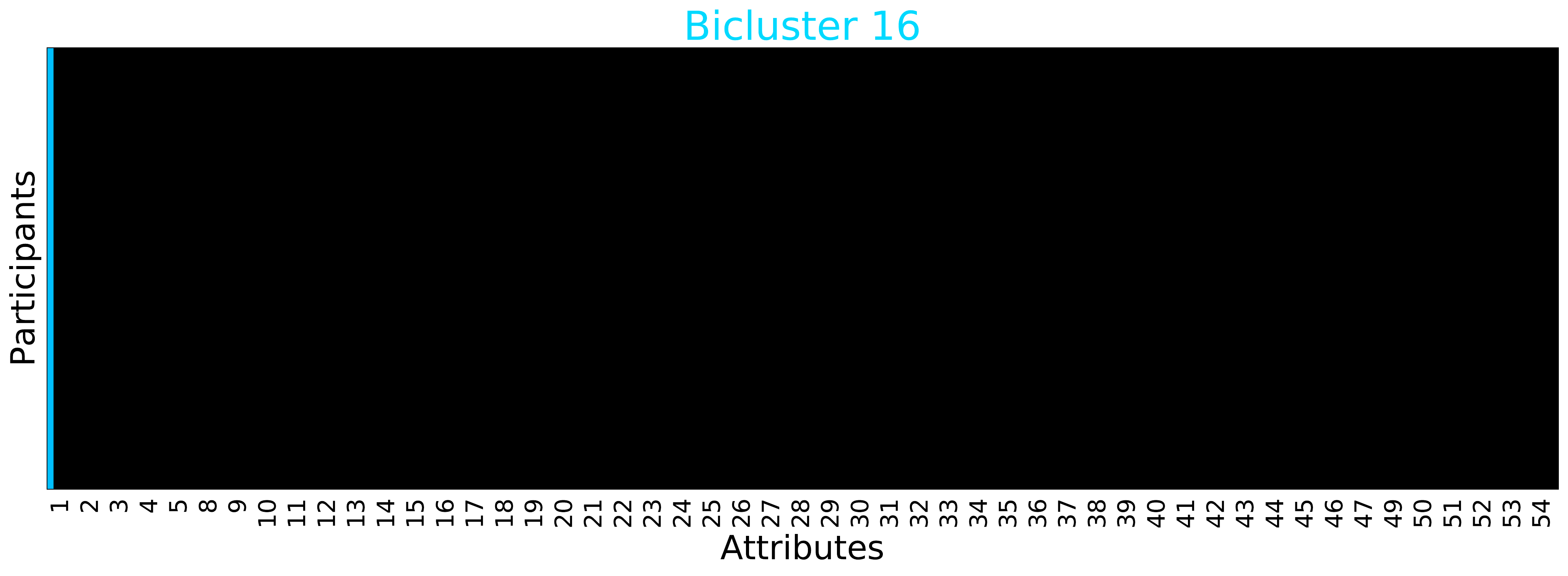}
  \includegraphics[height=0.2\hsize]{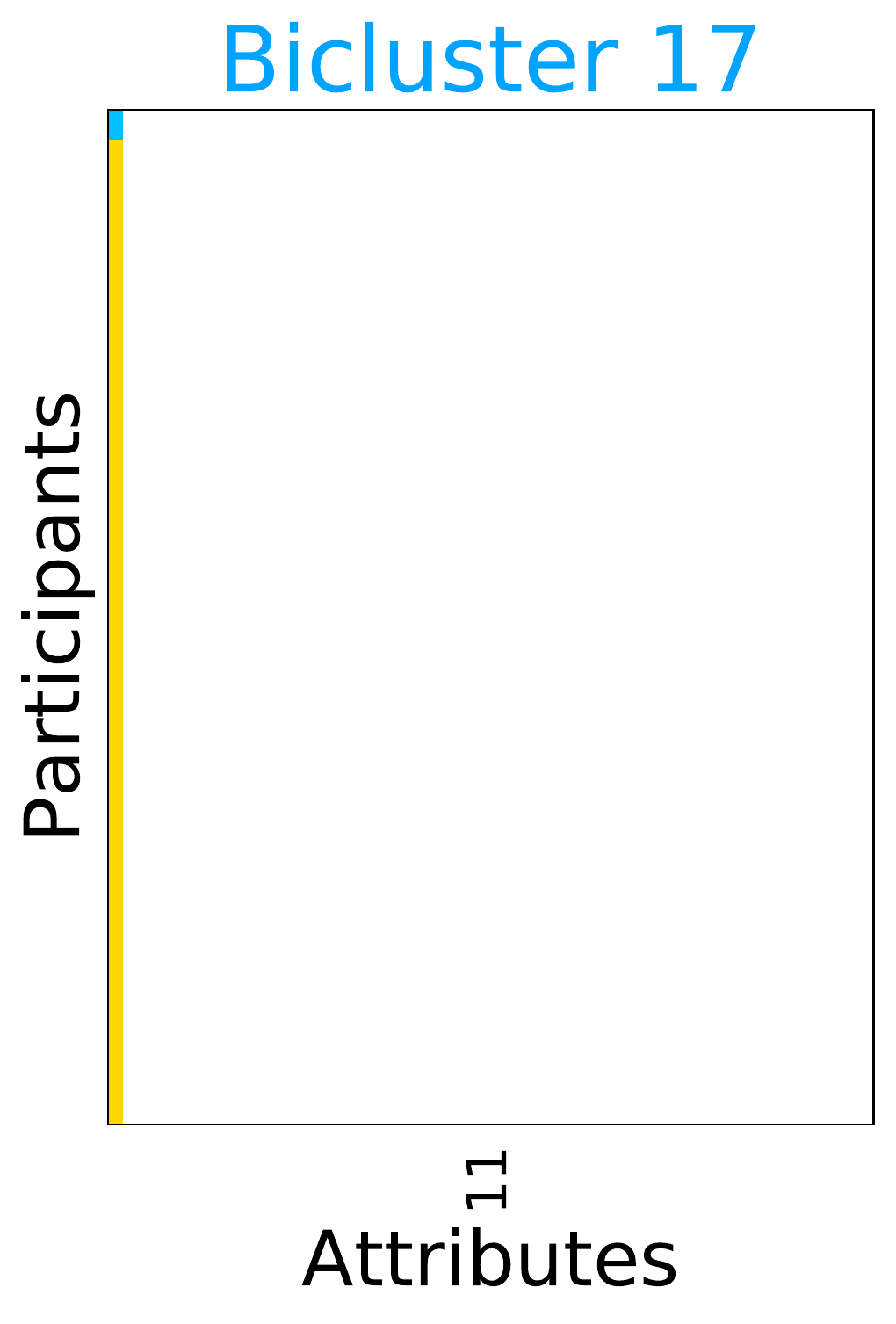}
  \includegraphics[height=0.2\hsize]{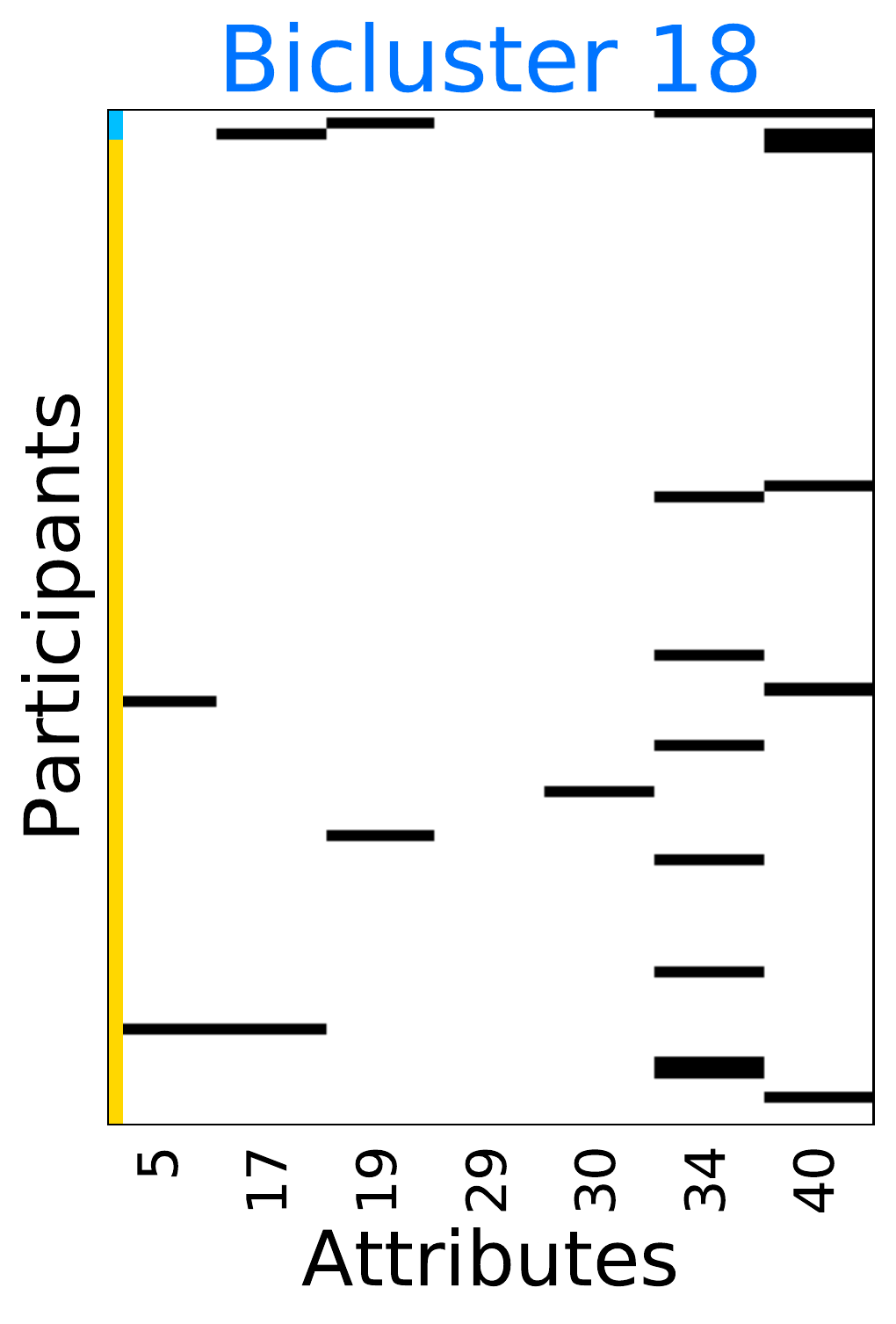}
  \includegraphics[height=0.2\hsize]{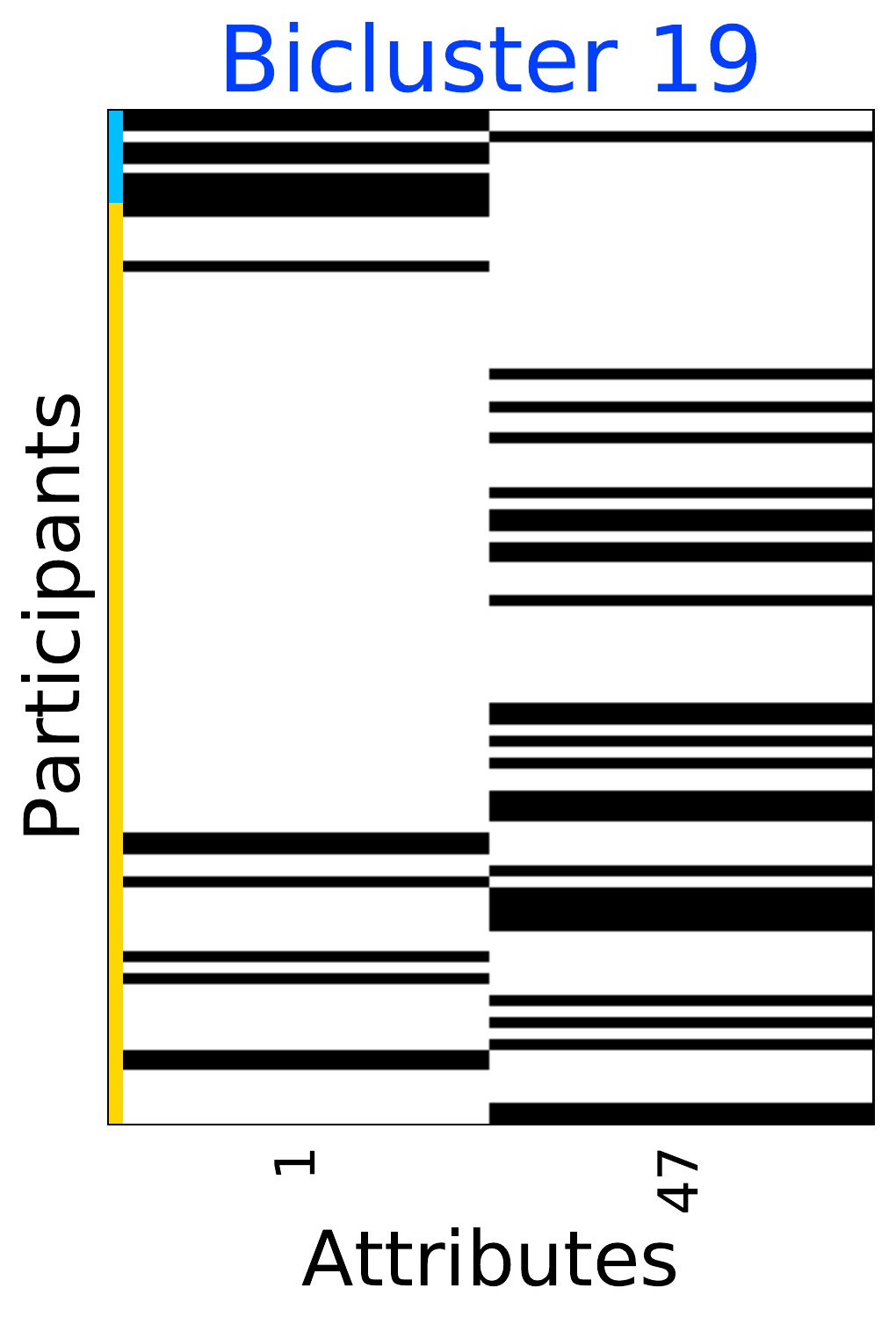}
  \includegraphics[height=0.2\hsize]{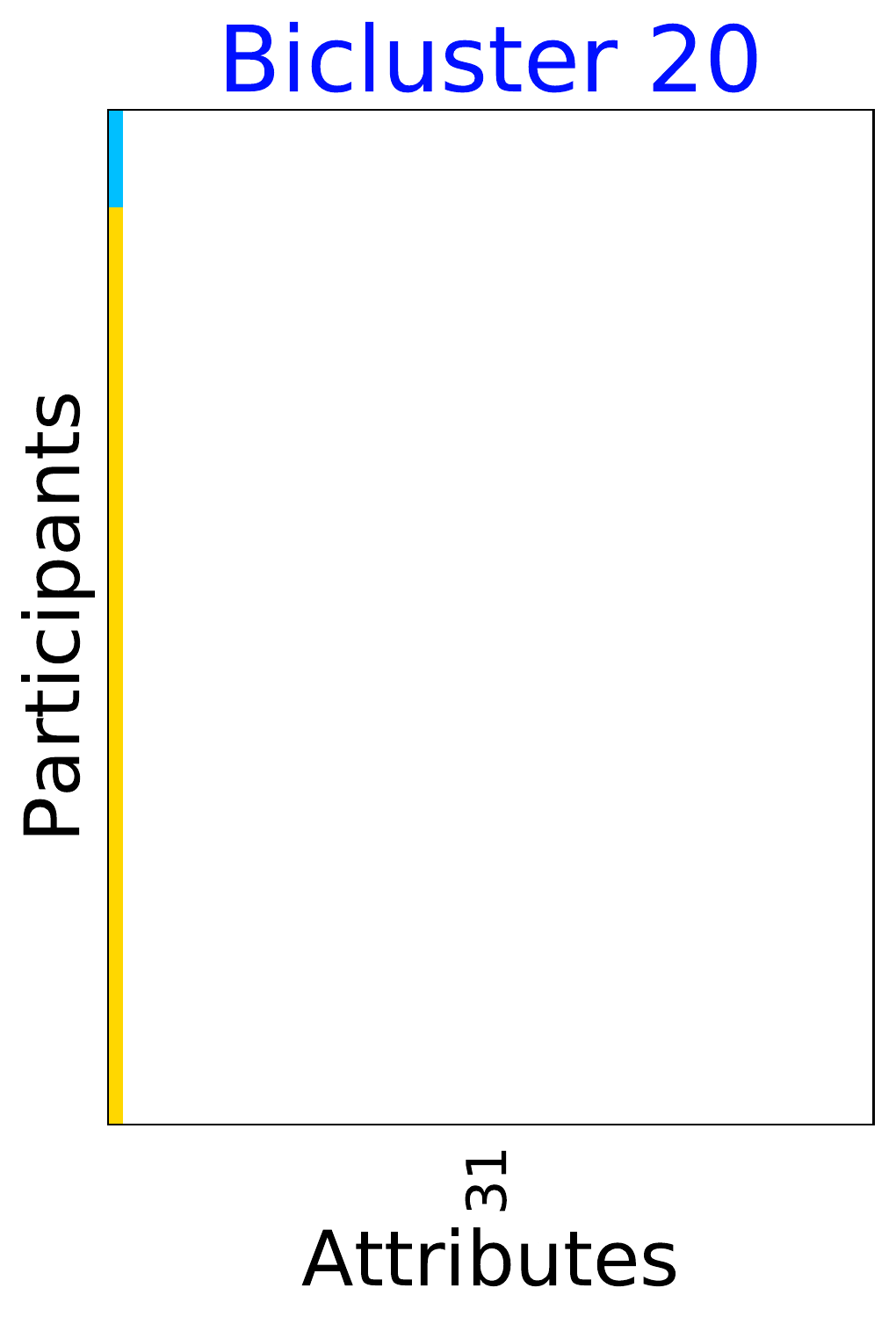}
  \includegraphics[height=0.2\hsize]{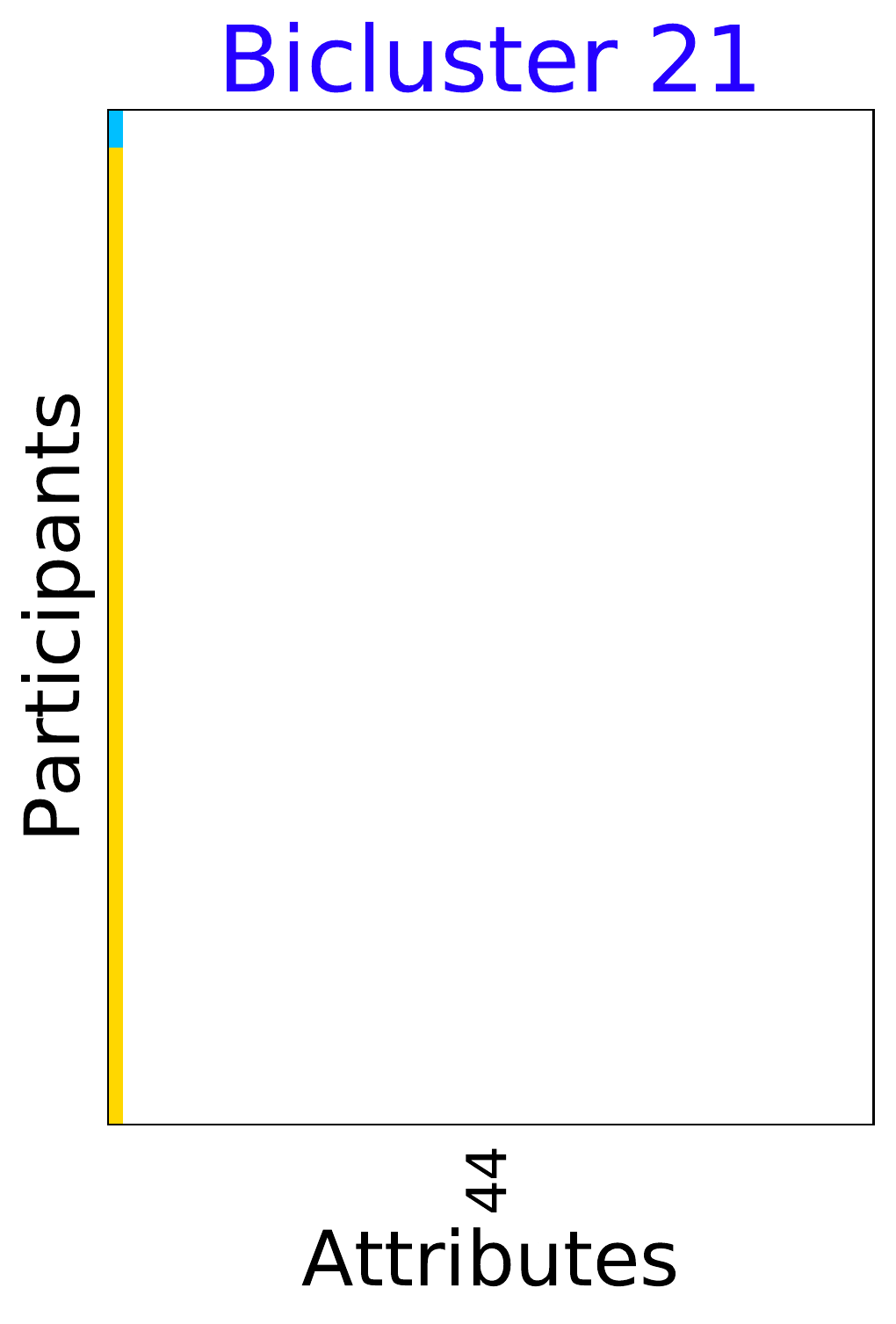}
  \includegraphics[height=0.2\hsize]{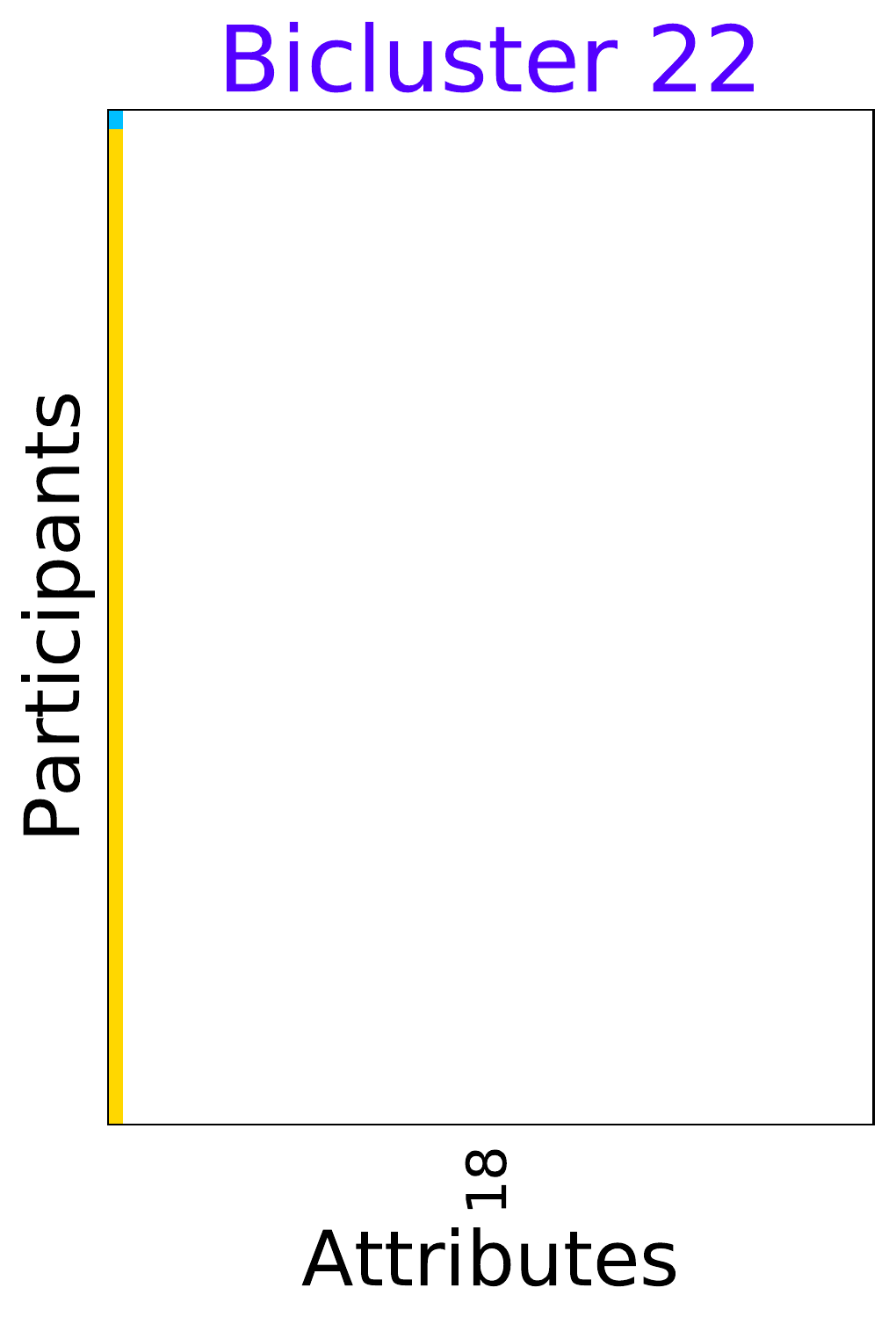}
  \includegraphics[height=0.2\hsize]{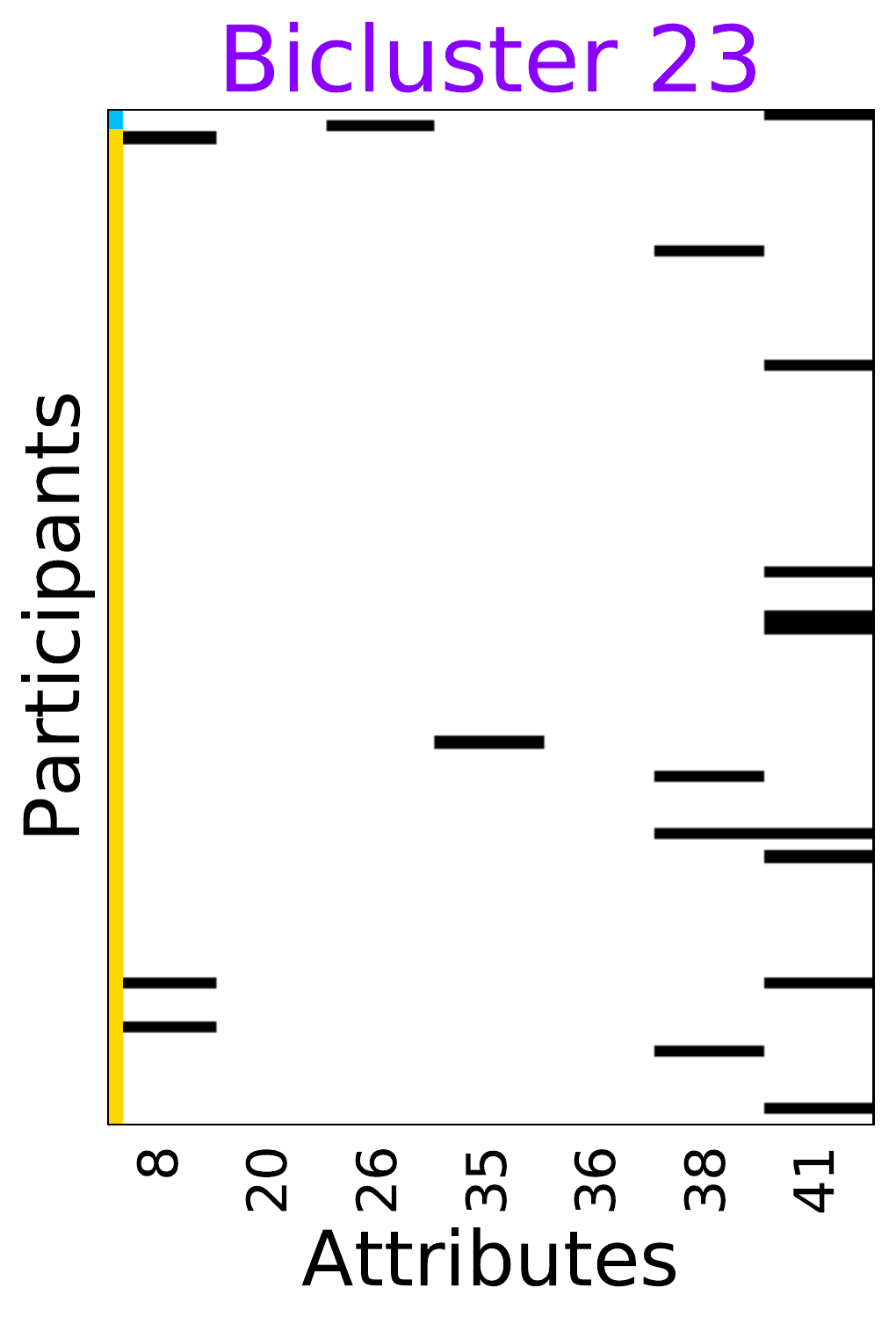}
  \includegraphics[height=0.2\hsize]{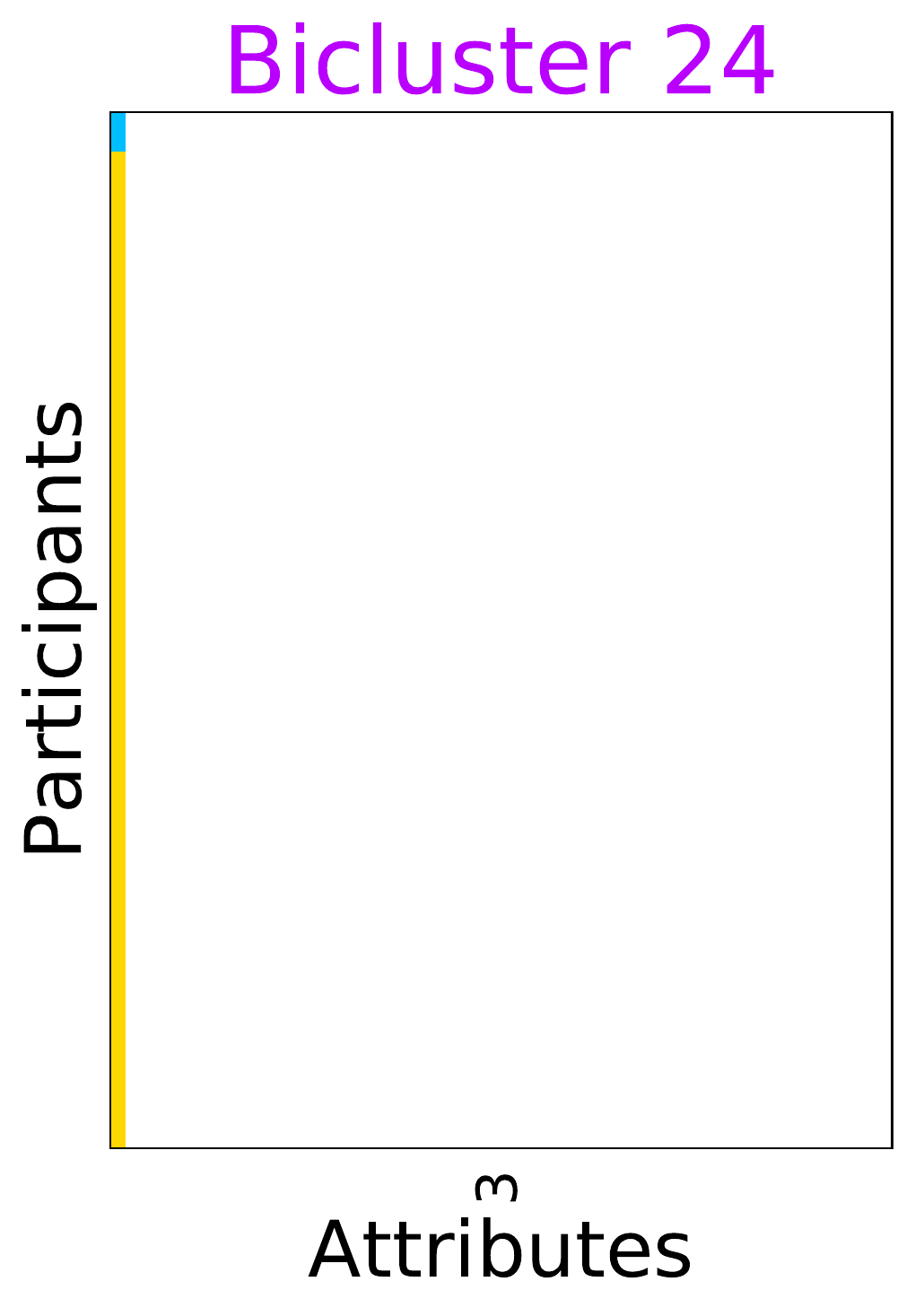}
  \includegraphics[height=0.2\hsize]{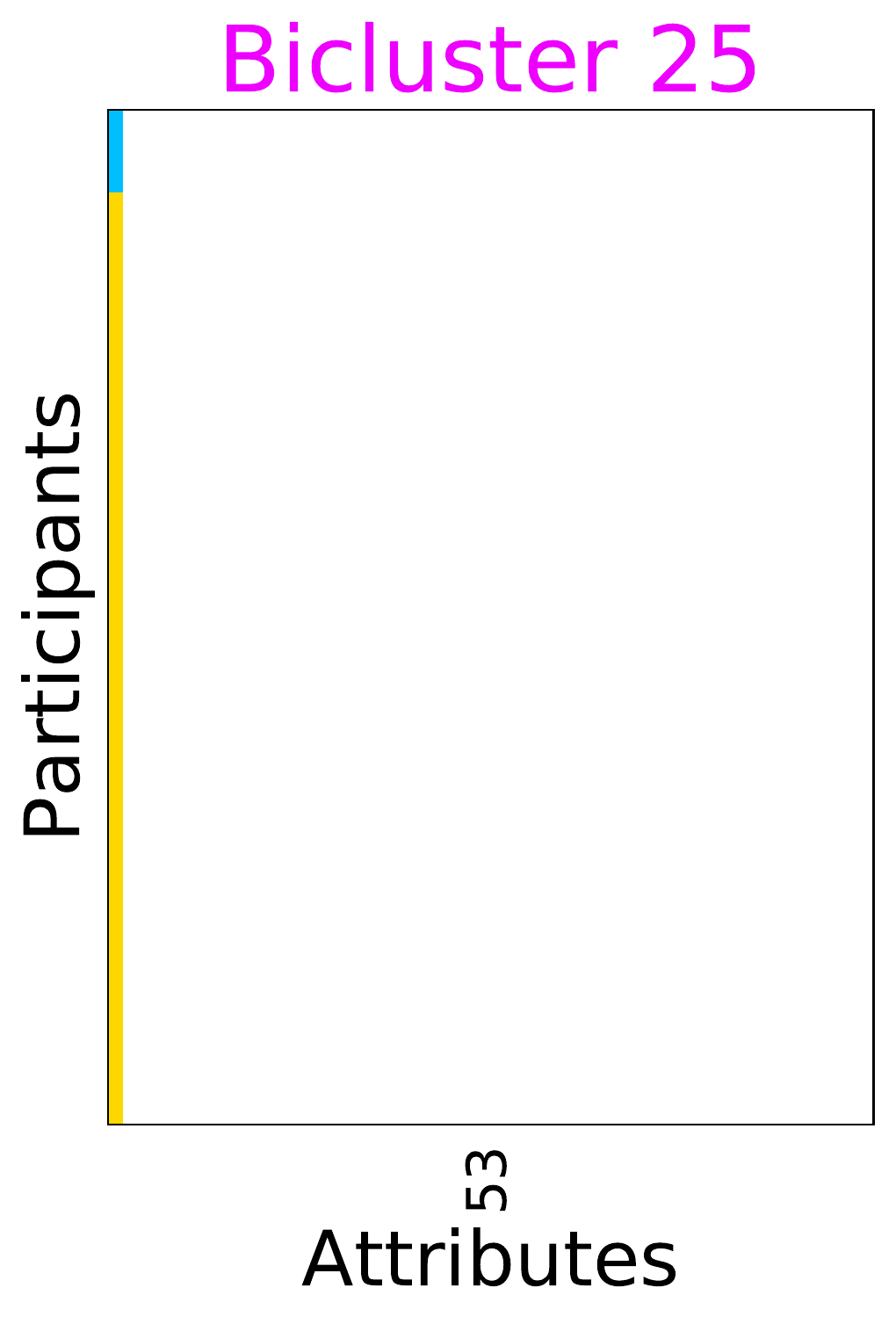}
  \includegraphics[height=0.2\hsize]{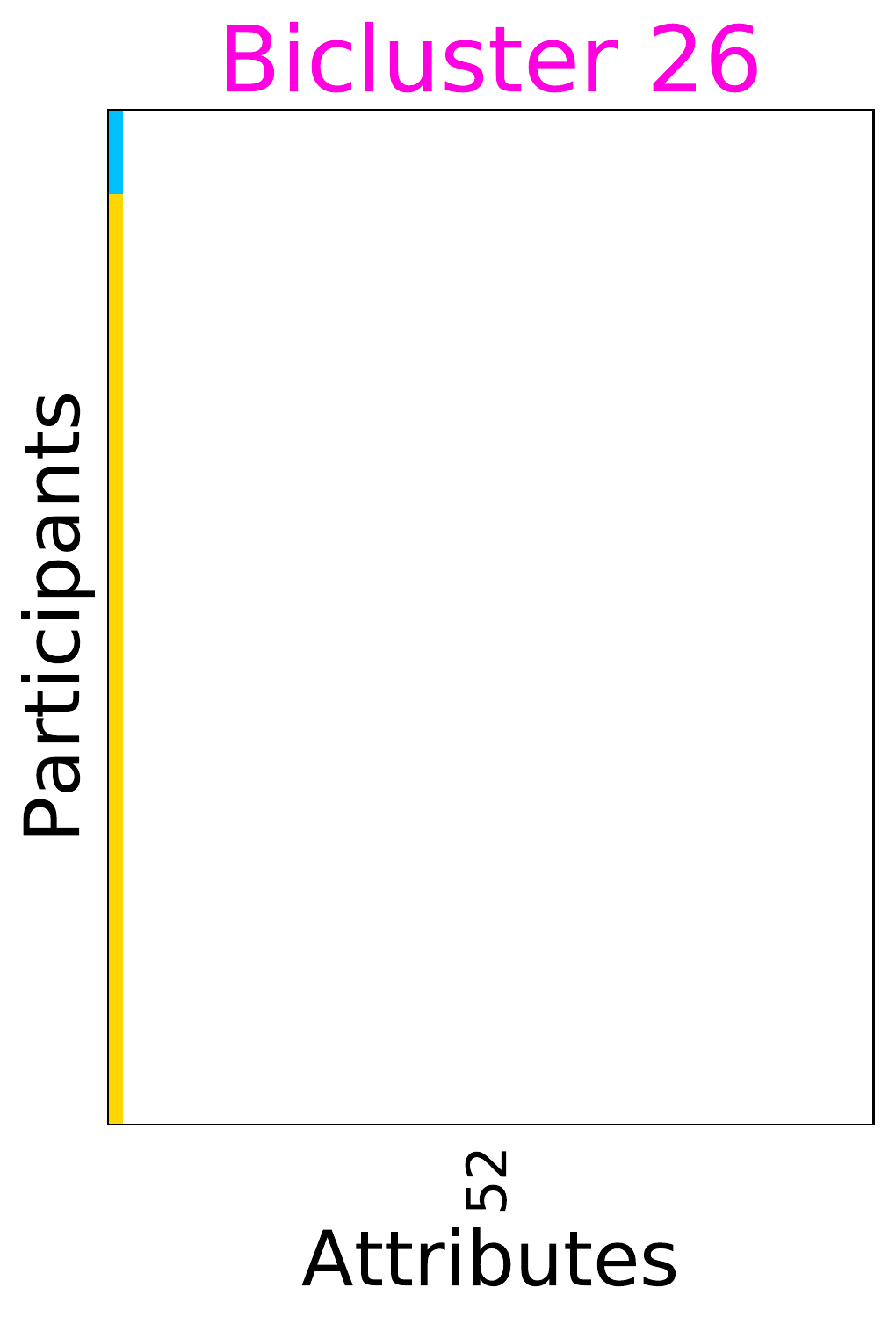}
  \includegraphics[height=0.2\hsize]{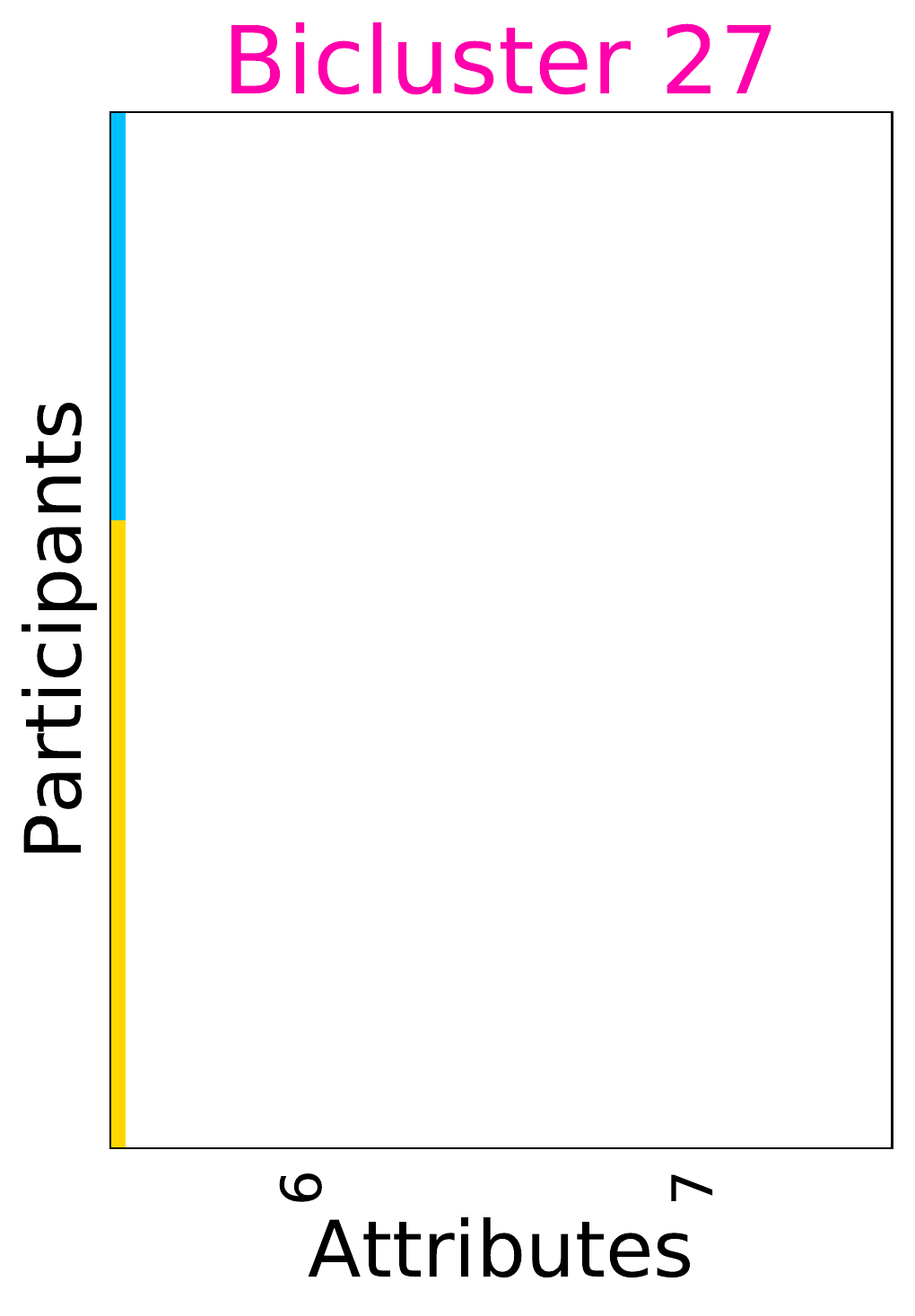}
  \includegraphics[height=0.2\hsize]{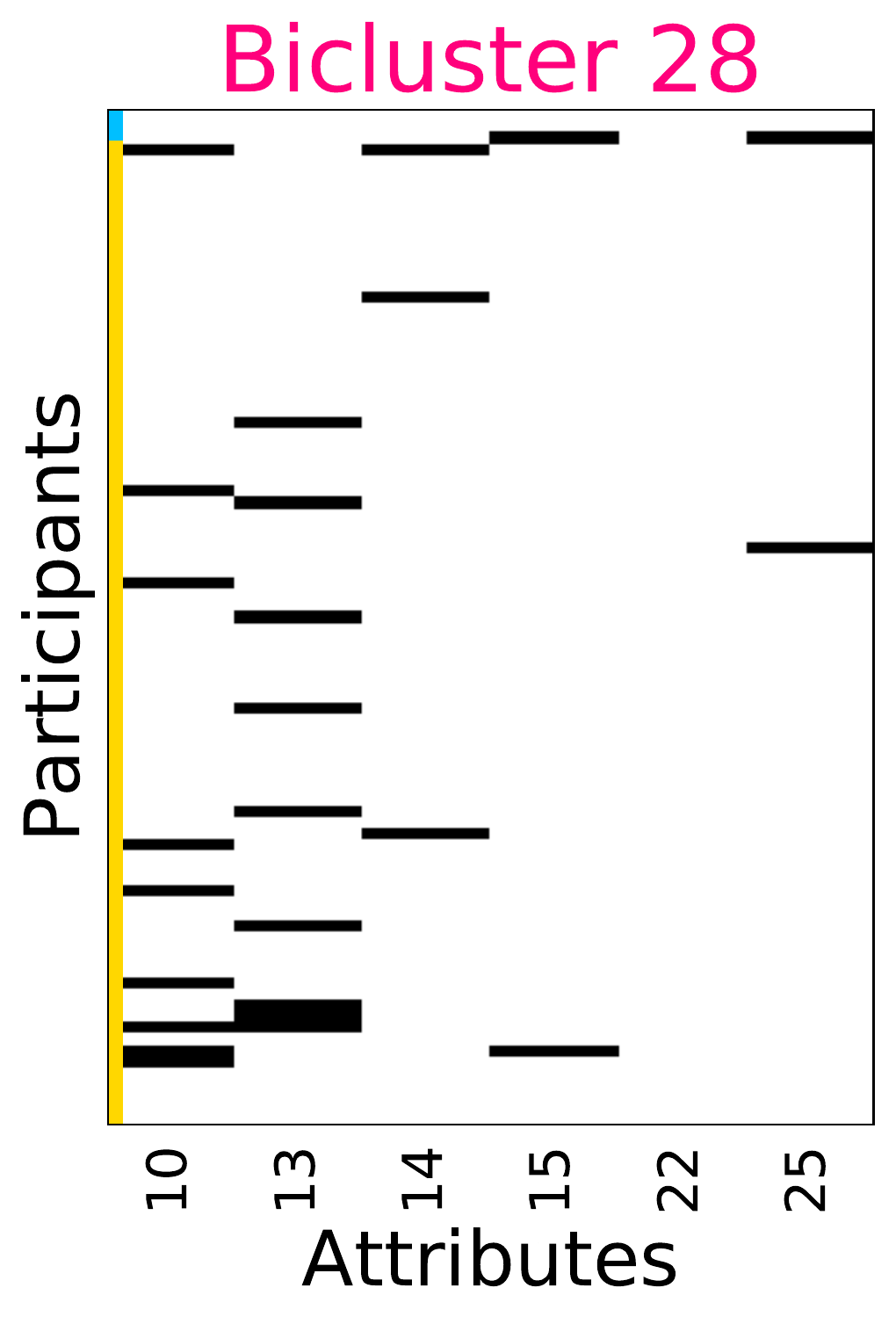}
  \includegraphics[height=0.2\hsize]{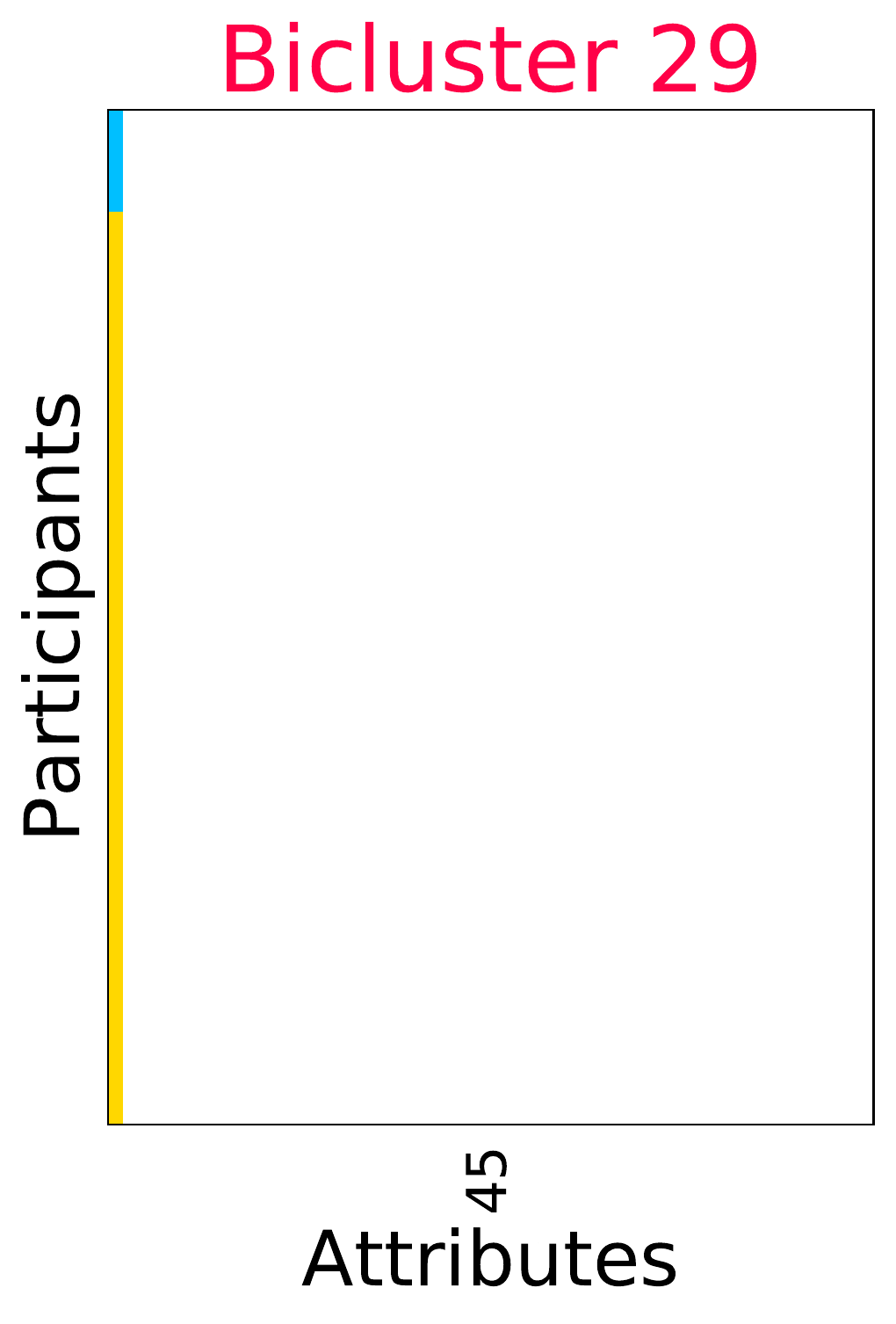}
  \caption{Estimated biclusters of the observed matrix when the proposed test accepted the null hypothesis. The font color of the title corresponds to the bicluster color in Figure \ref{fig:practical_A}, and the left side color for each row indicates the class of each participant in the bicluster, as in the legend of Figure \ref{fig:practical_A}.}
  \label{fig:practical_bic_proposed}
\end{figure}

\begin{table}[t]
\caption{Attributes of the Divorce Predictors data set \cite{Yontem2019}.}\vspace{2mm}
\scalebox{0.7}{\begin{tabular}{|c|l|}\hline
1 & If one of us apologizes when our discussion deteriorates, the discussion ends.\\
2 & I know we can ignore our differences, even if things get hard sometimes.\\
3 & When we need it, we can take our discussions with my spouse from the beginning and correct it.\\
4 & When I discuss with my spouse, to contact him will eventually work.\\
5 & The time I spent with my wife is special for us.\\
6 & We don't have time at home as partners.\\
7 & We are like two strangers who share the same environment at home rather than family.\\
8 & I enjoy our holidays with my wife.\\
9 & I enjoy traveling with my wife.\\
10 & Most of our goals are common to my spouse.\\
11 & I think that one day in the future, when I look back, I see that my spouse and I have been in harmony with each other.\\
12 & My spouse and I have similar values in terms of personal freedom.\\
13 & My spouse and I have similar sense of entertainment.\\
14 & Most of our goals for people (children, friends, etc.) are the same.\\
15 & Our dreams with my spouse are similar and harmonious.\\
16 & We're compatible with my spouse about what love should be.\\
17 & We share the same views about being happy in our life with my spouse.\\
18 & My spouse and I have similar ideas about how marriage should be.\\
19 & My spouse and I have similar ideas about how roles should be in marriage.\\
20 & My spouse and I have similar values in trust.\\
21 & I know exactly what my wife likes.\\
22 & I know how my spouse wants to be taken care of when she/he sick.\\
23 & I know my spouse's favorite food.\\
24 & I can tell you what kind of stress my spouse is facing in her/his life.\\
25 & I have knowledge of my spouse's inner world.\\
26 & I know my spouse's basic anxieties.\\
27 & I know what my spouse's current sources of stress are.\\
28 & I know my spouse's hopes and wishes.\\
29 & I know my spouse very well.\\
30 & I know my spouse's friends and their social relationships.\\
31 & I feel aggressive when I argue with my spouse.\\
32 & When discussing with my spouse, I usually use expressions such as ‘you always’ or ‘you never.’\\
33 & I can use negative statements about my spouse's personality during our discussions.\\
34 & I can use offensive expressions during our discussions.\\
35 & I can insult my spouse during our discussions.\\
36 & I can be humiliating when we discussions.\\
37 & My discussion with my spouse is not calm.\\
38 & I hate my spouse's way of open a subject.\\
39 & Our discussions often occur suddenly.\\
40 & We're just starting a discussion before I know what's going on.\\
41 & When I talk to my spouse about something, my calm suddenly breaks.\\
42 & When I argue with my spouse, I only go out and I don't say a word.\\
43 & I mostly stay silent to calm the environment a little bit.\\
44 & Sometimes I think it's good for me to leave home for a while.\\
45 & I'd rather stay silent than discuss with my spouse.\\
46 & Even if I'm right in the discussion, I stay silent to hurt my spouse.\\
47 & When I discuss with my spouse, I stay silent because I am afraid of not being able to control my anger.\\
48 & I feel right in our discussions.\\
49 & I have nothing to do with what I've been accused of.\\
50 & I'm not actually the one who's guilty about what I'm accused of.\\
51 & I'm not the one who's wrong about problems at home.\\
52 & I wouldn't hesitate to tell my spouse about her/his inadequacy.\\
53 & When I discuss, I remind my spouse of her/his inadequacy.\\
54 & I'm not afraid to tell my spouse about her/his incompetence.\\ \hline
\end{tabular}}
\label{tb:divorce_attribute}
\end{table}


\section{Discussion}
\label{sec:discussion}

This section explicates both the theoretical and practical perspectives on the proposed statistical test. 

We derived the asymptotic behavior of the proposed test statistic $T$ in both the null and alternative cases, where the null number of biclusters might increase with the matrix size (as the condition given in \ref{asmp:block_size}). As in the previous regular-grid based test \cite{Watanabe2021}, it is an important future work to reveal the non-asymptotic property of the test statistic $T$, that is, its convergence rate to the $TW_1$ distribution. To solve this problem, we need to derive the behavior of $T$ in case that the submatrix localization algorithm does \textit{not} output the correct bicluster structure, which requires more careful analysis. 

From a practical perspective, some studies \cite{Bendor2002,Liu2004} demonstrate the effectiveness of analyzing a gene expression data matrix by assuming the existence of \textit{order-preserving biclusters}, in which a set of rows (i.e., genes) has a similar linear ordering of columns (i.e., conditions). Such a definition of homogeneousness is different from ours, whereby we assume that each entry in a bicluster is generated in the i.i.d. sense. Additionally, some practical relational data matrices (e.g., MovieLens \cite{Harper2015} and Jester \cite{Goldberg2001} data sets) contain missing entries. It is an important topic in future research to construct a statistical test on $K$ for such cases deriving its theoretical guarantee. 


\section{Conclusion}
\label{sec:conclusion}

In this study, we proposed a new statistical test on the number of biclusters in a given relational data matrix and showed the asymptotic behavior of the proposed test statistic $T$ in both null and alternative cases. Unlike the previous study \cite{Watanabe2021}, we can apply the proposed method when the underlying bicluster structure is not necessarily represented by a regular grid. By sequentially testing the hypothetical numbers of biclusters in an ascending order, we can select an appropriate number of biclusters in a given observed matrix. We experimentally showed the asymptotic behavior of the proposed test statistic $T$ and its accuracy in selecting the correct number of biclusters with synthetic data matrices. Moreover, we analyzed the test result with a practical data set. 


\section*{Acknowledgments}

TS was partially supported by JSPS KAKENHI (18K19793, 18H03201, and 20H00576), Japan Digital Design, Fujitsu Laboratories Ltd., and JST CREST. 
We would like to thank Editage (\url{www.editage.com}) for English language editing. 


\clearpage
\begin{appendices}

\section{Proof of Lemma \ref{lm:zop2_eq1}}
\label{sec:lmd_1_upper}

\begin{proof}
From the assumption \ref{asmp:block_size}, the background can be divided to $H$ disjoint submatrices, whose row and column sizes are equal to or larger than $n_{\mathrm{min}}$ and $p_{\mathrm{min}}$, respectively. Let $X^{(k)} \in \mathbb{R}^{|I_k| \times |J_k|}$ be a submatrix of matrix $X \in \mathbb{R}^{n \times p}$ corresponding to the row and column indices $(I_k, J_k)$ of the $k$th bicluster. To distinguish the indices of the biclusters from those of the background submatrices, let $X^{(K + h)} \in \mathbb{R}^{|I_{K+h}| \times |J_{K+h}|}$ be a submatrix of matrix $X \in \mathbb{R}^{n \times p}$ corresponding to the row and column indices $(I_{K+h}, J_{K+h})$ of the $h$th background submatrix. 
We define a bicluster-wise constant matrix $Q^{(k)}$ for each $k$th bicluster ($k = 1, \dots, K$), 
\begin{align}
\label{eq:Q_k}
Q^{(k)} &\equiv Z^{(k)} - \frac{\tilde{s}_k}{s_k} \tilde{Z}^{(k)} = \frac{1}{s_k} \left( \tilde{P}^{(k)} - P^{(k)} \right) \nonumber \\
&= \frac{1}{|\mathcal{I}_k|} \left[ \sum_{(i, j) \in \mathcal{I}_k} \frac{A_{ij} - P_{ij}}{s_k} \right] \begin{bmatrix}
1 & \cdots & 1 \\
\vdots & & \vdots \\
1 & \cdots & 1
\end{bmatrix} \nonumber \\
&= \left( \frac{1}{|\mathcal{I}_k|} \sum_{(i, j) \in \mathcal{I}_k} Z_{ij} \right) \begin{bmatrix}
1 & \cdots & 1 \\
\vdots & & \vdots \\
1 & \cdots & 1
\end{bmatrix} \in \mathbb{R}^{|I_k| \times |J_k|}, 
\end{align}
and a submatrix-wise constant matrix $Q^{(K+h)}$ for each $h$th background submatrix ($h = 1, \dots, H$), 
\begin{align}
\label{eq:Q_prime_k}
Q^{(K+h)} &\equiv Z^{(K+h)} - \frac{\tilde{s}_0}{s_0} \tilde{Z}^{(K+h)} = \frac{1}{s_0} \left( \tilde{P}^{(K+h)} - P^{(K+h)} \right) \nonumber \\
&= \frac{1}{|\mathcal{I}_0|} \left[ \sum_{(i, j) \in \mathcal{I}_0} \frac{A_{ij} - P_{ij}}{s_0} \right] \begin{bmatrix}
1 & \cdots & 1 \\
\vdots & & \vdots \\
1 & \cdots & 1
\end{bmatrix} \nonumber \\
&= \left( \frac{1}{|\mathcal{I}_0|} \sum_{(i, j) \in \mathcal{I}_0} Z_{ij} \right) \begin{bmatrix}
1 & \cdots & 1 \\
\vdots & & \vdots \\
1 & \cdots & 1
\end{bmatrix} \in \mathbb{R}^{|I_{K+h}| \times |J_{K+h}|}. 
\end{align}

Based on these matrices, let $\underline{Z}^{(k)}$, $\underline{\tilde{Z}}^{(k)}$, and $\underline{Q}^{(k)}$, respectively, be $n \times p$ matrices whose entries in the $k$th bicluster (i.e., $\{ (i, j): i \in I_k, j \in J_k \}$) are $Z^{(k)}$, $\tilde{Z}^{(k)}$ and $Q^{(k)}$ and whose all the other entries are zero. Similarly, let $\underline{Z}^{(K+h)}$, $\underline{\tilde{Z}}^{(K+h)}$, and $\underline{Q}^{(K+h)}$, respectively, be $n \times p$ matrices whose entries in the $h$th background submatrix (i.e., $\{ (i, j): i \in I_{K+h}, j \in J_{K+h} \}$) are $Z^{(K+h)}$, $\tilde{Z}^{(K+h)}$ and $Q^{(K+h)}$ and whose all the other entries are zero. Figure \ref{fig:decompose_sub} shows an example of $\{ \underline{Z}^{(k)} \}$, where $k = 1, \dots, K+H$. 
Finally, we define matrix $Q$ by $Q \equiv \sum_{k = 1}^{K+H} \underline{Q}^{(k)}$. 
\begin{figure}[t]
  \centering
  \includegraphics[width=0.99\hsize]{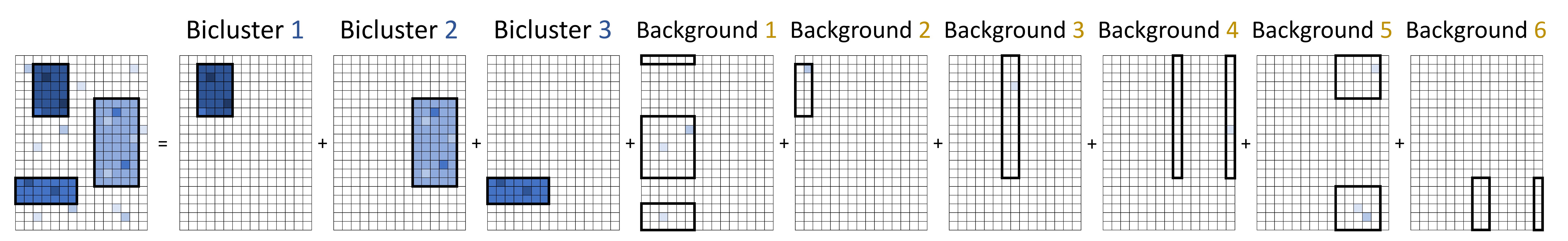}
  \caption{Decomposition of the noise matrix $Z$ to the $K$ biclusters $\{ \underline{Z}^{(k)} \}$, $k = 1, \dots, K$ and the $H$ background submatrices $\{ \underline{Z}^{(K+h)} \}$, $h = 1, \dots, H$. In the case of this figure, $K=3$ and $H=6$.} 
  \label{fig:decompose_sub}
\end{figure}

Let $\bm{v}_1$ be the normalized eigenvector of matrix $Z^{\top} Z$ corresponding to the maximum eigenvalue $\lambda_1$, as defined in (\ref{eq:v_eigenvec}). Since the largest singular value $\sqrt{\tilde{\lambda}_1}$ of matrix $\tilde{Z}$ is equal to the operator norm of $\tilde{Z}$, we have
\begin{align}
\label{eq:tilde_lmd1}
&\tilde{\lambda}_1 = \left( \sup_{\bm{u}} \frac{\| \tilde{Z} \bm{u} \|}{\| \bm{u} \|} \right)^2 
\geq \left( \frac{\| \tilde{Z} \bm{v}_1 \|}{\| \bm{v}_1 \|} \right)^2 
= \| \tilde{Z} \bm{v}_1 \|^2 
= \norm{ \sum_{k = 1}^{K+H} \underline{\tilde{Z}}^{(k)} \bm{v}_1 }^2 \nonumber \\
&= \norm{ \left[ \sum_{k = 1}^K \left( \frac{s_k}{\tilde{s}_k} \right) (\underline{Z}^{(k)} - \underline{Q}^{(k)}) \bm{v}_1 \right] + \left[ \sum_{h = 1}^H \left( \frac{s_0}{\tilde{s}_0} \right) (\underline{Z}^{(K+h)} - \underline{Q}^{(K+h)}) \bm{v}_1 \right] }^2 \nonumber \\
&= \norm{ \left[ \sum_{k = 1}^{K+H} (\underline{Z}^{(k)} - \underline{Q}^{(k)}) \bm{v}_1 \right] + \left[ \sum_{k = 1}^{K+H} \left( \frac{s_k}{\tilde{s}_k} - 1 \right) (\underline{Z}^{(k)} - \underline{Q}^{(k)}) \bm{v}_1 \right] }^2 \nonumber \\
&\geq \left[ \norm{ \sum_{k = 1}^{K+H} (\underline{Z}^{(k)} - \underline{Q}^{(k)}) \bm{v}_1 } - \norm{ \sum_{k = 1}^{K+H} \left( 1 - \frac{s_k}{\tilde{s}_k} \right) (\underline{Z}^{(k)} - \underline{Q}^{(k)}) \bm{v}_1 } \right]^2 \nonumber \\
&= \left[ \| (Z - Q) \bm{v}_1 \| - \norm{ \sum_{k = 1}^{K+H} \left( 1 - \frac{s_k}{\tilde{s}_k} \right) (\underline{Z}^{(k)} - \underline{Q}^{(k)}) \bm{v}_1 } \right]^2 \nonumber \\
&\geq \| (Z - Q) \bm{v}_1 \|^2 - 2 \| (Z - Q) \bm{v}_1 \| \norm{ \sum_{k = 1}^{K+H} \left( 1 - \frac{s_k}{\tilde{s}_k} \right) (\underline{Z}^{(k)} - \underline{Q}^{(k)}) \bm{v}_1 } \nonumber \\
&\geq \| (Z - Q) \bm{v}_1 \|^2 - 2 \| (Z - Q) \bm{v}_1 \| \left[ \sum_{k = 1}^{K+H} \left| 1 - \frac{s_k}{\tilde{s}_k} \right| \| (\underline{Z}^{(k)} - \underline{Q}^{(k)}) \bm{v}_1 \| \right] \nonumber \\
&\geq \| (Z - Q) \bm{v}_1 \|^2 - 2 \| (Z - Q) \bm{v}_1 \| \left[ \sum_{k = 1}^{K+H} \left| 1 - \frac{s_k}{\tilde{s}_k} \right| \left( \| \underline{Z}^{(k)} \bm{v}_1 \| + \| \underline{Q}^{(k)} \bm{v}_1 \| \right) \right] \nonumber \\
&\geq \| (Z - Q) \bm{v}_1 \|^2 - 2 \| (Z - Q) \bm{v}_1 \| \left[ \sum_{k = 1}^{K+H} \left| 1 - \frac{s_k}{\tilde{s}_k} \right| \left( \sqrt{\lambda_1^{(k)}} + \| \underline{Q}^{(k)} \bm{v}_1 \| \right) \right] \nonumber \\
&\geq \lambda_1 - 2 \sqrt{\lambda_1} \| Q \bm{v}_1 \| - 2 (\sqrt{\lambda_1} + \| Q \bm{v}_1 \|) \left[ \sum_{k = 1}^{K+H} \left| 1 - \frac{s_k}{\tilde{s}_k} \right| \left( \sqrt{\lambda_1^{(k)}} + \| \underline{Q}^{(k)} \bm{v}_1 \| \right) \right]. 
\end{align}
where $\lambda_1^{(k)}$ is the maximum eigenvalue of matrix $(\underline{Z}^{(k)})^{\top} \underline{Z}^{(k)}$ (which is equal to that of matrix $(Z^{(k)})^{\top} Z^{(k)}$). 
From the third line in (\ref{eq:tilde_lmd1}), we used the notation that $s_{K+1} = \dots = s_{K+H} = s_0$ for simplicity. 

Subsequently, we show the probabilistic orders of $\| \underline{Q}^{(k)} \bm{v}_1 \|$ and $\| Q \bm{v}_1 \|$. The non-zero entries in matrix $(\underline{Q}^{(k)})^{\top} \underline{Q}^{(k)}$ is only located in a submatrix $\{ (i, j): i \in J_k, j \in J_k \}$, and all of their values are $|I_k| \eta_k^2$ by (\ref{eq:Q_k}) and (\ref{eq:Q_prime_k}), where 
\begin{align}
\label{eq:nu_k}
\eta_k \equiv \frac{1}{|\mathcal{I}_k|} \sum_{(i, j) \in \mathcal{I}_k} Z_{ij} = O_p \left( \frac{1}{\sqrt{|\mathcal{I}_k|}} \right). 
\end{align}
Therefore, we have 
\begin{align}
\label{eq:underline_QQv}
(\underline{Q}^{(k)})^{\top} \underline{Q}^{(k)} \bm{v}_1 = |I_k| |J_k| \eta_k^2 (\bm{v}_1^{\top} \bm{u}^{(k)}) \bm{u}^{(k)}, 
\end{align}
where $\bm{u}^{(k)} \in \mathbb{R}^p$ is a vector whose entries are defined by $\bm{u}^{(k)}_j = \frac{1}{\sqrt{|J_k|}}$ if $j \in J_k$ and $\bm{u}^{(k)}_j = 0$ otherwise. Note that this vector satisfies $\| \bm{u}^{(k)} \| = 1$. From (\ref{eq:underline_QQv}), we have 
\begin{align}
\label{eq:underline_Qv_norm}
\| \underline{Q}^{(k)} \bm{v}_1 \| = \sqrt{\bm{v}_1^{\top} (\underline{Q}^{(k)})^{\top} \underline{Q}^{(k)} \bm{v}_1} 
= \sqrt{|I_k| |J_k| \eta_k^2 (\bm{v}_1^{\top} \bm{u}^{(k)})^2}. 
\end{align}

To upper bound the right side of (\ref{eq:underline_Qv_norm}), we refer to the following important property of each $j$th eigenvector $\bm{v}_j$ of matrix $Z^{\top} Z$, which has been proven in \cite{Bloemendal2016}. 
\begin{theorem}[Delocalization property of an eigenvector of a sample covariance matrix \cite{Bloemendal2016}]
Under the assumptions in Sect.~\ref{sec:method}, from Theorem 2.17 in \cite{Bloemendal2016}, a normalized eigenvector $\bm{v}_j$ of matrix $Z^{\top} Z$ (i.e., $\| \bm{v}_j \| = 1$) has a \textit{delocalization property}, that is, for all $\tilde{d} \in \mathbb{N}$, for any deterministic vectors $\{\bm{w}^{(i)}\}$ that satisfies $\| \bm{w}^{(i)} \| = 1$ for $i = 1, \dots, m^{\tilde{d}}$, for all $\epsilon > 0$, 
\begin{align}
\label{eq:delocalization}
\max_{i \in 1, \dots, m^{\tilde{d}}} \max_{j = 1, \dots, p} | \bm{v}_j^{\top} \bm{w}^{(i)} | = O_p \left( m^{-\frac{1}{2} + \epsilon} \right). 
\end{align}
\end{theorem}

Based on the above delocalization property of vector $\bm{v}_1$ and (\ref{eq:underline_Qv_norm}), we have
\begin{align}
\label{eq:Q_k_v1}
\| \underline{Q}^{(k)} \bm{v}_1 \| = \sqrt{|\mathcal{I}_k| O_p \left( \frac{1}{|\mathcal{I}_k|} \right) O_p \left( m^{-1 + 2 \epsilon} \right)} = O_p \left( m^{-\frac{1}{2} + \epsilon} \right), \ \mathrm{for\ all}\ \epsilon > 0. 
\end{align}

As for $\| Q \bm{v}_1 \|$, we can derive its upper bound by
\begin{align}
\label{eq:Q_v1}
&\| Q \bm{v}_1 \| = \norm{ \sum_{k = 1}^{K+H} \underline{Q}^{(k)} \bm{v}_1 } 
\leq \sum_{k = 1}^{K+H} \| \underline{Q}^{(k)} \bm{v}_1 \| 
= \sum_{k = 1}^{K+H} \sqrt{|\mathcal{I}_k| \eta_k^2 (\bm{v}_1^{\top} \bm{u}^{(k)})^2} \nonumber \\
&= \sum_{k = 1}^{K+H} \sqrt{|\mathcal{I}_k| O_p \left( \frac{1}{|\mathcal{I}_k|} \right)} |\bm{v}_1^{\top} \bm{u}^{(k)}| 
= \sum_{k = 1}^{K+H} O_p (1) O_p \left( m^{-\frac{1}{2} + \epsilon} \right) \nonumber \\
&= (K+H) O_p \left( m^{-\frac{1}{2} + \epsilon} \right). 
\end{align}
Here, we used the fact that (\ref{eq:delocalization}) holds from \cite{Bloemendal2016}. 

From Lemma \ref{lm:sigma_tilde} in Appendix \ref{sec:ap_sigma_tilde}, $| \tilde{s}_k - s_k | = O_p \left( \frac{1}{\sqrt{|\mathcal{I}_k|}} \right)$ holds, which results in
\begin{align}
\label{eq:1_s_tildes}
\left| 1 - \frac{s_k}{\tilde{s}_k} \right| = O_p \left( \frac{1}{\sqrt{|\mathcal{I}_k|}} \right). 
\end{align}

By substituting (\ref{eq:Q_k_v1}), (\ref{eq:Q_v1}), (\ref{eq:1_s_tildes}), and the fact that $\sqrt{\lambda_1^{(k)}} = O_p \left( |\mathcal{I}_k|^{\frac{1}{4}} \right)$ from (\ref{eq:T_true}), into (\ref{eq:tilde_lmd1}), we obtain
\begin{align}
\label{eq:tilde_lmd_lower_0}
\tilde{\lambda}_1 &\geq \lambda_1 - 2 (K+H) O_p \left( m^{\frac{1}{2}} \right) O_p \left( m^{-\frac{1}{2} + \epsilon} \right) \nonumber \\
&- 2 \left[ O_p \left( m^{\frac{1}{2}} \right) + (K+H) O_p \left( m^{-\frac{1}{2} + \epsilon} \right) \right] \nonumber \\
&\left\{ \sum_{k = 1}^{K+H} O_p \left( |\mathcal{I}_k|^{-\frac{1}{2}} \right) \left[ O_p \left( |\mathcal{I}_k|^{\frac{1}{4}} \right) + O_p \left( m^{-\frac{1}{2} + \epsilon} \right) \right] \right\} \nonumber \\
&= \lambda_1 - 2 (K+H) O_p \left( m^{\epsilon} \right) 
- 2 \left[ O_p \left( m^{\frac{1}{2}} \right) + (K+H) O_p \left( m^{-\frac{1}{2} + \epsilon} \right) \right] \nonumber \\
&\left\{ \sum_{k = 1}^{K+H} \left[ O_p \left( |\mathcal{I}_k|^{-\frac{1}{4}} \right) + O_p \left( |\mathcal{I}_k|^{-\frac{1}{2}} m^{-\frac{1}{2} + \epsilon} \right) \right] \right\}. 
\end{align}
By taking $\epsilon < \frac{1}{2}$, the lower bound in (\ref{eq:tilde_lmd_lower_0}) can be simplified as follows: 
\begin{align}
\tilde{\lambda}_1 &\geq \lambda_1 - 2 (K+H) O_p \left( m^{\epsilon} \right) 
- 2 \left[ O_p \left( m^{\frac{1}{2}} \right) + (K+H) O_p \left( m^{-\frac{1}{2} + \epsilon} \right) \right] \nonumber \\
&\left\{ \sum_{k = 1}^{K+H} \left[ O_p \left( |\mathcal{I}_k|^{-\frac{1}{4}} \right) + O_p \left( |\mathcal{I}_k|^{-\frac{1}{2}} |\mathcal{I}_k|^{-\frac{1}{4} + \frac{1}{2} \epsilon} \right) \right] \right\} \nonumber \\
&= \lambda_1 - 2 (K+H) O_p \left( m^{\epsilon} \right) \nonumber \\
&- 2 \left[ O_p \left( m^{\frac{1}{2}} \right) + (K+H) O_p \left( m^{-\frac{1}{2} + \epsilon} \right) \right] \left[ \sum_{k = 1}^{K+H} O_p \left( |\mathcal{I}_k|^{-\frac{1}{4}} \right) \right] \nonumber \\
&\geq \lambda_1 - 2 (K+H) O_p \left( m^{\epsilon} \right) 
- 2 \left[ O_p \left( m^{\frac{1}{2}} \right) + (K+H) O_p \left( m^{-\frac{1}{2} + \epsilon} \right) \right] \nonumber \\
&(K+H) O_p \left[ \left( \min_{k = 1, \dots, K+H} |\mathcal{I}_k| \right)^{-\frac{1}{4}} \right] \nonumber \\
&= \lambda_1 - 2 (K+H) O_p \left( m^{\epsilon} \right) - 2 (K+H) O_p \left( m^{\frac{1}{2} + \epsilon} \right)\nonumber \\
&O_p \left[ \left( \min_{k = 1, \dots, K+H} |\mathcal{I}_k| \right)^{-\frac{1}{4}} \right] \nonumber \\
&= \lambda_1 - 2 (K+H) O_p \left( m^{\epsilon} \right) \left\{ O_p (1) + O_p \left[ m^{\frac{1}{2}} \left( \min_{k = 1, \dots, K+H} |\mathcal{I}_k| \right)^{-\frac{1}{4}} \right] \right\} \nonumber \\
&= \lambda_1 - 2 (K+H) O_p \left( m^{\epsilon} \right) O_p \left[ m^{\frac{1}{2}} \left( \min_{k = 1, \dots, K+H} |\mathcal{I}_k| \right)^{-\frac{1}{4}} \right]\nonumber \\ 
&= \lambda_1 - 2 (K+H) O_p \left[ m^{\frac{1}{2} + \epsilon} \left( \min_{k = 1, \dots, K+H} |\mathcal{I}_k| \right)^{-\frac{1}{4}} \right]
\end{align}
From the assumption that $(K+H) \left( \min_{k = 1, \dots, K+H} |\mathcal{I}_k| \right)^{-\frac{1}{4}} = O \left( m^{-\frac{1}{6} - \epsilon_1} \right)$ for some $\epsilon_1 > 0$ in (\ref{eq:KH_condition_KHminIk}), by taking $\epsilon < \epsilon_1$, we have
\begin{align}
\lambda_1 \leq \tilde{\lambda}_1 + O_p \left( m^{\frac{1}{3} - (\epsilon_1 - \epsilon)} \right), 
\end{align}
which concludes the proof. 
\end{proof}


\section{Proof of $| \tilde{s}_k - s_k | = O_p \left( \frac{1}{\sqrt{|\mathcal{I}_k|}} \right)$.}
\label{sec:ap_sigma_tilde}

Let $A^{(k)}$, $P^{(k)}$, and $\tilde{P}^{(k)}$, respectively, be the $k$th \textbf{null} bicluster ($k = 1, \dots, K$) or background ($k = 0$) of matrices $A$, $P$, and $\tilde{P}$. 
\setcounter{lemma}{0}
\renewcommand{\thelemma}{\Alph{section}\arabic{lemma}}
\begin{lemma}
\label{lm:sigma_tilde}
Under the assumption that $\mathbb{E} [Z_{ij}^4] < \infty$, 
\begin{align}
\label{eq:sigma_tilde_op}
| \tilde{s}_k - s_k | = O_p \left( \frac{1}{\sqrt{|\mathcal{I}_k|}} \right), 
\end{align}
where $\tilde{s}_k$ is defined as in (\ref{eq:tilde_Z}) and $\mathcal{I}_k \equiv \{ (i, j): g_{ij} = k \}$ (i.e., the set of entries in the $k$th group). 
\end{lemma}
\begin{proof}
By definition, we have 
\begin{align}
\label{eq:tilde_S2_def}
\tilde{s}_k^2 &\equiv \frac{1}{|\mathcal{I}_k|} \sum_{(i, j) \in \mathcal{I}_k} \left( A^{(k)}_{ij} - \tilde{b}_k \right)^2 
= \frac{1}{|\mathcal{I}_k|} \sum_{(i, j) \in \mathcal{I}_k} \left[ \left( A^{(k)}_{ij} \right)^2 - \tilde{b}_k^2 \right] \nonumber \\
&= \frac{1}{|\mathcal{I}_k|} \sum_{(i, j) \in \mathcal{I}_k} \left[ \left( A^{(k)}_{ij} \right)^2 - \tilde{b}_k^2 \right] - \frac{1}{|\mathcal{I}_k|} 2 b_k \sum_{(i, j) \in \mathcal{I}_k} \left( A^{(k)}_{ij} - \tilde{b}_k \right) \nonumber \\
&=\ \frac{1}{|\mathcal{I}_k|} \sum_{(i, j) \in \mathcal{I}_k} \left( A^{(k)}_{ij} - b_k \right)^2 - \left( b_k - \tilde{b}_k \right)^2, 
\end{align}
where $\tilde{b}_k \equiv \frac{1}{|\mathcal{I}_k|} \sum_{(i, j) \in \mathcal{I}_k} A^{(k)}_{ij}$. 

From (\ref{eq:tilde_S2_def}), we have
\begin{align}
\label{eq:sigma_tilde_diff}
\tilde{s}_k^2 - s_k^2 &= \frac{1}{|\mathcal{I}_k|} \sum_{(i, j) \in \mathcal{I}_k} \left( A^{(k)}_{ij} - b_k \right)^2 - s_k^2 - \left( b_k - \tilde{b}_k \right)^2 \nonumber \\
&= \frac{1}{|\mathcal{I}_k|} \sum_{(i, j) \in \mathcal{I}_k} Y^{(k)}_{ij} - \left( b_k - \tilde{b}_k \right)^2, 
\end{align}
where used the notation that $Y^{(k)}_{ij} \equiv \left( A^{(k)}_{ij} - b_k \right)^2 - s_k^2$. Based on the assumption that the entries $\left( A^{(k)}_{ij} \right)_{(i, j) \in \mathcal{I}_k}$ are generated in the i.i.d. sense in each $k$th group, the random variables $\left( Y^{(k)}_{ij} \right)_{(i, j) \in \mathcal{I}_k}$ are also independent, and their expectations and variances satisfy
\begin{align}
\label{eq:sigma_tilde_ev}
\mathbb{E} \left[ Y^{(k)}_{ij} \right] &= \mathbb{E} \left[ \left( A^{(k)}_{ij} - b_k \right)^2 \right] - s_k^2 = 0, \nonumber \\
\mathbb{V} \left[ Y^{(k)}_{ij} \right] &= \mathbb{E} \left[ \left( Y^{(k)}_{ij} \right)^2 \right] 
= \mathbb{E} \left[ \left\{ \left( A^{(k)}_{ij} - b_k \right)^2 - s_k^2 \right\}^2 \right] \nonumber \\
&= s_k^4 \left( \mathbb{E} \left[ \left( Z^{(k)}_{ij} \right)^4  \right] - 1 \right), 
\end{align}
which results in
\begin{align}
\label{eq:sigma_tilde_ev_mean}
\mathbb{E} \left[ \frac{1}{|\mathcal{I}_k|} \sum_{(i, j) \in \mathcal{I}_k} Y^{(k)}_{ij} \right] &= 0, \nonumber \\
\mathbb{V} \left[ \frac{1}{|\mathcal{I}_k|} \sum_{(i, j) \in \mathcal{I}_k} Y^{(k)}_{ij} \right] &= \frac{1}{|\mathcal{I}_k|} s_k^4 \left( \mathbb{E} \left[ \left( Z^{(k)}_{ij} \right)^4  \right] - 1 \right). 
\end{align}

From (\ref{eq:sigma_tilde_ev_mean}) and Chebyshev's inequality, for all $t>0$, we have
\begin{align}
\label{eq:chebyshev_abs}
\mathrm{Pr} \left[ \left| \frac{1}{|\mathcal{I}_k|} \sum_{(i, j) \in \mathcal{I}_k} Y^{(k)}_{ij} \right| \geq t \sqrt{ \frac{1}{|\mathcal{I}_k|} s_k^4 \left( \mathbb{E} \left[ \left( Z^{(k)}_{ij} \right)^4  \right] - 1 \right)} \right] \leq \frac{1}{t^2}, 
\end{align}
which results in 
\begin{align}
\label{eq:S_tilde_term1}
\left| \frac{1}{|\mathcal{I}_k|} \sum_{(i, j) \in \mathcal{I}_k} Y^{(k)}_{ij} \right| = O_p \left( \frac{1}{\sqrt{|\mathcal{I}_k|}} \right). 
\end{align}
from the assumption of $\mathbb{E} \left[ \left( Z^{(k)}_{ij} \right)^4  \right] < \infty$. 

As for the second term in (\ref{eq:sigma_tilde_diff}), we have
\begin{align}
\label{eq:sigma_tilde_pp}
\left( b_k - \tilde{b}_k \right)^2 &= \left[ \frac{1}{|\mathcal{I}_k|} \sum_{(i, j) \in \mathcal{I}_k} \left( P^{(k)}_{ij} - A^{(k)}_{ij} \right) \right]^2  
= \frac{s_k^2}{|\mathcal{I}_k|^2} \left( \sum_{(i, j) \in \mathcal{I}_k} Z^{(k)}_{ij} \right)^2. 
\end{align}

From (\ref{eq:sigma_tilde_pp}), we have 
\begin{align}
\label{eq:sigma_tilde_e_pp}
\mathbb{E} \left[ \left( b_k - \tilde{b}_k \right)^2 \right] 
&= \frac{s_k^2}{|\mathcal{I}_k|^2} \mathbb{E} \left[ \left( \sum_{(i, j) \in \mathcal{I}_k} Z^{(k)}_{ij} \right)^2 \right] 
= \frac{s_k^2}{|\mathcal{I}_k|^2} \mathbb{V} \left[ \sum_{(i, j) \in \mathcal{I}_k} Z^{(k)}_{ij} \right] \nonumber \\
&= \frac{s_k^2}{|\mathcal{I}_k|}, 
\end{align}
since $Z^{(k)}_{ij}$ has a unit variance. 

From (\ref{eq:sigma_tilde_e_pp}) and Markov's inequality, we have
\begin{align}
\label{eq:sigma_tilde_markov}
&\forall t>0, \ \mathrm{Pr} \left[ \left( b_k - \tilde{b}_k \right)^2 \geq t \right] \leq \frac{s_k^2}{|\mathcal{I}_k|} \frac{1}{t} \nonumber \\
\iff &\forall t'>0, \ \mathrm{Pr} \left[ \left( b_k - \tilde{b}_k \right)^2 \geq \frac{s_k^2}{|\mathcal{I}_k|} t' \right] \leq \frac{1}{t'}, 
\end{align}
which results in 
\begin{align}
\label{eq:S_tilde_term2}
\left( b_k - \tilde{b}_k \right)^2 = O_p \left( \frac{1}{|\mathcal{I}_k|} \right). 
\end{align}

Using (\ref{eq:S_tilde_term1}), (\ref{eq:S_tilde_term2}), and (\ref{eq:sigma_tilde_diff}), we finally obtain
\begin{align}
\label{eq:sigma_tilde_diff2}
| \tilde{s}_k^2 - s_k^2 | \leq | \frac{1}{|\mathcal{I}_k|} \sum_{(i, j) \in \mathcal{I}_k} Y^{(k)}_{ij} | + | \left( b_k - \tilde{b}_k \right)^2| 
= O_p \left( \frac{1}{\sqrt{|\mathcal{I}_k|}} \right), 
\end{align}
which results in
\begin{align}
\label{eq:sigma_tilde_diff3}
| \tilde{s}_k - s_k | = \frac{| \tilde{s}_k^2 - s_k^2 |}{| \tilde{s}_k + s_k |}. 
\end{align}

From (\ref{eq:sigma_tilde_diff2}), we see that $\tilde{s}_k$ converges in probability to $s_k$, and thus $\frac{1}{| \tilde{s}_k + s_k |}$ converges in probability to $\frac{1}{2 s_k} > 0$. Therefore, we have
\begin{align}
| \tilde{s}_k - s_k | = O_p \left( \frac{1}{\sqrt{|\mathcal{I}_k|}} \right), 
\end{align}
which concludes the proof. 
\end{proof}


\section{Proof of Lemma \ref{lm:zop2_eq2}}
\label{sec:lmd_1_lower}

\begin{proof}
Let $\tilde{\bm{v}}_1^{(k)} \in \mathbb{R}^{|J_k|}$ be a subvector of $\tilde{\bm{v}}_1$ corresponding to the columns of the $k$th submatrix in observed matrix $A$, and let $\tau_k \equiv \frac{s_k}{\tilde{s}_k}$. In (\ref{eq:1_s_tildes}), we have already shown that $\left| 1 - \tau_k \right| = O_p \left( \frac{1}{\sqrt{|\mathcal{I}_k|}} \right)$. The maximum eigenvalue $\tilde{\lambda}_1$ of matrix $\tilde{Z}^{\top} \tilde{Z}$ can be upper bounded as follows: 
\begin{align}
\label{eq:tilde_lambda_u1u2}
&\tilde{\lambda}_1 = \| \tilde{Z} \tilde{\bm{v}}_1 \|^2
= \norm{ \sum_{k=1}^{K+H} \tau_k \left( \underline{Z}^{(k)} - \underline{Q}^{(k)} \right) \tilde{\bm{v}}_1 }^2\ \ \ (\because (\ref{eq:Q_k})) \nonumber \\
&= \norm{ \sum_{k=1}^{K+H} \left[ \underline{Z}^{(k)} + (\tau_k - 1) \underline{Z}^{(k)} - \tau_k \underline{Q}^{(k)} \right] \tilde{\bm{v}}_1 }^2 \nonumber \\
&= \norm{ \left\{ Z + \sum_{k=1}^{K+H} \left[ (\tau_k - 1) \underline{Z}^{(k)} - \tau_k \underline{Q}^{(k)} \right] \right\} \tilde{\bm{v}}_1 }^2 \nonumber \\
&= \| Z \tilde{\bm{v}}_1 \|^2 + 2 \tilde{\bm{v}}_1^{\top} Z^{\top} \sum_{k=1}^{K+H} \left[ (\tau_k - 1) \underline{Z}^{(k)} - \tau_k \underline{Q}^{(k)} \right] \tilde{\bm{v}}_1 \nonumber \\
&+ \norm{ \sum_{k=1}^{K+H} \left[ (\tau_k - 1) \underline{Z}^{(k)} - \tau_k \underline{Q}^{(k)} \right] \tilde{\bm{v}}_1 }^2 \nonumber \\
&\leq \| Z \tilde{\bm{v}}_1 \|^2 + 2 \sqrt{\lambda_1} \sum_{k=1}^{K+H} |\tau_k- 1| \| \underline{Z}^{(k)} \tilde{\bm{v}}_1 \| - 2 \tilde{\bm{v}}_1^{\top} Z^{\top} \sum_{k=1}^{K+H} \tau_k \underline{Q}^{(k)} \tilde{\bm{v}}_1 \nonumber \\
&+ \norm{ \sum_{k=1}^{K+H} \left[ (\tau_k - 1) \underline{Z}^{(k)} - \tau_k \underline{Q}^{(k)} \right] \tilde{\bm{v}}_1 }^2 \nonumber \\
&= \| Z \tilde{\bm{v}}_1 \|^2 + 2 \sqrt{\lambda_1} \sum_{k=1}^{K+H} |\tau_k - 1| \| Z^{(k)} \tilde{\bm{v}}_1^{(k)} \| - 2 \tilde{\bm{v}}_1^{\top} Z^{\top} \sum_{k=1}^{K+H} \tau_k \underline{Q}^{(k)} \tilde{\bm{v}}_1 \nonumber \\
&+ \norm{ \sum_{k=1}^{K+H} \left[ (\tau_k - 1) \underline{Z}^{(k)} - \tau_k \underline{Q}^{(k)} \right] \tilde{\bm{v}}_1 }^2 \nonumber \\
&\leq \| Z \tilde{\bm{v}}_1 \|^2 + 2 \sqrt{\lambda_1} \sum_{k=1}^{K+H} |\tau_k - 1| \sqrt{\lambda_1^{(k)}} - 2 \tilde{\bm{v}}_1^{\top} Z^{\top} \sum_{k=1}^{K+H} \tau_k \underline{Q}^{(k)} \tilde{\bm{v}}_1 \nonumber \\
&+ \norm{ \sum_{k=1}^{K+H} \left[ (\tau_k - 1) \underline{Z}^{(k)} - \tau_k \underline{Q}^{(k)} \right] \tilde{\bm{v}}_1 }^2 \nonumber \\
&= \| Z \tilde{\bm{v}}_1 \|^2 + 2 O_p \left( m^{\frac{1}{2}} \right) \sum_{k=1}^{K+H} O_p \left( |\mathcal{I}_k|^{-\frac{1}{4}} \right) - 2 \tilde{\bm{v}}_1^{\top} Z^{\top} \sum_{k=1}^{K+H} \tau_k \underline{Q}^{(k)} \tilde{\bm{v}}_1 \nonumber \\
&+ \norm{ \sum_{k=1}^{K+H} \left[ (\tau_k - 1) \underline{Z}^{(k)} - \tau_k \underline{Q}^{(k)} \right] \tilde{\bm{v}}_1 }^2 \nonumber \\
&= \| Z \tilde{\bm{v}}_1 \|^2 + 2 (K+H) O_p \left[ m^{\frac{1}{2}} \left( \min_{k = 1, \dots, K+H} |\mathcal{I}_k| \right)^{-\frac{1}{4}} \right] - 2 \tilde{\bm{v}}_1^{\top} Z^{\top} \sum_{k=1}^{K+H} \tau_k \underline{Q}^{(k)} \tilde{\bm{v}}_1 \nonumber \\
&+ \norm{ \sum_{k=1}^{K+H} \left[ (\tau_k - 1) \underline{Z}^{(k)} - \tau_k \underline{Q}^{(k)} \right] \tilde{\bm{v}}_1 }^2 \nonumber \\
&\leq \| Z \tilde{\bm{v}}_1 \|^2 + 2 (K+H) O_p \left[ m^{\frac{1}{2}} \left( \min_{k = 1, \dots, K+H} |\mathcal{I}_k| \right)^{-\frac{1}{4}} \right] - 2 \tilde{\bm{v}}_1^{\top} Z^{\top} \sum_{k=1}^{K+H} \tau_k \underline{Q}^{(k)} \tilde{\bm{v}}_1 \nonumber \\
&+ \left[ \norm{ \sum_{k=1}^{K+H} (\tau_k - 1) \underline{Z}^{(k)} \tilde{\bm{v}}_1 } + \norm{ \sum_{k=1}^{K+H} \tau_k \underline{Q}^{(k)} \tilde{\bm{v}}_1 } \right]^2 \nonumber \\
&\leq \| Z \tilde{\bm{v}}_1 \|^2 + 2 (K+H) O_p \left[ m^{\frac{1}{2}} \left( \min_{k = 1, \dots, K+H} |\mathcal{I}_k| \right)^{-\frac{1}{4}} \right] - 2 \tilde{\bm{v}}_1^{\top} Z^{\top} \sum_{k=1}^{K+H} \tau_k \underline{Q}^{(k)} \tilde{\bm{v}}_1 \nonumber \\
&+ \left[ \sum_{k=1}^{K+H} |\tau_k - 1| \| \underline{Z}^{(k)} \tilde{\bm{v}}_1 \| + \sum_{k=1}^{K+H} \tau_k \| \underline{Q}^{(k)} \tilde{\bm{v}}_1 \| \right]^2 \nonumber \\
&\leq \| Z \tilde{\bm{v}}_1 \|^2 + 2 (K+H) O_p \left[ m^{\frac{1}{2}} \left( \min_{k = 1, \dots, K+H} |\mathcal{I}_k| \right)^{-\frac{1}{4}} \right] - 2 \sum_{k=1}^{K+H} \tau_k \tilde{\bm{v}}_1^{\top} Z^{\top} \underline{Q}^{(k)} \tilde{\bm{v}}_1 \nonumber \\
&+ \left[ \sum_{k=1}^{K+H} O_p \left( |\mathcal{I}_k|^{-\frac{1}{4}} \right) + \sum_{k=1}^{K+H} \tau_k \| \underline{Q}^{(k)} \tilde{\bm{v}}_1 \| \right]^2 \nonumber \\
&\left( \because \| \underline{Z}^{(k)} \tilde{\bm{v}}_1 \| \leq \sqrt{\lambda_1^{(k)}} = O_p \left( |\mathcal{I}_k|^{\frac{1}{4}} \right) \right) \nonumber \\
&= \| Z \tilde{\bm{v}}_1 \|^2 + O_p \left[ m^{\frac{1}{2}} C^{(K, H)} \right] - 2 \sum_{k=1}^{K+H} \tau_k \tilde{\bm{v}}_1^{\top} Z^{\top} \underline{Q}^{(k)} \tilde{\bm{v}}_1 \nonumber \\
&+ \left[ O_p \left( C^{(K, H)} \right) + \sum_{k=1}^{K+H} \tau_k \| \underline{Q}^{(k)} \tilde{\bm{v}}_1 \| \right]^2, 
\end{align}
where we denote $C^{(K, H)} \equiv (K+H) \left( \min_{k = 1, \dots, K+H} |\mathcal{I}_k| \right)^{-\frac{1}{4}}$. 

The eigenvectors $\{ \bm{v}_j \}$ of symmetric matrix $Z^{\top} Z$ form an orthonormal system, and thus there exists a unique set of coefficients $\{ c_j \}$ such that
\begin{align}
\label{eq:tilde_v_decompose}
\tilde{\bm{v}}_1 = \sum_{j = 1}^p c_j \bm{v}_j = \tilde{\bm{v}}^{\mathrm{L}} + \tilde{\bm{v}}^{\mathrm{S}}, 
\end{align}
where
\begin{align}
\label{eq:tilde_v1v2}
&\tilde{\bm{v}}^{\mathrm{L}} \equiv \sum_{j = 1}^t c_j \bm{v}_j, \ \ \ 
\tilde{\bm{v}}^{\mathrm{S}} \equiv \sum_{j = t + 1}^p c_j \bm{v}_j, \nonumber \\
&\lambda_t \geq \lambda_1 - n^d, \ \ \ 
\lambda_{t + 1} < \lambda_1 - n^d, \ \ \ 
d = \frac{5}{7}. 
\end{align}
By substituting (\ref{eq:tilde_v_decompose}) into the last term in (\ref{eq:tilde_lambda_u1u2}) and from the similar discussion as in (\ref{eq:underline_QQv}), 
\begin{align}
\label{eq:Qv_norm2}
&\| \underline{Q}^{(k)} \tilde{\bm{v}}_1 \|^2 = \tilde{\bm{v}}_1^{\top} (\underline{Q}^{(k)})^{\top} \underline{Q}^{(k)} \tilde{\bm{v}}_1 
= \sum_{j = 1}^p \sum_{j' = 1}^p c_j c_{j'} \bm{v}_j^{\top} (\underline{Q}^{(k)})^{\top} \underline{Q}^{(k)} \bm{v}_{j'} \nonumber \\
&= \sum_{j = 1}^p \sum_{j' = 1}^p c_j c_{j'} |I_k| |J_k| \eta_k^2 (\bm{v}_j^{\top} \bm{u}^{(k)}) (\bm{v}_{j'}^{\top} \bm{u}^{(k)}) 
= |I_k| |J_k| \eta_k^2 \left[ \sum_{j = 1}^p c_j (\bm{v}_j^{\top} \bm{u}^{(k)}) \right]^2 \nonumber \\
&\leq |I_k| |J_k| \eta_k^2 \left[ \sqrt{\sum_{j = 1}^p c_j^2} \sqrt{\sum_{j = 1}^p (\bm{v}_j^{\top} \bm{u}^{(k)})^2} \right]^2 
= |I_k| |J_k| \eta_k^2 \| \tilde{\bm{v}}_1 \|^2 \left[ \sum_{j = 1}^p (\bm{v}_j^{\top} \bm{u}^{(k)})^2 \right] \nonumber \\
&= |I_k| |J_k| \eta_k^2 \left[ \sum_{j = 1}^p (\bm{v}_j^{\top} \bm{u}^{(k)})^2 \right] 
\leq |I_k| |J_k| \eta_k^2\ p \max_{j = 1, \dots, p} (\bm{v}_j^{\top} \bm{u}^{(k)})^2 \nonumber \\
&= |\mathcal{I}_k| O_p \left( |\mathcal{I}_k|^{-1} \right) p\ O_p \left( m^{-1 + 2\epsilon} \right) 
= O_p \left( m^{2\epsilon} \right). 
\end{align}
Here, we used the fact that (\ref{eq:delocalization}) holds from \cite{Bloemendal2016}. 

By combining (\ref{eq:tilde_lambda_u1u2}) and (\ref{eq:Qv_norm2}), 
\begin{align}
\label{eq:tilde_lambda_u1u2_2}
&\tilde{\lambda}_1 \leq \| Z \tilde{\bm{v}}_1 \|^2 - 2 \sum_{k=1}^{K+H} \tau_k \tilde{\bm{v}}_1^{\top} Z^{\top} \underline{Q}^{(k)} \tilde{\bm{v}}_1 + O_p \left[ m^{\frac{1}{2}} C^{(K, H)} \right] \nonumber \\
&+ \left[ O_p \left( C^{(K, H)} \right) + \sum_{k=1}^{K+H} \tau_k O_p \left( m^{\epsilon} \right) \right]^2 \nonumber \\
&= \| Z \tilde{\bm{v}}_1 \|^2 - 2 \sum_{k=1}^{K+H} \tau_k \tilde{\bm{v}}_1^{\top} Z^{\top} \underline{Q}^{(k)} \tilde{\bm{v}}_1 + O_p \left[ m^{\frac{1}{2}} C^{(K, H)} \right] \nonumber \\
&+ \left[ O_p \left( C^{(K, H)} \right) + \sum_{k=1}^{K+H} \left( 1 + O_p \left( |\mathcal{I}_k|^{-\frac{1}{2}} \right) \right) O_p \left( m^{\epsilon} \right) \right]^2 \nonumber \\
&= \| Z \tilde{\bm{v}}_1 \|^2 - 2 \sum_{k=1}^{K+H} \tau_k \tilde{\bm{v}}_1^{\top} Z^{\top} \underline{Q}^{(k)} \tilde{\bm{v}}_1 + O_p \left[ m^{\frac{1}{2}} C^{(K, H)} \right] \nonumber \\
&+ \left[ O_p \left( C^{(K, H)} \right) + (K+H) O_p \left( m^{\epsilon} \right) \right]^2 \nonumber \\
&= \| Z \tilde{\bm{v}}_1 \|^2 - 2 \sum_{k=1}^{K+H} \tau_k \tilde{\bm{v}}_1^{\top} Z^{\top} \underline{Q}^{(k)} \tilde{\bm{v}}_1 + O_p \left[ (K+H) m^{\frac{1}{2}} \left( \min_{k = 1, \dots, K+H} |\mathcal{I}_k| \right)^{-\frac{1}{4}} \right] \nonumber \\
&+ O_p \left[ (K+H)^2 m^{2\epsilon} \right]. 
\end{align}

As for the third term in (\ref{eq:tilde_lambda_u1u2_2}), based on the assumption in (\ref{eq:KH_condition_KHminIk}), 
\begin{align}
\label{eq:KH_min_I_k}
O_p \left[ (K+H) m^{\frac{1}{2}} \left( \min_{k = 1, \dots, K+H} |\mathcal{I}_k| \right)^{-\frac{1}{4}} \right] = O_p \left( m^{\frac{1}{3} - \epsilon_1} \right). 
\end{align}
With regard to the fourth term in (\ref{eq:tilde_lambda_u1u2_2}), based on the assumption that $K+H = O \left( m^{\frac{1}{42} - \epsilon_1} \right)$ for some $\epsilon_1 > 0$ in (\ref{eq:KH_condition_KH_sum}), 
by taking $\epsilon < \epsilon_1$, 
\begin{align}
\label{eq:KH_2_m_2}
O_p \left[ (K+H)^2 m^{2\epsilon} \right] = O_p \left( m^{\frac{1}{21} - 2(\epsilon_1 - \epsilon)} \right). 
\end{align}

An upper bound of the first and second terms in (\ref{eq:tilde_lambda_u1u2_2}) is given by
\begin{align}
\label{eq:Zv_vZQv}
&\| Z \tilde{\bm{v}}_1 \|^2 - 2 \sum_{k=1}^{K+H} \tau_k \tilde{\bm{v}}_1^{\top} Z^{\top} \underline{Q}^{(k)} \tilde{\bm{v}}_1 \nonumber \\
&= (\tilde{\bm{v}}^{\mathrm{L}} + \tilde{\bm{v}}^{\mathrm{S}})^{\top} Z^{\top} Z (\tilde{\bm{v}}^{\mathrm{L}} + \tilde{\bm{v}}^{\mathrm{S}}) - 2 \sum_{k=1}^{K+H} \tau_k \tilde{\bm{v}}_1^{\top} Z^{\top} \underline{Q}^{(k)} \tilde{\bm{v}}_1 \nonumber \\
&= \left( \tilde{\bm{v}}^{\mathrm{L}} \right)^{\top} Z^{\top} Z \tilde{\bm{v}}^{\mathrm{L}} + \left( \tilde{\bm{v}}^{\mathrm{S}} \right)^{\top} Z^{\top} Z \tilde{\bm{v}}^{\mathrm{S}} - 2 \sum_{k=1}^{K+H} \tau_k \tilde{\bm{v}}_1^{\top} Z^{\top} \underline{Q}^{(k)} \tilde{\bm{v}}_1 \nonumber \\
&= \left( \sum_{j = 1}^t c_j \bm{v}_j \right)^{\top} \left( \sum_{j = 1}^t c_j Z^{\top} Z \bm{v}_j \right) + \left( \sum_{j = t+1}^p c_j \bm{v}_j \right)^{\top} \left( \sum_{j = t+1}^p c_j Z^{\top} Z \bm{v}_j \right) \nonumber \\
&- 2 \sum_{k=1}^{K+H} \tau_k \tilde{\bm{v}}_1^{\top} Z^{\top} \underline{Q}^{(k)} \tilde{\bm{v}}_1 \nonumber \\
&= \left( \sum_{j = 1}^t c_j \bm{v}_j \right)^{\top} \left( \sum_{j = 1}^t c_j \lambda_j \bm{v}_j \right) + \left( \sum_{j = t+1}^p c_j \bm{v}_j \right)^{\top} \left( \sum_{j = t+1}^p c_j \lambda_j \bm{v}_j \right) \nonumber \\
&- 2 \sum_{k=1}^{K+H} \tau_k \tilde{\bm{v}}_1^{\top} Z^{\top} \underline{Q}^{(k)} \tilde{\bm{v}}_1 \nonumber \\
&= \sum_{j = 1}^t c_j^2 \lambda_j \| \bm{v}_j \|^2 + \sum_{j = t+1}^p c_j^2 \lambda_j \| \bm{v}_j \|^2 - 2 \sum_{k=1}^{K+H} \tau_k \tilde{\bm{v}}_1^{\top} Z^{\top} \underline{Q}^{(k)} \tilde{\bm{v}}_1  \nonumber \\
&= \sum_{j = 1}^t c_j^2 \lambda_j + \sum_{j = t+1}^p c_j^2 \lambda_j - 2 \sum_{k=1}^{K+H} \tau_k \tilde{\bm{v}}_1^{\top} Z^{\top} \underline{Q}^{(k)} \tilde{\bm{v}}_1 \nonumber \\
&\leq \lambda_1 \sum_{j = 1}^t c_j^2 + \lambda_{t+1} \sum_{j = t+1}^p c_j^2 - 2 \sum_{k=1}^{K+H} \tau_k \tilde{\bm{v}}_1^{\top} Z^{\top} \underline{Q}^{(k)} \tilde{\bm{v}}_1 \nonumber \\
&\leq \lambda_1 \sum_{j = 1}^t c_j^2 + (\lambda_1 - n^d) \sum_{j = t+1}^p c_j^2 - 2 \sum_{k=1}^{K+H} \tau_k \tilde{\bm{v}}_1^{\top} Z^{\top} \underline{Q}^{(k)} \tilde{\bm{v}}_1\ \ \ (\because (\ref{eq:tilde_v1v2})) \nonumber \\
&= \lambda_1 \sum_{j = 1}^p c_j^2 - n^d \sum_{j = t+1}^p c_j^2 - 2 \sum_{k=1}^{K+H} \tau_k \tilde{\bm{v}}_1^{\top} Z^{\top} \underline{Q}^{(k)} \tilde{\bm{v}}_1 \nonumber \\
&= \lambda_1 \| \tilde{\bm{v}}_1 \|^2 - n^d \| \tilde{\bm{v}}^{\mathrm{S}} \|^2 - 2 \sum_{k=1}^{K+H} \tau_k \tilde{\bm{v}}_1^{\top} Z^{\top} \underline{Q}^{(k)} \tilde{\bm{v}}_1 \nonumber \\
&= \lambda_1 - n^d \| \tilde{\bm{v}}^{\mathrm{S}} \|^2 - 2 \sum_{k=1}^{K+H} \tau_k \tilde{\bm{v}}_1^{\top} Z^{\top} \underline{Q}^{(k)} \tilde{\bm{v}}^{\mathrm{L}} - 2 \sum_{k=1}^{K+H} \tau_k \tilde{\bm{v}}_1^{\top} Z^{\top} \underline{Q}^{(k)} \tilde{\bm{v}}^{\mathrm{S}}. 
\end{align}

Let $\bm{u}^{(k)} \in \mathbb{R}^p$ be a vector whose entries are defined by $\bm{u}^{(k)}_j = \frac{1}{\sqrt{|J_k|}}$ if $j \in J_k$ and $\bm{u}^{(k)}_j = 0$ otherwise. 
As for the third term in (\ref{eq:Zv_vZQv}), using the fact that $\underline{Q}^{(k)} \bm{v}_j = \eta_k |J_k| (\bm{v}_j^{\top} \bm{u}^{(k)}) \bm{u}^{(k)}$, for all $\epsilon > 0$, 
\begin{align}
\label{vZQv1_upper}
&- \tilde{\bm{v}}_1^{\top} Z^{\top} \underline{Q}^{(k)} \tilde{\bm{v}}^{\mathrm{L}} 
\leq | \tilde{\bm{v}}_1^{\top} Z^{\top} \underline{Q}^{(k)} \tilde{\bm{v}}^{\mathrm{L}} | 
= \left| \sum_{j = 1}^t c_j \tilde{\bm{v}}_1^{\top} Z^{\top} \underline{Q}^{(k)} \bm{v}_j \right| \nonumber \\
&= \left| \eta_k |J_k| \tilde{\bm{v}}_1^{\top} Z^{\top} \bm{u}^{(k)} \sum_{j = 1}^t c_j (\bm{v}_j^{\top} \bm{u}^{(k)}) \right| \nonumber \\
&= O_p \left( |\mathcal{I}_k|^{-\frac{1}{2}} \right) |J_k| \left| \sum_{j = 1}^t c_j (\bm{v}_j^{\top} \bm{u}^{(k)}) \right| \left| \tilde{\bm{v}}_1^{\top} Z^{\top} \bm{u}^{(k)} \right| \nonumber \\
&= O_p (1) \left| \sum_{j = 1}^t c_j (\bm{v}_j^{\top} \bm{u}^{(k)}) \right| \left| \tilde{\bm{v}}_1^{\top} Z^{\top} \bm{u}^{(k)} \right| \nonumber \\
&\leq O_p (1) \sqrt{\sum_{j = 1}^t c_j^2} \sqrt{\sum_{j = 1}^t | \bm{v}_j^{\top} \bm{u}^{(k)} |^2} \left| \tilde{\bm{v}}_1^{\top} Z^{\top} \bm{u}^{(k)} \right| \nonumber \\
&\leq O_p (1) \| \tilde{\bm{v}}_1 \| \sqrt{t}\ O_p \left( m^{-\frac{1}{2} + \epsilon} \right) \left| \tilde{\bm{v}}_1^{\top} Z^{\top} \bm{u}^{(k)} \right|\ \ \ 
\nonumber \\
&= \sqrt{t}\ O_p \left( m^{-\frac{1}{2} + \epsilon} \right) \left| \tilde{\bm{v}}_1^{\top} Z^{\top} \bm{u}^{(k)} \right| 
\leq \sqrt{t}\ O_p \left( m^{-\frac{1}{2} + \epsilon} \right) \| \tilde{\bm{v}}_1^{\top} Z^{\top} \| \| \bm{u}^{(k)} \| \nonumber \\
&= \sqrt{t}\ O_p \left( m^{-\frac{1}{2} + \epsilon} \right) \| \tilde{\bm{v}}_1^{\top} Z^{\top} \| 
\leq \sqrt{t}\ O_p \left( m^{-\frac{1}{2} + \epsilon} \right) \sqrt{\lambda_1} \nonumber \\
&= \sqrt{t}\ O_p \left( m^{-\frac{1}{2} + \epsilon} \right) O_p \left( m^{\frac{1}{2}} \right) 
= \sqrt{t}\ O_p \left( m^{\epsilon} \right). 
\end{align}

With regard to the fourth term in (\ref{eq:Zv_vZQv}), we have
\begin{align}
\label{vZQv2_upper}
&-\tilde{\bm{v}}_1^{\top} Z^{\top} \underline{Q}^{(k)} \tilde{\bm{v}}^{\mathrm{S}} 
\leq |\tilde{\bm{v}}_1^{\top} Z^{\top} \underline{Q}^{(k)} \tilde{\bm{v}}^{\mathrm{S}}| 
\leq \| \tilde{\bm{v}}_1 \| \| Z^{\top} \underline{Q}^{(k)} \tilde{\bm{v}}^{\mathrm{S}} \| 
= \| Z^{\top} \underline{Q}^{(k)} \tilde{\bm{v}}^{\mathrm{S}} \| \nonumber \\
&\leq \| Z^{\top} \underline{Q}^{(k)} \|_{\mathrm{op}} \| \tilde{\bm{v}}^{\mathrm{S}} \| 
\leq \| Z \|_{\mathrm{op}} \| \underline{Q}^{(k)} \|_{\mathrm{F}} \| \tilde{\bm{v}}^{\mathrm{S}} \| 
= \sqrt{\lambda_1 |\mathcal{I}_k|}\ |\eta_k|\ \| \tilde{\bm{v}}^{\mathrm{S}} \|. 
\end{align}

By substituting (\ref{vZQv1_upper}) and (\ref{vZQv2_upper}) into (\ref{eq:Zv_vZQv}), 
\begin{align}
\label{eq:Zv_vZQv2}
&\| Z \tilde{\bm{v}}_1 \|^2 - 2 \sum_{k=1}^{K+H} \tau_k \tilde{\bm{v}}_1^{\top} Z^{\top} \underline{Q}^{(k)} \tilde{\bm{v}}_1 \nonumber \\
&\leq \lambda_1 - n^d \| \tilde{\bm{v}}^{\mathrm{S}} \|^2 + 2 \sum_{k=1}^{K+H} \tau_k \sqrt{t}\ O_p \left( m^{\epsilon} \right) + 2 \sum_{k=1}^{K+H} \tau_k \sqrt{\lambda_1 |\mathcal{I}_k|}\ |\eta_k|\ \| \tilde{\bm{v}}^{\mathrm{S}} \| \nonumber \\
&= \lambda_1 - n^d \| \tilde{\bm{v}}^{\mathrm{S}} \|^2 + 2 \sum_{k=1}^{K+H} \left[ 1 + O_p \left( |\mathcal{I}_k|^{-\frac{1}{2}} \right) \right] \sqrt{t}\ O_p \left( m^{\epsilon} \right) \nonumber \\
&+ 2 \sum_{k=1}^{K+H} \tau_k \sqrt{\lambda_1 |\mathcal{I}_k|}\ |\eta_k|\ \| \tilde{\bm{v}}^{\mathrm{S}} \| \nonumber \\
&= \lambda_1 - n^d \| \tilde{\bm{v}}^{\mathrm{S}} \|^2 + 2 \sum_{k=1}^{K+H} \sqrt{t}\ O_p \left( m^{\epsilon} \right) + 2 \sum_{k=1}^{K+H} \tau_k \sqrt{\lambda_1 |\mathcal{I}_k|}\ |\eta_k|\ \| \tilde{\bm{v}}^{\mathrm{S}} \| \nonumber \\
&= \lambda_1 - n^d \| \tilde{\bm{v}}^{\mathrm{S}} \|^2 + \sqrt{t}\ O_p \left[ (K+H) m^{\epsilon} \right] + 2 \sum_{k=1}^{K+H} \tau_k \sqrt{\lambda_1 |\mathcal{I}_k|}\ |\eta_k|\ \| \tilde{\bm{v}}^{\mathrm{S}} \|. 
\end{align}

From now on, we derive the probabilistic order of $t$. We denote the $j$th normalized eigenvalue of matrix $Z^{\top} Z$ as $\nu_j \equiv \frac{1}{n} \lambda_j$, and define the following variables: 
\begin{align}
\nu_+ \equiv \left( 1 + \sqrt{\frac{p}{n}} \right)^2, \ \ \ 
\nu_- \equiv \left( 1 - \sqrt{\frac{p}{n}} \right)^2, \ \ \ 
\epsilon_3 \equiv \nu_+ - \nu_1. 
\end{align}
Note that $|\epsilon_3| = O_p \left( \phi^C m^{-\frac{2}{3}} \right)$ holds for some constant $C>0$ and $\phi \equiv (\log p)^{\log \log p}$ from (4.1) of \cite{Pillai2014}. 
Since $\phi = o(m^{\epsilon_4})$ holds for any $\epsilon_4 > 0$, by taking $\epsilon_5 \equiv C \epsilon_4$, we have 
\begin{align}
\label{eq:eps3}
|\epsilon_3| = O_p \left( m^{-\frac{2}{3} + \epsilon_5} \right),\ \mathrm{for\ any}\ \epsilon_5 > 0. 
\end{align}

From (3.7) of \cite{Pillai2014}, we have
\begin{align}
\label{eq:bar_n_tp}
\left| \bar{n} - \frac{t}{p} \right| = O_p \left( m^{-1 + \epsilon_6} \right),\ \mathrm{for\ all}\ \epsilon_6 > 0, 
\end{align}
where $\bar{n} \equiv \int_{\nu_1 - n^{d-1}}^{\infty} q(x) \mathrm{d}x$ and 
\begin{align}
\label{eq:mp_dist}
q(x) = \frac{1}{2\pi} \frac{n}{p} \frac{\sqrt{\max \{ (\nu_+ - x)(x - \nu_-), 0 \}}}{x}. 
\end{align}

From (\ref{eq:mp_dist}), by taking $\epsilon_5 < d - \frac{1}{3} = \frac{8}{21}$, we have 
\begin{align}
\label{eq:q_nu_nd}
&q(\nu_1 - n^{d-1}) = q(\nu_+ - n^{d-1} - \epsilon_3) \nonumber \\
&= \frac{\sqrt{\nu_+ - \nu_-}}{\nu_+} \left[ n^{\frac{d-1}{2}} + O_p \left( m^{-\frac{1}{3} + \frac{\epsilon_5}{2}} \right) \right] \left[ 1 + O \left( m^{\frac{d-1}{2}} \right) + O_p \left( m^{-\frac{1}{3} + \frac{\epsilon_5}{2}} \right) \right] \nonumber \\
&= \frac{\sqrt{\nu_+ - \nu_-}}{\nu_+} n^{\frac{d-1}{2}} + O_p \left( m^{\frac{d-1}{2}} \right). 
\end{align}

From (\ref{eq:eps3}) and (\ref{eq:q_nu_nd}), by setting $\epsilon_5 < d - \frac{1}{3}$, 
\begin{align}
\label{eq:bar_n}
\bar{n} &= \int_{\nu_1 - n^{d-1}}^{\infty} q(x) \mathrm{d}x 
\leq \left| \int_{\nu_1 - n^{d-1}}^{\nu_+} q(x) \mathrm{d}x \right| + \left| \int_{\nu_+}^{\infty} q(x) \mathrm{d}x \right| \nonumber \\
&= \left| \int_{\nu_1 - n^{d-1}}^{\nu_+} q(x) \mathrm{d}x \right| 
\leq |\epsilon_3 + n^{d-1}|\ q(\nu_1 - n^{d-1}) \nonumber \\
&= O_p \left( m^{d-1} \right) O_p \left( m^{\frac{d-1}{2}} \right) 
= O_p \left( m^{\frac{3(d-1)}{2}} \right). 
\end{align}

From (\ref{eq:bar_n}) and (\ref{eq:bar_n_tp}), by setting $\epsilon_6 < \frac{3}{2} d - \frac{1}{2}$, 
\begin{align}
\label{eq:t_op_3_2}
t = O_p \left( m^{\frac{3}{2} d - \frac{1}{2}} \right). 
\end{align}

By substituting (\ref{eq:t_op_3_2}) into (\ref{eq:Zv_vZQv2}) and from the assumption in (\ref{eq:tilde_v1v2}) that $d = \frac{5}{7}$, for all $\epsilon > 0$, 
\begin{align}
\label{eq:Zv_vZQv2_2}
&\| Z \tilde{\bm{v}}_1 \|^2 - 2 \sum_{k=1}^{K+H} \tau_k \tilde{\bm{v}}_1^{\top} Z^{\top} \underline{Q}^{(k)} \tilde{\bm{v}}_1 \nonumber \\
&\leq \lambda_1 + O_p \left[ (K+H) m^{\frac{2}{7} + \epsilon} \right] + \| \tilde{\bm{v}}^{\mathrm{S}} \| \left( 2 \sqrt{\lambda_1}\ \sum_{k=1}^{K+H} \tau_k \sqrt{|\mathcal{I}_k|}\ |\eta_k| - n^d \| \tilde{\bm{v}}^{\mathrm{S}} \| \right) \nonumber \\
&= \lambda_1 + O_p \left[ (K+H) m^{\frac{2}{7} + \epsilon} \right] + \| \tilde{\bm{v}}^{\mathrm{S}} \| \left( 2 \sqrt{\lambda_1}\ \varpi - n^d \| \tilde{\bm{v}}^{\mathrm{S}} \| \right), 
\end{align}
where 
\begin{align}
\label{eq:varpi}
\varpi &\equiv \sum_{k=1}^{K+H} \tau_k \sqrt{|\mathcal{I}_k|}\ |\eta_k| &= \sum_{k=1}^{K+H} \left[ 1 + O_p \left( \frac{1}{\sqrt{|\mathcal{I}_k|}} \right) \right] \sqrt{|\mathcal{I}_k|}\ O_p \left( \frac{1}{\sqrt{|\mathcal{I}_k|}} \right) \nonumber \\
&= O_p (K+H). 
\end{align}

By using $d = \frac{5}{7}$ and (\ref{eq:varpi}), the third term in the right side of (\ref{eq:Zv_vZQv2_2}) can be upper bounded by
\begin{align}
\label{eq:Zv_op1_agmean}
&\| \tilde{\bm{v}}^{\mathrm{S}} \| \left( 2 \sqrt{\lambda_1}\ \varpi - n^d \| \tilde{\bm{v}}^{\mathrm{S}} \| \right) 
= 2 \| \tilde{\bm{v}}^{\mathrm{S}} \| \sqrt{\lambda_1}\ \varpi - n^d \| \tilde{\bm{v}}^{\mathrm{S}} \|^2 - \frac{\lambda_1 \varpi^2}{n^d} + \frac{\lambda_1 \varpi^2}{n^d} \nonumber \\
&= - \frac{(\sqrt{\lambda_1}\ \varpi - n^d \| \tilde{\bm{v}}^{\mathrm{S}} \|)^2}{n^d} + \frac{\lambda_1 \varpi^2}{n^d} 
\leq \frac{\lambda_1 \varpi^2}{n^d} = O_p \left[ (K+H)^2 m^{\frac{2}{7}} \right], 
\end{align}
which results in that
\begin{align}
\label{eq:eq:Zv_op1_4}
&\| Z \tilde{\bm{v}}_1 \|^2 - 2 \sum_{k=1}^{K+H} \tau_k \tilde{\bm{v}}_1^{\top} Z^{\top} \underline{Q}^{(k)} \tilde{\bm{v}}_1 \nonumber \\
&\leq \lambda_1 + O_p \left[ (K+H) m^{\frac{2}{7} + \epsilon} \right] + O_p \left[ (K+H)^2 m^{\frac{2}{7}} \right] \nonumber \\
&\leq \lambda_1 + O_p \left[ (K+H)^2 m^{\frac{2}{7} + \epsilon} \right],\ \mathrm{for\ all}\ \epsilon > 0. 
\end{align}
Therefore, from (\ref{eq:tilde_lambda_u1u2_2}), (\ref{eq:KH_min_I_k}), and (\ref{eq:KH_2_m_2}), for all $\epsilon > 0$, 
\begin{align}
&\tilde{\lambda}_1 \leq \lambda_1 + O_p \left[ (K+H)^2 m^{\frac{2}{7} + \epsilon} \right] + O_p \left( m^{\frac{1}{3} - \epsilon_1} \right) + O_p \left( m^{\frac{1}{21} - 2(\epsilon_1 - \epsilon)} \right). 
\end{align}

From the assumption that $K+H = O \left( m^{\frac{1}{42} - \epsilon_1} \right)$ for some $\epsilon_1 > 0$ in (\ref{eq:KH_condition_KH_sum}), by taking $\epsilon < \epsilon_1$, we finally obtain
\begin{align}
&\tilde{\lambda}_1 \leq \lambda_1 + O_p \left( m^{\frac{1}{3} - \tilde{\epsilon}} \right),\ \mathrm{for\ some}\ \tilde{\epsilon} > 0, 
\end{align}
which concludes the proof. 
\end{proof}


\section{Disjoint submatrix localization algorithm based on simulated annealing}
\label{sec:consistency_approx}

In this section, we develop a simulated annealing (SA) algorithm to find the bicluster structure of a given observed matrix. As in \cite{Flynn2020}, the proposed algorithm is based on the (generalized) profile likelihood \cite{Murphy2000}. Given an estimated bicluster assignment $\hat{g}$, the generalized profile-likelihood criterion for an exponential family model is given by
\begin{align}
\label{eq:gen_PL_def}
F (\hat{g}) \equiv \sum_{k = 0}^K \hat{p}_k f \left( \frac{1}{|\hat{\mathcal{I}}_k|} \sum_{(i, j) \in \hat{\mathcal{I}}_k} A_{ij} \right), 
\end{align}
where $\hat{p}_k \in \mathbb{R}$ is the proportion of entries in the $k$th group ($k \in \{ 0, 1, \dots, K \}$) in the estimated bicluster structures to all the $np$ entries and $f: \mathbb{R} \mapsto \mathbb{R}$ is a given function. The specific definition of function $f$ for each experiment is given in Sect.~\ref{sec:experiments}. 

\subsubsection{The naive implementation of SA-based submatrix localization}
\label{sec:naive_sa}

Let $\mathcal{G}_K$ be a set of all bicluster structures with (non-empty) $K$ biclusters, which are disjoint, but which are not necessarily bi-disjoint (as shown in Figure \ref{fig:submatrix} (b)). In SA, we first define a sequence of temperatures $\{ T_t \}_{t=0}^{\infty}$, a threshold $\epsilon^{\mathrm{SA}}$, and the initial state (i.e., bicluster assignment) $\hat{g}^{(0)} \in \mathcal{G}_K$. For each state $g \in \mathcal{G}_K$, we also define a set of its neighbors $N(g) \subseteq \mathcal{G}_K$ and a transition probability $R(g, g') \in [0, 1]$ to a given state $g' \in \mathcal{G}_K$. Here, we set $R(g, g') = 0$ iff $g' \notin N(g)$. 

For each step $t = 0, 1, 2, \dots$, if $T_t < \epsilon$, we stop the algorithm and output the final state $\hat{g}^{(t)}$. If $T_t \geq \epsilon$, we randomly choose a candidate for the next state $\tilde{g} \in N(\hat{g}^{(t)})$ with probability $R(\hat{g}^{(t)}, \tilde{g})$, and compute the difference of the objective function value $\Delta F \equiv F(\tilde{g}) - F(\hat{g}^{(t)})$. 
If $\Delta F > 0$, we set the next state at $\hat{g}^{(t+1)} = \tilde{g}$. If $\Delta F \leq 0$, we set the next state at $\hat{g}^{(t+1)} = \tilde{g}$ with probability $\exp \left( \frac{\Delta F}{T_t} \right)$, and set it at the current state $\hat{g}^{(t+1)} = \hat{g}^{(t)}$ with probability $1 - \exp \left( \frac{\Delta F}{T_t} \right)$. 

Specifically, we propose Algorithm \ref{algo:max_PL_SA} as an example of SA for approximately maximizing the generalized profile likelihood $F$. In Algorithm \ref{algo:max_PL_SA}, we define that the neighbors $N(g)$ of a state $g$ is a set of all possible bicluster assignments that can be obtained by adding/removing one row or column to/from one bicluster in $g$. As for the transition probability, we define that one of the elements in $N(g)$ is chosen from the uniform distribution on $N(g)$ (i.e., $R(g, g') = 1/|N(g)|$ for $g' \in N(g)$). 

We can easily check that the above settings satisfy the following \textit{irreducibility} and \textit{weak reversibility}: 
\begin{itemize}
\item Irreducibility: for any pair $g, g' \in \mathcal{G}_K$, there exists some sequence of transitions from $g$ to $g'$ with non-zero probability. 
\item Weak reversibility: for any pair $g, g' \in \mathcal{G}_K$ and $\tilde{F} \in \mathbb{R}$, the following two propositions (P1) and (P2) are mutually equivalent: 
\begin{itemize}
\item (P1) there exists some sequence of transitions $g_1 = g, g_2, \dots, g_p = g'$ with non-zero probability that satisfies $F(g_t) \geq \tilde{F}$ for all $t \in \{1, \dots, p\}$. 
\item (P2) there exists some sequence of transitions $g_1 = g', g_2, \dots, g_p = g$ with non-zero probability that satisfies $F(g_t) \geq \tilde{F}$ for all $t \in \{1, \dots, p\}$. 
\end{itemize}
\end{itemize}
We define that a state $g$ is \textit{locally optimal} if there is no state $g' \in \mathcal{G}_K$ that satisfies the following two conditions simultaneously: $F(g') > F(g)$, and there exists some sequence of transitions $g_1 = g, g_2, \dots, g_p = g'$ with non-zero probability that satisfies $F(g_t) \geq F(g)$ for all $t \in \{1, \dots, p\}$. For a locally but not globally optimal solution $g$, we define its \textit{depth} as the minimum $r$ that satisfies the following condition: there exists some $g'$ such that $F(g') > F(g)$ and there exists some sequence of transitions $g_1 = g, g_2, \dots, g_p = g'$ with non-zero probability that satisfies $F(g_t) \geq F(g) - r$ for all $t \in \{1, \dots, p\}$. 
By setting the sequence of temperatures at $T_t = [\max_{g \in \mathcal{G}_K} F(g) - \min_{g \in \mathcal{G}_K} F(g)] / \log (t+2)$ for all $t \geq 0$ \cite{Hajek1988}, for example, the proposed Algorithm \ref{algo:max_PL_SA} also satisfies the following conditions: 
\begin{itemize}
\item $T_t \geq T_{t + 1}$ holds for all $t \geq 0$, and $\lim_{t \to \infty} T_t = 0$. 
\item $\sum_{t=0}^{\infty} \exp \left( -\frac{r^*}{T_t} \right) = +\infty$, where $r^*$ is the maximum depth of all the locally but not globally optimal solutions. 
\end{itemize}

It has been proven that under the above conditions, the probability that an SA algorithm outputs the global optimal solution converges to one in the limit of $t \to \infty$ \cite{Hajek1988}. 

\begin{algorithm}
\caption{A naive SA algorithm for finding the maximum profile likelihood solution $\hat{g}$. }         
\label{algo:max_PL_SA}
\begin{algorithmic}[1]
\REQUIRE A cooling schedule of temperature $\{ T_t \}_{t=0}^{\infty}$ and a threshold $\epsilon^{\mathrm{SA}}$. 
\ENSURE Approximated optimal bicluster assignment $\hat{g}$. 
\STATE $t \gets 0$. 
\STATE Randomly generate an initial bicluster assignment $\hat{g}$, which is disjoint but not necessarily bi-disjoint (as shown in Figure \ref{fig:submatrix} (b)). 
\WHILE{$T_t \geq \epsilon^{\mathrm{SA}}$}
\STATE Set $\tilde{g} \gets \hat{g}$ and randomly choose an index $k_0$ from the uniform distribution on $\{ 1, \dots, 2K \}$. 
  \IF{$k_0 \leq K$}
    \STATE Set bicluster index $k \gets k_0$. 
    \STATE Let $I_k$ and $J_k = \{ j_1, \dots, j_{|J_k|} \}$, respectively, be the sets of row and column indices in the $k$th bicluster. We define \textit{add} and \textit{remove lists} as follows. 
    \STATE For $i \in I_k$, let $\mathcal{I}^{\mathrm{rem}}_{ki}$ be the set of entries in the $i$th row of the $k$th bicluster (i.e., $\{ (i, j_1), (i, j_2), \dots, (i, j_{|J_k|}) \}$). We define the \textit{remove list} as $\mathcal{I}^{\mathrm{rem}}_k = \{ \mathcal{I}^{\mathrm{rem}}_{ki} \}_{i \in I_k}$. 
    \STATE Let $\bar{I}_k$ be the set of row indices $i$ that satisfies $\bigcap_{s = 1}^{|J_k|} \bigcap_{k' = 1}^K [(i, j_s) \notin \mathcal{I}_{k'}]$. For $i \in \bar{I}_k$, let $\mathcal{I}^{\mathrm{add}}_{ki}$ be the set of entries $\{ (i, j_1), (i, j_2), \dots, (i, j_{|J_k|}) \}$. We define the \textit{add list} as $\mathcal{I}^{\mathrm{add}}_k = \{ \mathcal{I}^{\mathrm{add}}_{ki} \}_{i \in \bar{I}_k}$. 
    \STATE Let $\mathcal{I}_0$ be the set of background entries in $\tilde{g}$. Set $y^{\mathrm{add}} \gets (|\bar{I}_k| \geq 2) \cup [(|\bar{I}_k| = 1) \cap (\mathcal{I}_0 \neq \mathcal{I}^{\mathrm{add}}_k)]$, which is a flag of whether or not we can execute ``add'' operation. This guarantees that the set of background entries is not null. 
    \IF{$|I_k| \geq 2$ and $y^{\mathrm{add}} = \mathrm{True}$}
      \STATE Randomly choose $i$ from the uniform distribution on $\{ 1, \dots, |I_k| + |\bar{I}_k| \}$. If $i \leq |I_k|$, remove $\mathcal{I}^{\mathrm{rem}}_{ki}$ from the $k$th bicluster and add it to the background in $\tilde{g}$. If $i > |I_k|$, remove $\mathcal{I}^{\mathrm{add}}_{k(i-|I_k|)}$ from the background and add it to the $k$th bicluster in $\tilde{g}$. 
    \ELSIF{$|I_k| \geq 2$}
      \STATE Randomly choose $i$ from the uniform distribution on $\{ 1, \dots, |I_k| \}$. Remove $\mathcal{I}^{\mathrm{rem}}_{ki}$ from the $k$th bicluster and add it to the background in $\tilde{g}$. 
    \ELSIF{$y^{\mathrm{add}} = \mathrm{True}$}
      \STATE Randomly choose $i$ from the uniform distribution on $\{ 1, \dots, |\bar{I}_k| \}$. Remove $\mathcal{I}^{\mathrm{add}}_{ki}$ from the background and add it to the $k$th bicluster in $\tilde{g}$. 
    \ENDIF
  \ELSE
    \STATE Set the bicluster index $k \gets k_0 - K$. 
    \STATE Execute lines $7$ to $17$ by swapping the rows and columns in all the operations. 
  \ENDIF
  \IF{$F (\tilde{g}) - F (\hat{g}) > 0$}
  \STATE $\hat{g} \gets \tilde{g}$. 
  \ELSE
  \STATE With probability $\exp \left( \frac{F (\tilde{g}) - F (\hat{g})}{T_t} \right)$, $\hat{g} \gets \tilde{g}$. 
  \ENDIF
\STATE $t \gets t + 1$. 
\ENDWHILE
\end{algorithmic}
\end{algorithm}

\subsubsection{A further approximated version of SA-based submatrix localization algorithm}
\label{sec:compress_algo}

Although the naive SA algorithm in Sect.~\ref{sec:naive_sa} is tractable compared to the exhaustive search, it still requires too many steps for the algorithm to converge. Therefore, in this subsection, we propose a further approximation of Algorithm \ref{algo:max_PL_SA}. The main idea here is to first compress an observed data matrix $A$ by using row-column clustering, and then apply an SA algorithm on the compressed data matrix. 

Remark that the \textbf{null} group-wise mean matrix $P$ with $K$ biclusters has at most $2^K$ distinct rows, depending on whether or not it includes each $k$th bicluster ($k = 1, \dots, K$). Based on this fact, we first apply a clustering method (e.g., hierarchical clustering) to the rows of matrix $A$, by setting the number of clusters at $L_1 \in \mathbb{N}$, which satisfies $\min \{2^K, n\} \leq L_1 \leq n$. 
Based on a similar discussion, we also perform column clustering with number of clusters $L_2$  that satisfies $\min \{2^K, p\} \leq L_2 \leq p$. 
Then, we define the compressed observed matrix $A^{\mathrm{comp}} \in \mathbb{R}^{L_1 \times L_2}$ and matrix $M \in \mathbb{N}^{L_1 \times L_2}$ as follows: 
\begin{align}
\label{eq:A_comp_M}
&A^{\mathrm{comp}} = (A^{\mathrm{comp}}_{hh'})_{1 \leq h \leq L_1, 1 \leq h' \leq L_2}, \ \ \ \ \ 
A^{\mathrm{comp}}_{hh'} = \frac{1}{|\mathcal{I}^{\mathrm{comp}}_{hh'}|} \sum_{(i, j) \in \mathcal{I}^{\mathrm{comp}}_{hh'}} A_{ij}, \nonumber \\
&M = (M_{hh'})_{1 \leq h \leq L_1, 1 \leq h' \leq L_2}, \ \ \ \ \ 
M_{hh'} = |\mathcal{I}^{\mathrm{comp}}_{hh'}|, 
\end{align}
where $\mathcal{I}^{\mathrm{comp}}_{hh'}$ is the set of entries of matrix $A$ in the $h$th row cluster and the $h'$th column cluster. 

Next, we apply an SA algorithm to the compressed observed matrix $A^{\mathrm{comp}}$. Let $\hat{g}^{\mathrm{comp}}_{hh'} \in \{ 0, 1, \dots, K \}$ be the estimated group index of the $(h, h')$th entry of matrix $A^{\mathrm{comp}}$, and let $\mathcal{J}^{\mathrm{comp}}_k \subseteq \{ (1, 1), \dots, (L_1, L_2) \}$ be the set of entries in the $k$th estimated group ($k = 0, 1, \dots, K$) of matrix $A^{\mathrm{comp}}$. Note that we have $\mathcal{J}^{\mathrm{comp}}_k = \{ (h, h'): \hat{g}^{\mathrm{comp}}_{hh'} = k \}$. 

The key insight is that we have
\begin{align}
&\hat{p}_k = \frac{1}{np} \sum_{(h, h') \in \mathcal{J}^{\mathrm{comp}}_k} M_{hh'}, \nonumber \\
&\frac{1}{|\hat{\mathcal{I}}_k|} \sum_{(i, j) \in \hat{\mathcal{I}}_k} A_{ij} = \frac{1}{np \hat{p}_k} \sum_{(h, h') \in \mathcal{J}^{\mathrm{comp}}_k} M_{hh'} A^{\mathrm{comp}}_{hh'}. 
\end{align}
Based on the above fact, we can compute the objective function value (i.e., profile likelihood) based on the matrices $A^{\mathrm{comp}}$ and $M$, and the bicluster assignment $\hat{g}^{\mathrm{comp}} = (\hat{g}^{\mathrm{comp}}_{hh'})_{1 \leq h \leq L_1, 1 \leq h' \leq L_2}$ by 
\begin{align}
\label{eq:gen_PL_compress}
F (\hat{g}^{\mathrm{comp}}) \equiv \sum_{k = 0}^K \left( \frac{1}{np} \sum_{(h, h') \in \mathcal{J}^{\mathrm{comp}}_k} M_{hh'} \right) f \left( \frac{1}{np \hat{p}_k} \sum_{(h, h') \in \mathcal{J}^{\mathrm{comp}}_k} M_{hh'} A^{\mathrm{comp}}_{hh'} \right). 
\end{align}
From these observations, Algorithm \ref{algo:max_PL_SA2} provides an approximated solution of Algorithm \ref{algo:max_PL_SA}. 

\begin{algorithm}[t]
\caption{Approximated SA algorithm for finding the maximum profile likelihood solution $\hat{g}$. }         
\label{algo:max_PL_SA2}
\begin{algorithmic}[1]
\REQUIRE A set of row and column cluster numbers $(L_1, L_2)$ that satisfies $L_1 \geq 2^K$ and $L_2 \geq 2^K$, a cooling schedule of temperature $\{ T_t \}_{t=0}^{\infty}$ and a threshold $\epsilon^{\mathrm{SA}}$. 
\ENSURE An approximated optimal bicluster assignment $\hat{g}$. 
\STATE Apply a clustering algorithm to the rows of observed matrix $A$ with the number of clusters $L_1$. 
\STATE Apply a clustering algorithm to the columns of observed matrix $A$ with the number of clusters $L_2$. 
\STATE Let $\mathcal{I}^{\mathrm{comp}}_{hh'}$ be the set of entries of matrix $A$ in the $h$th row cluster and the $h'$th column cluster, and let $\mathcal{I}^{\mathrm{comp}} = (\mathcal{I}^{\mathrm{comp}}_{hh'})_{1 \leq h \leq L_1, 1 \leq h' \leq L_2}$.  Based on the clustering result $\mathcal{I}^{\mathrm{comp}}$, define the matrices $A^{\mathrm{comp}}$ and $M$ by (\ref{eq:A_comp_M}). 
\STATE $t \gets 0$. 
\STATE Randomly generate initial (compressed) bicluster assignment $\hat{g}^{\mathrm{comp}}$, which is disjoint but not necessarily bi-disjoint. 
\STATE Execute lines $3$ to $28$ in Algorithm \ref{algo:max_PL_SA} by replacing $A$ and $\hat{g}$ with $A^{\mathrm{comp}}$ and $\hat{g}^{\mathrm{comp}}$, respectively. As for the objective function value, we can compute it by using (\ref{eq:gen_PL_compress}). 
\STATE Convert the set of results $\mathcal{I}^{\mathrm{comp}}$ and $\hat{g}^{\mathrm{comp}}$ into the bicluster assignment $\hat{g}$ of the original observed matrix $A$. 
\end{algorithmic}
\end{algorithm}

\end{appendices}


\clearpage
\bibliographystyle{abbrv}
\bibliography{template}

\end{document}